\documentclass[11pt,a4paper]{article}

\usepackage[utf8]{inputenc}   %
\usepackage[T1]{fontenc}      %

\usepackage{amsmath}
\usepackage{amsthm}
\usepackage{amsfonts}
\usepackage{geometry}
\usepackage{authblk}
\usepackage{fancyhdr}
\usepackage{graphicx} %
\usepackage{caption}

\usepackage{subcaption} %
\usepackage{xurl} %

\newtheorem*{thm*}{Theorem}
\newtheorem{prp}{Proposition}[section]

\newtheorem{lmm}[prp]{Lemma}

\usepackage{todonotes}

\def \FF {{\mathcal F}}

\def \Id {\mbox{Id}}

\usepackage{booktabs}
\usepackage{multirow}
\usepackage{longtable}
\usepackage{needspace}
\usepackage{array}

\DeclareMathOperator*{\argmax}{arg\,max}
\DeclareMathOperator*{\argmin}{arg\,min}
\DeclareMathOperator*{\shiv}{Sv}
\DeclareMathOperator*{\torad}{TRD}
\DeclareMathOperator{\score}{\mathsf{score}}

\theoremstyle{plain}
\newtheorem{theorem}{Theorem}

\newtheorem*{corollary*}{Corollary}

\theoremstyle{definition}
\newtheorem{definition}[theorem]{Definition}
\newtheorem*{definition*}{Definition}

\newtheorem*{hypothesis*}{Hypothesis}

\usepackage{hyperref}
\hypersetup{
    colorlinks=true,
    linkcolor=blue, %
    filecolor=magenta,      
    urlcolor=cyan,
    citecolor=black, %
    pdfpagemode=FullScreen, %
    }

\begin{document}

\title{Network node immunization: improving Netshield algorithm through random rooted forests}

\author[1]{Luca Avena}
\author[2]{Alexandre Gaudillière}
\author[3]{Irina Gurewitsch}
\author[4]{Adoré Randriamandroso}
\author[5]{Alessio Troiani}
\affil[1]{Department of Mathematics and Computer Science ``Ulisse Dini'',
  University of Florence, Florence, Italy.
  \textit{Corresponding author}: \texttt{luca.avena@unifi.it}}
\affil[2]{Aix Marseille Univ, CNRS, I2M, Marseille, France}
\affil[3]{\texttt{irina\_gur@web.de}}
\affil[4]{University of Fianarantsoa, Madagascar;
  Centre de Recherche sur l'Enseignement des Mathématiques (CREM) Laboratory}
\affil[5]{Department of Mathematics and Computer Science,
  University of Perugia, Perugia, Italy}

\date{June 2nd, 2026}

\maketitle

\begin{abstract}
We are interested in the so-called multiple-node immunization problem for complex networks under attack by a viral agent. It consists in identifying and removing a set of nodes of size $k$ in a graph to maximize the impeding of virus spread. A few approaches have been proposed in the literature based on numerical and theoretical insights on how classical models for virus spread evolve on graphs.
Based on the analysis of these models, the maximal eigenvalue of the adjacency matrix of the graph has become a classical measure of how resilient the network is. Thus, a clear, well-explored approach for multiple-node immunization consists of identifying a set of $k$ nodes in such a way that the reduced network, obtained by removing these nodes, has a minimal largest eigenvalue. This spectral optimization problem turns out to be a computationally hard problem for which only greedy algorithms offer good solutions at efficient computational time. 
Among those, the so-called Netshield algorithm
represents one of the reference choices. The latter is, in fact, a clearly defined algorithm aiming at optimizing a certain sub-modular functional, called shield-value, which approximates the original optimization problem.  
We propose here a novel procedure, based on random walk kernels and related random spanning forests, to build a new algorithm, referred to as K-shield, which enhances Netshield searching performance at the same computational complexity. We give theoretical insights behind this novel method, which could also be used for other optimization problems, and then test it via numerical showcase experiments on various benchmarks.
\end{abstract}



\section{Introduction: multiple-node immunization, SIS \& maximal eigendrop}\label{sec:intro}
In the last decades, we have been witnessing an enormous growth of the scientific literature concerning problems of information spreading or diffusion in huge heterogeneous network models. These types of problems emerge naturally in a variety of applied fields such as epidemiology, social-studies, cyber-security, engineering, information theory, and neuroscience, see e.g. \cite{BBV08,DM10,KMS17,P-SCvMV15}.
Within this broad framework, we are interested in the so-called {\em multiple-node immunization} problem which can be described as follows.\\
\noindent{\bf Multiple-node immunization.}
Consider a big finite graph on $n$ nodes (where nodes might be thought of as e.g. individuals in a given population, or cities in a certain country, or computer hardware's linked to each others), assume that a virus or a malicious rumor is spreading within the graph, for a given $k\leq n$, the goal is in identifying the {\em best subset of $k$ nodes} to be removed from the network to stop the viral agent invading the entire vertex set. Despite this multiple-node immunization problem is simple to describe and is clearly of great interest for many real-world problems, it is far from obvious even to agree on what the right concept of the best subset of $ k$-nodes to be identified and immunized is.
In fact, any reasonable proposal would depend on how one describes the evolution of the contagion and 
on the underlying geometrical structure.
As such, this field can still be considered in active development, and we are far from having a clear understanding of the best strategy to implement, as well as efficient and scalable computational methods.
Yet, given the relevance of the problem, several approaches have emerged from quite different standpoints, either driven by real-world data or by the analysis of stylized mathematical models of information spread

\noindent{\bf SIS on graphs.} A key celebrated mathematical model for virus spread in a graph, which has gained popularity both in pure and applied fields, is the so-called {\em SIS dynamics} or {\em contact process}\footnote{SIS stands for Susceptible-Infected-Susceptible as this is the type of evolution a given "individual/node can experience". While the contact process is the classical name within the realm of so-called interacting particle systems as studied in probability theory, see e.g.\cite{liggett1985interacting}}.
From a microscopic perspective, this SIS dynamics on a graph\footnote{We will restrict to undirected graphs without multiple edges and self-loops.} $G=(V, E)$ can be described as the continuous time Markov chain on the state space $\{0,1\}^V$ (where the two possible states $0$ and $1$ of a node stand for susceptible and infected, respectively), where the allowed transitions are only those from a given configuration $\eta\in \{0,1\}^V$  to a new one $\eta^x \in \{0,1\}^V$, for $x\in V$, which coincides with the vector $\eta$ everywhere except at $x$ (i.e. $\eta^x(x)=1-\eta(x)$ and $\eta^x(y)=\eta(y)$ for any $y\in V\setminus\{x\}$). These allowed transitions occur at rate
\begin{equation}\label{CPrates}
    q(\eta,\eta^x) := \delta \eta(x) + 
(1-\eta(x)) \left[ \displaystyle\sum_{y: (y,x)\in E} \eta(y)\beta(y,x) \right],
\end{equation}
where $\delta>0$ is a positive 
parameter which represents a {\em healing rate}, while $\{\beta(x,y)\in(0,\infty): (x,y)\in E\}$ is a given collection of positive {\em infections rates}, which we assume to be symmetric, $\beta(x,y)=\beta(y,x)$. In words, the above dynamics is saying that an infected node $x$ transmits the virus to a neighboring node $y$  at rate $\beta(x,y)$, while recovery happens spontaneously across individuals at homogeneous\footnote{We could have allowed for non-homogeneous healing rates $(\delta_x)_{x\in V}\in (0,\infty)^V$; the homogeneous choice is just for simplicity of exposition and to give a neat interpretation to the data sets on which we run the experiments.} rate $\delta$.

\noindent{\bf Largest eigenvalue \& eigendrop maximization.}
The Markov evolution of the SIS dynamics defined in~\eqref{CPrates} depends on the initial distribution of the infected nodes. Yet, as the graph is finite and the configuration with all zeros (i.e., no infections present) represents the only absorbing state in $\{0,1\}^V$, it follows from basic theory that for our Markov process, regardless of the starting configuration and the graph structure, in finite time the infections will disappear as the system will reach the absorbing state. However, from a practical perspective, one would wish to control the time the system takes to heal completely as a function of the topology, the size of the network, and the infection parameters. In this direction, it is straightforward to obtain quantitative bounds accounting for the graph structure encoded in the adjacency matrix. Indeed, given the graph $G=(V, E)$  where the SIS is defined, and the collections of infection rates associated to $E$ as in \eqref{CPrates}, we call $A=A(\beta)$ the symmetric $|V|\times|V|$ weighted adjacency matrix with entries 
$a_{x,y}:=\beta(x,y)\mathrm{1}_{\{(x,y)\in E\}}$
and\footnote{Adding non-trivial diagonal terms to the adjacency matrix is equivalent to considering a non-homogeneous healing rate $\delta$.} $a_{x, x} = 0$.
Let us further denote by
$\vec{\rho}(t):=(\rho_x(t))_{x\in V}$ the vector of the means of the infected sites 
$\rho_x(t):=\mathbb{E}_{\mu}[\eta_t(x)]$ at arbitrary time $t$ when starting from an arbitrary distribution $\mu$ on $\{0,1\}^V$. It is standard to show\footnote{Via a simple domination from above of the contact process by a similar process which would allow for more than one particle per site with transitions almost as in \eqref{CPrates} except that in the first term, responsible for the infection growth, $(1-\eta(x))$ gets replaced by $1$, see e.g. \cite{P-SCvMV15}} that 
$\vec{\rho}(t)\leq \vec{\rho}(0)e^{(A-  \delta Id)t}$. Thus, if we denote by $\lambda_{\rm{max}}$ the largest eigenvalue of $A$, for any node $x\in V$, 
if $\lambda_{\rm{max}}<\delta$, $\rho_x(t)$ decays exponentially fast to $0$, as $t$ grows. This general observation on the SIS dynamics on graphs suggests that one can take $\lambda_{\rm{max}}$ as a measure of how vulnerable the finite network is to the viral agent spread (the larger $\lambda_{\rm{max}}$ is, the more vulnerable the whole
graph) and in fact this has become one of the clear and standard measures of network resilience in the literature, see \cite{Chack08, WvM08, MKE17, HvHP26}. It is worth stressing that there are other ways to study epidemiological models and data on graphs and hence various other measures of 
network vulnerability have been proposed and explored in the literature, see \cite{CP08, Pellis10, Ma16, vM17, TCTE-R21, H24, HvHP26}.
Getting back to the multiple-node immunization problem above, if one takes this $\lambda_{\rm{max}}$ view based on the contact process in measuring the network vulnerability, it is thus natural to face and phrase the starting $k$-node immunization problem in terms of the following optimization problem. Given a subset $S\subset V$ of cardinality $k\leq |V|$, denote by $\lambda(S)$ the largest eigenvalue of the weighted adjacency matrix as above, but associated to the graph obtained from the original one $G$ by removing the nodes in $S$ and all edges incident to these nodes, and call {\em eigendrop associated to } $S\subset V$
the difference $\lambda_{\rm{max}}-\lambda(S)$. The goal is to find among all subsets $S$ of nodes of cardinality $k$ those minimizing $\lambda(S)$, or equivalently, maximizing the eigendrop. In other words, one looks for 
\begin{equation}\label{SeigenGoal} S_*= \argmin_{S\subset V: |S|=k}\lambda(S) = 
\argmax_{S\subset V: |S|=k}\left[\lambda_{\rm{max}}-\lambda(S) \right].
\end{equation}
This standpoint on epidemiological models and a clear optimization perspective for multiple-node immunization were adopted by \cite{Chen16}, among others, and will also be our starting point. In fact, our main contribution of this paper will be to offer a new randomized algorithm which improves the performance of the clear algorithmic procedures proposed and analyzed in \cite{Chen16}.

 \noindent{\bf Paper content structure \& our contribution}
The rest of the paper is organized as follows. 
In Section~\ref{sec:preliminaries}, we briefly recall the analysis of the eigendrop optimization problem through the so-called \emph{shield value functional} presented by Chen et al. in \cite{Chen16}, and the related deterministic greedy Netshield algorithm they propose, which serves as our starting point.
Section \ref{sec:Kforest} is devoted to introducing what we will refer to as the \emph{Kirchhoff forest}, and its key properties in relation to immunization and sampling, as derived in \cite{AG18, ACGM18}. The latter is a certain random spanning rooted forest associated with an arbitrary graph, and it represents the main random combinatorial object on which we will build to propose a novel algorithm to improve the search and exceed the performance of Netshield.
This novel algorithm, which we will call \emph{K-shield algorithm}, represents the main contribution of the present paper; it can be seen as an enrichment of Netshield by adding a proper randomized search through the Kirchhoff forest, without altering the computational cost of the algorithm. This new \emph{K-shield algorithm} is presented in Section \ref{sec:OurAlgorithm}, where we further discuss its advantages and, in what sense and for what datasets, one should expect improvements for the eigendrop maximization problem with respect to the use of the sole Netshield.
In Section~\ref{sec:experiments}, we present numerical results by collecting various experiments on different synthetic and real-world network benchmarks to showcase the practical performances of this novel method. The concluding Sections~\ref{sec:Netmore} and~\ref{sec:Kmore} offer a more elaborate mathematical discussion on some technical details, respectively, in relation to the properties of  Netshield and of the Kirchhoff forest. Finally, Appendix~\ref{sec:graphs_details} collects details on the data sets where we run the experiments and an exhaustive list of tables reporting all the outcomes of the performed experiments, whereas Appendix~\ref{sec:charts} provides a graphical representation of the results presented in Section~\ref{sec:experiments}.

\section{Brief recap on ShieldValue and Netshield algorithm}\label{sec:preliminaries}

The shield value $\shiv(S)$ of a set of nodes $S \subset V$, as studied in \cite{Chen16}, is defined as 
\begin{equation}
\label{ShieldValue}
\shiv(S) := \sum_{x\in S} 2\lambda_{\rm{max}} u(x)^2 -\sum_{x,y \in S} a_{x, y}u(x)u(y),
\end{equation}
where $u$ represents the (Perron-Frobenius) normalized right-eigenvector associated to $\lambda_{\rm{max}}$. As it turns out, $\shiv(S)$ gives an upper bound for the eigendrop $\lambda_{\rm{max}} - \lambda(S)$, and it can actually be modified 
by considering the \emph{adjusted shield value}: 
\begin{equation}\label{ASV}
        \widetilde\shiv(S) := {\shiv(S) - \sum_{x \in S} \lambda_{\rm{max}}u(x)^2 \over 1 - \sum_{x \in S} u(x)^2},
\end{equation}
so that, the following bounds are true for any subset $S \subset V$:

\begin{equation}\label{Shieldbounds}
    \lambda_{\rm{max}} - \lambda(S) \leq \widetilde\shiv(S) \leq \shiv(S).
\end{equation}
For completeness, a proof of this inequality is presented in Lemma \ref{clo} in Section \ref{sec:Netmore}.

\medskip\par\noindent
It is worth noticing that the computational cost of both $\shiv(S)$ and $\widetilde\shiv(S)$
is in $O(k^2)$ for $S$ of size $k$ and we recall from~\cite{Chen16}
that $\shiv$ is a \emph{submodular function}. These facts together with \eqref{Shieldbounds} 
 suggest ---see for example~\cite{krause2014submodular}---
to use a greedy algorithm
in order to build a set $S$ of given size $k$ 
with a large shield value,
which may be associated with a large eigendrop once removed. This greedy procedure is what in \cite{Chen16} is called the Netshield algorithm
which can be implemented in such a way that its complexity results in $O(m + nk)$ operations,
with $m$ being the edge number of the adjacency matrix $A$, and $n$ and $k$, respectively, the cardinality of the vertex set of the graph and of the set of nodes $S$ to be removed. 
See Table \ref{algo:Netshield} and Section \ref{sec:Netmore} for more details.

\section{Brief recap on Kirchhoff forests and their set of roots}\label{sec:Kforest}
We are now ready to introduce what we call \emph{Kirchhoff forest}, a random spanning rooted forest of a given graph whose law has nice properties which will serve us to propose the novel algorithm in Section~\ref{sec:OurAlgorithm} to improve the search performances of Netshield at the same computational cost, see Section \ref{sec:Keigen-complexity}.

Consider a graph $G=(V, E)$ with finite vertex set $V$ and set of edges $E$, and  
let $\mathcal{F}$ denote the space of rooted spanning forests of $G$. A rooted spanning forest $F \in \mathcal{F}$ of $G$ is
a collection of vertex-disjoint rooted trees spanning its vertex set $V$, where a rooted tree is seen as a collection of
directed edges pointing towards the root. That is, a rooted forest $F$ is a subset of $E$ such that:
\begin{itemize} \item each vertex has at most one outgoing edge in $F$;
\item if there exists a directed path in $F$ from vertex $x$ to vertex $y$, then no such path exists from $y$ to $x$.
\end{itemize}
The set of roots of $F$, denoted as $\rho(F)\subseteq V$, is constituted by those vertices in $F$ without an outgoing edge. See Fig.\ref{torus} for an example. 
\begin{figure}
\begin{center}
\setlength{\unitlength}{1.5cm}
\begin{picture}(6,3)
\linethickness{0.01mm}
\multiput(1,0)(1,0){5}{\line(0,1){3}}

\linethickness{0.01mm}
\multiput(1,0)(0,1){4}
{\line(1,0){4}}

\put(3,1){\circle*{0.15}}
\put(2,3){\circle*{0.15}}
\put(4,2){\circle*{0.15}}
\put(5,2){\circle*{0.15}}
\put(5,1){\circle*{0.15}}

\linethickness{0.5mm}
\put(1,0){\vector(1,0){0.55}}
\put(1.55,0){\line(1,0){0.45}}

\linethickness{0.5mm}
\put(1,2){\vector(1,0){0.55}}
\put(1.55,2){\line(1,0){0.45}}

\linethickness{0.5mm}
\put(3,3){\vector(1,0){0.55}}
\put(3.55,3){\line(1,0){0.45}}

\linethickness{0.5mm}
\put(2,0){\vector(1,0){0.55}}
\put(2.55,0){\line(1,0){0.45}}

\linethickness{0.5mm}
\put(4,0){\vector(-1,0){0.55}}
\put(3.45,0){\line(-1,0){0.45}}

\linethickness{0.5mm}
\put(4,1){\vector(-1,0){0.55}}
\put(3.45,1){\line(-1,0){0.45}}

\linethickness{0.5mm}
\put(5,3){\vector(-1,0){0.55}}
\put(4.45,3){\line(-1,0){0.45}}

\linethickness{0.5mm}
\put(1,1){\vector(0,1){0.55}}
\put(1,1.55){\line(0,1){0.45}}

\linethickness{0.5mm}
\put(5,0){\vector(0,1){0.55}}
\put(5,0.55){\line(0,1){0.45}}

\linethickness{0.5mm}
\put(2,2){\vector(0,1){0.55}}
\put(2,2.55){\line(0,1){0.45}}

\linethickness{0.5mm}
\put(3,0){\vector(0,1){0.55}}
\put(3,0.55){\line(0,1){0.45}}

\linethickness{0.5mm}
\put(2,1){\vector(0,-1){0.55}}
\put(2,0.45){\line(0,-1){0.45}}

\linethickness{0.5mm}
\put(1,3){\vector(0,-1){0.55}}
\put(1,2.45){\line(0,-1){0.45}}

\linethickness{0.5mm}
\put(4,3){\vector(0,-1){0.55}}
\put(4,2.45){\line(0,-1){0.45}}

\linethickness{0.5mm}
\put(3,2){\vector(0,-1){0.55}}
\put(3,1.45){\line(0,-1){0.45}}
\end{picture}

\end{center}
\caption{An example of a spanning rooted forest $F$ with 5 roots
on a $2$-dimensional $5 \times 4$ square-grid graph. The forest $F$ is the subgraph depicted in boldface and is composed of 5 rooted trees. Edge orientation is represented by oriented arrows. Note that edges in each tree are directed towards the corresponding root (black circle) and that a single disconnected root is considered a degenerate tree.
}
\label{torus}
\end{figure}
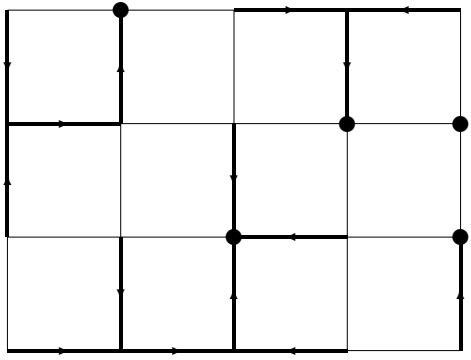

\begin{definition}[\bf{Kirchhoff forest of intensity $q$}]\label{KforestDef}. Given a finite graph $G=(V,E)$ and a weight function $w: E\rightarrow [0,\infty)$.
Fix a positive parameter $q > 0$ and let $\Phi_q$ be the random
variable with values in $\FF$ with law:
$$\mathbb{P}(\Phi_q = F) = \frac{q^{|\rho(F)|} w(F)}
{Z(q)},\qquad F \in \FF,$$ 
where $$w(F) := \prod_{e\in F}w(e),$$ stands for the forest weight, $|\rho(F)|$ denotes the cardinality of the set of roots (or equivalently the number of trees in $F$), and $Z(q)$ is a normalizing constant. 
\end{definition}
Note from the definition that the tuning parameter $q>0$ is such that for large $q$, $\Phi_q$ tends to favor the presence of many trees (or roots), while for small $q$ few trees will typically be present in $\Phi_q$.

The random spanning forest in Def. \ref{KforestDef} is a generalization of the celebrated Uniform Spanning Tree (UST), see e.g. \cite{B16}, which is recovered in the constant-weight case (say $w(e)\equiv 1$) in the limit $q\to 0$. As such, this object shares several structural and sampling properties with the UST, on top, the presence of the tunable parameter $q>0$ makes it a flexible tool to partition the vertex set of a network into random sub-blocks and to identify subsets of well-distributed vertices, see \cite[Theorem 1]{AG18}, from the perspective of the continuous time random walk associated to the given graph. We postpone to Section \ref{sec:Kmore} some in-depth mathematical discussion on the link with this random walk and classical basic sampling, see Section \ref{forestSampling}, and on key results which shed light on why Kirchhoff forests are suited for immunization purposes, especially for heterogeneous networks, see Section \ref{trace}. 
In the remainder of this section, we discuss further some basic theoretical features that will be relevant within the immunization context for the K-shield algorithm proposed in Section \ref{sec:OurAlgorithm}.

\subsection{A determinantal sub-set $S$ of nodes for those Susceptible to be kept}\label{sec:roots}

We have seen that the Netshield algorithm aims at optimizing the Shieldvalue functional, which builds on the idea that the optimal set(/s) $S$ should have nodes with a high eingenscore and at the same time that are in some sense far apart from each other, which would correspond to the two terms, respectively, in the right-hand side of Eq. \eqref{ShieldValue}.

By analogy, we can exploit Kirchhoff forest to identify other random sets of points having similar collective features (i.e. ``high eingenscore and repelling points")
which might be good candidates to beat in terms of eingendrop, the deterministic set identified by the greedy Netshield. 
Indeed, two of the fascinating observables of the Kirchhoff forest $\Phi_q$ are represented by its set of roots
$\rho(\Phi_q)$ and its complements set 
$\rho(\Phi_q)^{c}=V\setminus\rho(\Phi_q)$. These complementary random subsets of points in $V$ have nice properties that we next describe and that can help us move out of the regions where Netshield searches.

The first interesting property is that,  see \cite[Prop. 2.2]{AG18} for a proof, for any given $q>0$, the random set of roots $\rho(\Phi_q)$ of the forest  $\Phi_q$ is a \emph{determinatal point process} with \emph{kernel}
\begin{equation}\label{kernel}
		K_{q}(x,y):=\mathbb{P}_x\bigl(X(T_q) = y\bigr), \quad x, y \in V,
\end{equation}
where $\{X(t)\}_{t\geq 0}$ is the RW (Random Walk) described below \eqref{graphLaplacian} and $T_q$ is an exponential independent random variable of rate $q$.  
This determinantal property means that the law of the random subset of nodes $\rho(\Phi_q) \subseteq V$ can be characterized by the following statistics:
	\begin{equation}\label{detRoots}
		{\mathbb P}\left( A\subseteq \rho(\Phi_q) \right)
		= {\rm det} \left[K_{q}\right]_{A},\qquad \text{ for any } A\subseteq V,
	\end{equation}
with $\left[K_{q}\right]_{A}$ being the restriction of the matrix $K_q$ in Eq.\eqref{kernel} to the set of indices in $A$.
Notice that the entries $(x,y)$ of the matrix $K_{q}$ have a neat probabilitstic interpretation, that is, Eq.\eqref{kernel}
represents the probability that the RW $X$, when starting from $x$, is at location $y$ when observed at an independent exponential time $T_q$. Determinantal point processes have appeared in many contexts in probability theory and statistical physics . Point processes of this type have an elegant algebraic closed structure that allows explicit computation of their statistics in terms of determinants of the associated restricted kernel, and, more relevant for our discussion, they have a non-trivial correlation structure. In particular, in the symmetric adjacency matrix setting that we are considering (i.e. $w(x,y)=w(y,x)$), an immediate consequence of \eqref{detRoots} is that pairs of points in $\rho(\Phi_q)$ tend to repel each other, being \emph{negatively correlated}. This can be seen by setting $A=\{x,y\}$ in Eq.~\eqref{detRoots} and noticing that this implies that
\begin{equation}\label{NegCorr}	{\mathbb P}\left( \{x,y\}\subset \rho(\Phi_q)\right)%
\leq %
 {\mathbb P}\left( x\in \rho(\Phi_q)\right){\mathbb P}\left( y\in \rho(\Phi_q)\right).\end{equation}

As it turns out in general, also the complementary set 
of a determinantal point process is itself a determinantal point process, see e.g. \cite{TBUA23}.
Thus the negative correlation in~\eqref{NegCorr} holds true also for the complementary set $\rho(\Phi_q)^{c}=V\setminus\rho(\Phi_q)$. 

If we aim at immunizing a given number of nodes in a network of individuals to impede the global spread of a viral agent, it is reasonable to try to immunize infected individuals that are not too close (in graph distance) to each other, as well as to keep far apart from each other, those individuals that are susceptible. In this respect, the fact that both $\rho(\Phi_q)$ and $\rho(\Phi_q)^{c}$ have negative correlations suggests that we may use one of these two sets for the susceptibles and the other for those to be immunized. This way, we in particular guarantee that points to be immunized should be far apart from each other (analogously to the second term in the r.h.s. of \eqref{ShieldValue}).
At this point, we may ask which of the two sets to use to identify those to be kept and those to be immunized. As in the Shieldvalue one looks for points ``with a high eingenscore" (see first term in the second term in the r.h.s. of \eqref{ShieldValue}) in building the set to be immunized. While using Kirchhoff forests, it is natural to set $\rho(\Phi_q)$ for those susceptible to be kept and $\rho(\Phi_q)^{c}$ for those to be immunized.
Indeed, the set $\rho(\Phi_q)^{c}$ tends to be concentrated on points with a high weighted degree. This is clear, for example, when looking at the probability that the complementary root-set is constituted by a single node $x\in V$, which is indeed proportional to the weighted degree of $x$, that is, if \begin{equation}\label{wDeg}
w(x):=\displaystyle\sum_{(x,y)=e\in E}w(e),
\end{equation} 
by summing the forest probabilities in \eqref{KforestDef} over all $F\in \mathcal{F}$ such that $\rho(F)=\{x\}$ , 
we obtain that ${\mathbb P}(\rho(\Phi_q)^{c}=\{x\})\propto w(x)$. These first observations inspired us to try to search for a good set to be immunized by means of $\rho(\Phi_q)^{c}$.
From a computational cost point of view, it is also much more convenient to sample a small $\rho(\Phi_q)^c$ of size $k$ with $n - k$ trees than a small $\rho(\Phi_q)$ of size $k$ with $k$ trees.
In applications, we typically seek to remove a small number $k$ of nodes. To this end we will leverage on Wilson's sampling algorithm (see Section~\ref{sec:Kmore} ) that is more efficient for large $q$ (i.e., many trees).

Yet, it still remains far from obvious how to use this random set to identify a set of cardinality $k$ that could be better than the deterministic one output by Netshield for the original maximal eigendrop problem. There are, in fact, two main issues which we have not yet discussed. First, for a given $q$, the random set $\rho(\Phi_q)^{c}$ has not a fixed cardinality while we wish a set of a given size $|V|-k$. For this first problem, it turns out that the mean size of $\rho(\Phi_q)^c$ is known and given by $|V| -\sum_{i < |V|}q/(q+\lambda_i)$, see \cite[Prop. 2.1]{AG18}, with $\lambda_i$'s denoting the eigenvalues of the graph Laplacian matrix in \eqref{graphLaplacian}. Thus, to solve this first problem, one can use this knowledge on the mean number of roots, as done in \cite[Theorem 2, Section 2.3]{AG18}, to cook up an efficient procedure to sample a forest with exactly the desired number of roots. In Figure~\ref{fig:Karate2}, we present a first qualitative plot on a specific benchmark, just to show that the distribution on the set of roots with a given cardinality tends to have peaks in correspondence to sets with good eigendrop.
In practice, we will not use this forest $\Phi_q$ conditioned to a given root number, but an approximation of this conditioned forest that is faster to sample.

\begin{figure}
    \centering
    \includegraphics[width=1\linewidth]{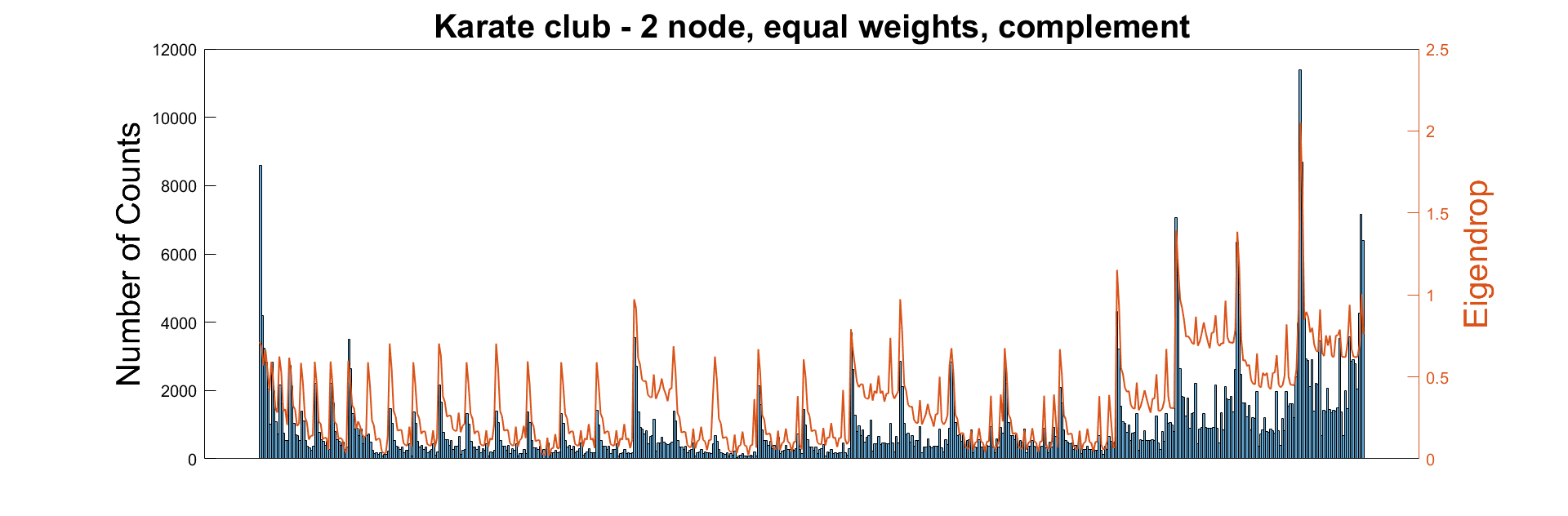}
    \caption{An empirical distribution of the pair $\rho(\Phi_q)^c$
    conditioned to $|\rho(\Phi_q)| = n - 2$
    together with its associated eigendrop
    for the 34 vertex ``karate club'' network.}\label{fig:Karate2}
\end{figure}

The second problem is that even if we can sample the set of roots with a given cardinality, it still remains a random subset of points, and thus, one should clarify how many and which samples of this special random set to use to possibly find better sets than the output of Netshield.
As discussed in Section \ref{sec:OurAlgorithm}, where we present the K-shield algorithm, to solve this other issue, we will prescribe computationally \emph{cheap} tests on \emph{a few} samples of the forest, where \emph{a few} and \emph{cheap} will be specified by keeping in mind that we do not want the algorithm to exceed the magnitude of the running time of Netshield.  
Before moving into these algorithmic considerations, let us anticipate that the idea of identifying the complement of the set of roots for the set of nodes to be immunized might indeed be a good idea to make the network less vulnerable. Indeed, as we will see in Section \ref{trace}, removing from a graph this complementary set has in mean the property that the more heterogeneous the spectrum of the considered graph, the stronger is the decay of the average infection rate.

\section{K-shield algorithm: a novel method to improve Netshield search }\label{sec:OurAlgorithm}
Our new algorithm (K-shield) complements Netshield search by providing two additional (randomly selected) sets of nodes to test for removal from the graph $\mathcal{G}$ whose eigendrop may be higher than the eigendrop associated with the set of nodes removed with Netshield. These two additional sets of nodes are determined by first sampling a collection of candidate sets through the complement of the set of roots in the Kirchhoff forest, and, among these, by selecting the \emph{best two} on the basis of the value of two proxies of the eigendrop: the \emph{adjusted shield value} defined in~\eqref{ASV}) and the \emph{total rate drop}.
The $\torad$ (Total Rate Drop) associated with a set $S$ is the sum of the jump rates along the edges incident to the nodes of set $S$:
\begin{equation}\label{fabienne}
    \begin{split}
   \torad(S)&:=\sum_x w(x) - \sum_{x, y \not\in S} w(x, y)
        = \sum_{x, y \in S} w(x, y) + 2 \sum_{x \in S, y \not\in S} w(x, y)\\
        &= \sum_{x, y \in S} w(x, y) + 2 \left[\sum_{x \in S} w(x) - \sum_{x, y \in S} w(x, y)\right]= 2 \sum_{x \in S} w(x) - \sum_{x, y \in S} w(x, y)
    \end{split}
\end{equation}

In particular, K-shield selects $\phi = \phi(n, m, k)$ sets of nodes $S_{1}, S_{2}, \dots S_{\phi}$, each with cardinality $k$, and for each of these sets computes the adjusted shield value and the total rate drop. Then the sets associated with the highest adjusted shield value $S_{s,k}$ and the highest total rate drop $S_{r,k}$ are selected, together with the set determined by Netshield $S_{N,k}$. Finally, the eigendrop associated with the removal of these three sets of nodes (the one selected by Netshield and the two additional ones selected by K-shield) is computed, and the set with the largest eigendrop is removed from $\mathcal{G}$.

Instead of sampling 
$S_1, ~ \dots, S_d$ according to
the law of $\Phi_q$ conditioned to $|\rho(\Phi_q)| = n - k$,
we use the coupled forest algorithm presented in~\cite[Theorem 2, Section 2.3]{AG18},
which samples a forest trajectory $t \geq 0 \mapsto \Phi_{1 / t}$
issued form the forest $\Phi_\infty$ made of $n$ trivial trees,
with time marginals distributed like $\Phi_q$ at all times $t = 1 / q$
and that crosses all sets
$\mathcal F_k = \left\{ F \in \mathcal F : |\rho(F)| = k \right\}$.

Since the numerical cost of sampling $\Phi_q$ along this coupled forest algorithm
further studied in \cite{BCGMQT25} is in
$$
    O\bigg(\sum_x {q + w(x) \over q}\biggr)
    = O\biggl(\sum_{j < n} {q + \lambda_j \over q}\biggr)
    = O\biggl(n\biggl(1 + {\bar\lambda \over q}\biggr)\biggr)
$$
with $\lambda_j$, $j < n$, the eigenvalues of the graph Laplacian $L$ in \eqref{graphLaplacian},  and
\begin{equation}
    \bar\lambda
    = {1 \over n} \sum_{j < n} \lambda_j
    = {1 \over n} {\rm Tr}(L),
\end{equation}
and since
$$
    E\bigl[|\rho(\Phi_q)|\bigr] = \sum_{j < n} {q \over q + \lambda_j},
$$
by stopping it when it reaches $\mathcal F_{n - k}$ we get an approximated
sample of $\Phi_q$ conditioned to $|\rho(\Phi_q)| = n - k$
at a numerical cost in $O(n)$ for
\begin{equation}\label{eq:julie}
    n - k \geq E[|\rho(\Phi_{\bar\lambda})|]
    = \sum_{j < n} {\bar\lambda \over \bar\lambda + \lambda_j}.
\end{equation}
By setting $\alpha = \max_{x \in V} w(x)$
and using $\lambda_j \leq 2\alpha$, this condition is in particular satisfied for
\begin{equation}
    k \leq \biggl(1 - {\bar\lambda \over 4\alpha}\biggr) n.
\end{equation}
Note that $\bar\lambda / \alpha$
is the ratio between the mean and the largest weighted degree in the network.

\subsection{Complexity of the K-shield algorithm}\label{sec:Keigen-complexity}

As mentioned above and discussed in Section~\ref{sec:Netmore}, the total algorithmic complexity of the Netshield algorithm is $O(m + nk)$ where $n$ is the number of nodes in the graph, $m$ is the number of unoriented edges, and $k$ is the number of nodes to remove.
The computational cost of finding the largest eigenvalue and the associated eigenvector of the adjacency matrix $A$ of some graph with $m$ edges is $O(m)$.

By assuming \eqref{eq:julie}, the cost of sampling a rooted random forest with Wilson’s algorithm, and, hence, a set of $k$ nodes to remove from $\mathcal{G}$ is $O(n)$. Further, the cost of determining the adjusted shield value associated with the graph obtained by removing from $\mathcal{G}$ the nodes in the set $S$, with $|S| = k$, is $O(k^{2})$, once the eigenvector associated with the largest eigenvalue of the adjacency matrix has been computed. Therefore, the total computational cost of determining the adjusted shield value for a graph obtained by removing $k$ nodes sampled with Wilson’s algorithm is $O(n + k^2)$. This means that sampling $\frac{m + nk}{n + k^2}$ sets $S_{1}, S_{2}, \dots S_{\frac{m +nk}{n+k^{2}}}$ using Wilson algorithm and computing, for each of these sets, the associated adjusted shield value has the same computational cost of Netshield.
Moreover, note that the total rate drop associated with the removal of $k$ nodes from $\mathcal{G}$ can be computed in $O(k^{2})$ steps (see \eqref{fabienne}) so that the total cost of the algorithm stays $O(n + k^{2})$.

\subsection{Comparison: Netshield \& K-shield algorithms}\label{sec:comparison}

While nodes with a high degree tend to be selected for removal by the K-shield algorithm,
Netshield uses a finer criterion, which is related to the Perron-Frobenius eigenvector.
However, this vector tends to be strongly localized and leads Netshield to remove nodes in the same region, whereas the K-shield algorithm follows a more global approach
with a much stronger ability to select nodes in far-apart regions.
For this reason, it performs better on graphs with community structures
where Netshield will tend to act on one community only. 
To fix  this issue, one can instead implement, at a higher numerical cost, the Neshield+ algorithm (see \cite{Chen16})
that recomputes the Perron-Frobenius eigenvector of the reduced graph after each removal
of a set of size $b \leq k$ nodes.
This leads, at higher numerical cost, to significant improvement of performances, and the same is true 
with the same kind of comparison for the analogous K-shield+ version 
of the K-shield algorithm.
However, both algorithms tend to perform better for small $b$ and lead to simple deterministic
greedy algorithms with higher computational cost in $O(km)$ for $b = 1$:
for Netshield+, one removes at each step the node with the largest shield value
and for the K-shield+ algorithm, the node with the highest adjusted shield-value
or total rate drop, depending on which of the two corresponds to the largest eigendrop.

\section{Numerical results}\label{sec:experiments}
To keep the asymptotic complexity of K-shield equal to the complexity of Netshield, 
we sampled
$\phi = 1+4\frac{2 m+n k}{n+k^2}$ rooted forests in a single run of K-shield
(note that this number is in $O(\frac{m + nk}{n + k^{2}})$).

In our experiments, we considered several benchmark graphs with a symmetric adjacency matrix that are studied in the literature, together with some “synthetic” graphs we created to evaluate our algorithm in some specific circumstances.
Appendix~\ref{sec:graphs_details} provides a description of the benchmark graphs and the details on how to obtain the graphs data.

Further, to assess how the K-shield performs on graphs with a strong community structure, we tested the algorithm on a synthetic graph (denoted by ``communities 01'' obtained by joining, with a clique of size $5$, five communities consisting of $5$ independent Erdös-Rényi graphs with, respectively, 20, 25, 28, 30, and 30 vertices and edge probability $p = 0.7$ (see Fig.~\ref{fig:community_graph}).

 \begin{figure}
     \centering
     \caption{An illustration of the toy community graph}\label{fig:community_graph}
     \includegraphics[width=0.5\linewidth]{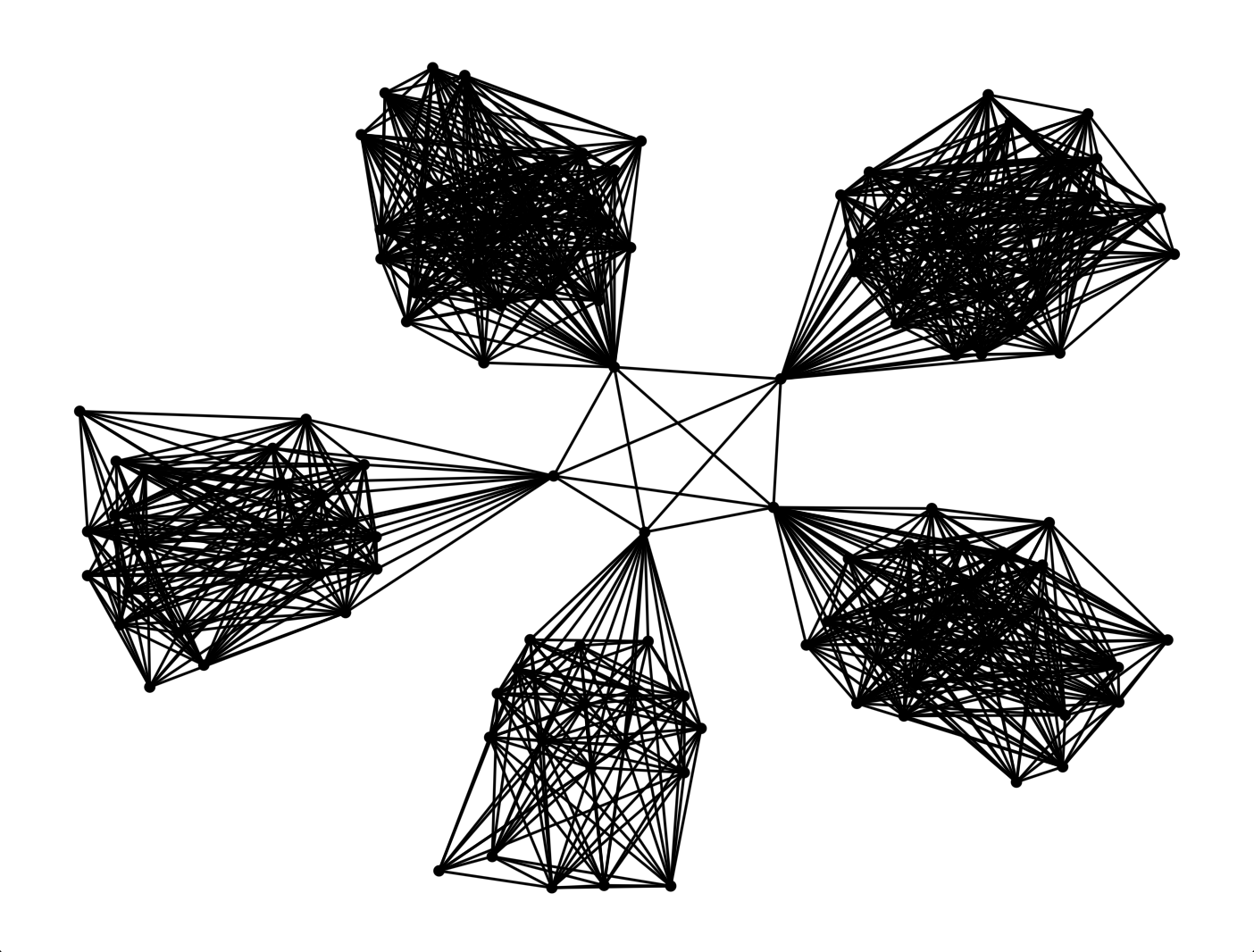}
 \end{figure}

We tested K-shield for several values of the number of ``removed nodes'': 
$k \in \{1, 2, 3, 5, N/10, N/4\}$ where $N$ is the number of nodes in the graph, and
with $N/10$ and $N/4$ rounded to the nearest integer.

To collect statistics on the performance of the K-shield, we ran it 400 times for each graph and each value of $k$. We recall that, for each graph and each value of $k$, one run requires sampling a given number of forests with Wilson's algorithm.

In Table~\ref{tab:eigendrop_long} we show the proportion of times the largest eigendrop is obtained considering the set $S_{N, k}$ (Netshield set of nodes), $S_{r, k}$, (set of nodes with highest total rate drop) or $S_{s, k}$ (set of nodes with highest adjusted shield value) and the average value of the eigendrop for each of these sets (normalized to the largest eigenvalue $\lambda$ of the adjacency matrix of the graph). Our results show that Netshield remains a valid choice for trying to immunize a network. However, it is quite common that the result of Netshield can be improved by using the K-shield. In particular, our analysis suggests that the K-shield is a viable option to obtain a more effective immunization, especially in the case of weighted graphs and graphs that present a strong community structure. 

For clarity, the same data is also provided in the form of a chart for a selection of the graphs taken into account\footnote{the charts for all the graphs taken into account are provided in the appendix (Appendix~\ref{sec:charts})}.

\setlength{\LTleft}{10pt}
\setlength{\LTright}{\fill}
\setlength{\LTcapwidth}{0.9\linewidth}
\small

\begin{longtable}{@{}lccccccc@{}}
\caption{Origin best eigendrop and normalized average eigendrop for each graph and $k$.The columns referring to the origin of best eigendrop give the proportion of times the largest eigendrop is obtained by removing the set of nodes chosen by Netshield (NS), those with highest total rate drop (TRD) and those with the highest adjusted shield value (ASV). The columns referring to the average eigendrop provide, for the same set of nodes, the average value of the eigendrop over the $1+4\frac{2 m+n k}{n+k^2}$ sampled forests. Values are normalized to the largest eigenvalue $\lambda_{\rm{max}}$ of the weighted adjacency matrix of the graph. Graphs marked with the ``(w)'' flag have weighted edges, whereas graphs marked with the ``(uw)'' flag have unweighted edges.}
\label{tab:eigendrop_long}\\
\toprule
\multirow{2}{*}{graph} & \multirow{2}{*}{$k$} & \multicolumn{3}{c}{origin best eigendrop} & \multicolumn{3}{c}{average eigendrop} \\
\cmidrule(lr){3-5}\cmidrule(lr){6-8}
& & TRD & ASV & NS & TRD & ASV & NS \\
\midrule
\endfirsthead

\toprule
\multirow{2}{*}{graph} & \multirow{2}{*}{$k$} & \multicolumn{3}{c}{origin best eigendrop} & \multicolumn{3}{c}{average eigendrop} \\
\cmidrule(lr){3-5}\cmidrule(lr){6-8}
& & TRD & ASV & NS & TRD & ASV & NS \\
\midrule
\endhead

\midrule
\multicolumn{8}{r}{continued on next page} \\
\endfoot

\bottomrule
\endlastfoot

\multirow{5}{*}{paper (uw)} & 1 & 0.003 & 0.728 & \textbf{1.000} & 0.029 & 0.104 & \textbf{0.131} \\*
 & 2 & 0.170 & 0.698 & \textbf{1.000} & 0.079 & 0.115 & \textbf{0.131} \\*
 & 3 & 0.790 & \textbf{0.877} & 0.595 & 0.163 & \textbf{0.167} & 0.131 \\*
 & 4 & 0.158 & 0.158 & \textbf{1.000} & 0.379 & 0.344 & \textbf{1.000} \\*
 & 5 & 0.378 & 0.378 & \textbf{1.000} & 0.593 & 0.565 & \textbf{1.000} \\

\midrule
\multirow{5}{*}{conference 1 (w)} & 1 & 0.980 & 0.980 & \textbf{1.000} & 0.206 & 0.207 & \textbf{0.209} \\*
 & 2 & 0.445 & 0.360 & \textbf{0.907} & 0.205 & 0.213 & \textbf{0.214} \\*
 & 3 & 0.450 & 0.258 & \textbf{0.562} & \textbf{0.226} & 0.221 & 0.214 \\*
 & 5 & \textbf{0.562} & 0.268 & 0.350 & \textbf{0.291} & 0.250 & 0.215 \\*
 & 25 & \textbf{0.522} & 0.258 & 0.412 & \textbf{0.639} & 0.560 & 0.631 \\

\midrule
\multirow{5}{*}{conference 2 (w)} & 1 & 0.985 & 0.985 & \textbf{1.000} & 0.490 & 0.490 & \textbf{0.490} \\*
 & 2 & \textbf{0.657} & 0.375 & 0.168 & 0.474 & \textbf{0.500} & 0.492 \\*
 & 3 & \textbf{0.627} & 0.468 & 0.015 & 0.485 & \textbf{0.505} & 0.492 \\*
 & 5 & 0.100 & 0.048 & \textbf{0.873} & 0.495 & 0.522 & \textbf{0.564} \\*
 & 26 & 0.203 & 0.043 & \textbf{0.780} & 0.673 & 0.657 & \textbf{0.726} \\

\midrule
\multirow{5}{*}{conference 3 (w)} & 1 & \textbf{1.000} & \textbf{1.000} & \textbf{1.000} & \textbf{0.715} & \textbf{0.715} & \textbf{0.715} \\*
 & 2 & \textbf{0.975} & 0.175 & 0.000 & \textbf{0.798} & 0.732 & 0.715 \\*
 & 3 & \textbf{0.978} & 0.037 & 0.000 & \textbf{0.814} & 0.730 & 0.715 \\*
 & 5 & \textbf{0.760} & 0.035 & 0.207 & \textbf{0.825} & 0.746 & 0.818 \\*
 & 24 & \textbf{0.757} & 0.107 & 0.155 & \textbf{0.897} & 0.843 & 0.885 \\

\midrule
\multirow{5}{*}{airport 1 (uw)} & 1 & 0.500 & 0.980 & \textbf{1.000} & 0.033 & 0.033 & \textbf{0.033} \\*
 & 2 & 0.175 & 0.175 & \textbf{1.000} & 0.065 & 0.065 & \textbf{0.066} \\*
 & 3 & 0.033 & 0.033 & \textbf{1.000} & 0.095 & 0.095 & \textbf{0.098} \\*
 & 5 & 0.000 & 0.000 & \textbf{1.000} & 0.150 & 0.150 & \textbf{0.163} \\*
 & 12 & 0.000 & 0.000 & \textbf{1.000} & 0.321 & 0.321 & \textbf{0.388} \\

\midrule
\multirow{5}{*}{airport 2 (w)} & 1 & 0.055 & 0.943 & \textbf{1.000} & 0.091 & 0.116 & \textbf{0.118} \\*
 & 2 & 0.245 & 0.245 & \textbf{1.000} & 0.167 & 0.177 & \textbf{0.210} \\*
 & 3 & 0.033 & 0.033 & \textbf{1.000} & 0.219 & 0.225 & \textbf{0.299} \\*
 & 5 & 0.000 & 0.000 & \textbf{1.000} & 0.290 & 0.299 & \textbf{0.422} \\*
 & 125 & 0.005 & 0.005 & \textbf{0.998} & 0.982 & 0.977 & \textbf{0.994} \\

\midrule
\multirow{6}{*}{airport 3 (uw)} & 1 & 0.080 & 0.443 & \textbf{1.000} & 0.013 & 0.014 & \textbf{0.014} \\*
 & 2 & 0.010 & 0.010 & \textbf{1.000} & 0.023 & 0.024 & \textbf{0.027} \\*
 & 3 & 0.000 & 0.000 & \textbf{1.000} & 0.031 & 0.033 & \textbf{0.041} \\*
 & 5 & 0.000 & 0.000 & \textbf{1.000} & 0.046 & 0.048 & \textbf{0.067} \\*
 & 186 & 0.000 & 0.000 & \textbf{1.000} & 0.643 & 0.651 & \textbf{0.766} \\*
 & 464 & 0.515 & \textbf{0.647} & 0.195 & 0.886 & \textbf{0.888} & 0.882 \\

\midrule
\multirow{6}{*}{airport 4 (w)} & 1 & 0.245 & 0.320 & \textbf{0.728} & 0.019 & 0.021 & \textbf{0.025} \\*
 & 2 & 0.000 & 0.000 & \textbf{1.000} & 0.028 & 0.031 & \textbf{0.056} \\*
 & 3 & 0.000 & 0.000 & \textbf{1.000} & 0.036 & 0.039 & \textbf{0.086} \\*
 & 5 & 0.000 & 0.000 & \textbf{1.000} & 0.051 & 0.055 & \textbf{0.140} \\*
 & 798 & 0.000 & 0.000 & \textbf{1.000} & 0.869 & 0.867 & \textbf{0.946} \\*
 & 1994 & \textbf{0.568} & 0.388 & 0.193 & \textbf{0.963} & 0.961 & 0.962 \\

\midrule
\multirow{6}{*}{rfid (w)} & 1 & 0.998 & 0.998 & \textbf{1.000} & 0.256 & 0.256 & \textbf{0.256} \\*
 & 2 & \textbf{0.958} & 0.172 & 0.155 & \textbf{0.311} & 0.306 & 0.308 \\*
 & 3 & \textbf{0.652} & 0.210 & 0.440 & \textbf{0.355} & 0.331 & 0.340 \\*
 & 5 & \textbf{0.890} & 0.388 & 0.072 & \textbf{0.441} & 0.395 & 0.349 \\*
 & 8 & \textbf{0.542} & 0.182 & 0.448 & \textbf{0.532} & 0.465 & 0.529 \\*
 & 19 & 0.018 & 0.015 & \textbf{0.983} & 0.750 & 0.713 & \textbf{0.836} \\

\midrule
\multirow{5}{*}{UKfaculty (w)} & 1 & 0.958 & 0.958 & \textbf{1.000} & 0.191 & 0.191 & \textbf{0.194} \\*
 & 2 & \textbf{0.825} & 0.455 & 0.347 & \textbf{0.213} & 0.206 & 0.207 \\*
 & 3 & \textbf{0.755} & 0.343 & 0.203 & \textbf{0.232} & 0.221 & 0.209 \\*
 & 5 & \textbf{0.863} & 0.480 & 0.020 & \textbf{0.272} & 0.252 & 0.210 \\*
 & 8 & \textbf{0.858} & 0.390 & 0.000 & \textbf{0.320} & 0.287 & 0.218 \\

\midrule
\multirow{6}{*}{enron (w)} & 1 & 0.998 & 0.998 & \textbf{1.000} & 0.389 & 0.389 & \textbf{0.389} \\*
 & 2 & \textbf{0.828} & 0.110 & 0.155 & \textbf{0.494} & 0.389 & 0.389 \\*
 & 3 & \textbf{0.848} & 0.102 & 0.102 & \textbf{0.540} & 0.413 & 0.389 \\*
 & 5 & \textbf{0.885} & 0.130 & 0.025 & \textbf{0.605} & 0.438 & 0.389 \\*
 & 18 & \textbf{0.640} & 0.220 & 0.200 & \textbf{0.755} & 0.647 & 0.741 \\*
 & 46 & \textbf{0.600} & 0.177 & 0.295 & \textbf{0.860} & 0.814 & 0.855 \\

\midrule
\multirow{6}{*}{communities 01 (uw)} & 1 & 0.090 & 0.487 & \textbf{0.912} & 0.003 & 0.026 & \textbf{0.026} \\*
 & 2 & \textbf{0.507} & 0.087 & 0.435 & 0.021 & \textbf{0.027} & 0.027 \\*
 & 3 & \textbf{0.662} & 0.165 & 0.207 & \textbf{0.033} & 0.030 & 0.027 \\*
 & 5 & \textbf{0.730} & 0.287 & 0.030 & \textbf{0.053} & 0.041 & 0.027 \\*
 & 13 & \textbf{0.845} & 0.195 & 0.000 & \textbf{0.119} & 0.084 & 0.027 \\*
 & 33 & \textbf{0.902} & 0.172 & 0.000 & \textbf{0.272} & 0.215 & 0.075 \\

\end{longtable}
\normalsize

\begin{figure}
    \centering
    \caption{Graph: ``enron'' (weighted)$ $}
    \begin{subcaptionblock}{0.45\textwidth}
        \includegraphics[width=\textwidth]{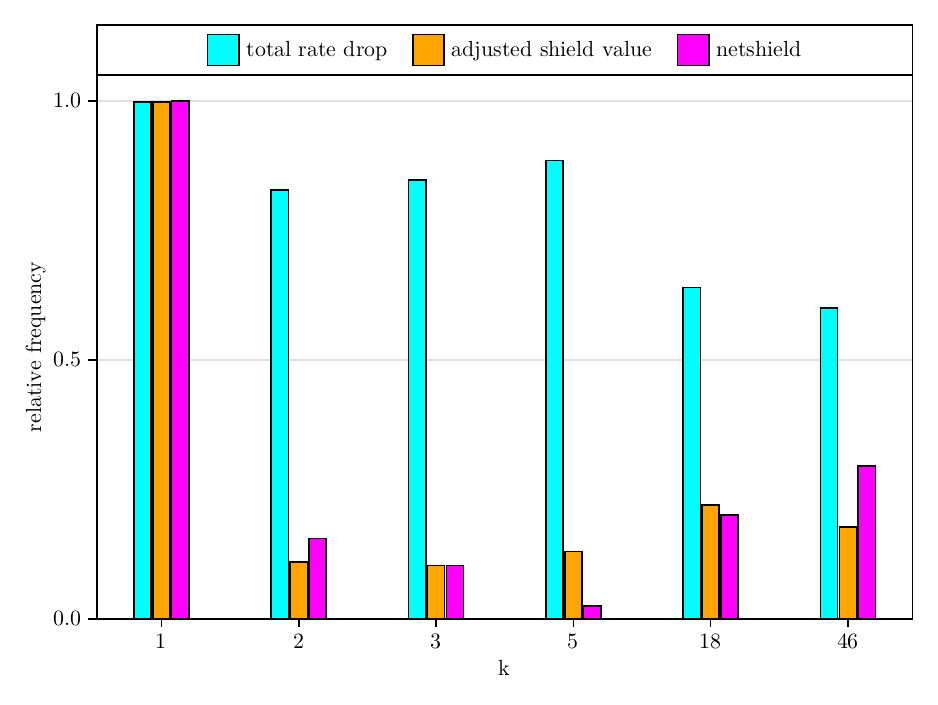}
        \caption{Origin of best eigendrop}
    \end{subcaptionblock}
    \begin{subcaptionblock}{0.45\textwidth}
        \includegraphics[width=\textwidth]{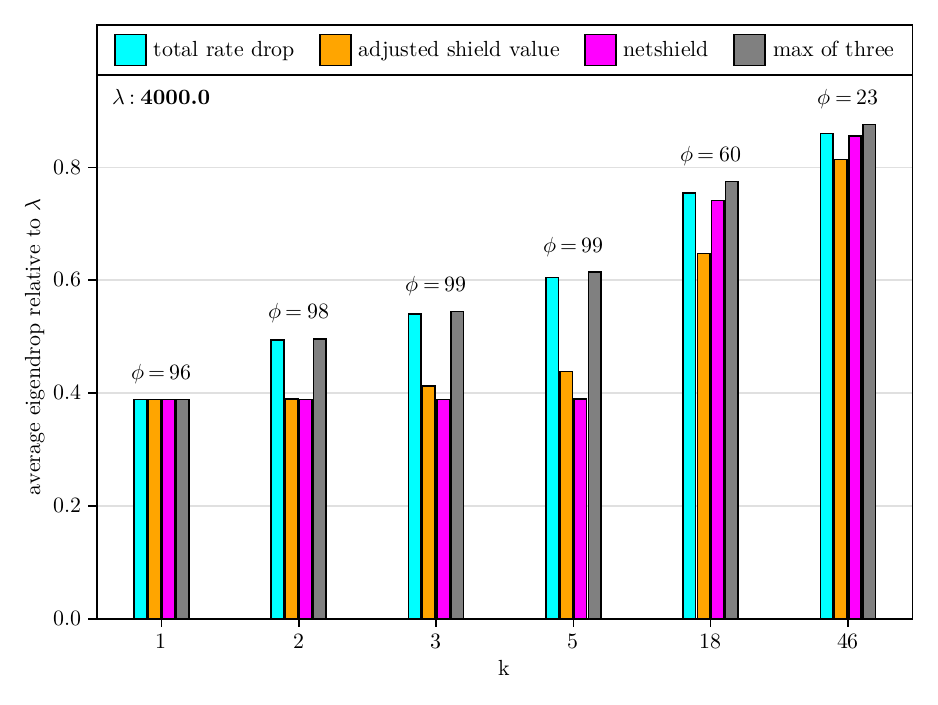}
        \caption{Average eigendrop}
    \end{subcaptionblock}
\end{figure}

\begin{figure}
    \centering
    \caption{Graph: ``communities 01'' (non-weighted)$ $}
    \begin{subcaptionblock}{0.45\textwidth}
        \includegraphics[width=\textwidth]{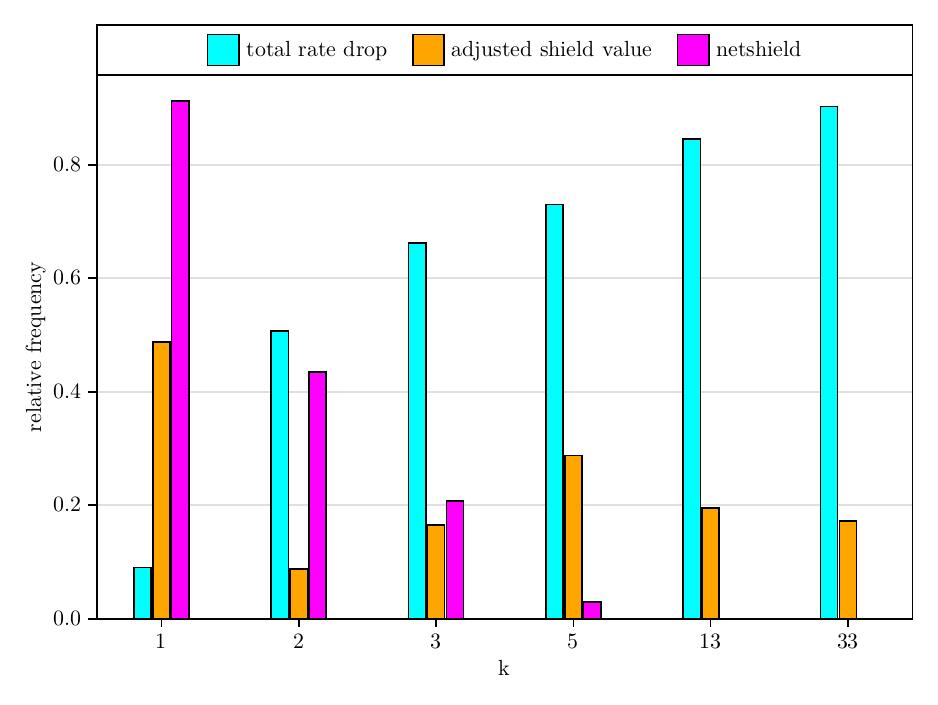}
        \caption{Origin of best eigendrop}
    \end{subcaptionblock}
    \begin{subcaptionblock}{0.45\textwidth}
        \includegraphics[width=\textwidth]{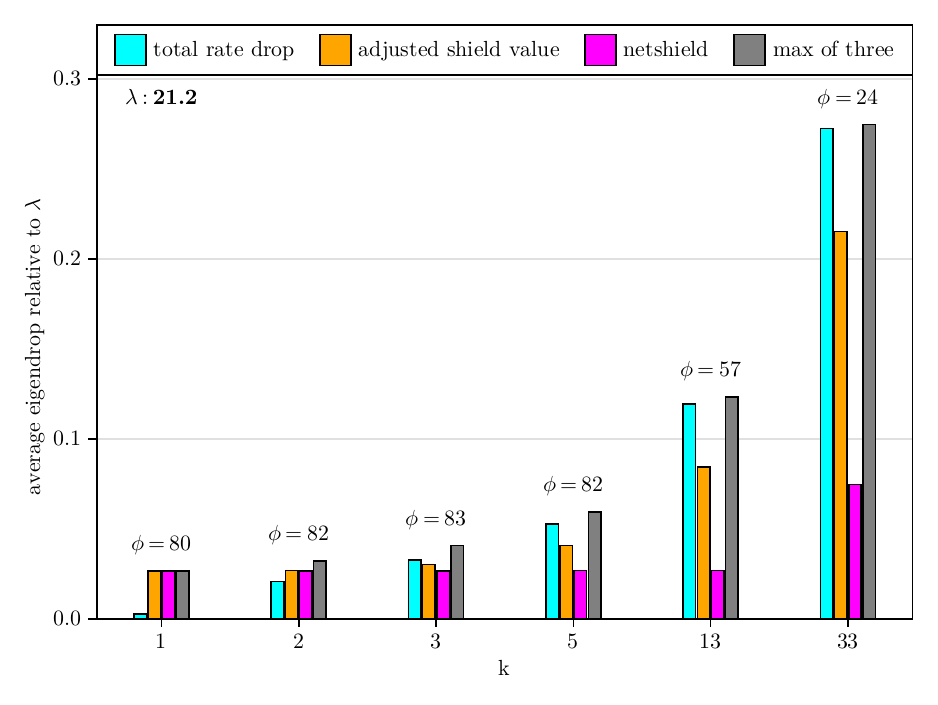}
        \caption{Average eigendrop}
    \end{subcaptionblock}
    \label{fig:Graph_communities}
\end{figure}

\subsection{Eigendrop distribution}\label{sec:eigendropDistribution}
We also considered the distribution of the eigendrop as a function of the number of removed nodes $k$ and compared it with the eigendrop obtained with Netshield for the same $k$. For this purpose, we considered, for each value of $k$, the eigendrop of 3600 forests sampled with Wilson's algorithm.
In Table~\ref{tab:eigendrop_quantiles_long} we present the quantiles of the eigendrop (normalized to the one obtained with the Netshield algorithm) associated with the removal of the complement of the root set sampled with Wilson's algorithm (also in this case, for some graphs, the results are provided in graphical form as well for the sake of clarity). 
We recall that computing the eigendrop for any set of nodes sampled with Wilson's algorithm is computationally costly, and these findings are provided only to give an insight into the behavior of the K-shield.
It is, however, worth mentioning that each time the eigendrop is better than that of Netshield for only 5\% of Kirchoff forests produced by Wilson's algorithm, it is sufficient to sample 60 forests to get some improvement of the Netshield result with probability .95 at least: 
$$
    \biggl(1 - {1 \over 20}\biggr)^{60} \leq e^{-3} \leq {1 \over 20}.
$$
Since the number of sampled forests is fixed, such an algorithm has computational cost $O(m)$.
From Table~\ref{tab:eigendrop_quantiles_long} it is possible to see that the best values (higher quantiles) of the eigendrop found with this procedure tend to be better than the eigendrop obtained with Netshield, especially in those graphs where the ``community structure'' is (or, at least, is likely to be) stronger.

\setlength{\LTleft}{\fill}
\setlength{\LTright}{\fill}
\setlength{\LTcapwidth}{0.9\linewidth}
\small

\begin{longtable}{@{}lcccc@{}}
\caption{Quantiles of the eigendrop obtained with K-shield algorithm (KS) for each graph and $k$. Values are normalized to the eigendrop obtained with Netshield (NS) for the same graph and $k$. The quantiles for each graph and $k$ have been determined by computing the eigendrop correspondig to 3600 set of nodes to immunize sampled with Wilson's algorithm. Graphs marked with the ``(w)'' flag have weighted edges, whereas graphs marked with the ``(uw)'' flag have unweighted edges.}
\label{tab:eigendrop_quantiles_long}\\
\toprule
\multirow{2}{*}{graph} & \multirow{2}{*}{$k$} & \multicolumn{3}{c}{quantiles of KS eignedrop} \\
\cmidrule(lr){3-5}
& & 0.9 & 0.95 & 0.99 \\
\midrule
\endfirsthead

\toprule
\multirow{2}{*}{graph} & \multirow{2}{*}{$k$} & \multicolumn{3}{c}{quantiles of KS eignedrop} \\
\cmidrule(lr){3-5}
& & 0.9 & 0.95 & 0.99 \\
\midrule
\endhead

\midrule
\multicolumn{5}{r}{continued on next page} \\
\endfoot

\bottomrule
\endlastfoot

\multirow{5}{*}{paper (uw)} & 1 & 0.421 & 1.000 & 1.000 \\*
 & 2 & \textbf{1.000} & \textbf{1.000} & \textbf{1.000} \\*
 & 3 & \textbf{1.323} & \textbf{1.885} & \textbf{1.885} \\*
 & 4 & 0.248 & 0.248 & 1.000 \\*
 & 5 & 0.386 & 0.386 & 1.000 \\

\midrule
\multirow{6}{*}{conference 1 (w)} & 1 & 0.408 & 1.000 & \textbf{1.000} \\*
 & 2 & 0.978 & 0.983 & \textbf{1.000} \\*
 & 3 & 0.989 & 0.998 & \textbf{1.260} \\*
 & 5 & \textbf{1.039} & \textbf{1.327} & \textbf{1.777} \\*
 & 10 & \textbf{1.889} & \textbf{2.066} & \textbf{2.460} \\*
 & 25 & \textbf{1.002} & \textbf{1.052} & \textbf{1.149} \\

\midrule
\multirow{6}{*}{conference 2 (w)} & 1 & 0.517 & 0.958 & \textbf{1.000} \\*
 & 2 & 0.958 & \textbf{1.001} & \textbf{1.049} \\*
 & 3 & \textbf{1.003} & \textbf{1.025} & \textbf{1.061} \\*
 & 5 & 0.914 & 0.925 & 0.965 \\*
 & 10 & 0.968 & 0.990 & \textbf{1.025} \\*
 & 26 & 0.959 & 0.980 & \textbf{1.010} \\

\midrule
\multirow{6}{*}{conference 3 (w)} & 1 & 1.000 & \textbf{1.000} & \textbf{1.000} \\*
 & 2 & \textbf{1.000} & \textbf{1.011} & \textbf{1.142} \\*
 & 3 & \textbf{1.024} & \textbf{1.112} & \textbf{1.143} \\*
 & 5 & 1.000 & \textbf{1.000} & \textbf{1.026} \\*
 & 10 & \textbf{1.029} & \textbf{1.032} & \textbf{1.043} \\*
 & 24 & \textbf{1.009} & \textbf{1.029} & \textbf{1.040} \\

\midrule
\multirow{5}{*}{airport 1 (uw)} & 1 & 0.963 & 0.971 & 1.000 \\*
 & 2 & 0.906 & 0.940 & 0.977 \\*
 & 3 & 0.859 & 0.898 & 0.953 \\*
 & 5 & 0.816 & 0.848 & 0.905 \\*
 & 12 & 0.769 & 0.788 & 0.827 \\

\midrule
\multirow{6}{*}{airport 2 (w)} & 1 & 0.620 & 0.763 & \textbf{1.000} \\*
 & 2 & 0.562 & 0.675 & 0.890 \\*
 & 3 & 0.514 & 0.611 & 0.730 \\*
 & 5 & 0.518 & 0.570 & 0.695 \\*
 & 50 & 0.849 & 0.865 & 0.893 \\*
 & 125 & 0.984 & 0.991 & 0.997 \\

\midrule
\multirow{6}{*}{airport 3 (uw)} & 1 & 0.762 & 0.880 & 0.949 \\*
 & 2 & 0.653 & 0.736 & 0.856 \\*
 & 3 & 0.581 & 0.653 & 0.772 \\*
 & 5 & 0.529 & 0.582 & 0.686 \\*
 & 186 & 0.805 & 0.820 & 0.856 \\*
 & 464 & \textbf{1.002} & \textbf{1.010} & \textbf{1.022} \\

\midrule
\multirow{6}{*}{airport 4 (w)} & 1 & 0.551 & 0.852 & \textbf{1.046} \\*
 & 2 & 0.395 & 0.472 & 0.615 \\*
 & 3 & 0.317 & 0.366 & 0.498 \\*
 & 5 & 0.267 & 0.318 & 0.411 \\*
 & 798 & 0.919 & 0.926 & 0.938 \\*
 & 1994 & \textbf{1.006} & \textbf{1.008} & \textbf{1.012} \\

\midrule
\multirow{6}{*}{rfid (w)} & 1 & 0.962 & 1.000 & \textbf{1.000} \\*
 & 2 & 0.878 & 0.925 & \textbf{1.015} \\*
 & 3 & 0.896 & 0.929 & \textbf{1.054} \\*
 & 5 & \textbf{1.059} & \textbf{1.132} & \textbf{1.263} \\*
 & 8 & 0.857 & 0.922 & \textbf{1.047} \\*
 & 19 & 0.835 & 0.878 & 0.945 \\

\midrule
\multirow{6}{*}{UKfaculty (w)} & 1 & 0.534 & 0.603 & \textbf{1.000} \\*
 & 2 & 0.707 & 0.966 & \textbf{1.071} \\*
 & 3 & 0.968 & \textbf{1.051} & \textbf{1.144} \\*
 & 5 & \textbf{1.134} & \textbf{1.219} & \textbf{1.401} \\*
 & 8 & \textbf{1.326} & \textbf{1.423} & \textbf{1.539} \\*
 & 20 & 0.909 & 0.955 & \textbf{1.028} \\

\midrule
\multirow{6}{*}{enron (w)} & 1 & 0.857 & \textbf{1.000} & \textbf{1.000} \\*
 & 2 & 1.000 & 1.000 & \textbf{1.413} \\*
 & 3 & 1.000 & \textbf{1.024} & \textbf{1.423} \\*
 & 5 & \textbf{1.117} & \textbf{1.423} & \textbf{1.726} \\*
 & 18 & \textbf{1.016} & \textbf{1.041} & \textbf{1.076} \\*
 & 46 & \textbf{1.012} & \textbf{1.026} & \textbf{1.054} \\

\midrule
\multirow{6}{*}{communities 01 (uw)} & 1 & 0.991 & 0.995 & \textbf{1.015} \\*
 & 2 & \textbf{1.127} & \textbf{1.393} & \textbf{1.606} \\*
 & 3 & \textbf{1.447} & \textbf{1.678} & \textbf{2.423} \\*
 & 5 & \textbf{2.401} & \textbf{2.575} & \textbf{2.778} \\*
 & 13 & \textbf{4.676} & \textbf{5.018} & \textbf{5.637} \\*
 & 33 & \textbf{3.717} & \textbf{3.876} & \textbf{4.153} \\

\end{longtable}
\normalsize

\begin{figure}
    \centering $ $
    \caption{Graph: ``enron'' (weighted) - eigendrop distribution (Netshield eigendrop is in red)}
      \begin{subcaptionblock}{0.45\textwidth}
          \includegraphics[width=\textwidth]{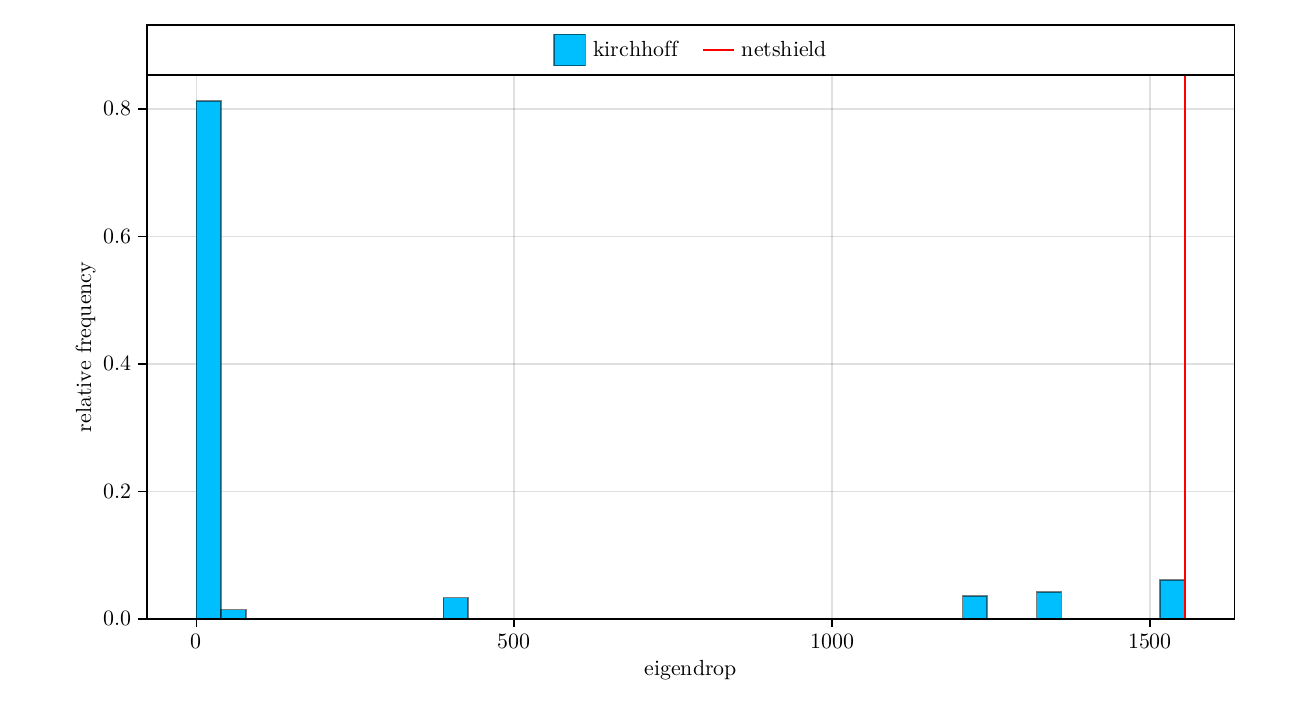}
          \caption{$k = 1$}
      \end{subcaptionblock}
      \begin{subcaptionblock}{0.45\textwidth}
          \includegraphics[width=\textwidth]{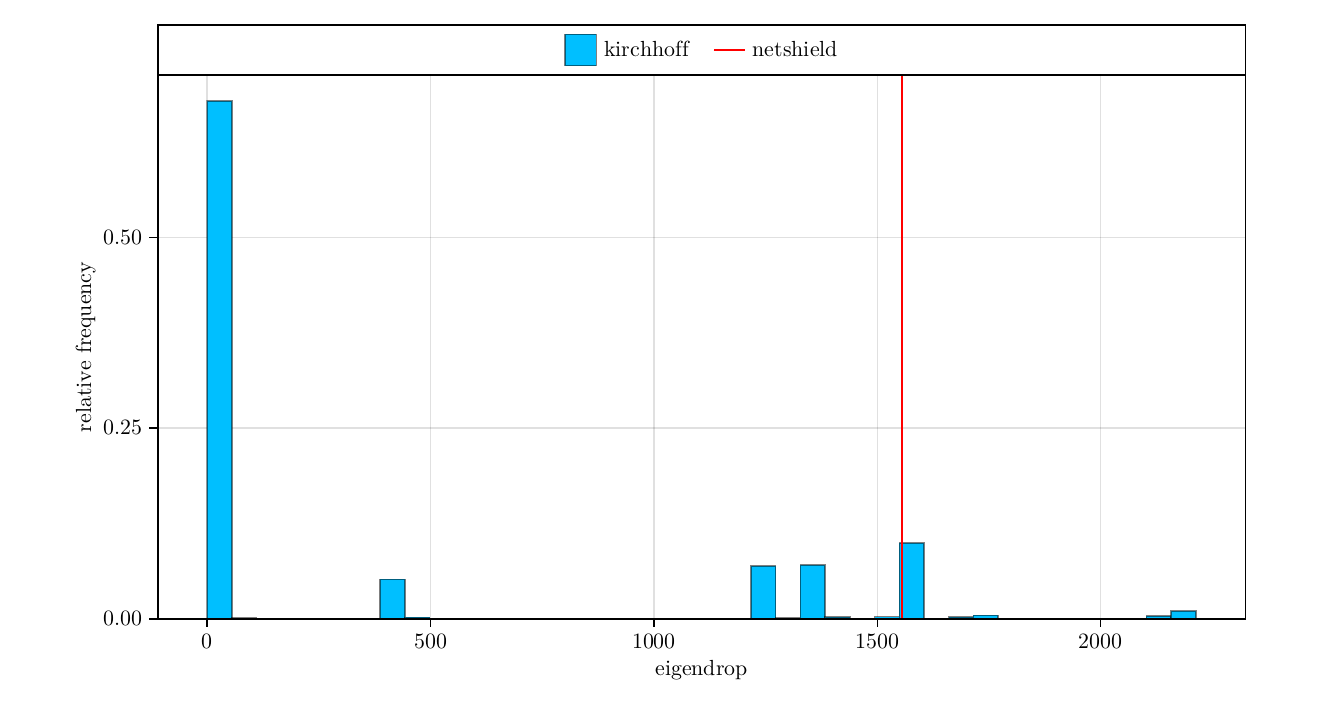}
          \caption{$k = 2$}
      \end{subcaptionblock}
      \\
      \begin{subcaptionblock}{0.45\textwidth}
          \includegraphics[width=\textwidth]{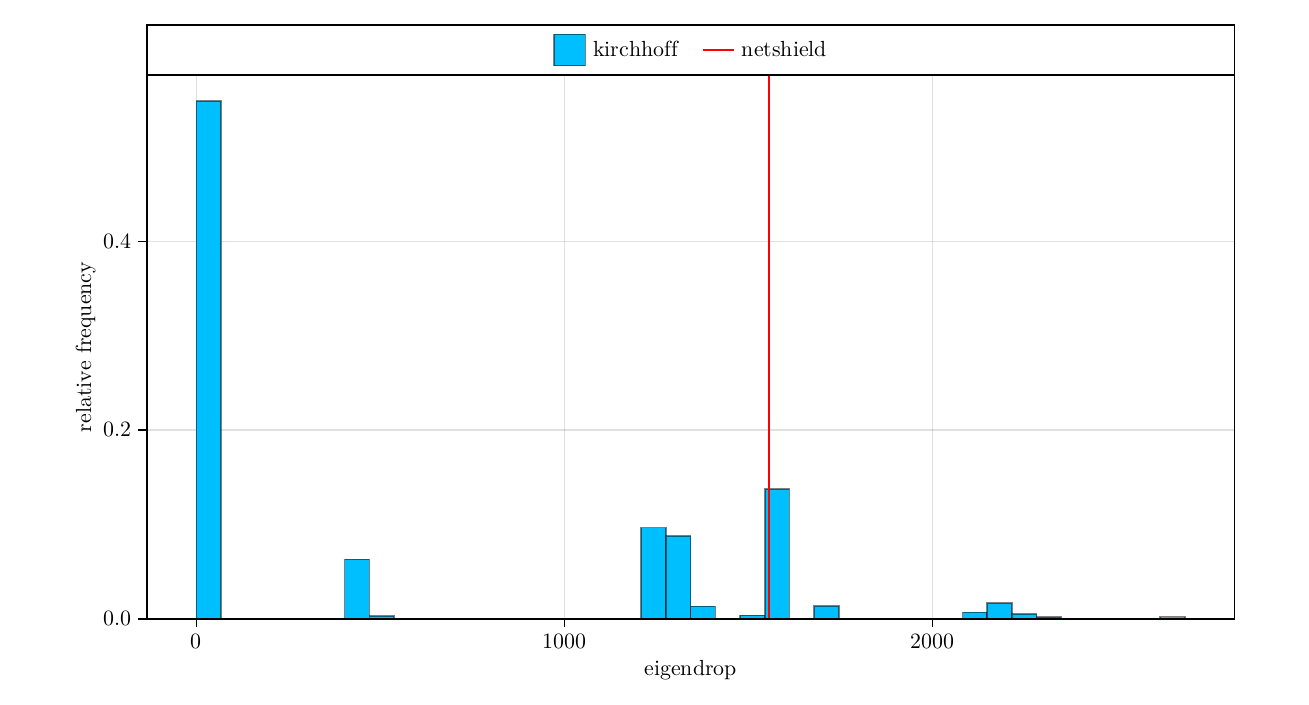}
          \caption{$k = 3$}
      \end{subcaptionblock}
      \begin{subcaptionblock}{0.45\textwidth}
          \includegraphics[width=\textwidth]{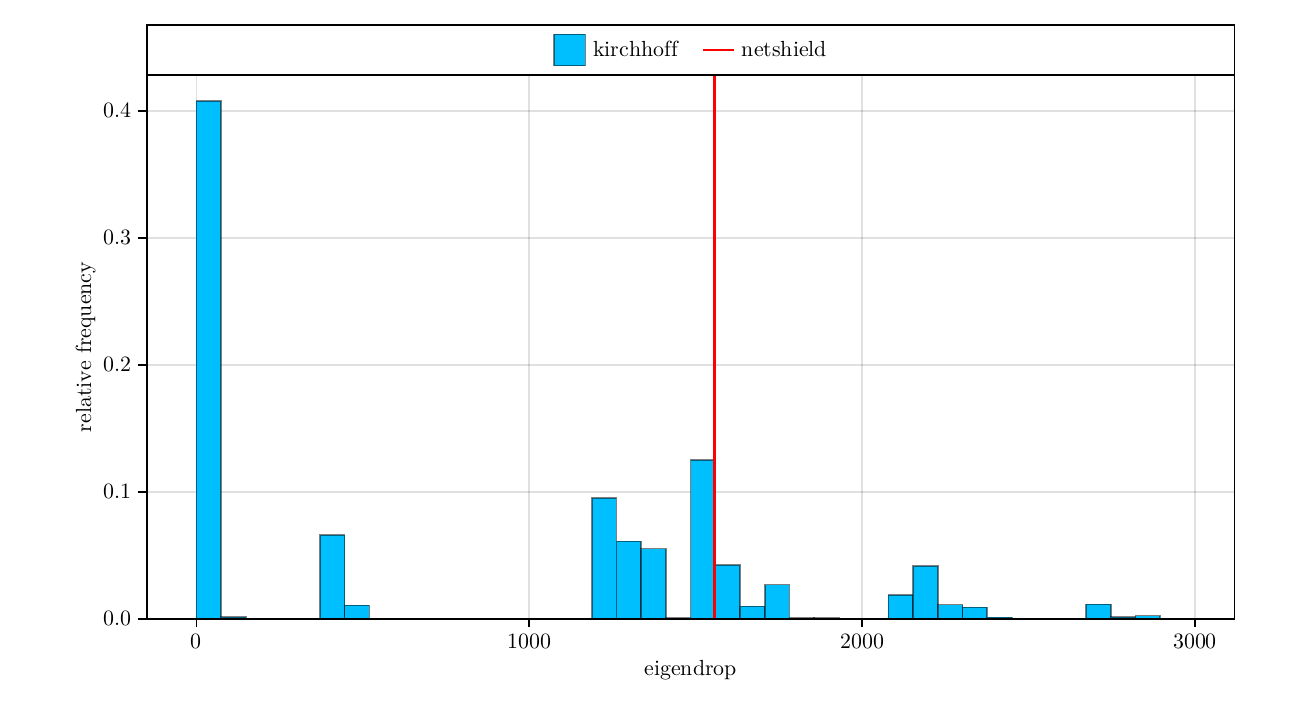}
          \caption{$k = 5$}
      \end{subcaptionblock}
      \\
      \begin{subcaptionblock}{0.45\textwidth}
          \includegraphics[width=\textwidth]{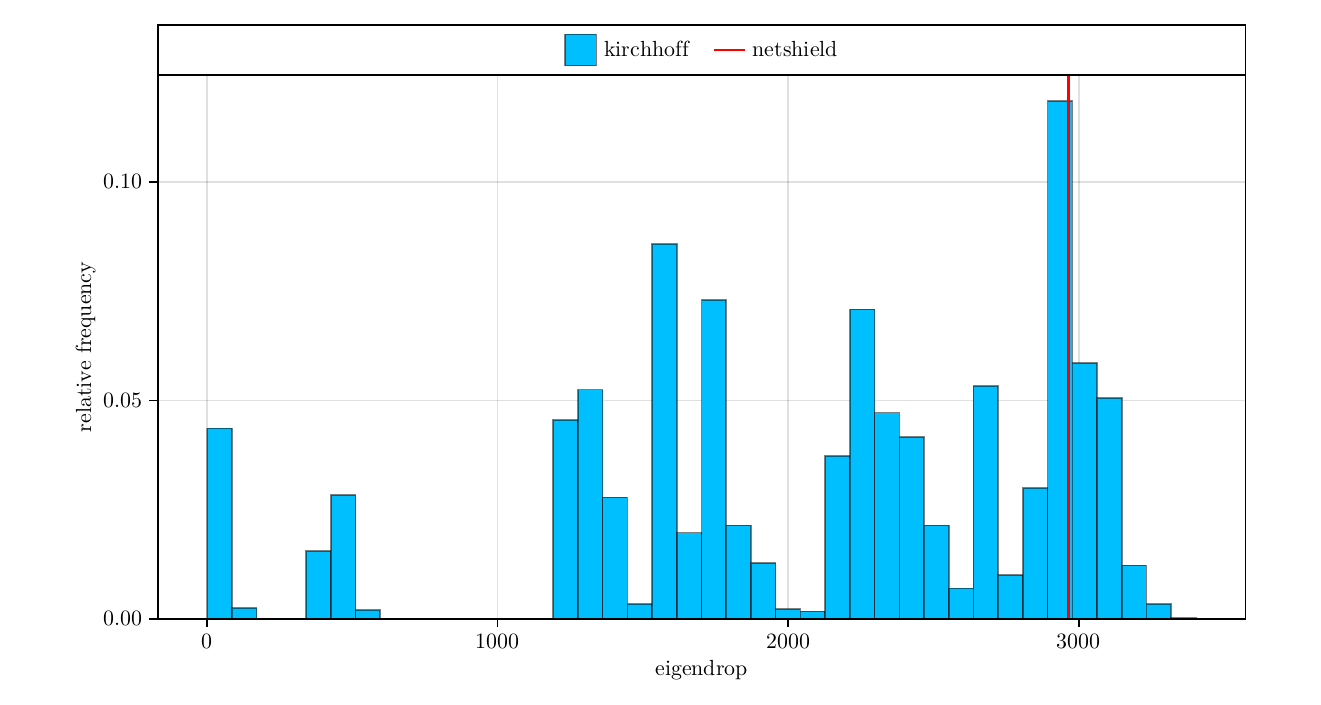}
          \caption{$k = 18$}
      \end{subcaptionblock}
      \begin{subcaptionblock}{0.45\textwidth}
          \includegraphics[width=\textwidth]{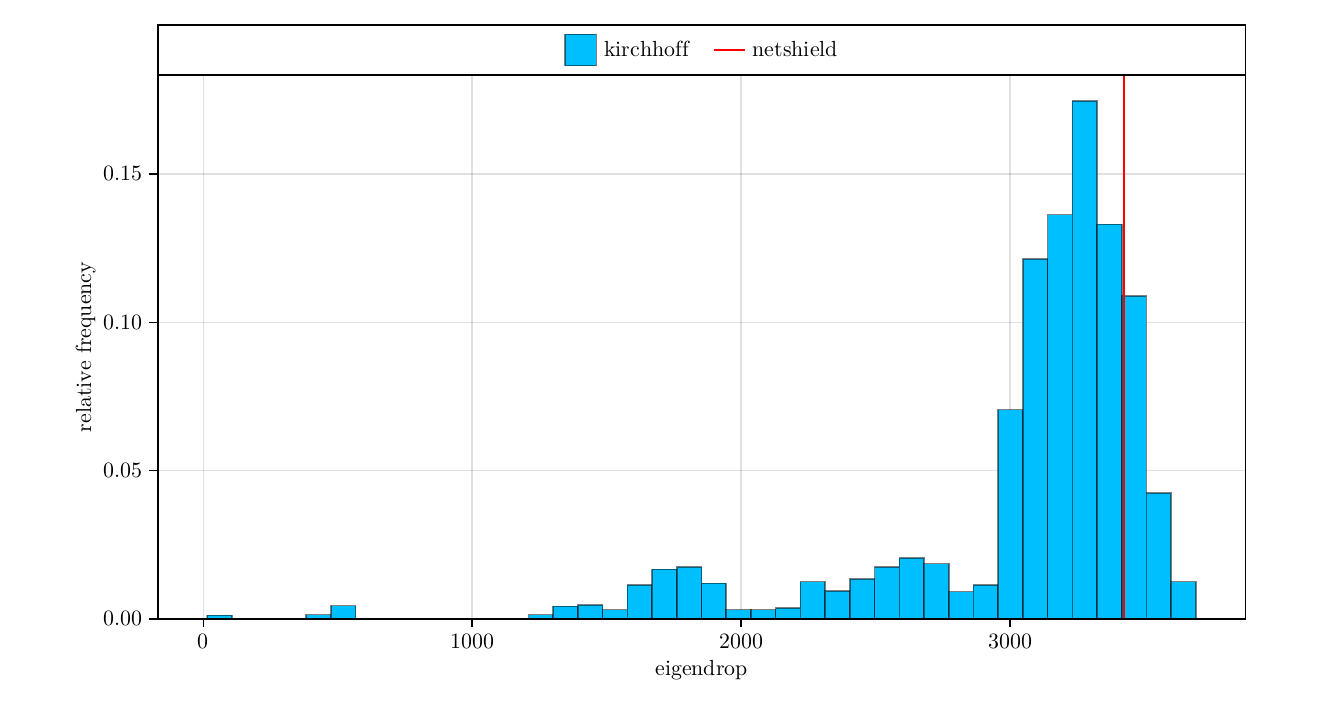}
          \caption{$k = 46$}
      \end{subcaptionblock}
\end{figure}

\begin{figure}
    \centering $ $
    \caption{Graph: ``communities 01'' (non-weighted) - eigendrop distribution}
      \begin{subcaptionblock}{0.45\textwidth}
          \includegraphics[width=\textwidth]{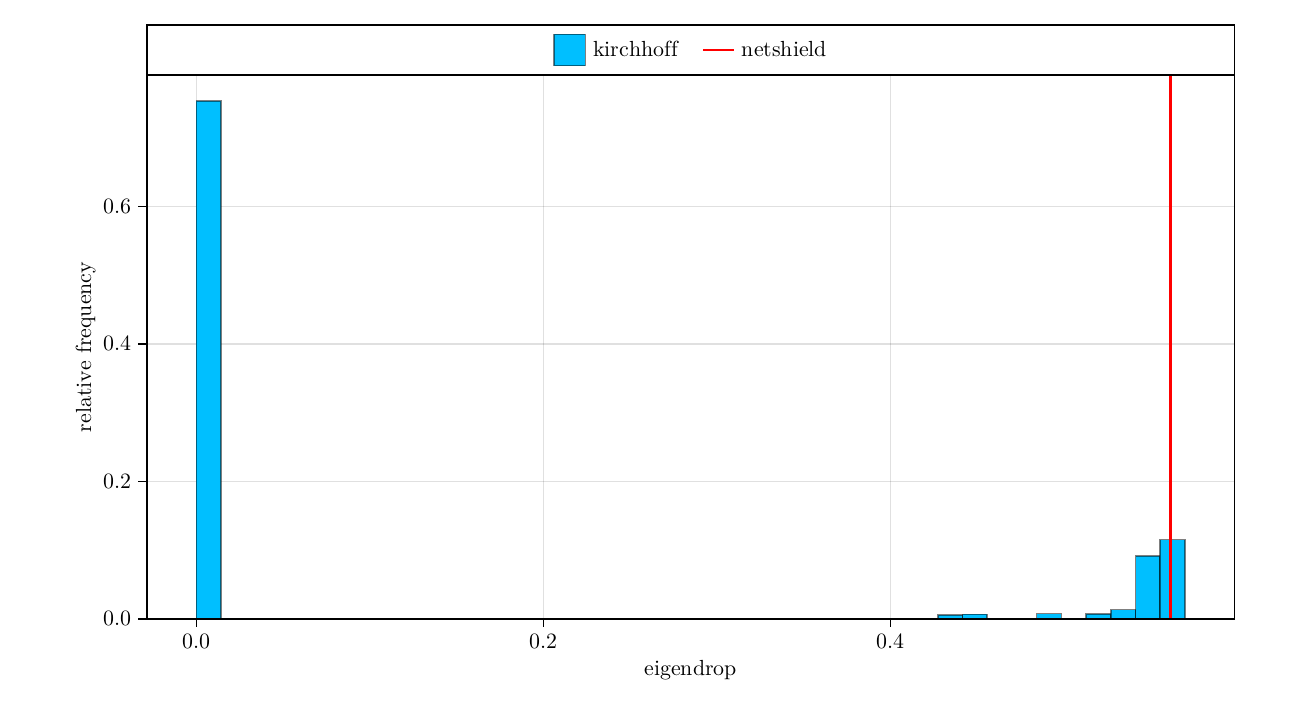}
          \caption{$k = 1$}
      \end{subcaptionblock}
      \begin{subcaptionblock}{0.45\textwidth}
          \includegraphics[width=\textwidth]{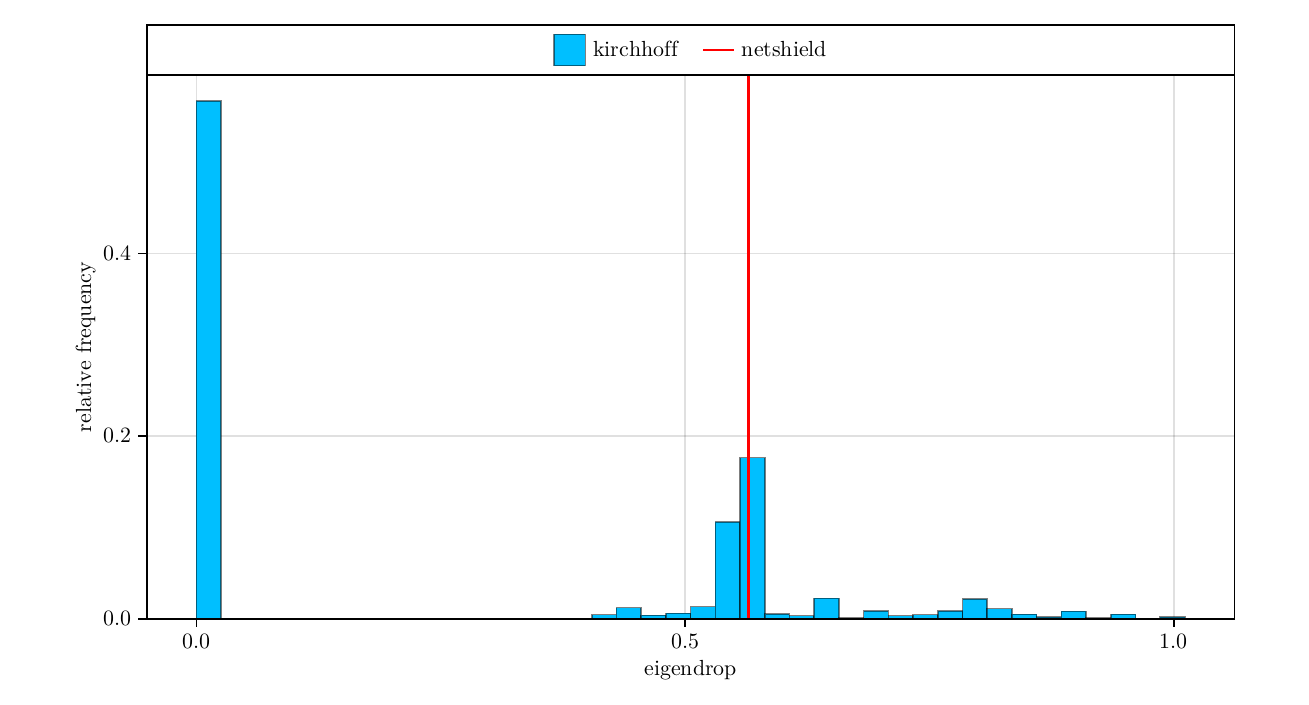}
          \caption{$k = 2$}
      \end{subcaptionblock}
      \\
      \begin{subcaptionblock}{0.45\textwidth}
          \includegraphics[width=\textwidth]{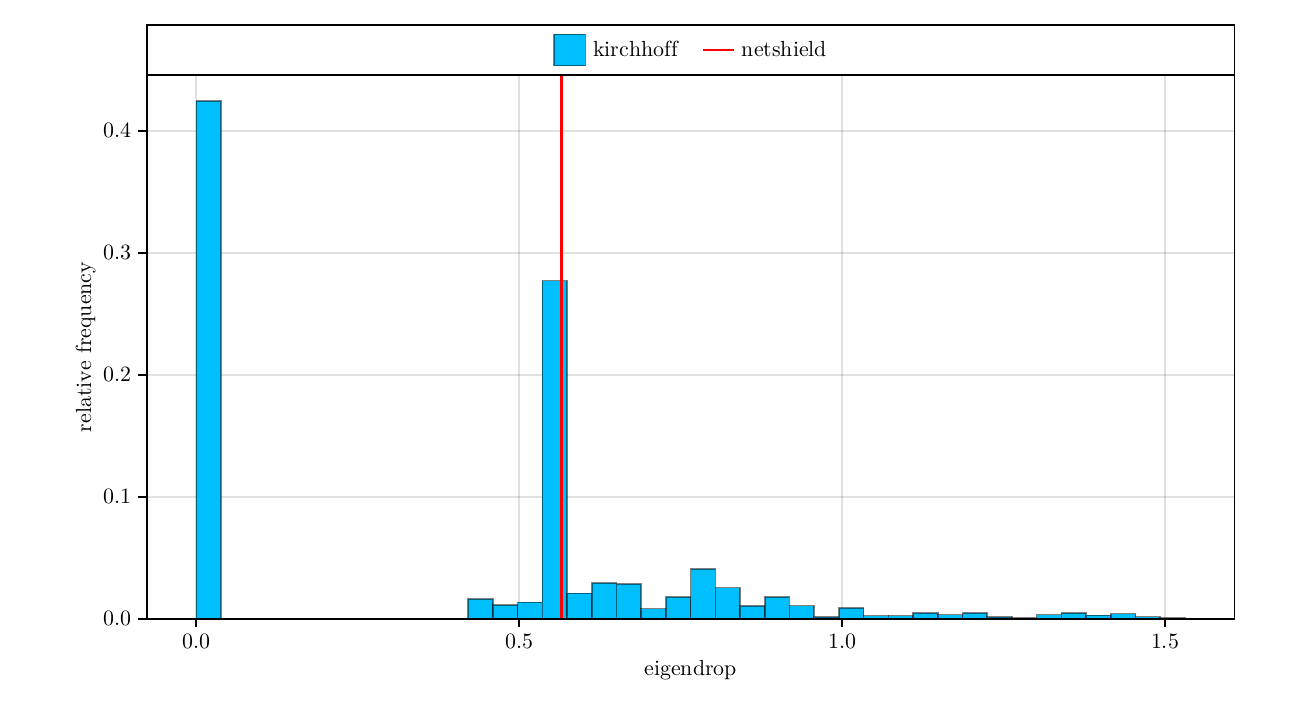}
          \caption{$k = 3$}
      \end{subcaptionblock}
      \begin{subcaptionblock}{0.45\textwidth}
          \includegraphics[width=\textwidth]{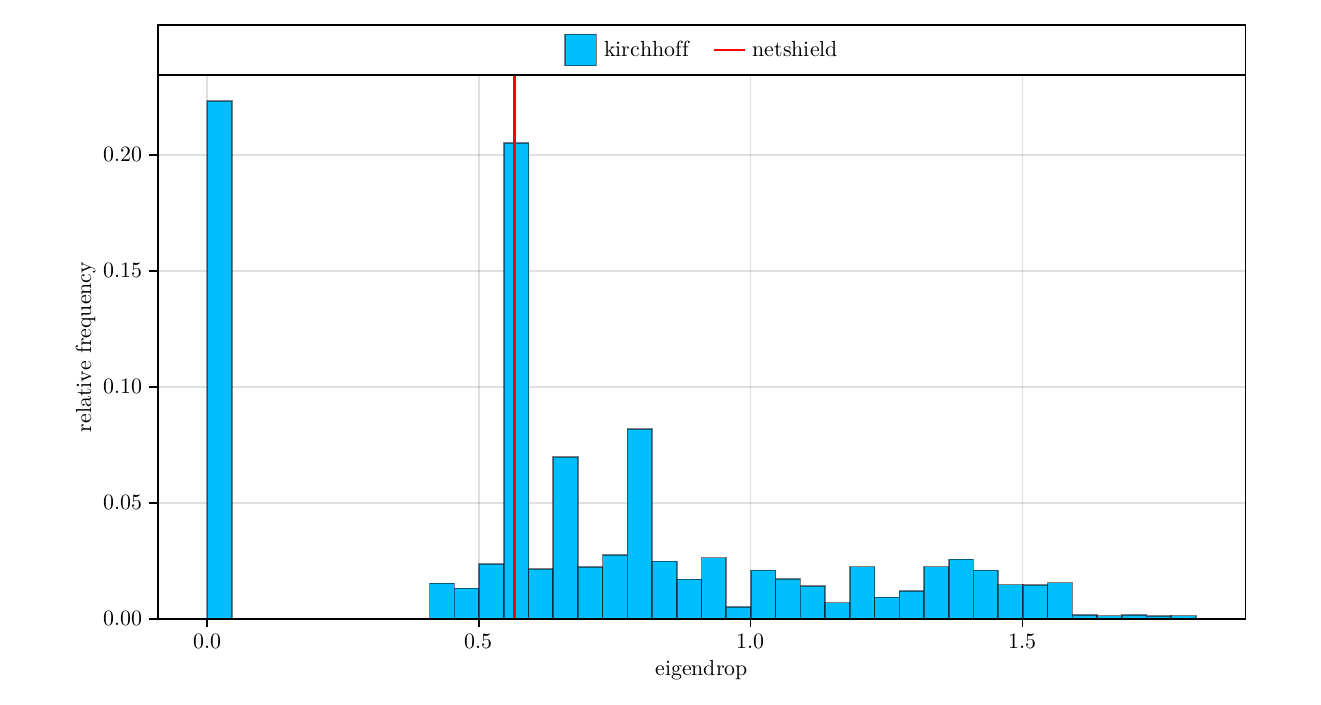}
          \caption{$k = 5$}
      \end{subcaptionblock}
      \\
      \begin{subcaptionblock}{0.45\textwidth}
          \includegraphics[width=\textwidth]{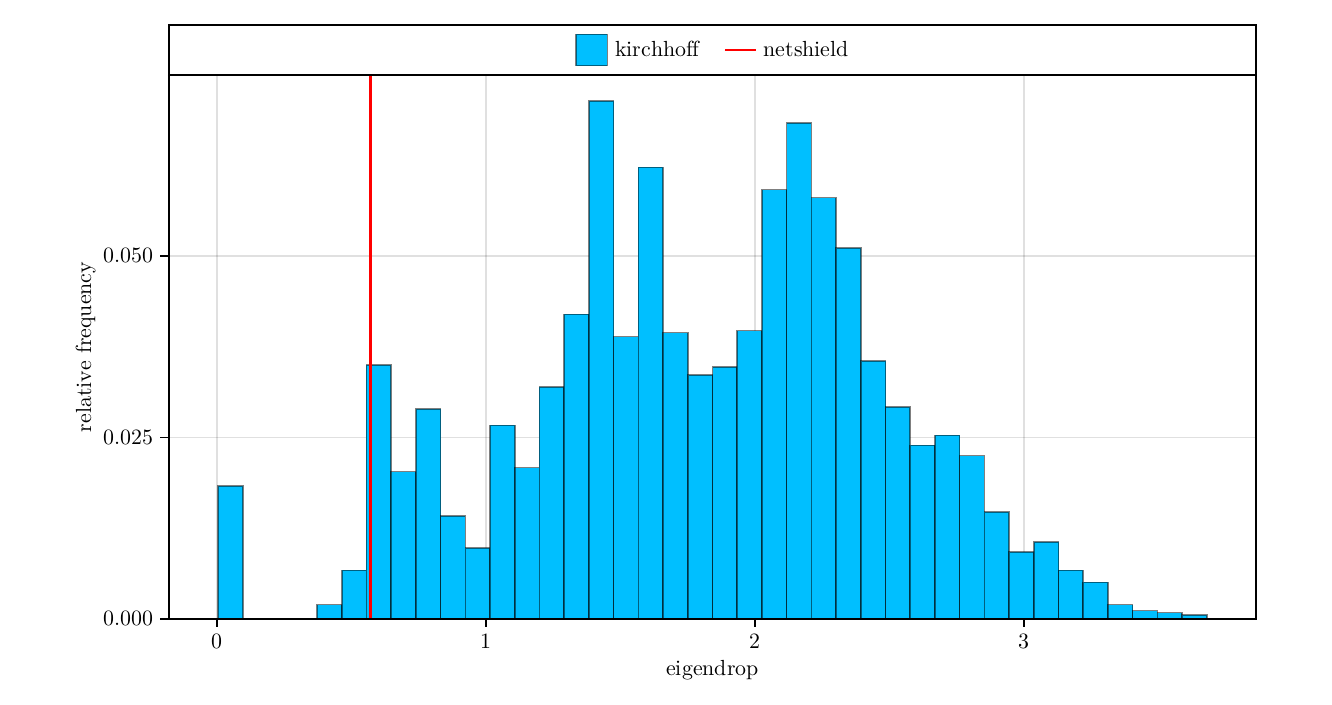}
          \caption{$k = 13$}
      \end{subcaptionblock}
      \begin{subcaptionblock}{0.45\textwidth}
          \includegraphics[width=\textwidth]{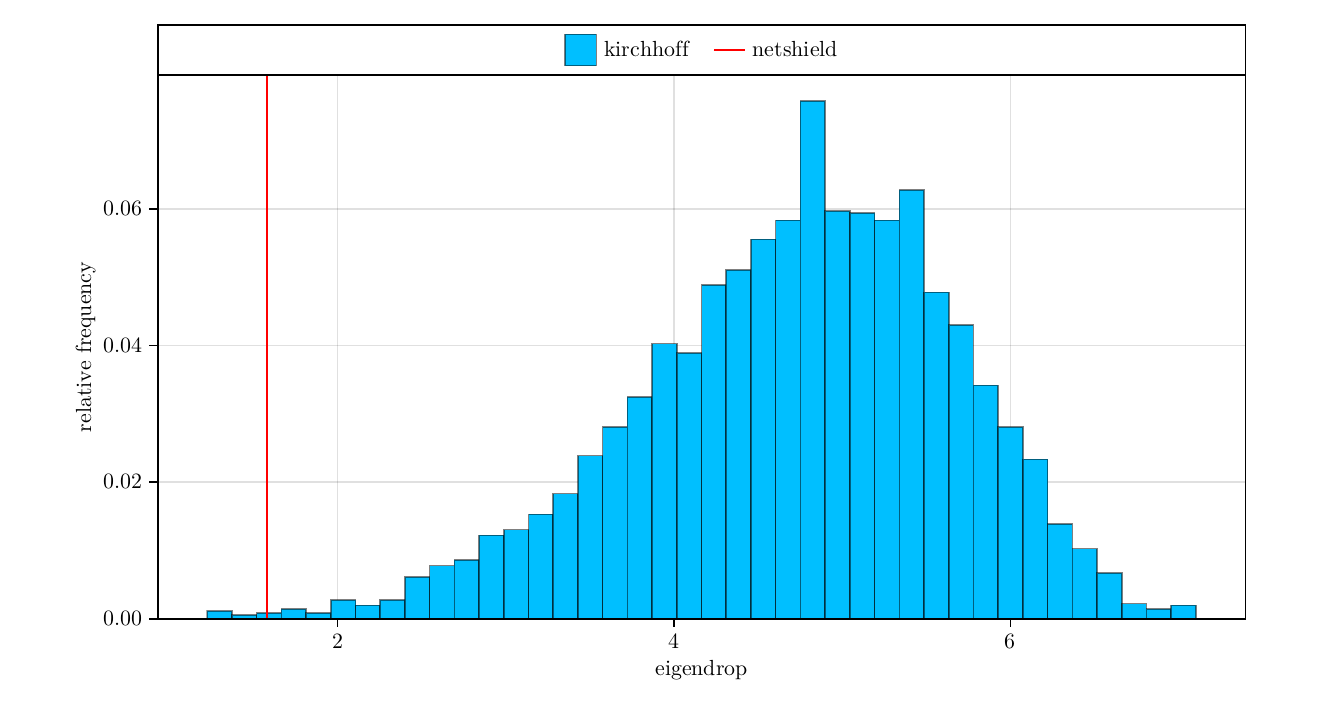}
          \caption{$k = 33$}
      \end{subcaptionblock}
\end{figure}

\section{Some mathematical details on Shieldvalue and Netshield}\label{sec:Netmore}
We show in the following lemma the validity of the bounds presented in \eqref{Shieldbounds}.

\begin{lmm}\label{clo}
For all $S \subset V$ it holds
$$
    \lambda_{\rm{max}} - \lambda(S) \leq \widetilde\shiv(S) \leq \shiv(S).
$$
\end{lmm}

\medskip\par\noindent\emph{Proof of Lemma~\ref{clo}.} In the following, denote by \(A_{BC}\) the submatrix of the adjacency matrix \(A\) obtained by selecting the rows in the subset of vertices \(B\) and the columns in the subset of vertices \(C\).
$$
    A = \left(
        \begin{array}{cc}
            A_{RR} & A_{RS}\\
            A_{SR} & A_{SS}
        \end{array}
    \right)
$$
We write in the same way $u_B$ for the vector obtained by selecting the coordinates in $B$ of a column vector $u$. For $u$ the normalized Perron-Frobenius vector associated with $\lambda_{\rm{max}}$ 
and with $R$ the complementary of $S$ we get from $Au = \lambda_{\rm{max}} u$
$$
    A_{SR}u_R + A_{SS} u_S = \lambda_{\rm{max}}u_S
$$
and
\begin{align*}
    \lambda_{\rm{max}}
    &= \langle u_R, A_{RR} u_R\rangle + 2\langle u_S, A_{SR} u_R\rangle + \langle u_S, A_{SS} u_S\rangle \\
    &= \langle u_R, A_{RR} u_R\rangle + 2\langle u_S, \lambda_{\rm{max}}u_S - A_{SS}u_S\rangle + \langle u_S, A_{SS} u_S\rangle \\
    &\leq \lambda(S) \|u_R\|^2 + 2\lambda_{\rm{max}}\|u_S\|^2 - \langle u_S, A_{SS} u_S\rangle \\
    &= \lambda(S) \|u_R\|^2 + \shiv(S),
\end{align*}
so that
$$
    \lambda(S)
    \geq {\lambda_{\rm{max}} - \shiv(S) \over \|u_R\|^2}
    \geq \lambda_{\rm{max}} - \shiv(S)
$$
and taking the opposite before adding $\lambda_{\rm{max}}$
$$
    \lambda_{\rm{max}} -  \lambda(S)
    \leq  \lambda_{\rm{max}} + {\shiv(S) - \lambda_{\rm{max}} \over \|u_R\|^2}
    = {\shiv(S) - \lambda_{\rm{max}} \|u_S\|^2 \over 1 - \|u_S\|^2}
    \leq \shiv(S).
$$
\qed

Table \ref{algo:Netshield} below describes the Netshield algorithm in a slightly different form with respect to the original description given in \cite{Chen16} to optimize the related complexity costs.

\begin{table}[h!]
\centering
\begin{tabular}{p{.8\textwidth}}
    \hline \\
    \textbf{Netshield Algorithm} \\
    \hline
    \textbf{Input}: the $n\times n$ adjacency matrix $A$; integer $k$\\
		\textbf{Output}: a subset $S$ with $k$ nodes\\
		Compute: \begin{tabular}[t]{l}
			$\lambda_{\rm{max}}$ largest eigenvalue of $A$\\
			$u$  its corresponding eigenvector\\ 
		\end{tabular}\\
		Initialize $S = \emptyset$\\
		\textbf{For} $j$ from $1$ to $n$\\
		\quad\begin{tabular}{|l}
			$\score(j) = \left(2 \cdot \lambda_{\rm{max}} - A(j , j)\right) \cdot u(j)^2$\\
		\end{tabular}\\
		\textbf{End For}\\
		\textbf{For} $iter$ from $1$ to $k$\\
		\quad\begin{tabular}{|l}
			$i = \argmax_{j \not\in S} \score(j)$\\
			$S = S\cup\left\lbrace i \right\rbrace$\\
            \textbf{If $|S| < k$}\\
            \quad\begin{tabular}{|l}
                \textbf{For} $j \not\in S$\\
                \quad\begin{tabular}{|l}
                        $\score(j) = \score(j) - 2 A(i, j)u(i)u(j)$\\
                \end{tabular}\\
                \textbf{End For}\\
            \end{tabular}\\
            \textbf{End If}
		\end{tabular}\\
		\textbf{End For}\\
		Return $S$\\
        \hline
\end{tabular}
\caption{Netshield Algorithm}
\label{algo:Netshield}
\end{table}

\newpage We see from the algorithm description that its complexity is in $O(m + nk)$ operations, as $m$ is the number of edges of the adjacency matrix $A$:
the cost of the computation of $\lambda$ and $u$
is in $O(m)$ by using the power method;
and there are $k$ iterations of cost in $O(n)$.

As mentioned above, our description in Table \ref{algo:Netshield}  is slightly different from the implementation originally suggested, see \cite{Chen16}, where the score table -- rather than its variation -- was directly computed at each iteration, which led to an extra matrix vector product and a complexity in $O(m + nk^2)$.

\section{Kirchhoff forest: mathematical structure, sampling and immunization }\label{sec:Kmore}
\subsection{Forests, random walks, loop-erasure and Wilson's algorithm }\label{forestSampling}

Given a finite graph $G=(V, E)$ and a weight function $w: E\rightarrow [0,\infty)$, it is called {\em graph Laplacian matrix} the $|V|\times|V|$ symmetric\footnote{The restriction to undirected weighted graphs is only for simplicity of exposition and because for the node-immunization application we have only considered symmetric graph benchmarks. However, it is worth stressing that an analogous immunization approach could also be developed in a directed graph setting. %
} matrix $L=(L_{x,y})_{x,y\in V}$
with entries: 
\begin{equation}\label{graphLaplacian}
L_{x,y}=\begin{cases}
    -w(e) & \text{if } (x,y)=e\in E \\
    \displaystyle\sum_{(x,y)=e\in E}w(e) & \text{if } x=y\\
    0 & \text{else}.
\end{cases}
\end{equation}
Such a matrix and its spectrum are intimately linked to the Kirchhoff forest; in fact, the normalizing constant (or partition function) in Def. \ref{KforestDef} can be expressed in terms of $L$ as follows: 
\begin{equation}\label{PartFn}
Z(q) =\det\left[q \Id+ L\right],
\end{equation} 
This relation is referred to in the literature as (generalized) \emph{matrix tree-theorem} and it can be traced back to the seminal work of Kirchhoff on electrical networks \cite{K1847}, which also explains our choice for the name of such a random forest. 

We call \emph{Random Walk (RW) associated to a graph} $G=(V,E)$ with weight function $w$, the continuous-time Markov chain 
$X =\{X(t): t\geq 0\}$, 
with state space $V$ and infinitesimal generator given by $-L$. The relevance of the RW $X$ is due to the fact that it can be used to efficiently sample the Kirchhoff forest $\Phi_q$ as well as to analyze its local statistics.
Indeed, to sample  $\Phi_q$ for a given $q$, one can use the celebrated Wilson's algorithm \cite{W96} which recursively builds the random forest by launching successive loop-erased trajectories of $X$ until exhaustion of the vertex set as briefly described next. 

Let us first recall that given a RW $X$, it is called the associated loop-erased RW, the process obtained by erasing cycles as soon as they appear along the trajectory of $X$, so that the resulting trajectories are self-avoiding. It is worth stressing that the trajectories of such a loop-erased process have closed statistics that can be expressed in terms of the determinant of $L$, see \cite{M00}.  
\\
\noindent {\bf Wilson's algorithm:}
\emph{\begin{enumerate}
\item[(1)] Let $T_q$ be an independent exponential clock of rate $q>0$. Fix an arbitrary vertex $x_0\in V$ and launch from there a loop-erased version of the RW $X$ until the clock $T_q$ rings and record the self-avoiding path $\gamma$ in $V$ covered by this first loop-erased trajectory, so that $\gamma$ is a finite sequence of distinct vertices.
\item[(2)] Take next any arbitrary vertex $x_1$ in $V\setminus\{x\in\gamma\}$ and launch, again, a loop-erased version of $X$ until the minimum between a new independent time $T_q$ and the first time that this loop-erased walk hits a node in $\gamma$. 
\item[(3)] Update the subset of nodes in $\gamma$ from step (1) by adding to it those vertices covered by the loop-erased trajectory launched in step (2).
\item[(4)] iterate steps (2) and (3) until exhaustion of the vertex set $V$.
\end{enumerate}}

It is not difficult to see that when the above algorithm stops, it outputs a random spanning sub-graph with no cycles, which corresponds exactly to the forest with the law as in Def.\ref{KforestDef}, see\cite{W96,AG18}.
What is quite remarkable in Wilson's algorithm is that in steps (1) and (2) above, the choice of the starting vertex is arbitrary. Such an \emph{exchangeable} property makes Wilson's algorithm also a powerful tool to analyze the interesting local features of the forest. As mentioned in Section \ref{sec:OurAlgorithm}, a dynamical variant of Wilson's algorithm  that allows to sample in one run a collection of forests having $k$ roots with $k$ ranging from $1$ to $n=|V|$ has been studied in \cite[Theorem 2]{AG18}.

\subsection{Forest immunization and heterogeneous graphs}\label{trace}

As mentioned in the introduction, depending on the perspective, we may consider other measures of network vulnerability than the largest eigenvalue of the adjacency matrix of the graph.
Another classical approach to impede the spread of a malicious agent in a network is to try to limit the presence of hubs, namely nodes with a high weighted degree (as defined in \eqref{wDeg}).

By adapting this perspective for the multiple-node immunization original problem, a natural proposal would then be to look for a set of nodes of a given size to be removed which makes minimal \emph{the average weighted degree}\footnote{Another similar related measure would be the maximal weighted degree $w_{max}=\max_{x\in V} w(x)$. Note that on any arbitrary simple weighted graph $\lambda_{\rm{max}}\leq w_{max}$.} of the remaining network. Kirchhoff forests turn out to have an interesting property concerning this idea of reducing the influence of the (weighted) hubs, which we next describe. Denote by 
\begin{equation}\label{wDegAv}
\overline{w}:=\frac{1}{|V|}\sum_{x\in V} w(x),
\end{equation}
the average weighted degree associated with the starting weighted graph $(G, w)$. Notice that if the weights of the graph represent the infection rates of an underlying SIS as defined in \eqref{CPrates}, then such an average weighted degree can be interpreted as the average infection rate.
Analogously, for a given set $R \subset V$, $V$ being the vertex set of $G$, let $\overline{[w]_R}$ denote the
average weighted degree for the \emph{reduced} weighted network $(G', [w]_R)$, where $G'$ is obtained from $G$ by removing all vertices in the set $S:=V\setminus R$ and edges incident to the nodes in $S$, and $[w]_R$ is the restriction of the weight function $w$ on the resulting reduced edge set.

Let us further introduce the so-called \emph{trace process} associated with a subset $R\subset V$. Formally, this is the continuous-time Markov chain, say $X^{(R)}$, with state space $R$ and with infinitesimal generator given by 
the $|R|\times|R|$ matrix $-L^{(R)}$ defined as
the Schur complement of $[-L]_R$ in $-L$, where we recall that $-L$ is minus the graph Laplacian matrix associated to the original graph $G$ in \eqref{graphLaplacian}, and $[-L]_R$ is its restriction to the set $R$.
See \cite[Lemma 8]{ACGM20} for more details.
In words, the trace process on $R$ moves as the original RW $X$ on $G$ inside the set $R$, while it can be thought as of a consistent infinite speed version of $X$ when the trajectories of $X$ cover points outside $R$.
In particular, such a RW $X^{(R)}$ is associated to the weighted reduced graph $(G'', w^{(R)})$, where the vertex set of $G''$ is $R$, the weight function 
$w^{(R)}(x,y)$ for $x,y\in S$ is determined by the off-diagonal entries of the matrix $-L^{(R)}$ above and the edge set of $G''$ is determined by the pair of points in $R$ for which $w^{(R)}(x,y)>0$. Thus, by comparison between the trace process on $R$ and the process $X$ restricted to $R$, or equivalently, between the weighted graphs $(G'', w^{(R)})$ and $(G', [w]_R)$, it follows by construction that $w^{(R)}(x,y)\geq [w]_R(x,y)$ for any $x,y\in S$ and in particular, for the corresponding average weighted degrees, the following domination holds: 

\begin{equation}\label{TraceComp}
\overline{w^{(R)}}:= \frac{1}{| R|}\sum_{x\in R} w^{(R)}(x)\geq \overline{[w]_R}:=\frac{1}{|R|}\sum_{x\in R} [w]_R(x) . 
\end{equation}

In view of this domination and the above discussion, while searching for good sets of vertices to be immunized (i.e., removed) for the multiple-node original problem, one may look to those sets $S\subseteq V$ of given size such that the average weighted degree of the trace process on $R=V\setminus S$ is minimal.

We can now finally describe the interesting property of Kirchhoff forests.
Consider for the set $S$ to be immunized/removed the random set $\rho(\Phi_q)^{c}$
(i.e., the complement of the set of roots in Kirchhoff forest) conditioned to have a given cardinality $k\leq |V|$. Under this standpoint, it turns out that at least in mean, Kirchhoff forest does a good job as a function of $k$ in lowering the average weighted degree if the considered graph has a heterogeneous spectrum, contrary to the complete graph. More precisely, if $|V|=n$ and we want to immunize/remove $k\leq n-1$ nodes, it turns out, see \cite[Prop. 20]{ACGM20} that, in mean, the average weighted degree of the trace process mentioned above has the following explicit expression
\begin{equation}\label{ak} 
    a_k:= \mathbb{E} \left[ \overline{w^{(\rho(\Phi_q)^c)}} \mid 
                            |\rho(\Phi_q)^{c}|=k 
                      \right]
        =\frac{k+1}{n-k}\frac{b_{k+1}}{b_k}, \qquad k\leq n-1,
\end{equation}
with $b_k$ given by
\begin{equation}\label{bk}
    b_k:=\mathbb{P}( |\rho(\Phi_q)|=n-k )=\mathbb{P}( |\rho(\Phi_q)^c|= k ),
\end{equation}
representing the coefficients of order $n-k$ of the partition function in \eqref{PartFn}, that is:
$$Z(q)=q\prod_{k=1}^n (q+\lambda_i)=\sum_{k=0}^n b_k q^{n-k},$$
$\lambda_i$'s standing for the positive (real) eigenvalues of the (symmetric) graph Laplacian in \eqref{graphLaplacian}.
By noticing that $b_n=0$ since with probability one Kirchhoff forest contains at least one tree (equiv. root), and using the classical Newton's inequalities applied twice to the coefficients $b_k$'s of the polynomial $Z(q)$, one gets that whatever the graph is, the average infection rate in \eqref{ak} is decreasing in $k$, indeed:
$$\frac{a_k}{a_{k-1}}\leq 1- \frac{1}{(n-k)^2}.$$
Furthermore, Newton's inequalities also guarantee that the above inequality is an identity if and only if the eigenvalues $\lambda_i$'s above coincide, which is the case only if the underlying graph is the complete graph. In conclusion, immunizing with the complementary set of Kirchhoff forest has, in mean, the property that \emph{the more heterogeneous the spectrum of the considered graph, the stronger is the decay of the average infection rate}.

\section*{Acknowledgment}

Insert the Acknowledgment text here.
L. Avena was supported by NWO Gravitation Grant 024.002.003-NETWORKS. 
The work of AT has been partially supported by PRIN 2022 PNRR: "RETINA: REmote sensing daTa INversion with multivariate functional modeling for essential climAte variables characterization" (Project Number: P20229SH29, CUP: J53D23015950001) funded by the European Union under the Italian National Recovery and Resilience Plan (NRRP) of NextGenerationEU, under the Italian Ministry of University and Research (MUR)

\bibliographystyle{abbrv}
\bibliography{biblio}

\appendix

\section{Graphs details}\label{sec:graphs_details}
The benchmark graphs considered in the paper that have been used to test the K-shield algorithm are described below:
\begin{description}
    \item [karate]: Zachary's karate club graph with 34 nodes and 77 unoriented edges introduced in \cite{zachary1977information}; the graphs; data obtained from \url{https://networks.skewed.de/net/karate}; 
    \item [Conference graphs]: a set of three graphs (labeled as ``conference~1'', ``conference~2'' and ``conference~3'') concerning the face-to-face interaction between pair of participants during the ACM Hypertext 2009 conference; each graph describes the interaction in one day of the conference; the weight of edge $\{i, j\}$ is proportional to total time individuals $i$ and $j$ interacted in a given day; the graphs has been first studied in \cite{Isella2011qo}; the data is obtained from that available at \url{http://www.sociopatterns.org/datasets/hypertext-2009-dynamic-contact-network/} and the edge weights, for each day of conference, are obtained by counting the number of 20 seconds interactions between two nodes available in the file \url{http://www.sociopatterns.org/files/datasets/003/ht09_contact_list.dat.gz}.
    \item [enron]: a graph with $183$ nodes related to the mail exchanges at Enron; the data has been made public by the U.S. Department of Justice; the weight of edge $\{i, j\}$ is given by the number of mail exchanges between $i$ and $j$; data retrieved from the "igraphdata" package for R (see \url{https://cran.r-project.org/web/packages/igraphdata/index.html});
    \item [rfid]: $75$ nodes graph (studied in \cite{vanhems2013estimating}) representing workers in a geriatric unit of a hospital in Lyon (France) in 2012; the weight of edge $\{i, j\}$ is given by the number of encounters between $i$ and $j$; data retrieved from the "igraphdata" package for R;
    \item [UK faculty]: graph with $81$ nodes each representing a faculty in a UK university; the the weighted adjacency matrix $W$ of the graph gives, for each $i$ and $j$ the strength of the connection between individuals $i$ and $j$; in our tests we considered the symmetrized version $W'$of the weighted adjacency matrix defined as $W' = W + W^{T}$; the graph has been introduced in \cite{nepuszFuzzyCommunitiesConcept2008}; data retrieved from the "igraphdata" package for R;
    \item [airports graphs]: a set of four graphs containing airport-related data:
        \begin{description}
            \item [airport~1]: a subset of ``airport~2'' obtained considering the $50$ airports (nodes) with higher degree (the edges of the graph are not weighted)
            \item [airport~2]: (also referred to as ``USairport500'' in the literature) the nodes of the graph represent the 500 busiest commercial airports in the United States with edges $\{ i, j \}$ denoting the presence of at least one scheduled flight between airport $i$ and $j$ (in 2002); the weight of edge $\{ i, j \}$ is given by the number of seats available between $i$ and $j$; this dataset has been studied in \cite{colizza2007reaction} and the data has been retrieved from \url{http://opsahl.co.uk/tnet/datasets/USairport_2010.txt}
            \item [airport~3]: graph with $1574$ nodes representing the complete set of US airports in 2010; an edge between nodes $i$ and $j$ exists if passengers flights have been operated between $i$ and $j$; the edges are not weighted; data was obtained from \url{[tnet-format](http://opsahl.co.uk/tnet/datasets/USairport_2010.txt)} where all weights have been set equal to $1$.
            \item [airport~4]: a graph with $7976$ nodes representing world airports considered in \cite{opsahl2010node}; the weight of edge $\{ i, j \}$ is given by the number of routes between $i$ and $j$; data has been downloaded from \url{http://opsahl.co.uk/tnet/datasets/openflights.txt}
        \end{description}
\end{description}

\section{Graphical representation of obtained results}\label{sec:charts}

\subsection{Origin of best eigendrop}\label{sec:figs-originBestEigendrop}

\begin{figure}
    \centering
    \caption{Graph: ``paper'' (non-weighted)$ $}
    \begin{subcaptionblock}{0.45\textwidth}
        \includegraphics[width=\textwidth]{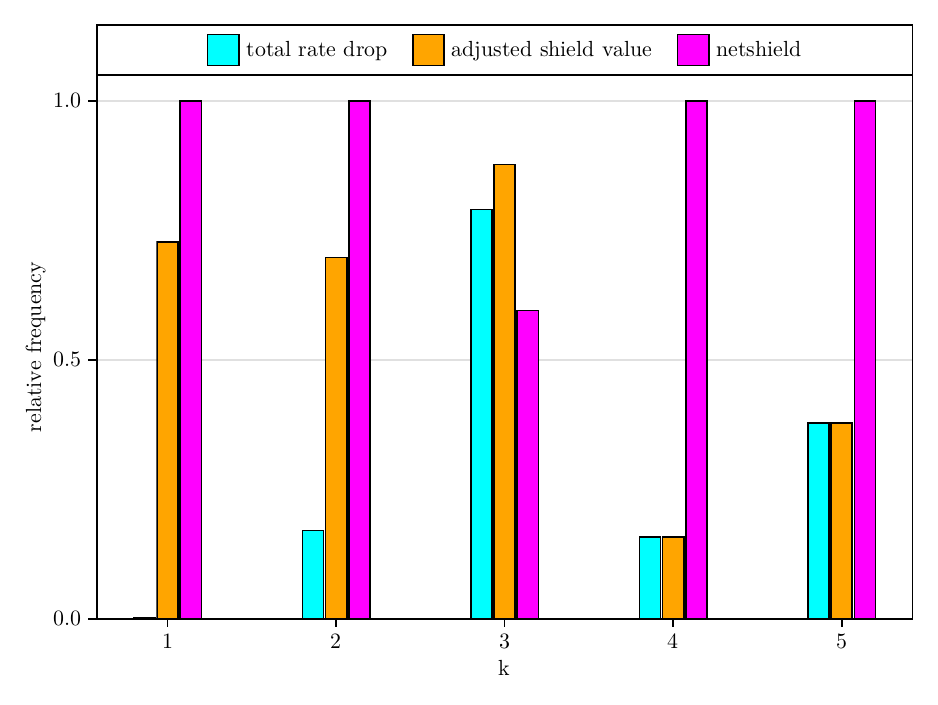}
        \caption{Origin of best eigendrop}
    \end{subcaptionblock}
    \begin{subcaptionblock}{0.45\textwidth}
        \includegraphics[width=\textwidth]{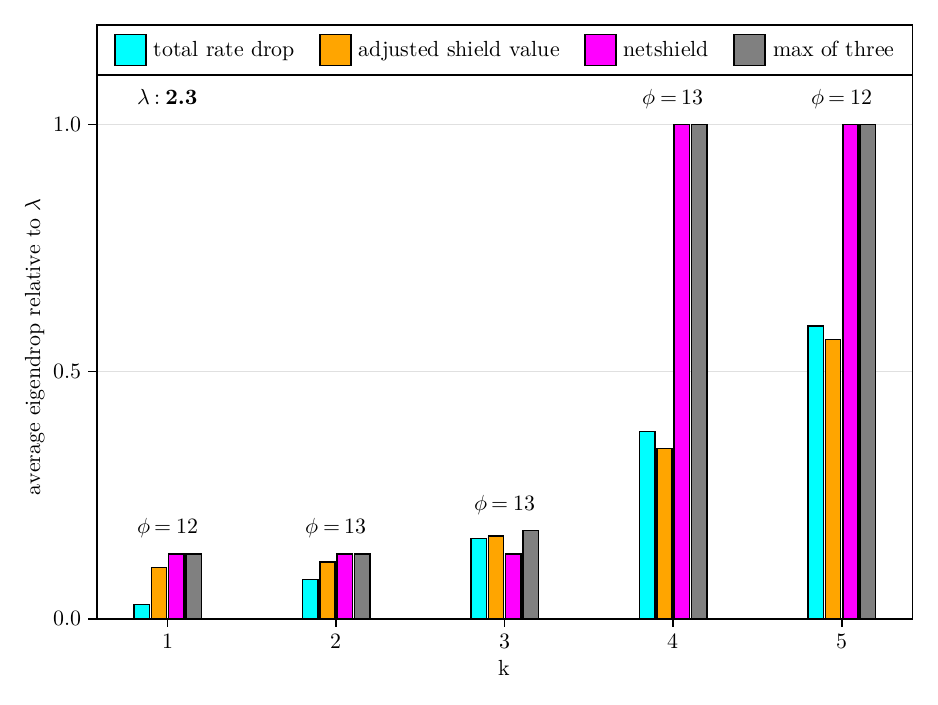}
        \caption{Average eigendrop}
    \end{subcaptionblock}
\end{figure}

\begin{figure}
    \centering
    \caption{Graph: ``karate club'' (non-weighted)$ $}
    \begin{subcaptionblock}{0.45\textwidth}
        \includegraphics[width=\textwidth]{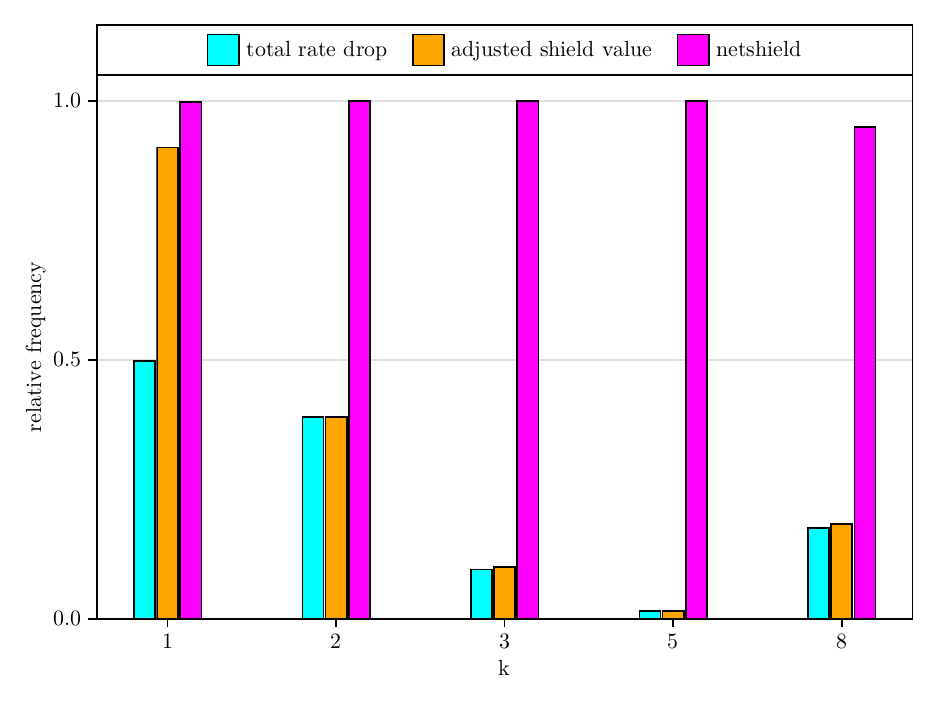}
        \caption{Origin of best eigendrop}
    \end{subcaptionblock}
    \begin{subcaptionblock}{0.45\textwidth}
        \includegraphics[width=\textwidth]{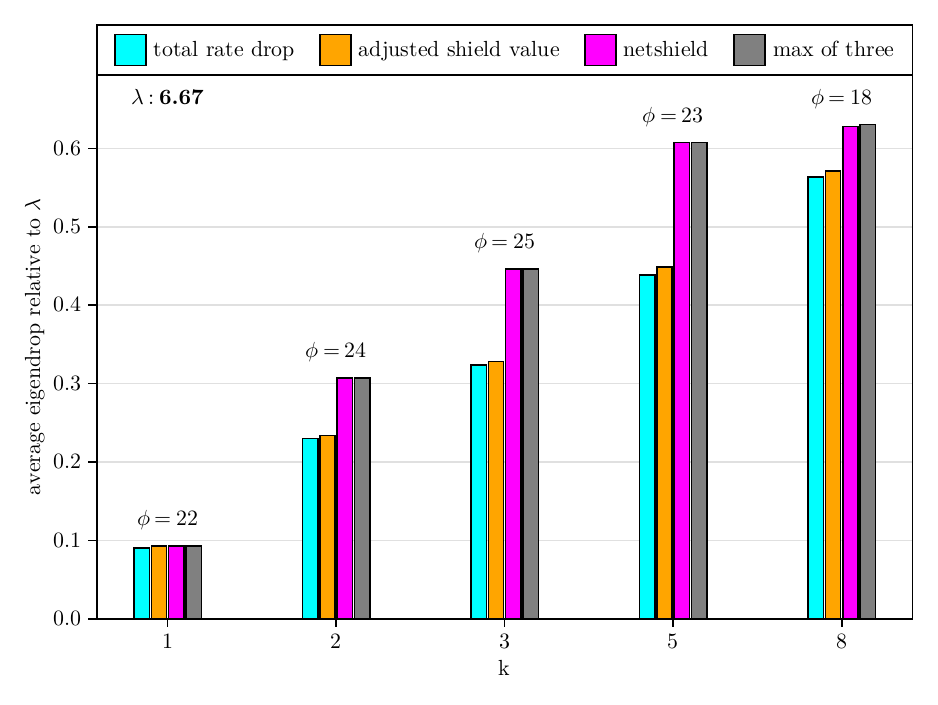}
        \caption{Average eigendrop}
    \end{subcaptionblock}
\end{figure}

\begin{figure}
    \centering
    \caption{Graph: ``conference 1'' (non-weighted)$ $}
    \begin{subcaptionblock}{0.45\textwidth}
        \includegraphics[width=\textwidth]{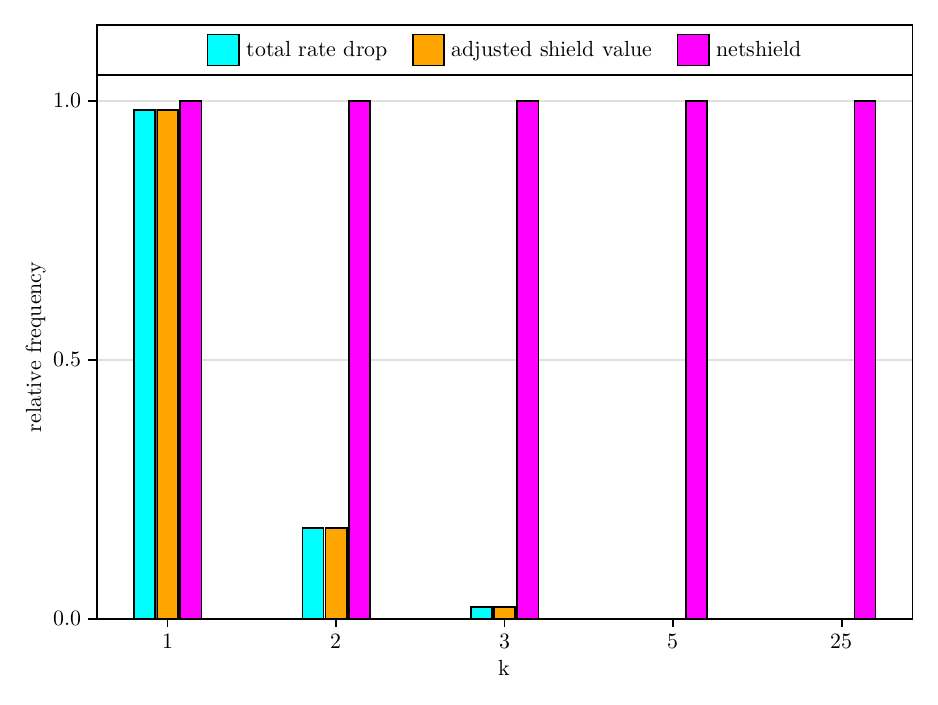}
        \caption{Origin of best eigendrop}
    \end{subcaptionblock}
    \begin{subcaptionblock}{0.45\textwidth}
        \includegraphics[width=\textwidth]{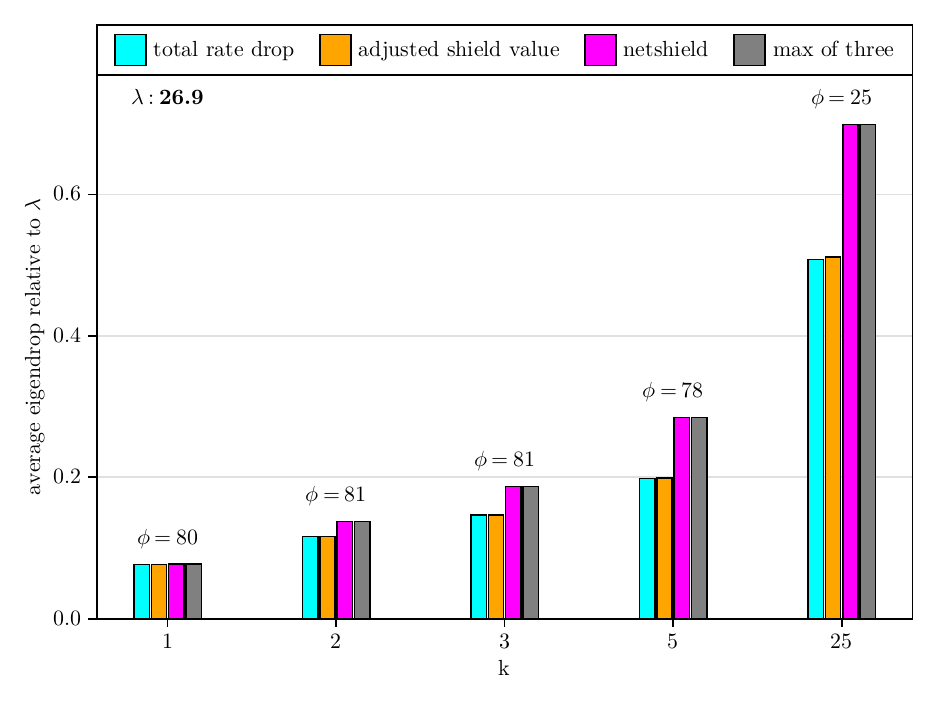}
        \caption{Average eigendrop}
    \end{subcaptionblock}
\end{figure}

\begin{figure}
    \centering
    \caption{Graph: ``conference 1'' (weighted)$ $}
    \begin{subcaptionblock}{0.45\textwidth}
        \includegraphics[width=\textwidth]{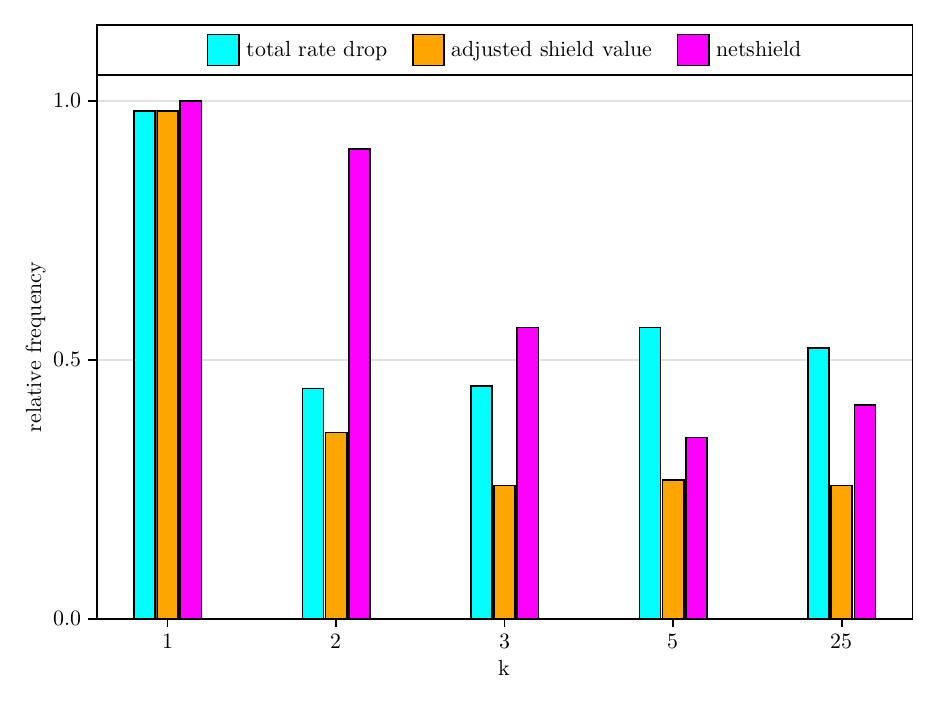}
        \caption{Origin of best eigendrop}
    \end{subcaptionblock}
    \begin{subcaptionblock}{0.45\textwidth}
        \includegraphics[width=\textwidth]{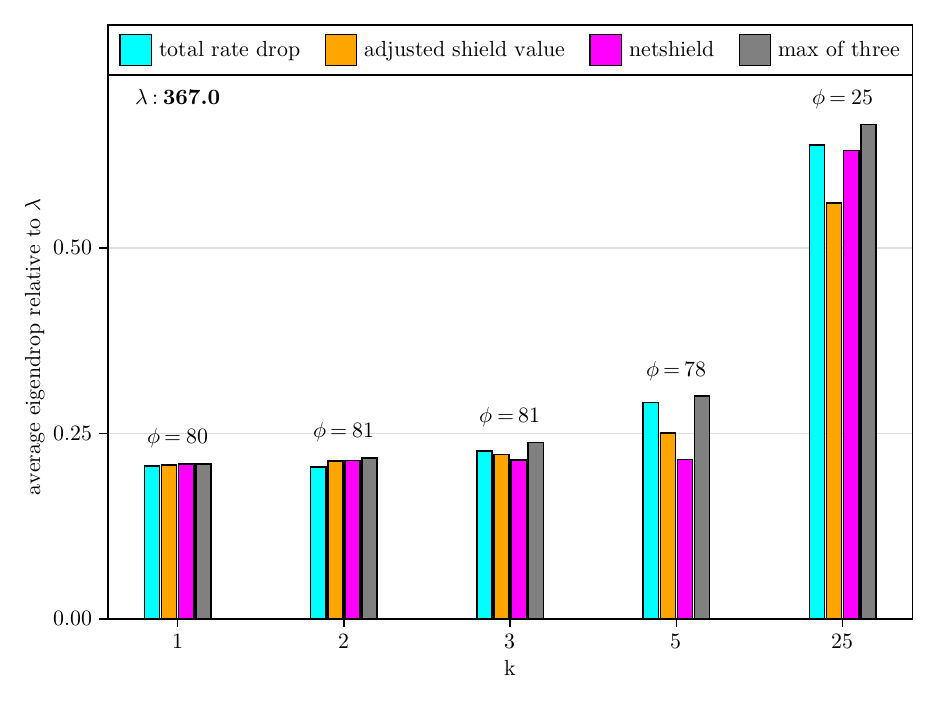}
        \caption{Average eigendrop}
    \end{subcaptionblock}
\end{figure}

\begin{figure}
    \centering
    \caption{Graph: ``conference 2'' (non-weighted)$ $}
    \begin{subcaptionblock}{0.45\textwidth}
        \includegraphics[width=\textwidth]{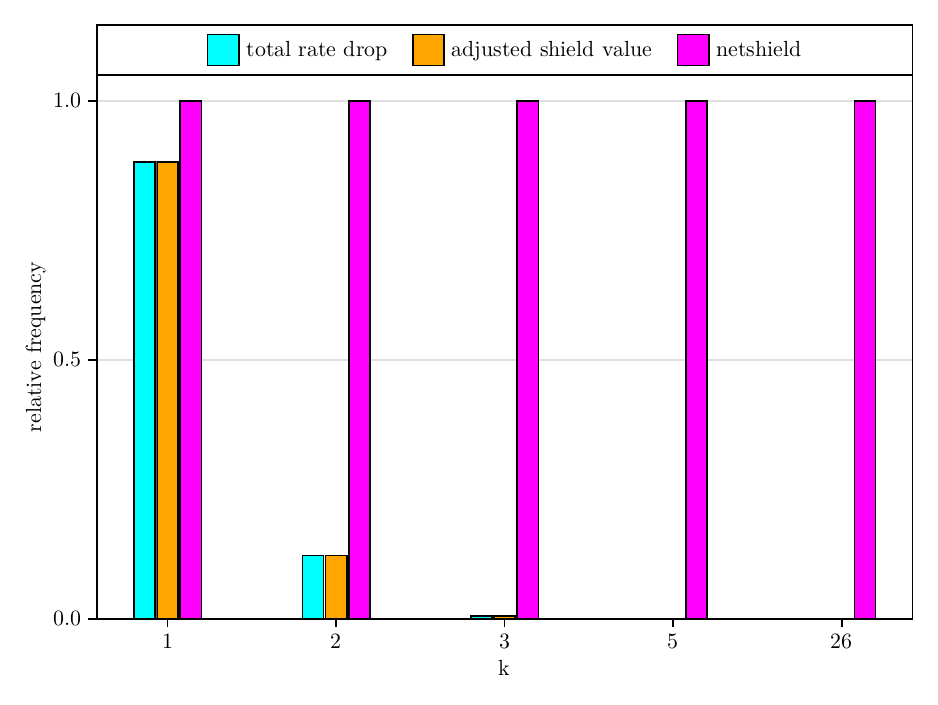}
        \caption{Origin of best eigendrop}
    \end{subcaptionblock}
    \begin{subcaptionblock}{0.45\textwidth}
        \includegraphics[width=\textwidth]{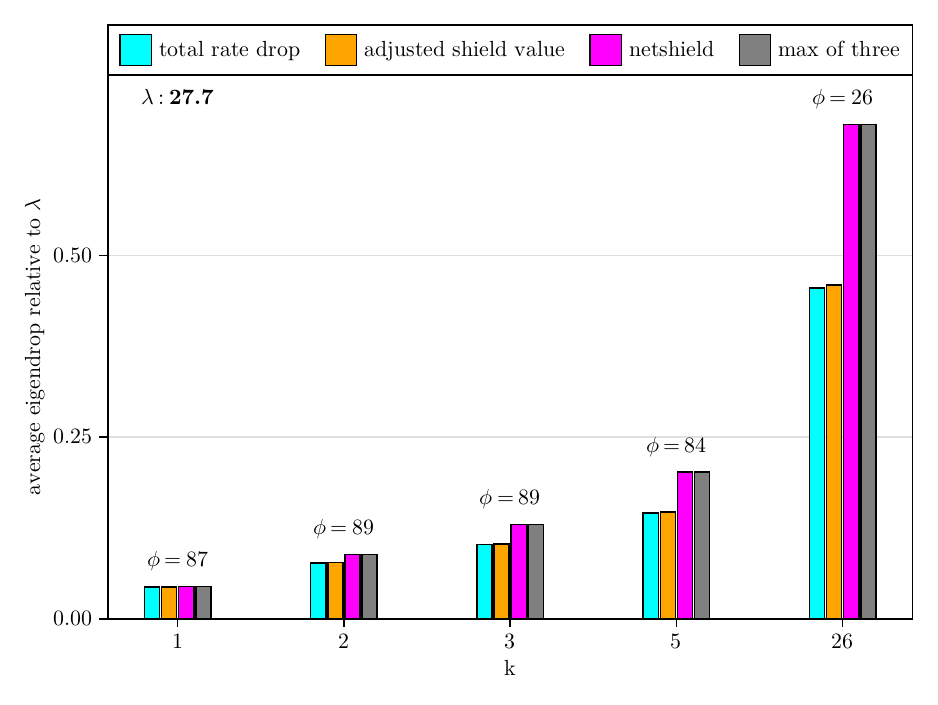}
        \caption{Average eigendrop}
    \end{subcaptionblock}
\end{figure}

\begin{figure}
    \centering
    \caption{Graph: ``conference 2'' (weighted)$ $}
    \begin{subcaptionblock}{0.45\textwidth}
        \includegraphics[width=\textwidth]{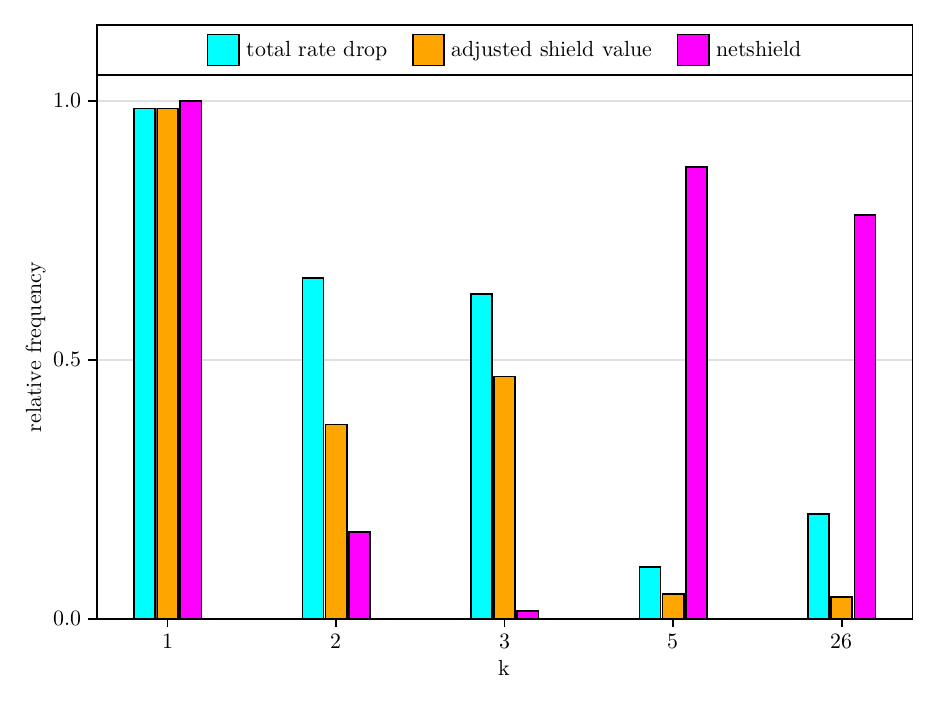}
        \caption{Origin of best eigendrop}
    \end{subcaptionblock}
    \begin{subcaptionblock}{0.45\textwidth}
        \includegraphics[width=\textwidth]{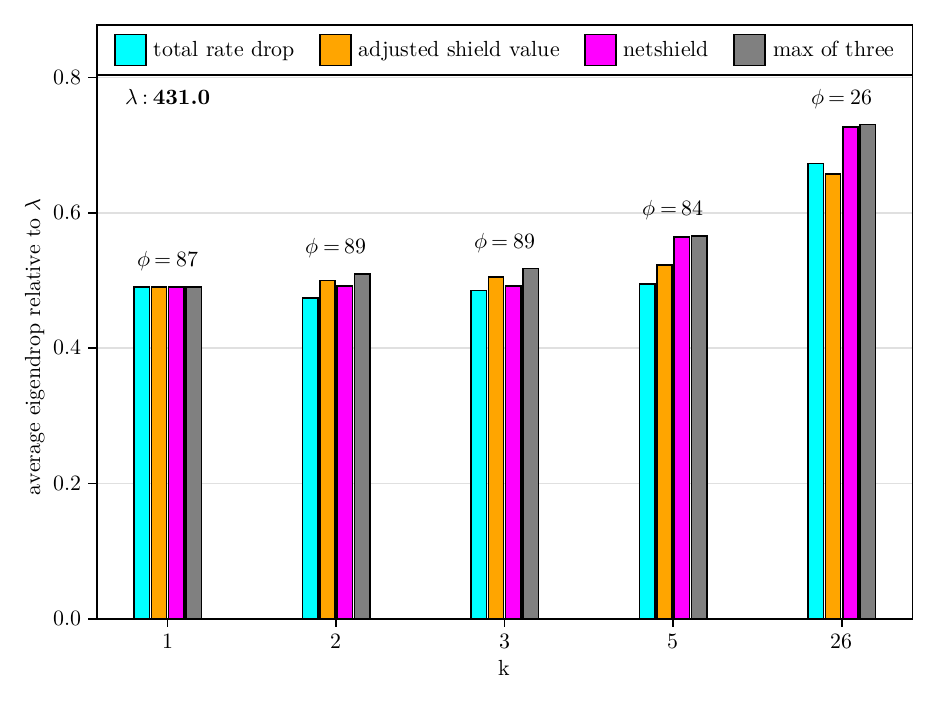}
        \caption{Average eigendrop}
    \end{subcaptionblock}
\end{figure}

\begin{figure}
    \centering
    \caption{Graph: ``conference 3'' (non-weighted)$ $}
    \begin{subcaptionblock}{0.45\textwidth}
        \includegraphics[width=\textwidth]{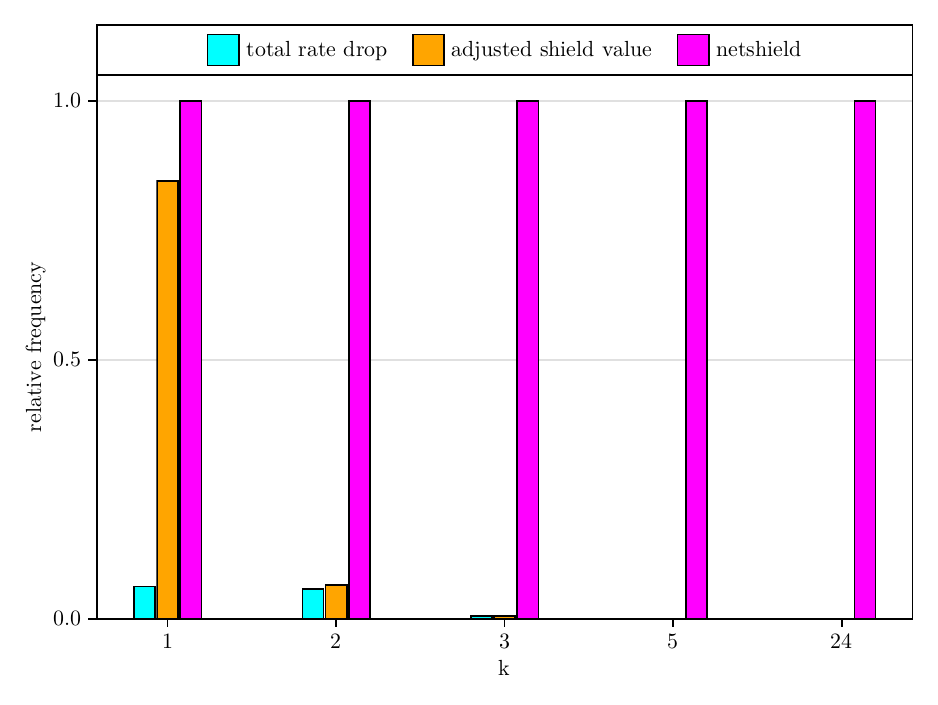}
        \caption{Origin of best eigendrop}
    \end{subcaptionblock}
    \begin{subcaptionblock}{0.45\textwidth}
        \includegraphics[width=\textwidth]{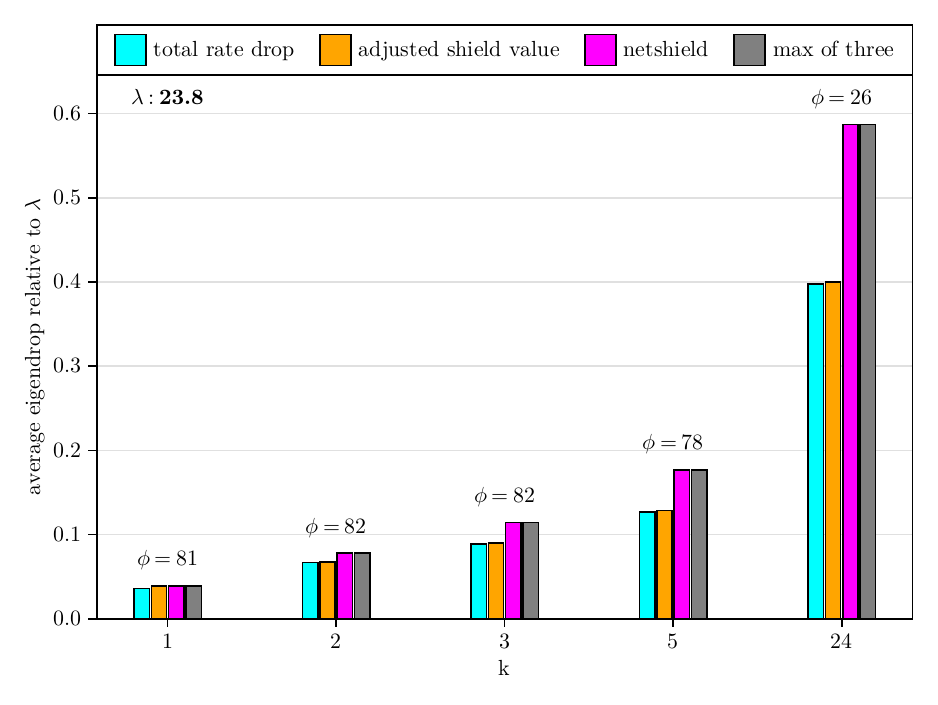}
        \caption{Average eigendrop}
    \end{subcaptionblock}
\end{figure}

\begin{figure}
    \centering
    \caption{Graph: ``conference 3'' (weighted)$ $}
    \begin{subcaptionblock}{0.45\textwidth}
        \includegraphics[width=\textwidth]{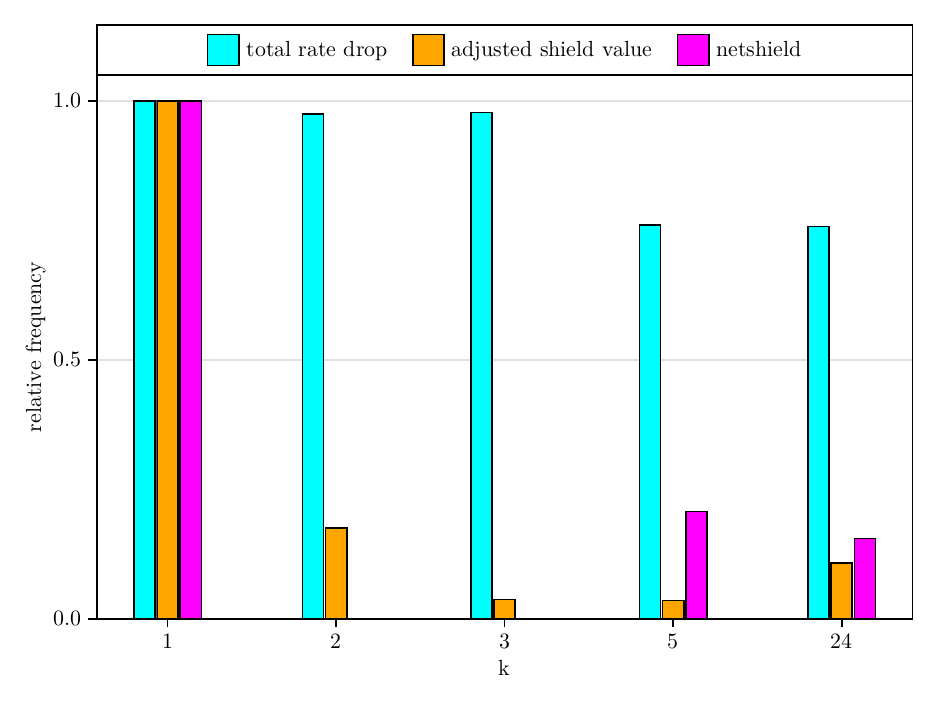}
        \caption{Origin of best eigendrop}
    \end{subcaptionblock}
    \begin{subcaptionblock}{0.45\textwidth}
        \includegraphics[width=\textwidth]{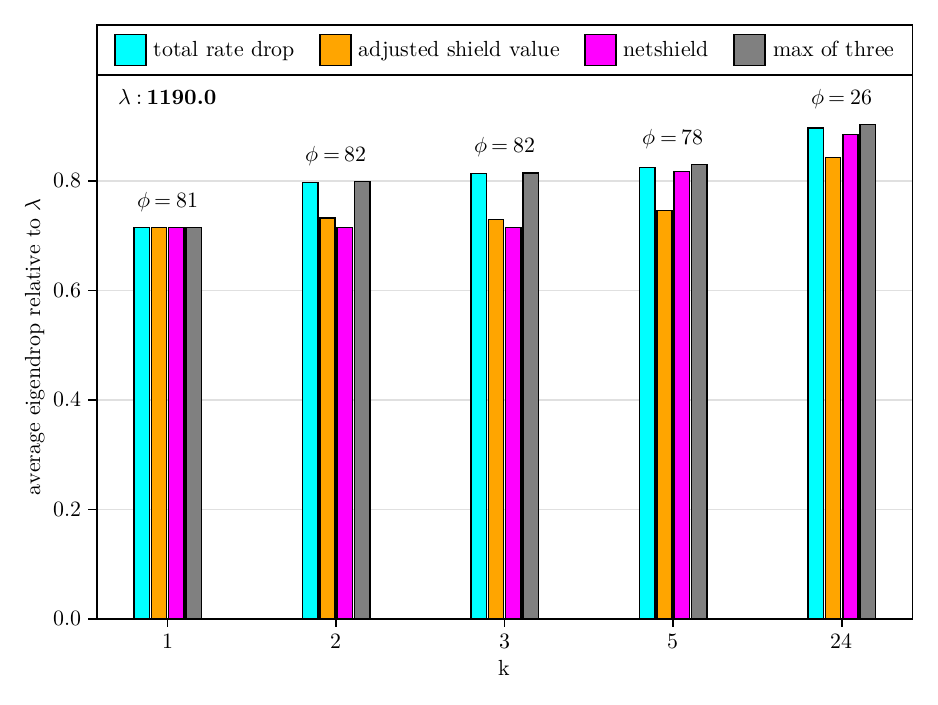}
        \caption{Average eigendrop}
    \end{subcaptionblock}
\end{figure}

\begin{figure}
    \centering
    \caption{Graph: ``airport 1'' (non-weighted)$ $}
    \begin{subcaptionblock}{0.45\textwidth}
        \includegraphics[width=\textwidth]{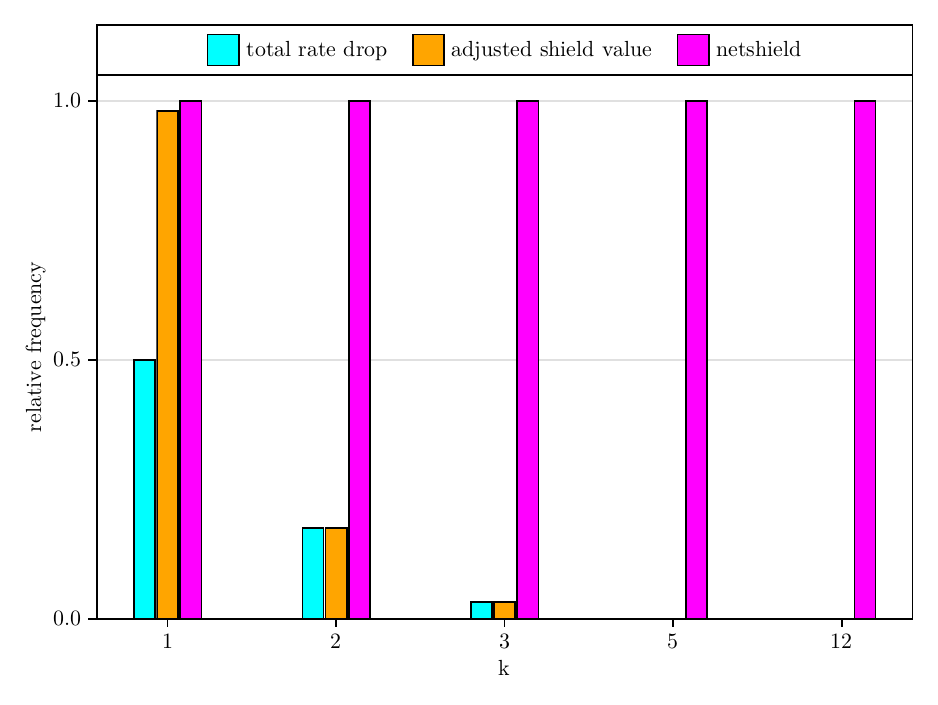}
        \caption{Origin of best eigendrop}
    \end{subcaptionblock}
    \begin{subcaptionblock}{0.45\textwidth}
        \includegraphics[width=\textwidth]{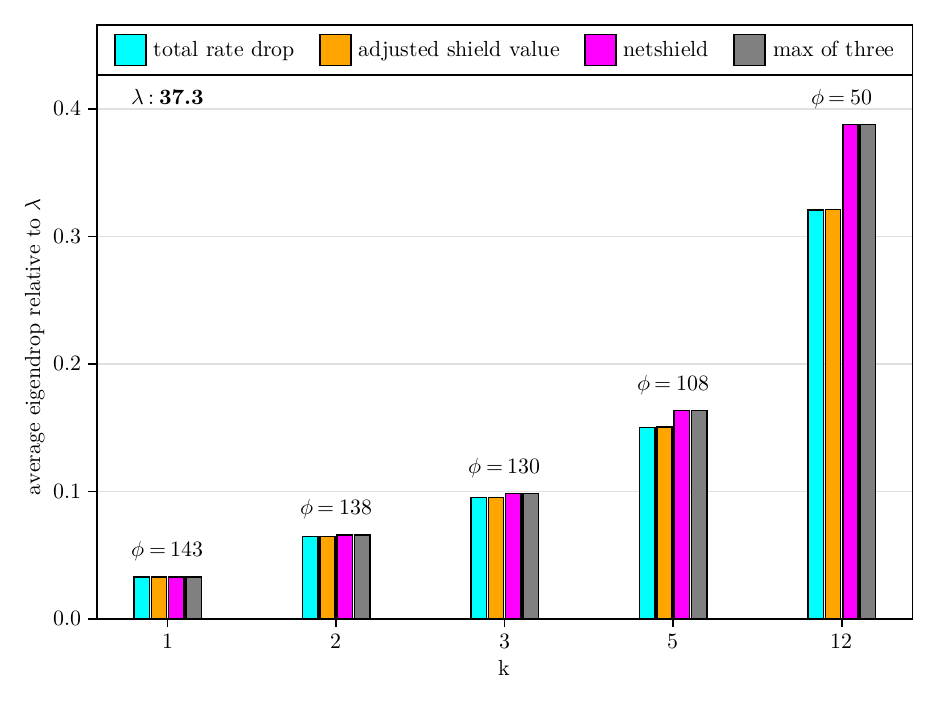}
        \caption{Average eigendrop}
    \end{subcaptionblock}
\end{figure}

\begin{figure}
    \centering
    \caption{Graph: ``airport 2'' (non-weighted)$ $}
    \begin{subcaptionblock}{0.45\textwidth}
        \includegraphics[width=\textwidth]{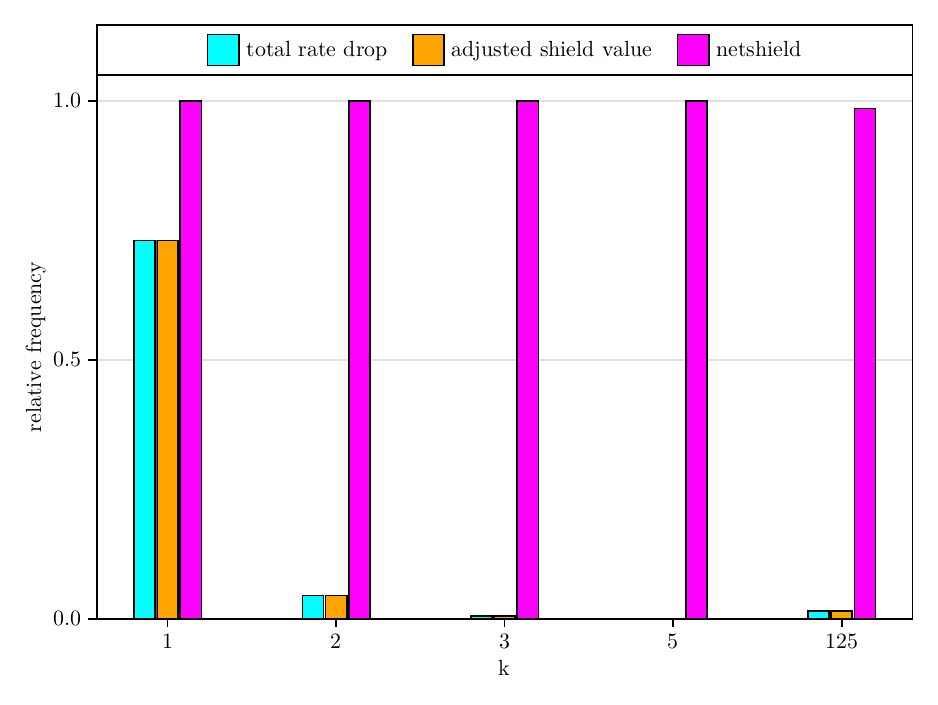}
        \caption{Origin of best eigendrop}
    \end{subcaptionblock}
    \begin{subcaptionblock}{0.45\textwidth}
        \includegraphics[width=\textwidth]{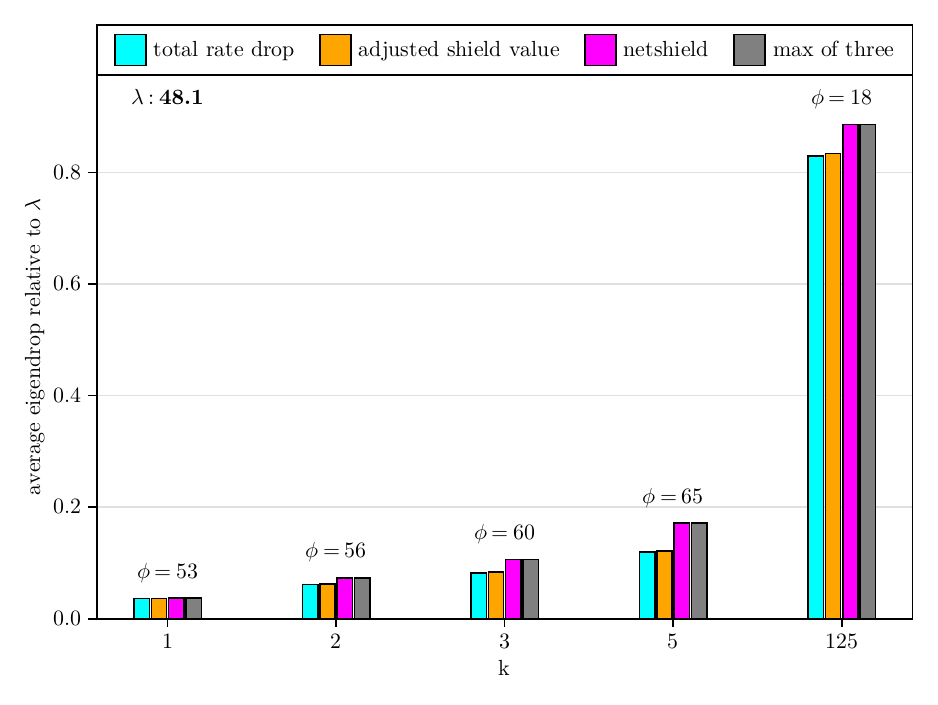}
        \caption{Average eigendrop}
    \end{subcaptionblock}
\end{figure}

\begin{figure}
    \centering
    \caption{Graph: ``airport 2'' (weighted)$ $}
    \begin{subcaptionblock}{0.45\textwidth}
        \includegraphics[width=\textwidth]{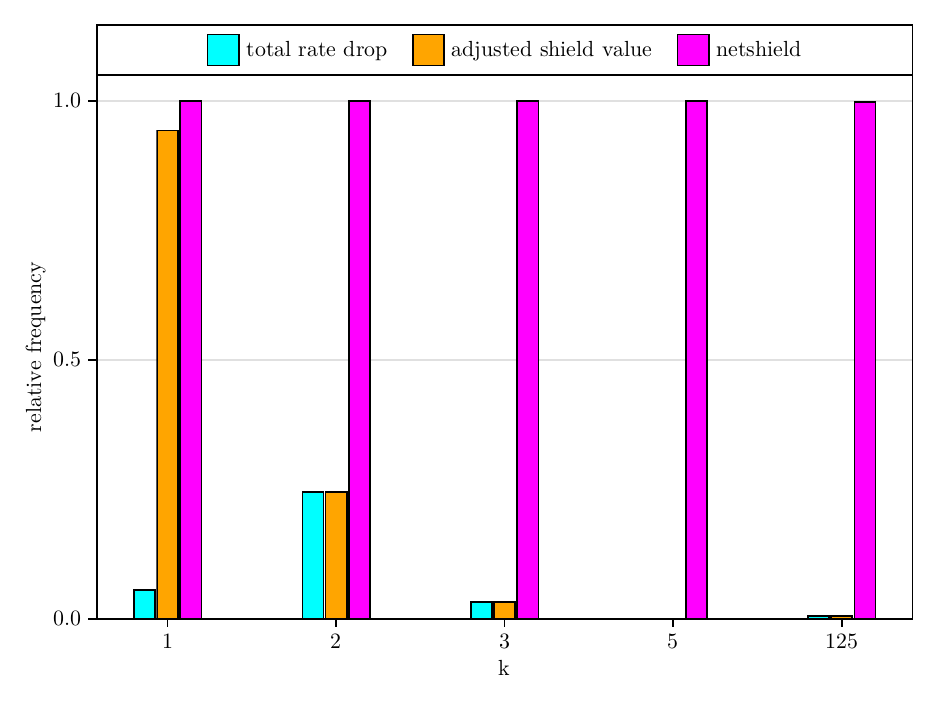}
        \caption{Origin of best eigendrop}
    \end{subcaptionblock}
    \begin{subcaptionblock}{0.45\textwidth}
        \includegraphics[width=\textwidth]{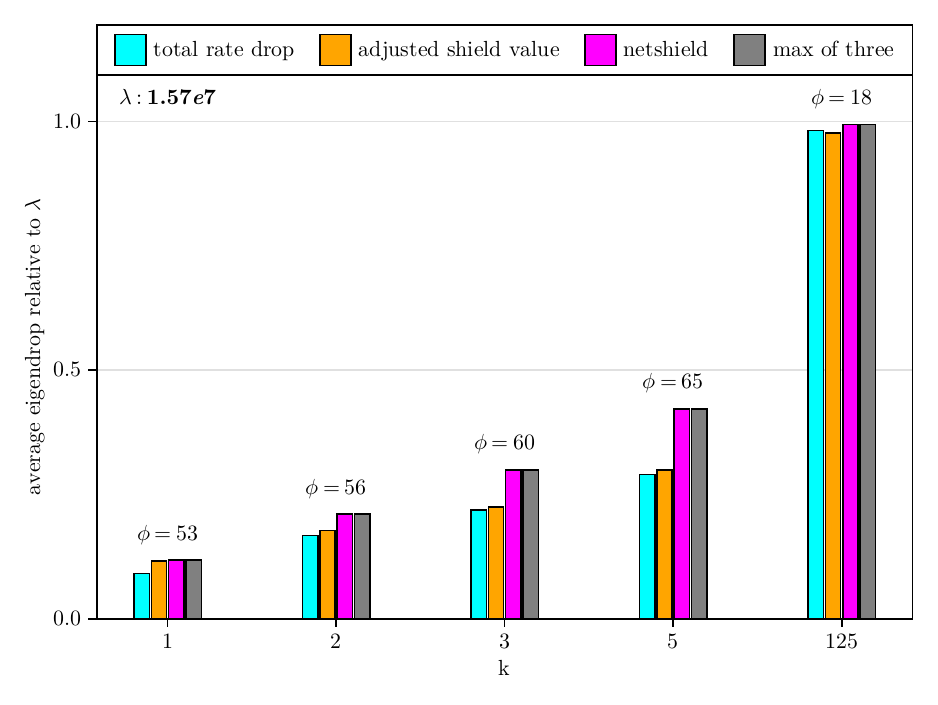}
        \caption{Average eigendrop}
    \end{subcaptionblock}
\end{figure}

\begin{figure}
    \centering
    \caption{Graph: ``airport 3'' (non-weighted)$ $}
    \begin{subcaptionblock}{0.45\textwidth}
        \includegraphics[width=\textwidth]{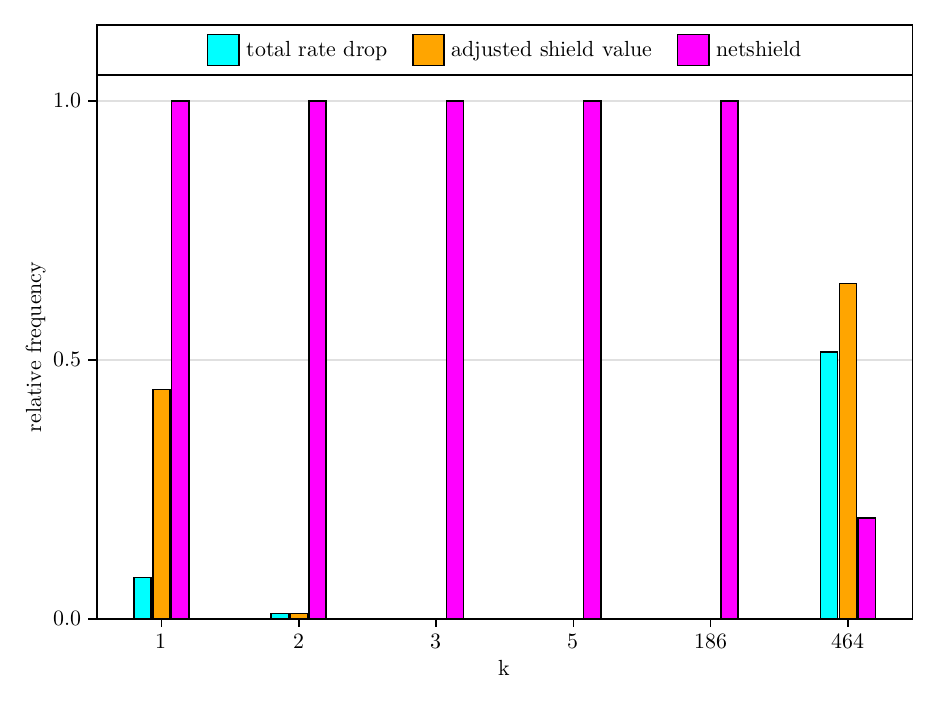}
        \caption{Origin of best eigendrop}
    \end{subcaptionblock}
    \begin{subcaptionblock}{0.45\textwidth}
        \includegraphics[width=\textwidth]{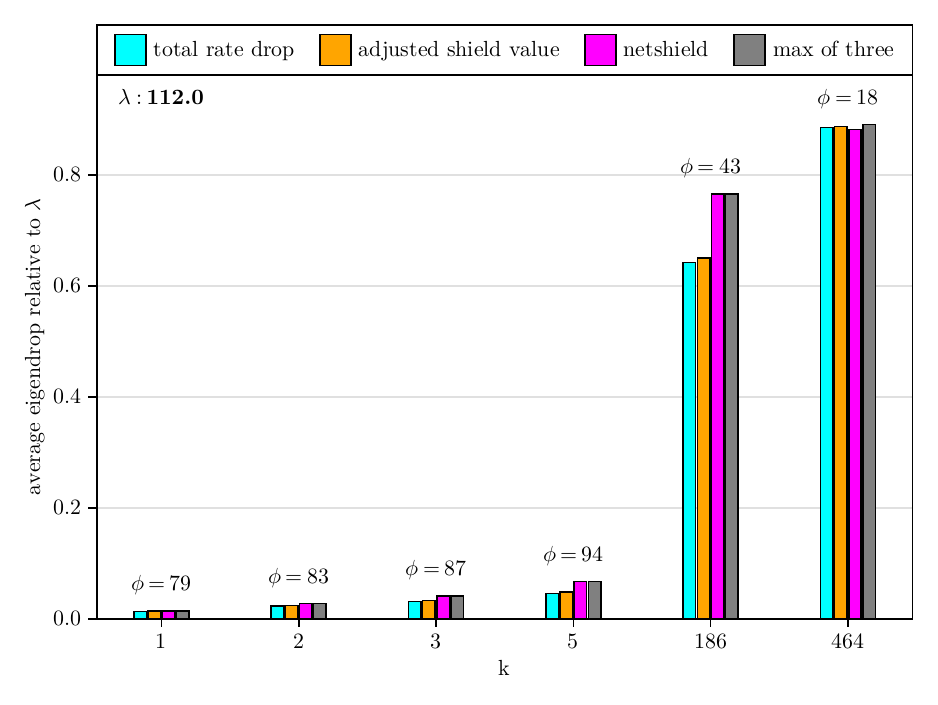}
        \caption{Average eigendrop}
    \end{subcaptionblock}
\end{figure}

\begin{figure}
    \centering
    \caption{Graph: ``airport 4'' (non-weighted)$ $}
    \begin{subcaptionblock}{0.45\textwidth}
        \includegraphics[width=\textwidth]{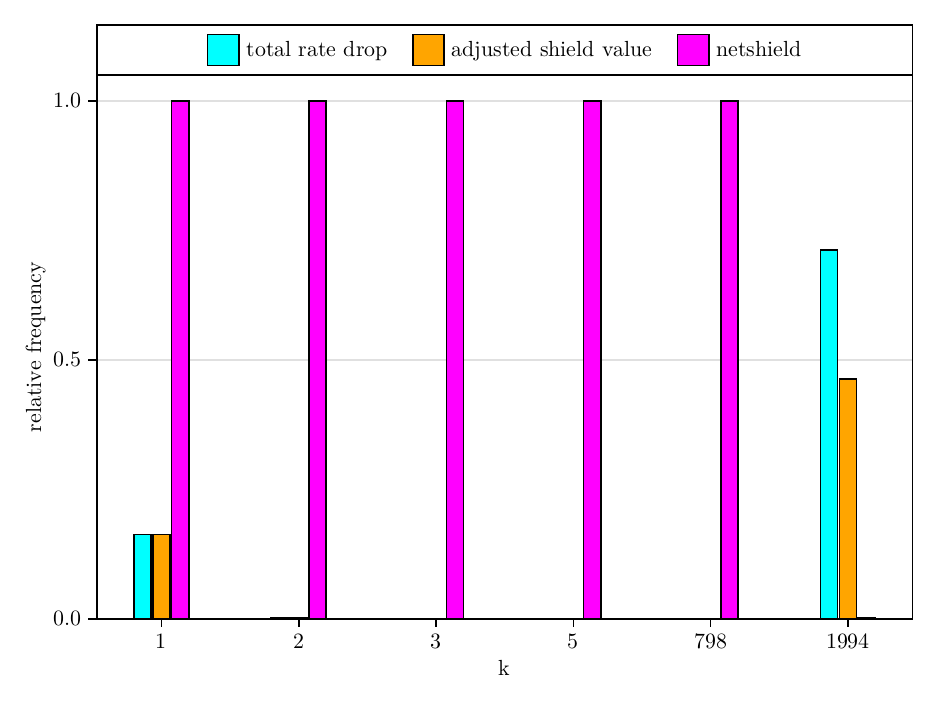}
        \caption{Origin of best eigendrop}
    \end{subcaptionblock}
    \begin{subcaptionblock}{0.45\textwidth}
        \includegraphics[width=\textwidth]{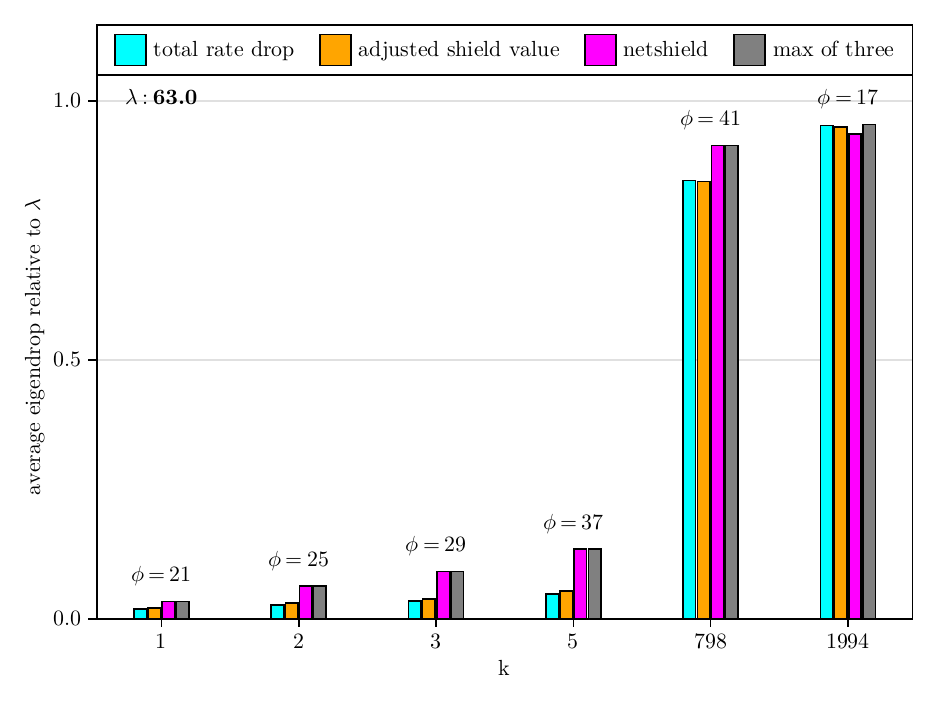}
        \caption{Average eigendrop}
    \end{subcaptionblock}
\end{figure}

\begin{figure}
    \centering
    \caption{Graph: ``airport 4'' (weighted)$ $}
    \begin{subcaptionblock}{0.45\textwidth}
        \includegraphics[width=\textwidth]{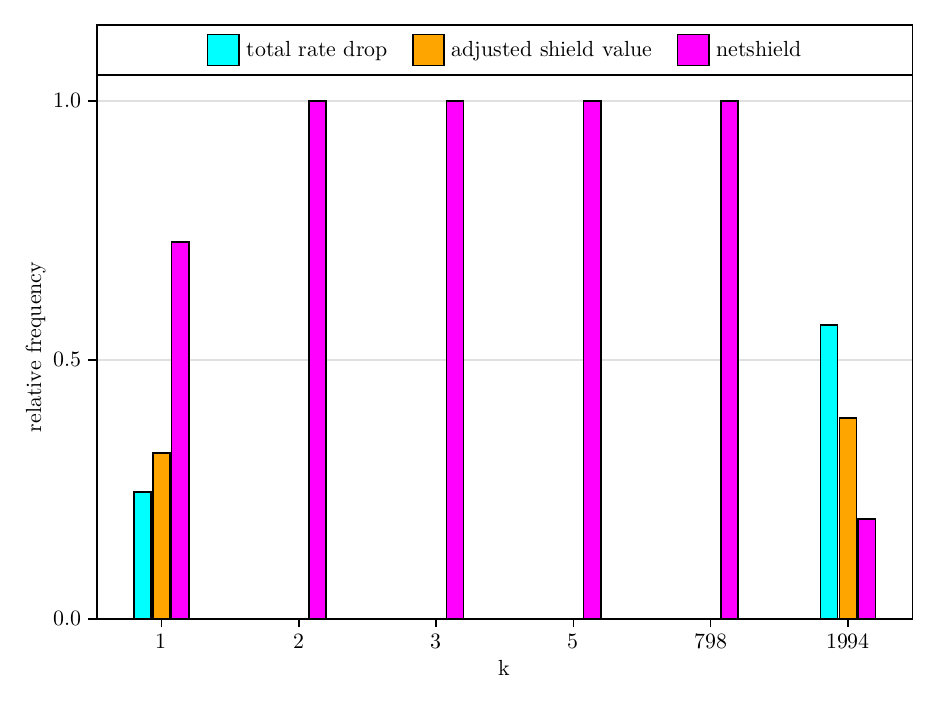}
        \caption{Origin of best eigendrop}
    \end{subcaptionblock}
    \begin{subcaptionblock}{0.45\textwidth}
        \includegraphics[width=\textwidth]{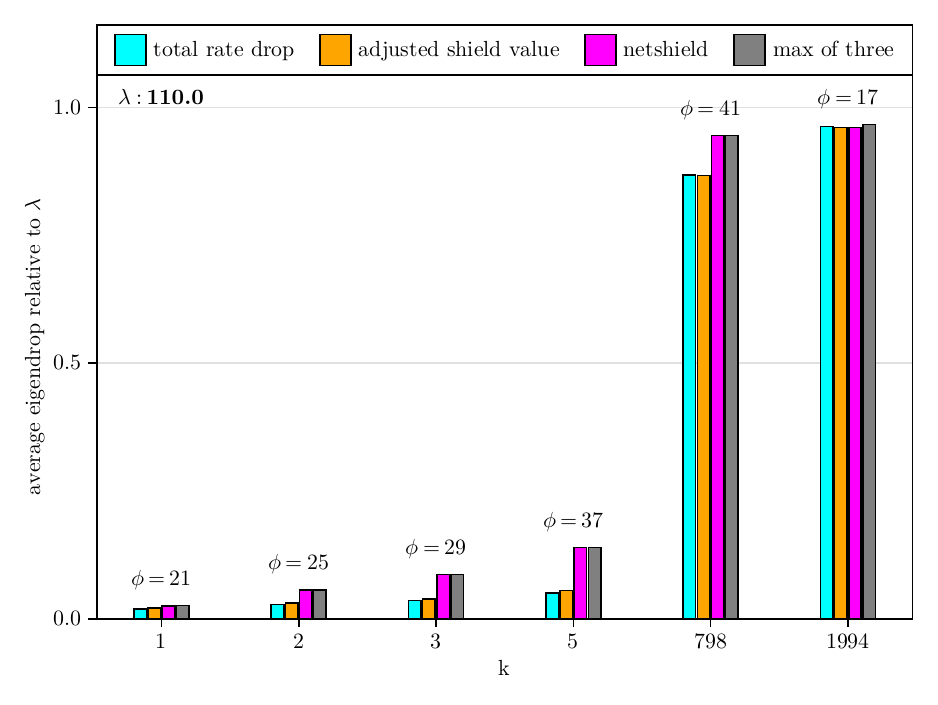}
        \caption{Average eigendrop}
    \end{subcaptionblock}
\end{figure}

\begin{figure}
    \centering
    \caption{Graph: ``rfid'' (non-weighted)$ $}
    \begin{subcaptionblock}{0.45\textwidth}
        \includegraphics[width=\textwidth]{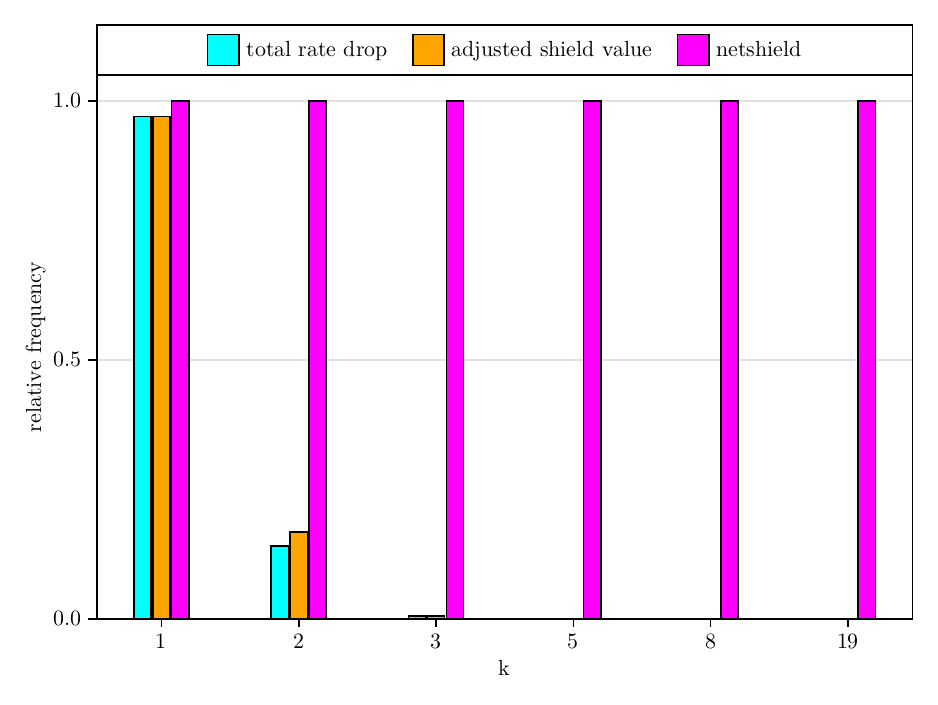}
        \caption{Origin of best eigendrop}
    \end{subcaptionblock}
    \begin{subcaptionblock}{0.45\textwidth}
        \includegraphics[width=\textwidth]{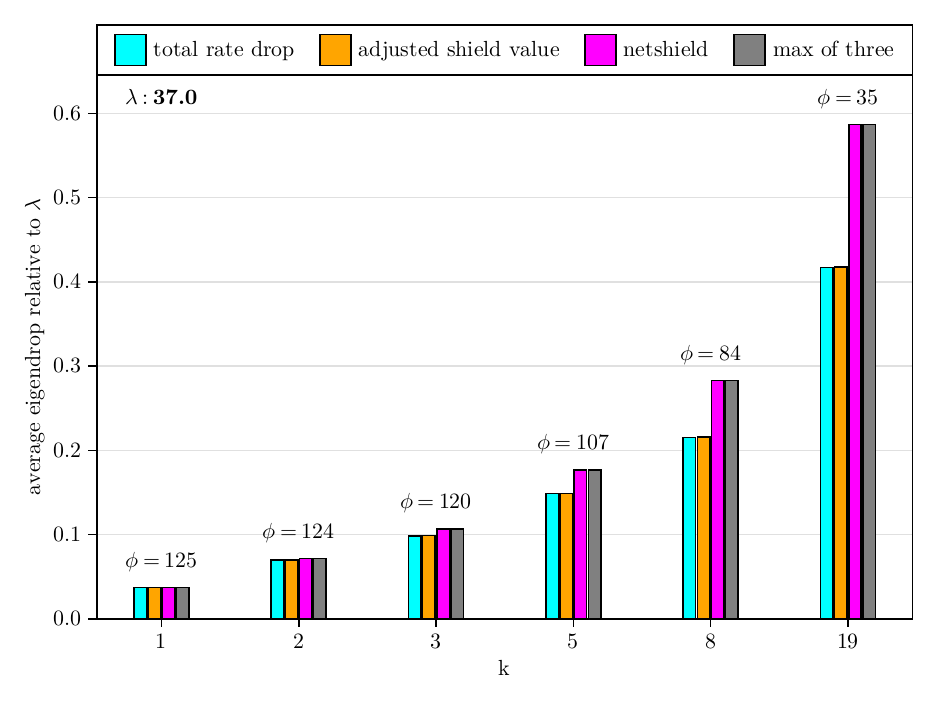}
        \caption{Average eigendrop}
    \end{subcaptionblock}
\end{figure}

\begin{figure}
    \centering
    \caption{Graph: ``rfid'' (weighted)$ $}
    \begin{subcaptionblock}{0.45\textwidth}
        \includegraphics[width=\textwidth]{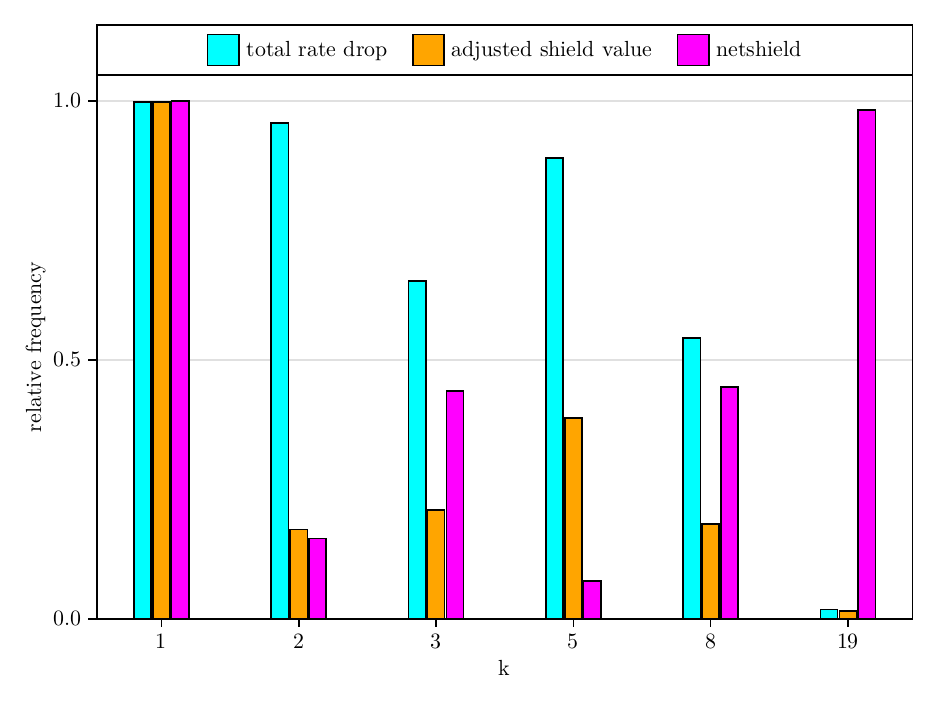}
        \caption{Origin of best eigendrop}
    \end{subcaptionblock}
    \begin{subcaptionblock}{0.45\textwidth}
        \includegraphics[width=\textwidth]{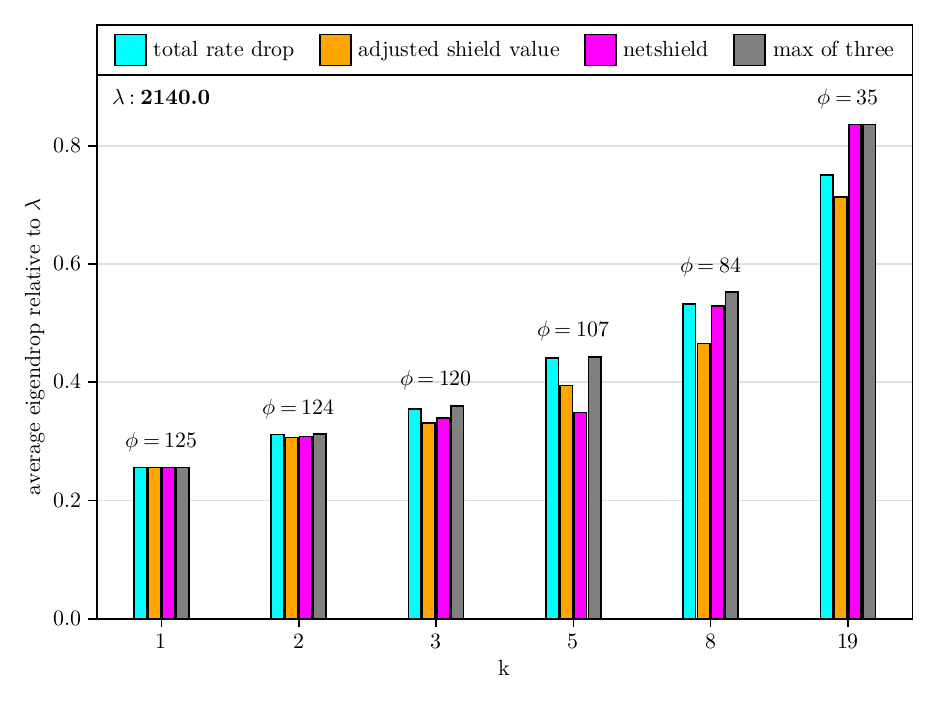}
        \caption{Average eigendrop}
    \end{subcaptionblock}
\end{figure}

\begin{figure}
    \centering
    \caption{Graph: ``UKfaculty'' (non-weighted)$ $}
    \begin{subcaptionblock}{0.45\textwidth}
        \includegraphics[width=\textwidth]{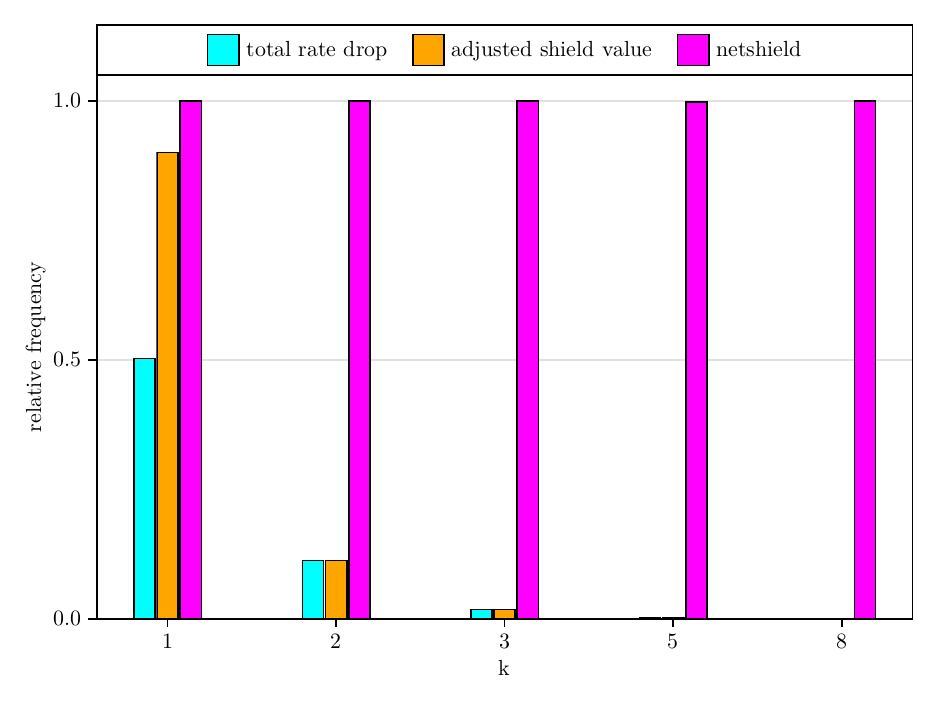}
        \caption{Origin of best eigendrop}
    \end{subcaptionblock}
    \begin{subcaptionblock}{0.45\textwidth}
        \includegraphics[width=\textwidth]{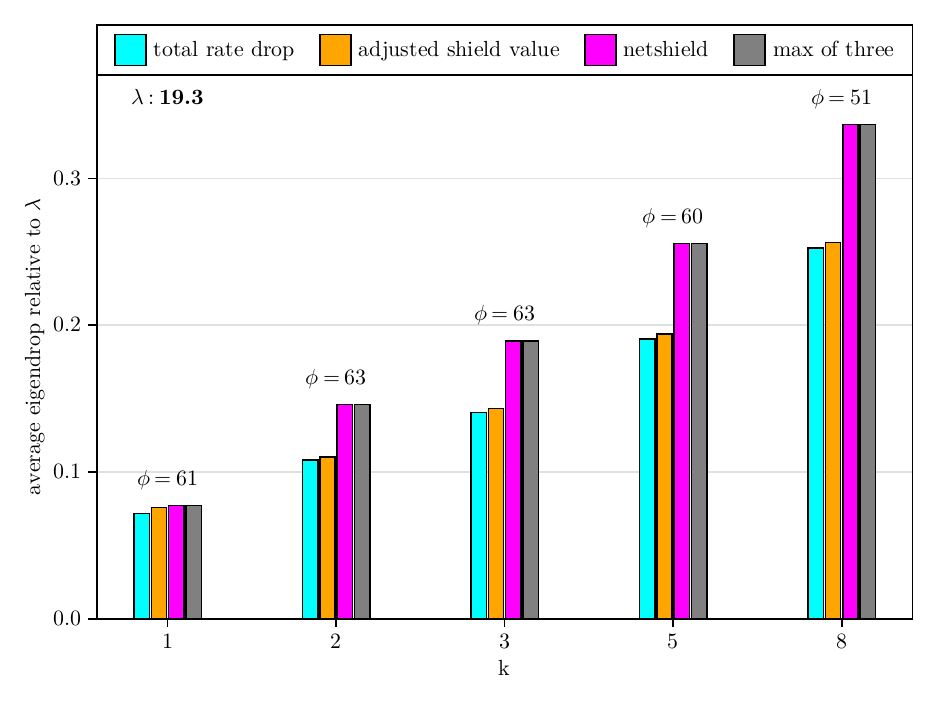}
        \caption{Average eigendrop}
    \end{subcaptionblock}
\end{figure}

\begin{figure}
    \centering
    \caption{Graph: ``UKfaculty'' (weighted)$ $}
    \begin{subcaptionblock}{0.45\textwidth}
        \includegraphics[width=\textwidth]{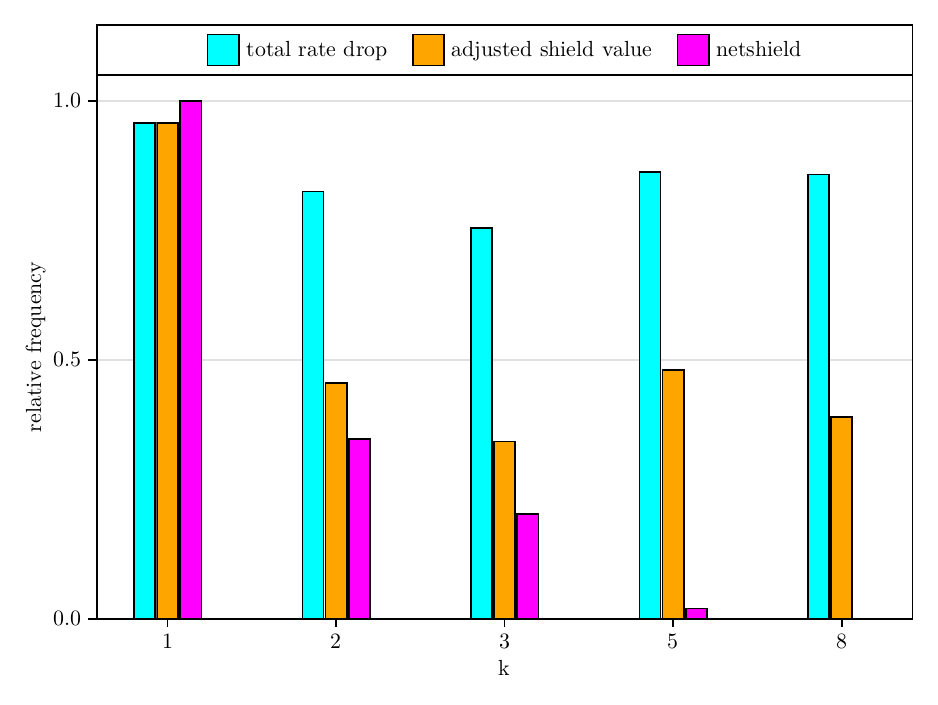}
        \caption{Origin of best eigendrop}
    \end{subcaptionblock}
    \begin{subcaptionblock}{0.45\textwidth}
        \includegraphics[width=\textwidth]{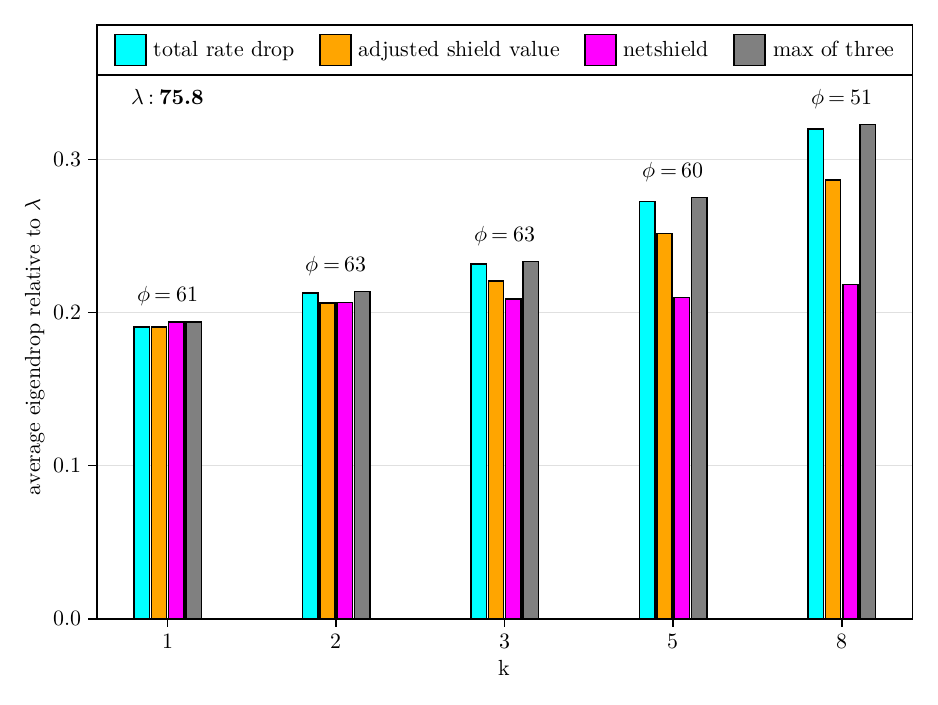}
        \caption{Average eigendrop}
RAN-eigendrop-origin    \end{subcaptionblock}
\end{figure}

\begin{figure}
    \centering
    \caption{Graph: ``enron'' (non-weighted)$ $}
    \begin{subcaptionblock}{0.45\textwidth}
        \includegraphics[width=\textwidth]{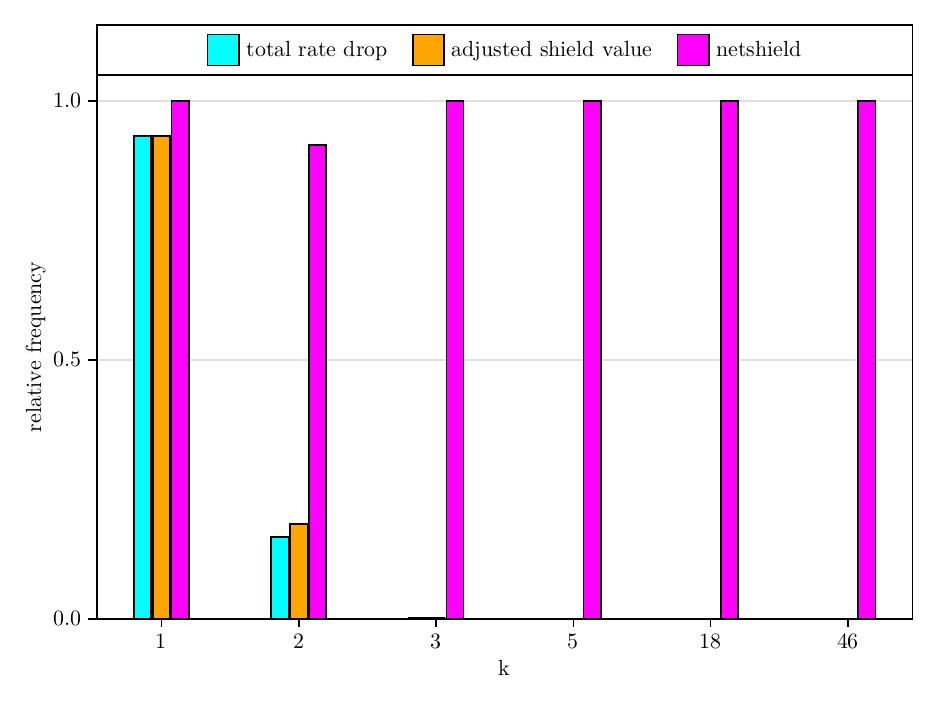}
        \caption{Origin of best eigendrop}
    \end{subcaptionblock}
    \begin{subcaptionblock}{0.45\textwidth}
        \includegraphics[width=\textwidth]{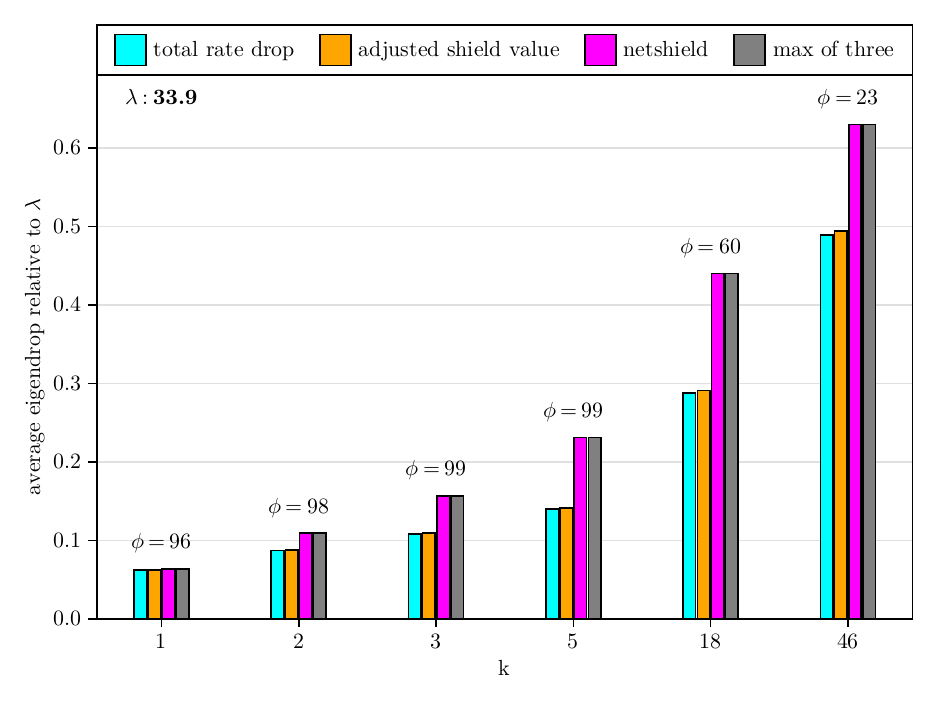}
        \caption{Average eigendrop}
    \end{subcaptionblock}
\end{figure}

\subsection{Eigendrop distribution}\label{sec:figs-eigendropDistribution}

\begin{figure}
    \centering $ $
    \caption{Graph: ``karate club'' (non-weighted) - eigendrop distribution}
      \begin{subcaptionblock}{0.45\textwidth}
          \includegraphics[width=\textwidth]{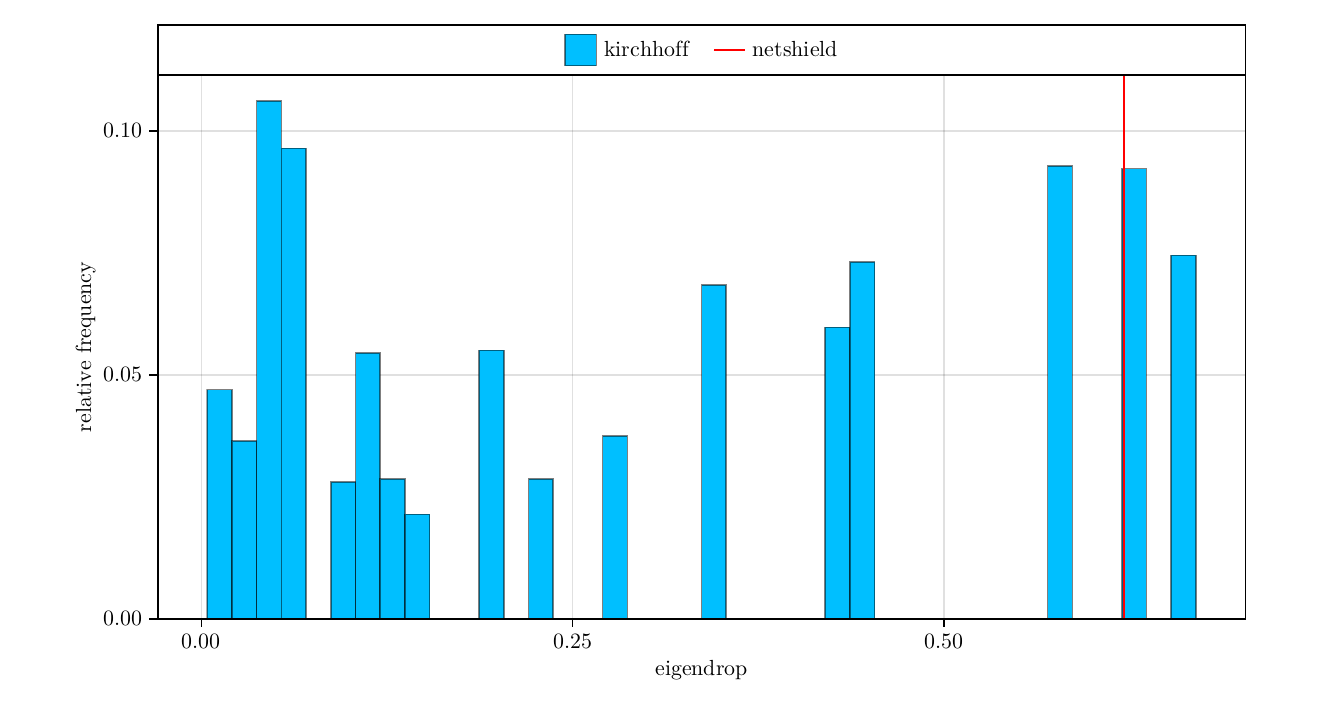}
          \caption{$k = 1$}
      \end{subcaptionblock}
      \begin{subcaptionblock}{0.45\textwidth}
          \includegraphics[width=\textwidth]{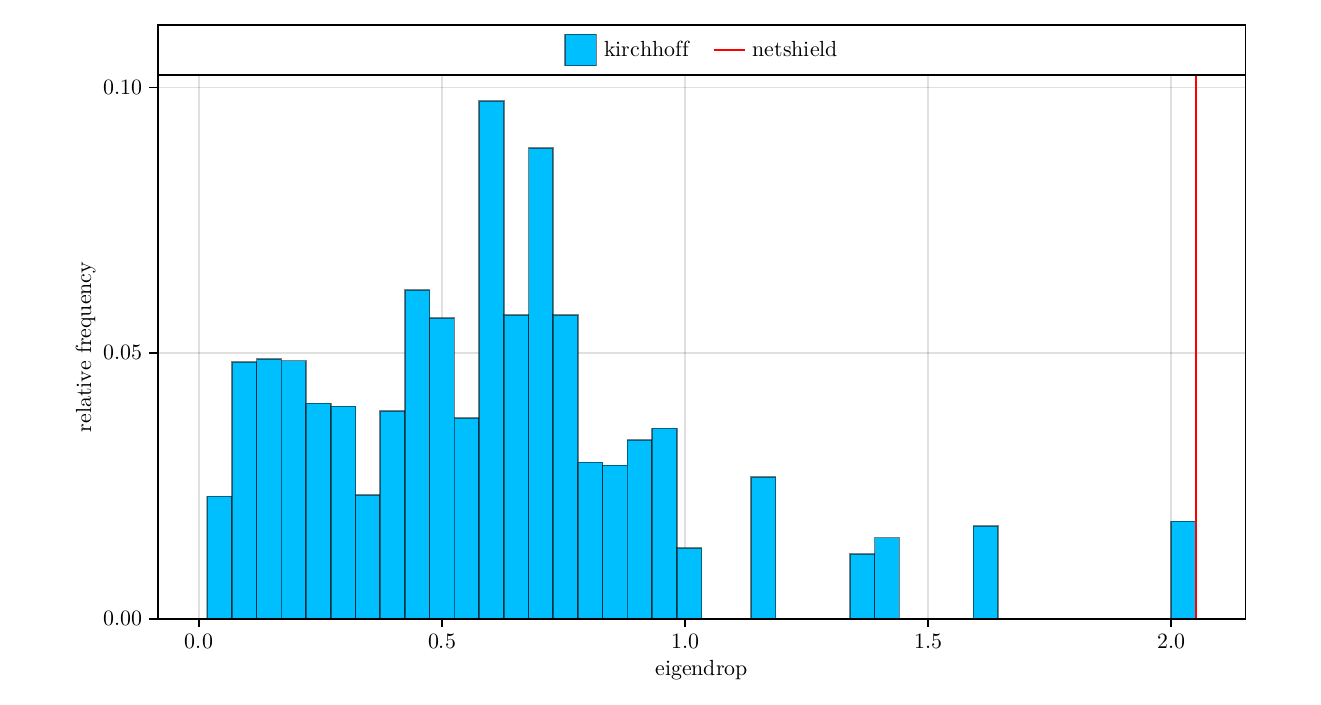}
          \caption{$k = 2$}
      \end{subcaptionblock}
      \\
      \begin{subcaptionblock}{0.45\textwidth}
          \includegraphics[width=\textwidth]{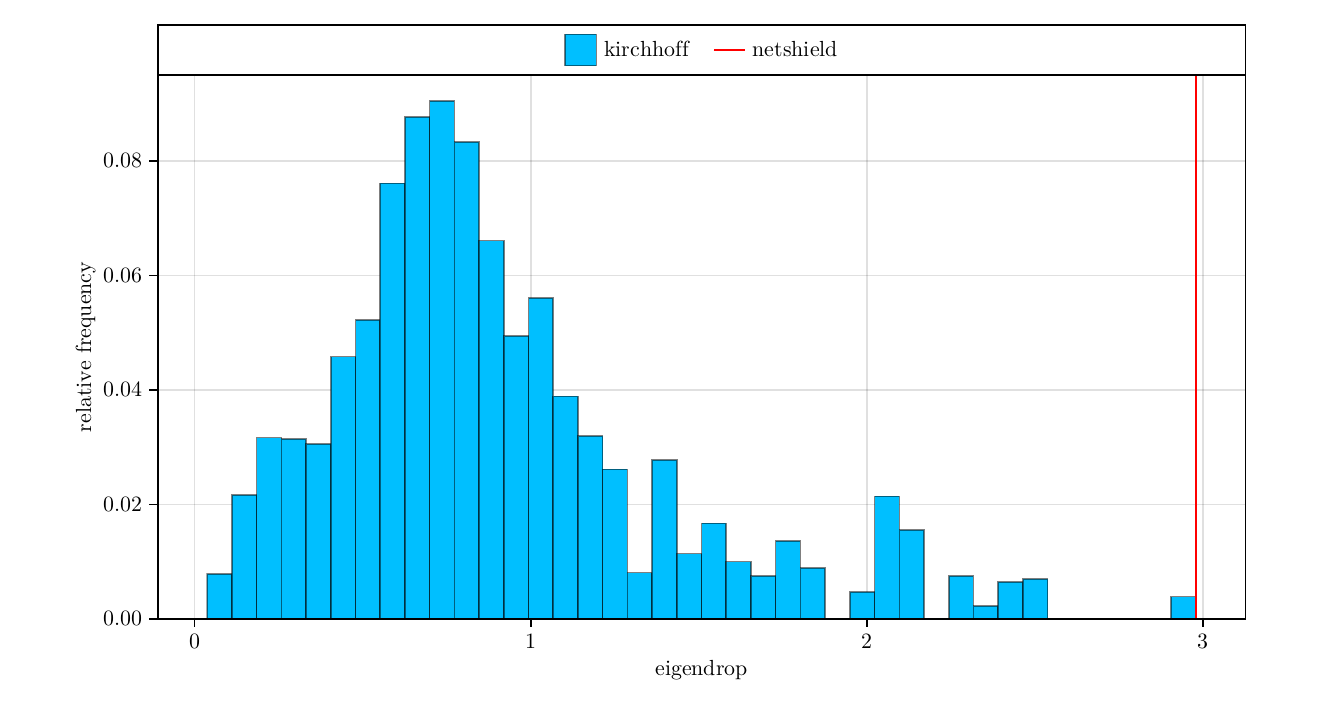}
          \caption{$k = 3$}
      \end{subcaptionblock}
      \begin{subcaptionblock}{0.45\textwidth}
          \includegraphics[width=\textwidth]{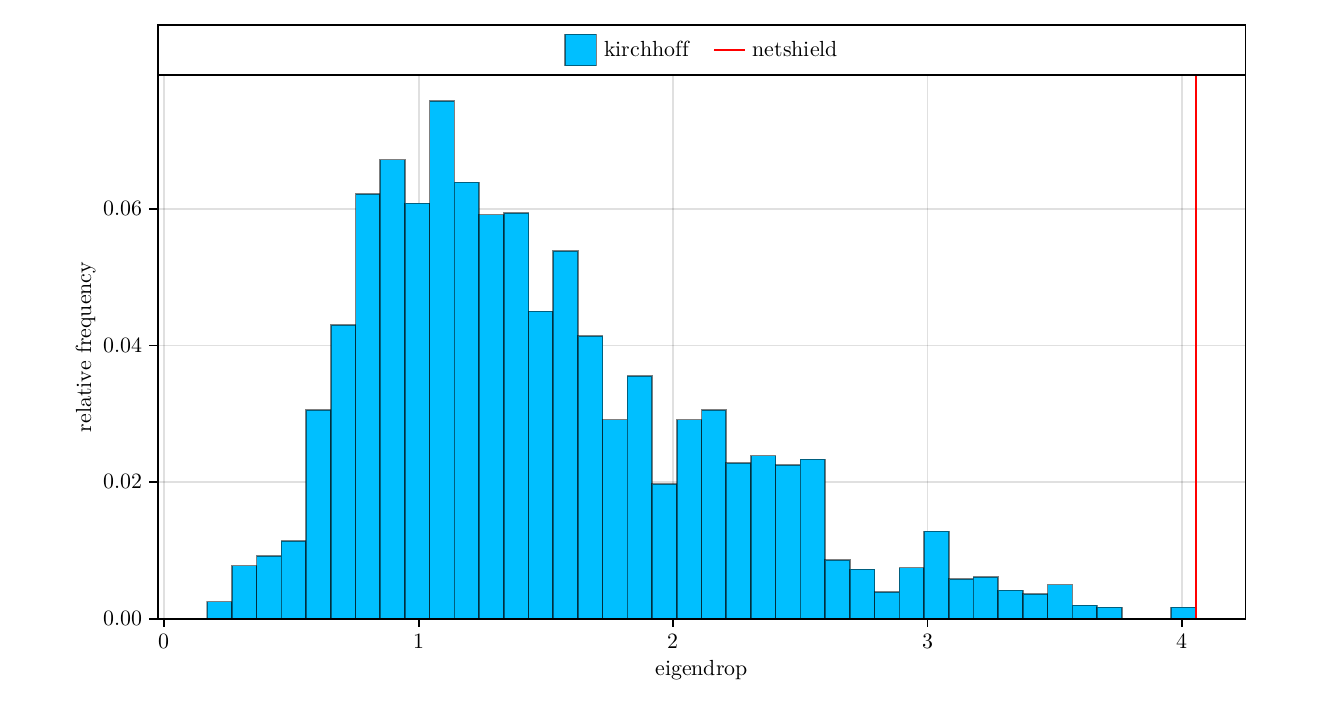}
          \caption{$k = 5$}
      \end{subcaptionblock}
      \\
      \begin{subcaptionblock}{0.45\textwidth}
          \includegraphics[width=\textwidth]{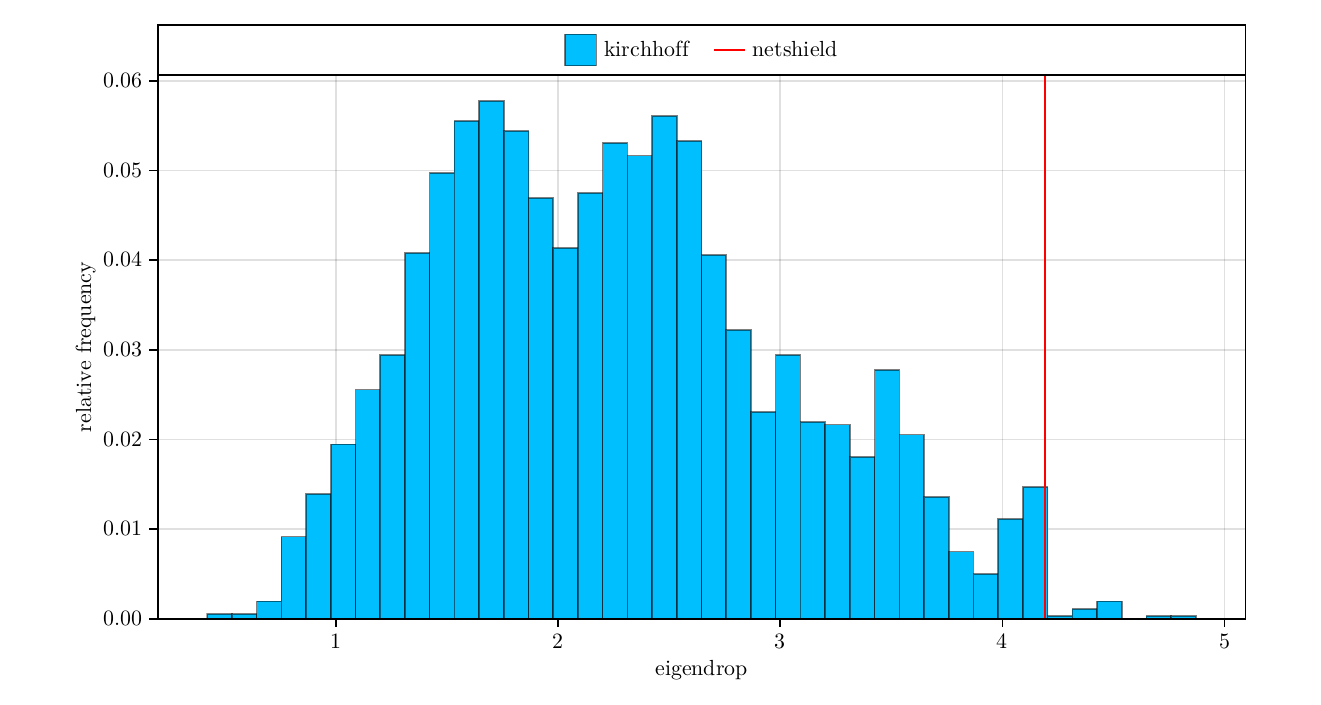}
          \caption{$k = 8$}
      \end{subcaptionblock}
\end{figure}

\begin{figure}
    \centering $ $
    \caption{Graph: ``conference 1'' (non-weighted) - eigendrop distribution}
      \begin{subcaptionblock}{0.45\textwidth}
          \includegraphics[width=\textwidth]{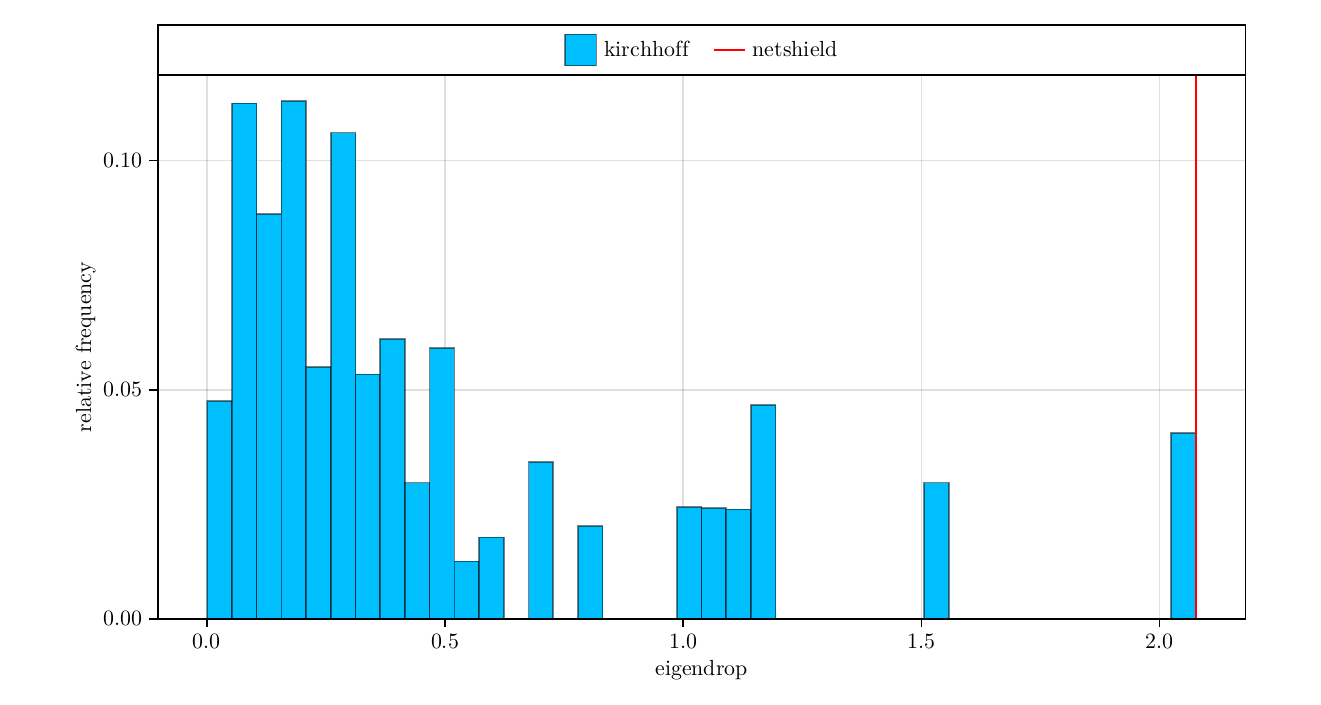}
          \caption{$k = 1$}
      \end{subcaptionblock}
      \begin{subcaptionblock}{0.45\textwidth}
          \includegraphics[width=\textwidth]{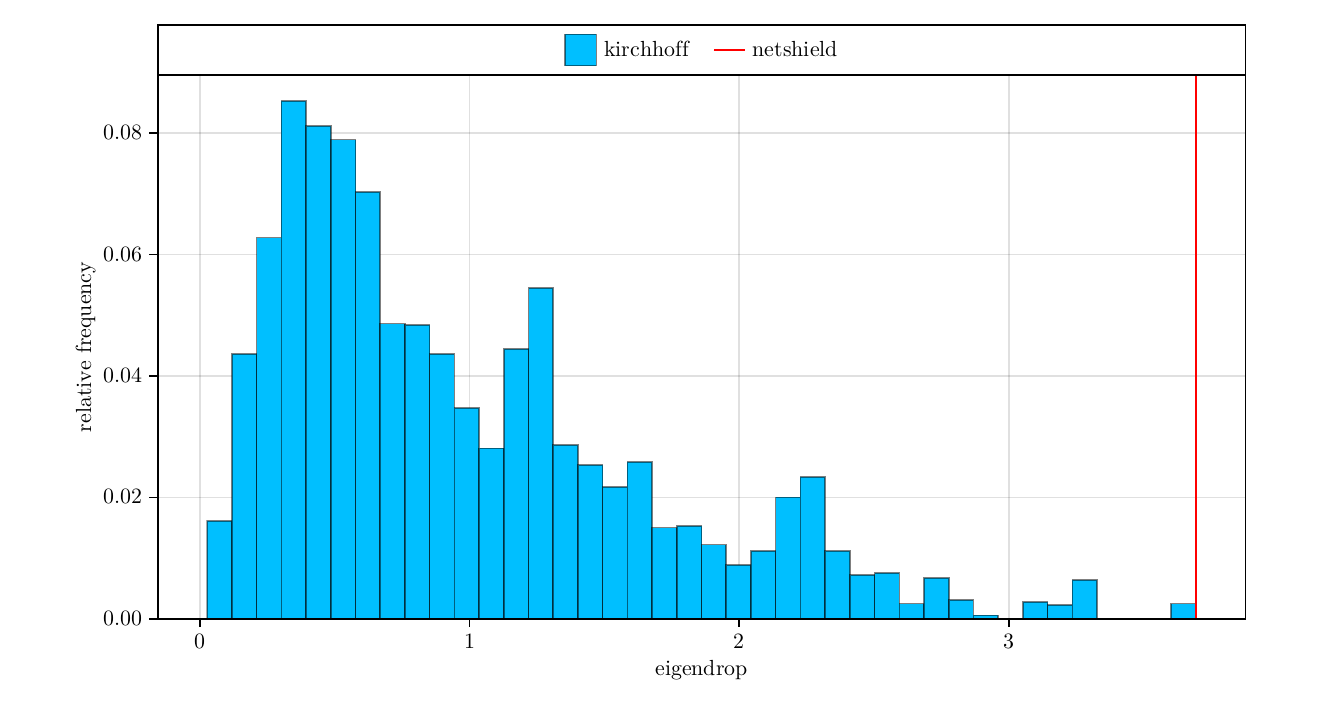}
          \caption{$k = 2$}
      \end{subcaptionblock}
      \\
      \begin{subcaptionblock}{0.45\textwidth}
          \includegraphics[width=\textwidth]{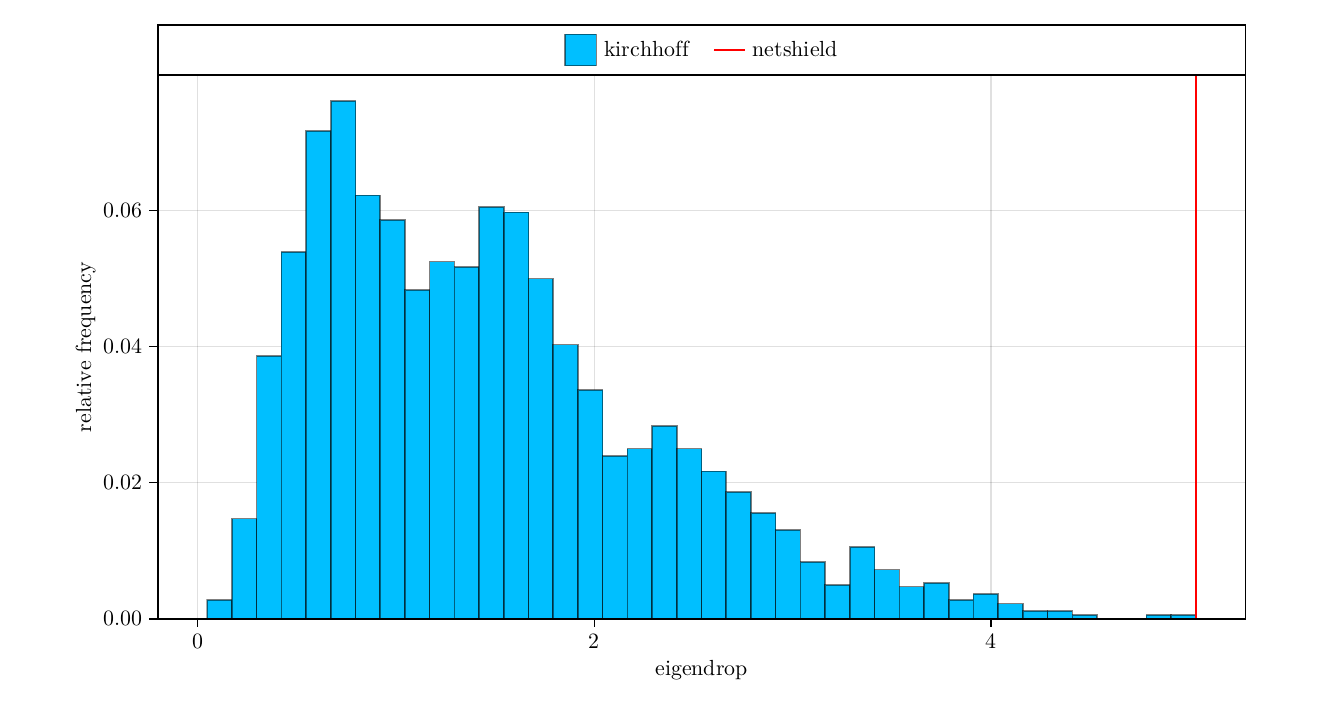}
          \caption{$k = 3$}
      \end{subcaptionblock}
      \begin{subcaptionblock}{0.45\textwidth}
          \includegraphics[width=\textwidth]{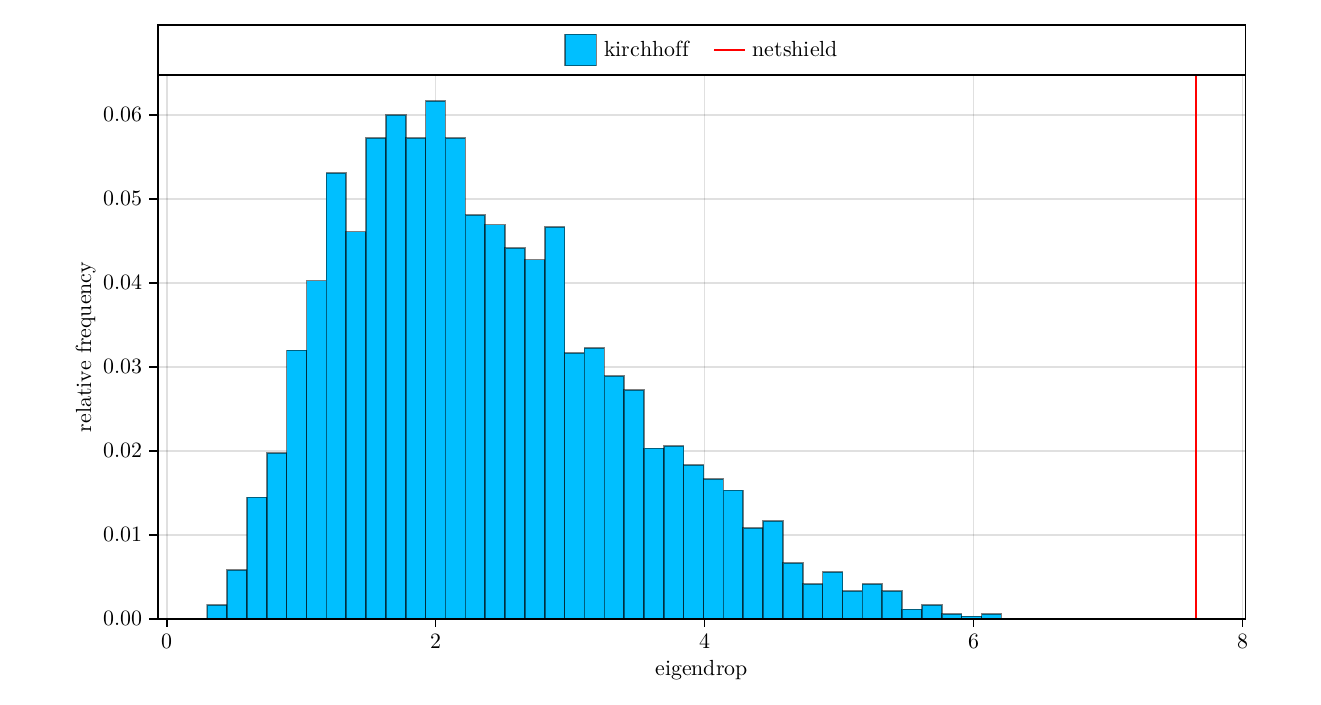}
          \caption{$k = 5$}
      \end{subcaptionblock}
      \\
      \begin{subcaptionblock}{0.45\textwidth}
          \includegraphics[width=\textwidth]{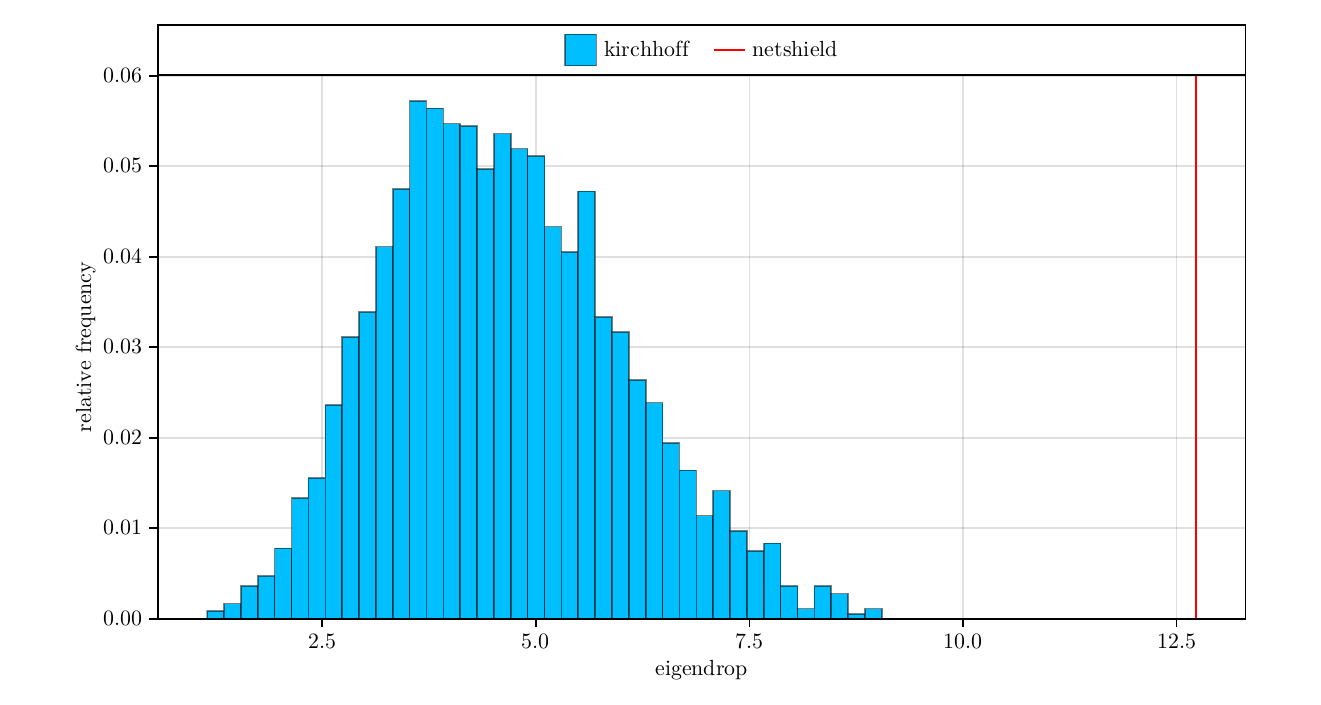}
          \caption{$k = 10$}
      \end{subcaptionblock}
      \begin{subcaptionblock}{0.45\textwidth}
          \includegraphics[width=\textwidth]{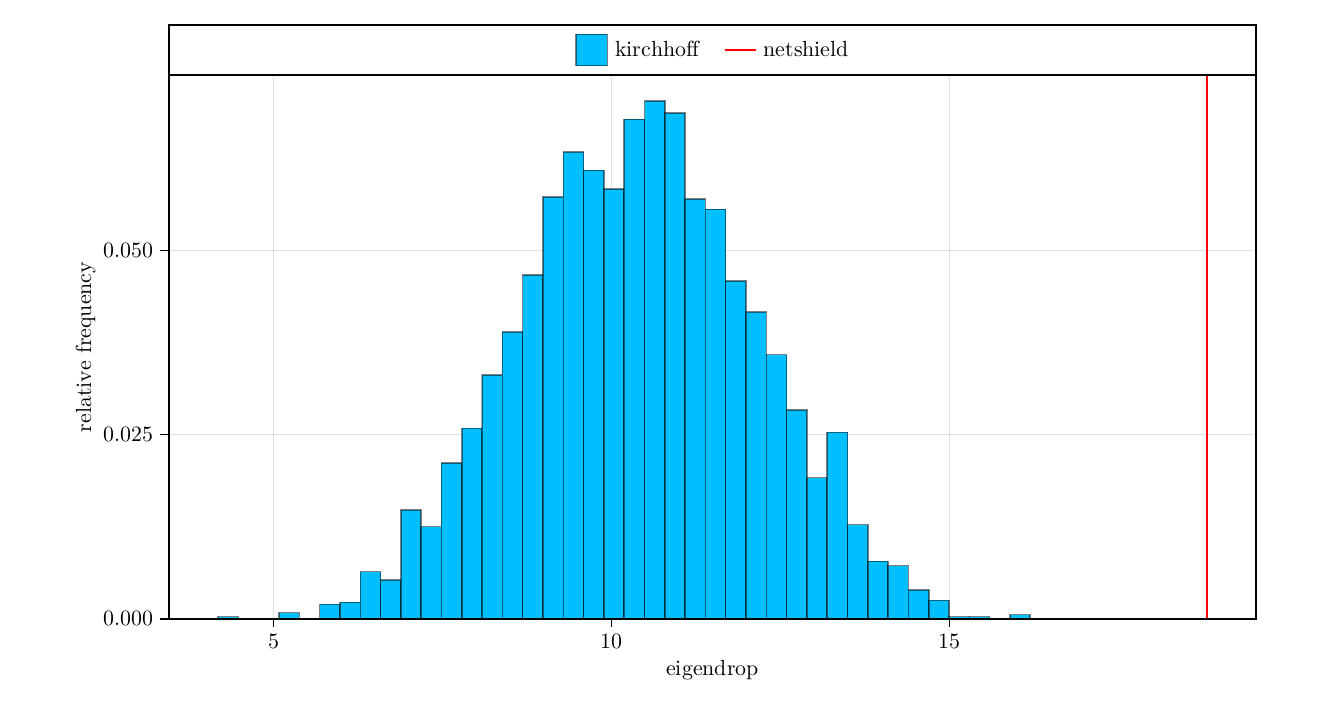}
          \caption{$k = 25$}
      \end{subcaptionblock}
\end{figure}

\begin{figure}
    \centering $ $
    \caption{Graph: ``conference 1'' (weighted) - eigendrop distribution}
      \begin{subcaptionblock}{0.45\textwidth}
          \includegraphics[width=\textwidth]{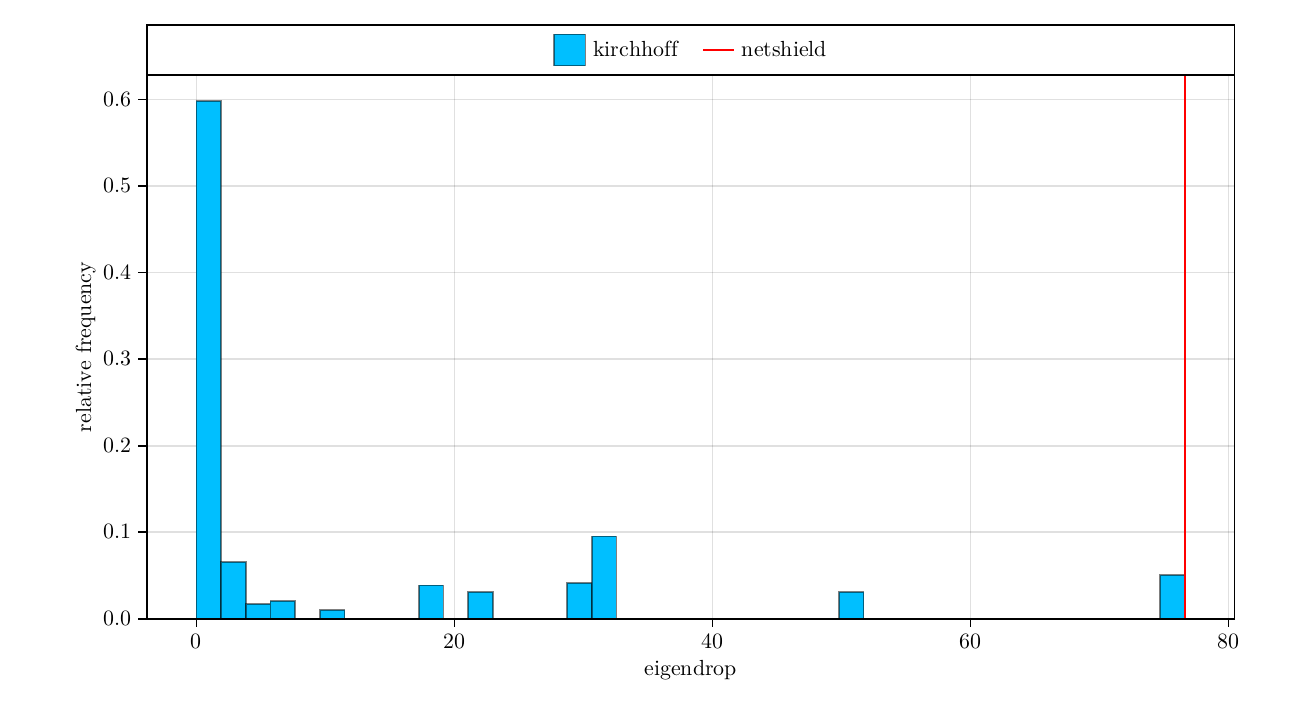}
          \caption{$k = 1$}
      \end{subcaptionblock}
      \begin{subcaptionblock}{0.45\textwidth}
          \includegraphics[width=\textwidth]{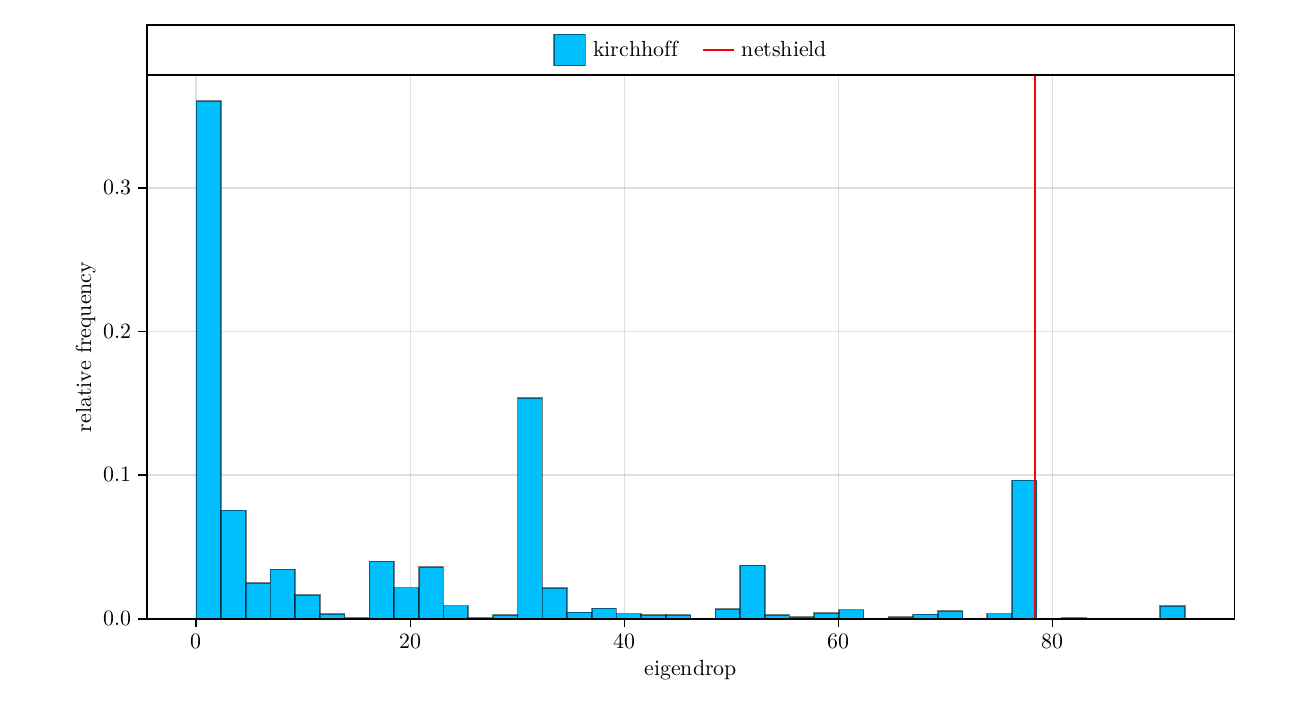}
          \caption{$k = 2$}
      \end{subcaptionblock}
      \\
      \begin{subcaptionblock}{0.45\textwidth}
          \includegraphics[width=\textwidth]{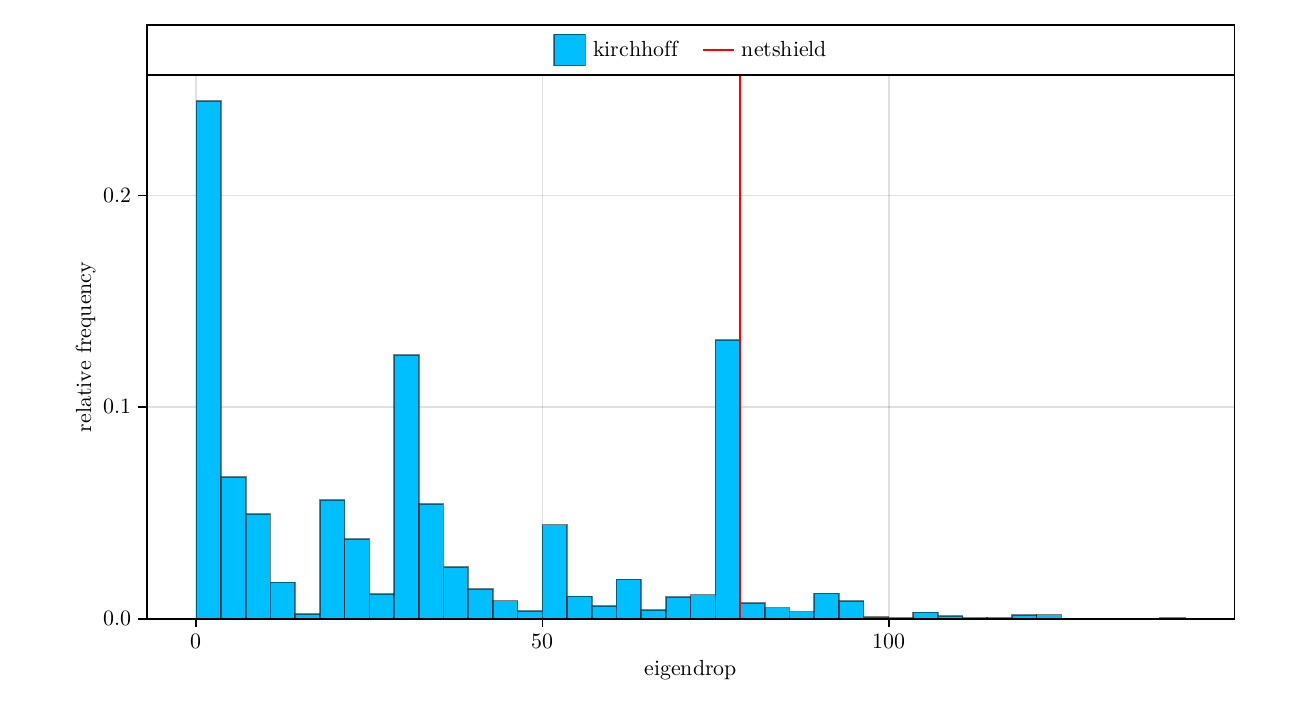}
          \caption{$k = 3$}
      \end{subcaptionblock}
      \begin{subcaptionblock}{0.45\textwidth}
          \includegraphics[width=\textwidth]{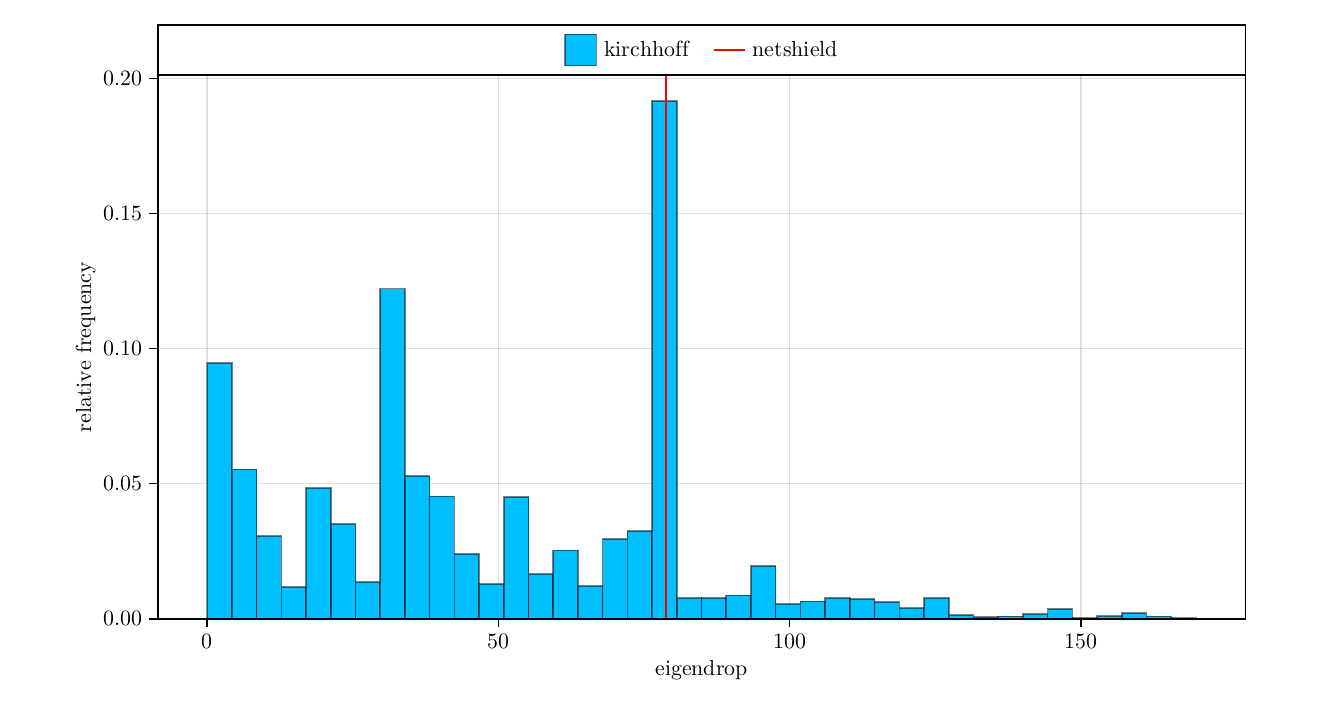}
          \caption{$k = 5$}
      \end{subcaptionblock}
      \\
      \begin{subcaptionblock}{0.45\textwidth}
          \includegraphics[width=\textwidth]{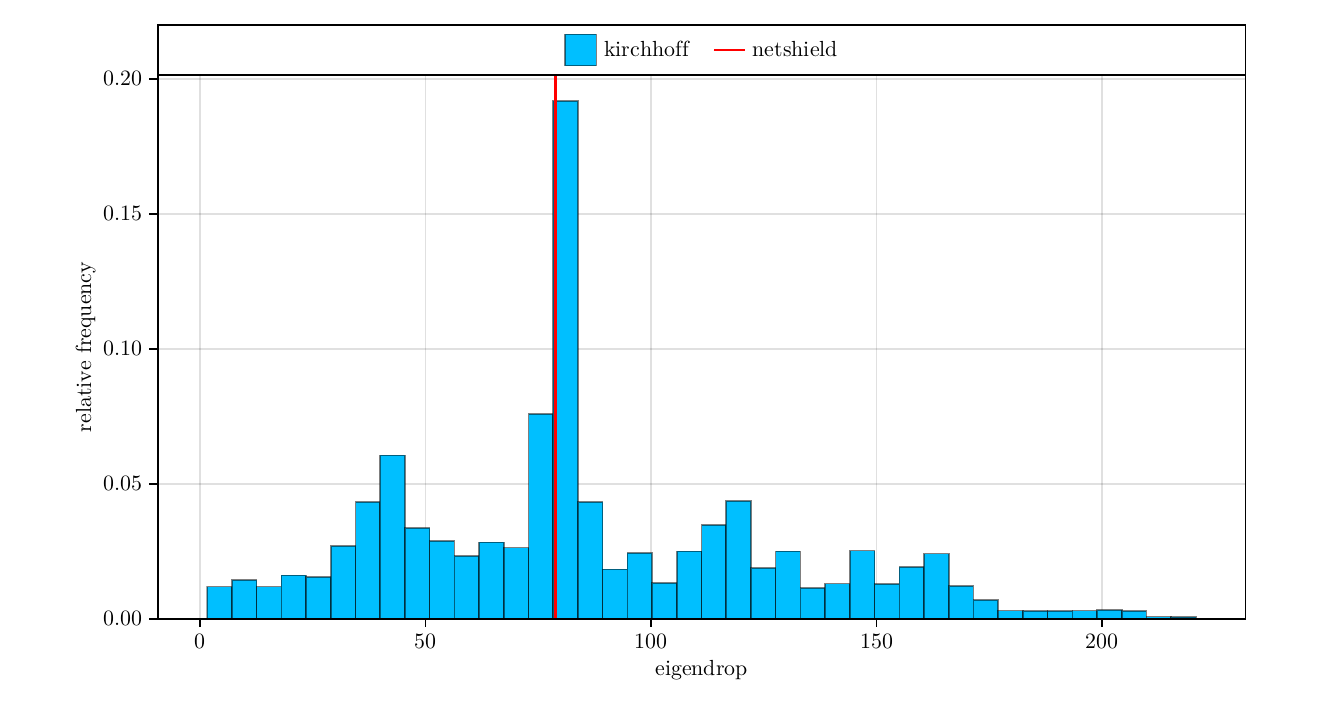}
          \caption{$k = 10$}
      \end{subcaptionblock}
      \begin{subcaptionblock}{0.45\textwidth}
          \includegraphics[width=\textwidth]{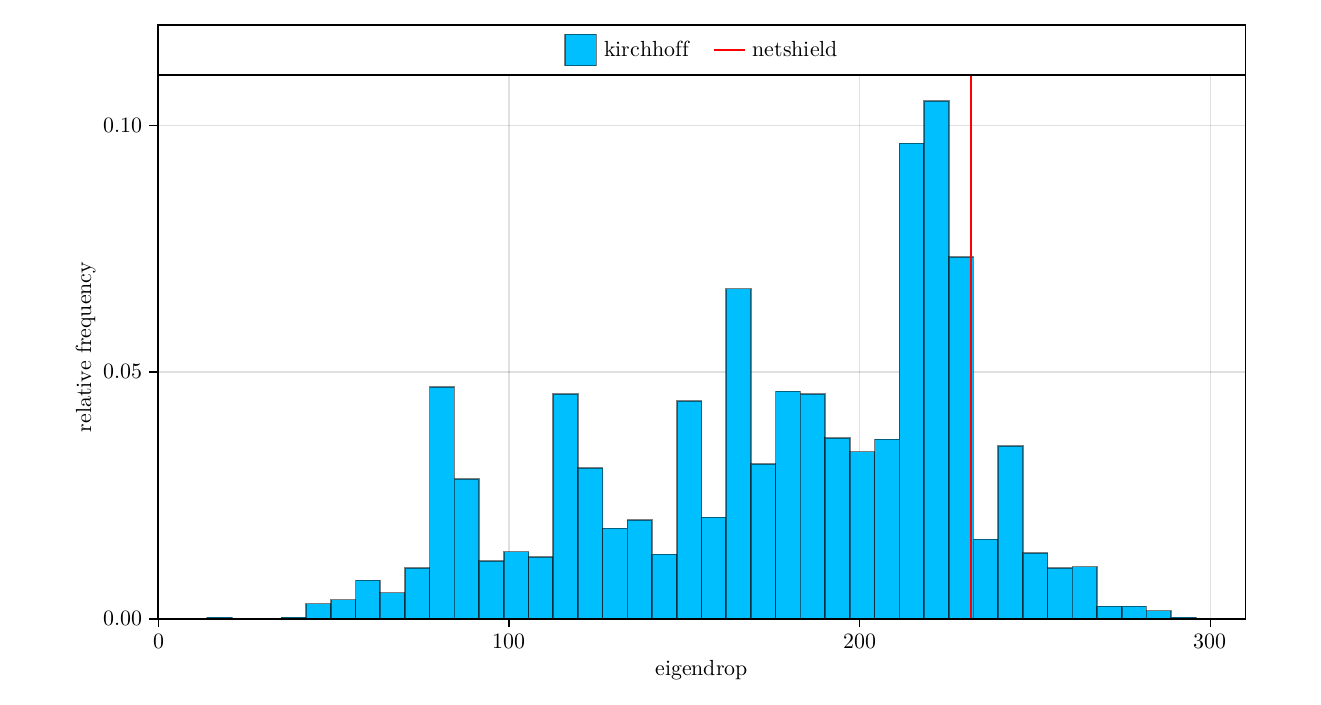}
          \caption{$k = 25$}
      \end{subcaptionblock}
\end{figure}

\begin{figure}
    \centering $ $
    \caption{Graph: ``conference 2'' (non-weighted) - eigendrop distribution}
      \begin{subcaptionblock}{0.45\textwidth}
          \includegraphics[width=\textwidth]{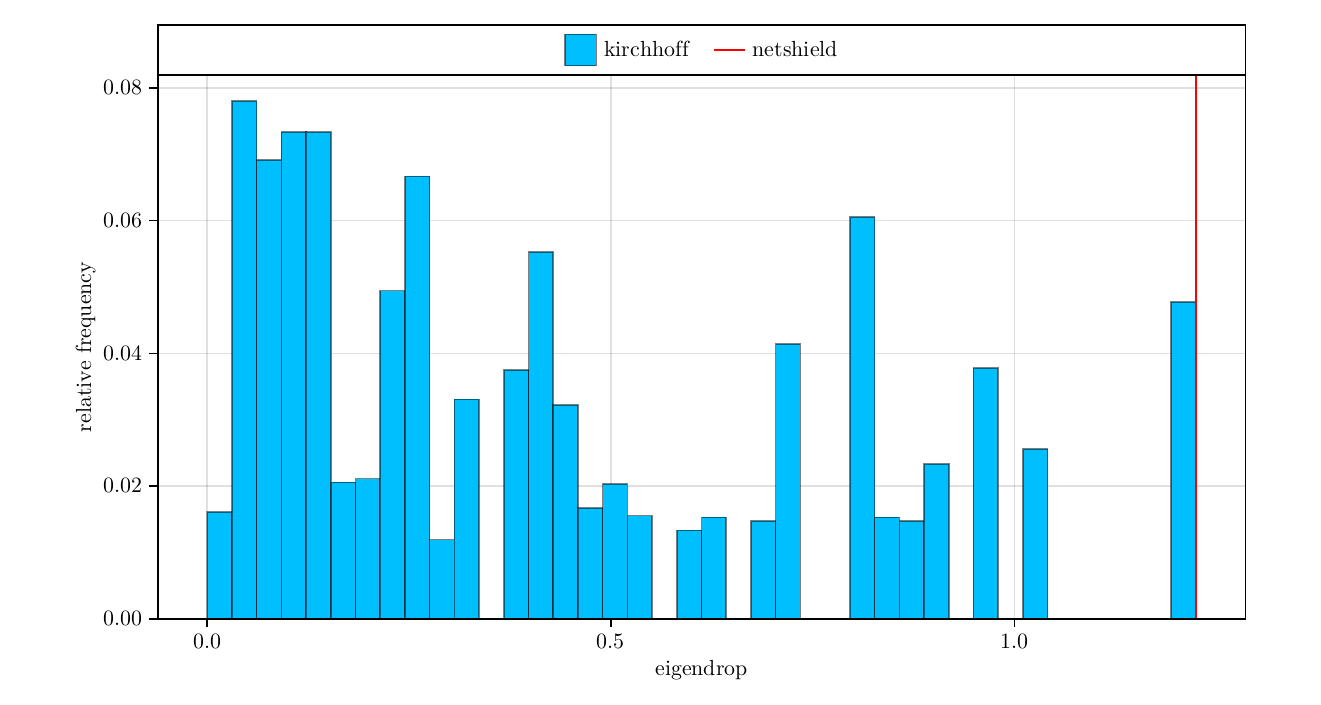}
          \caption{$k = 1$}
      \end{subcaptionblock}
      \begin{subcaptionblock}{0.45\textwidth}
          \includegraphics[width=\textwidth]{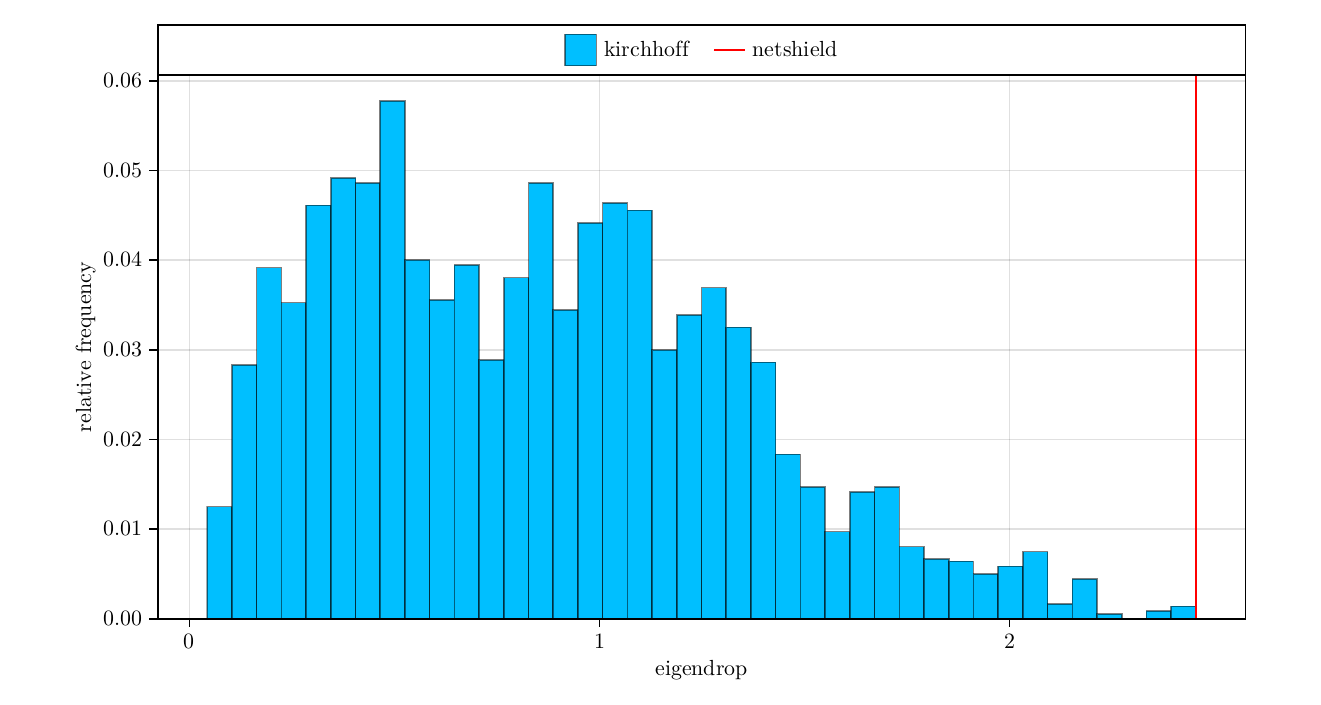}
          \caption{$k = 2$}
      \end{subcaptionblock}
      \\
      \begin{subcaptionblock}{0.45\textwidth}
          \includegraphics[width=\textwidth]{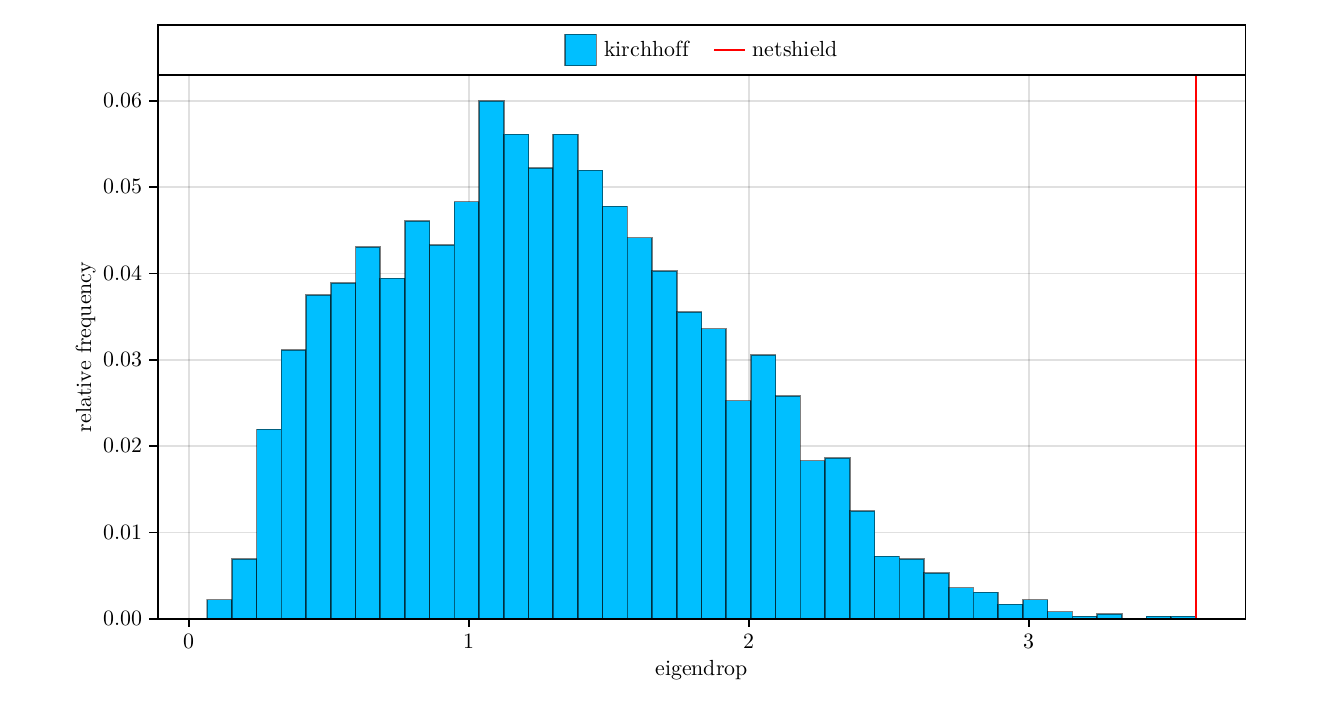}
          \caption{$k = 3$}
      \end{subcaptionblock}
      \begin{subcaptionblock}{0.45\textwidth}
          \includegraphics[width=\textwidth]{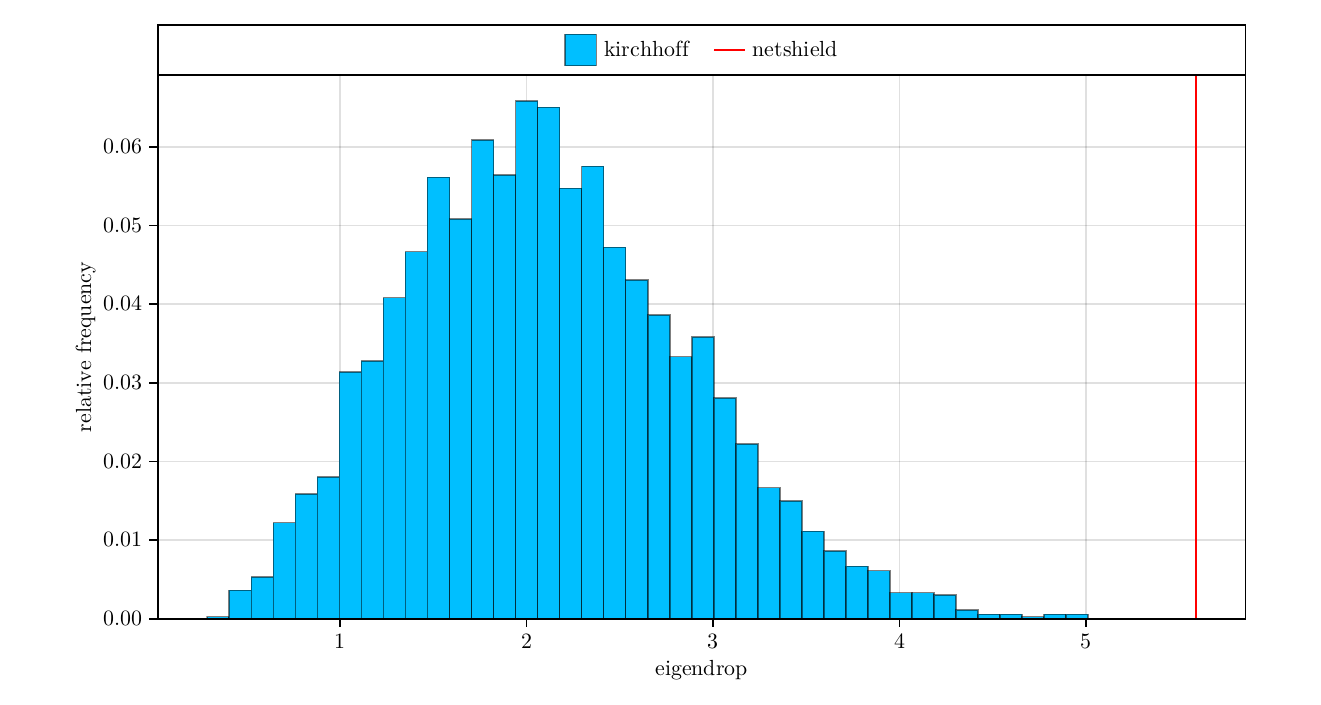}
          \caption{$k = 5$}
      \end{subcaptionblock}
      \\
      \begin{subcaptionblock}{0.45\textwidth}
          \includegraphics[width=\textwidth]{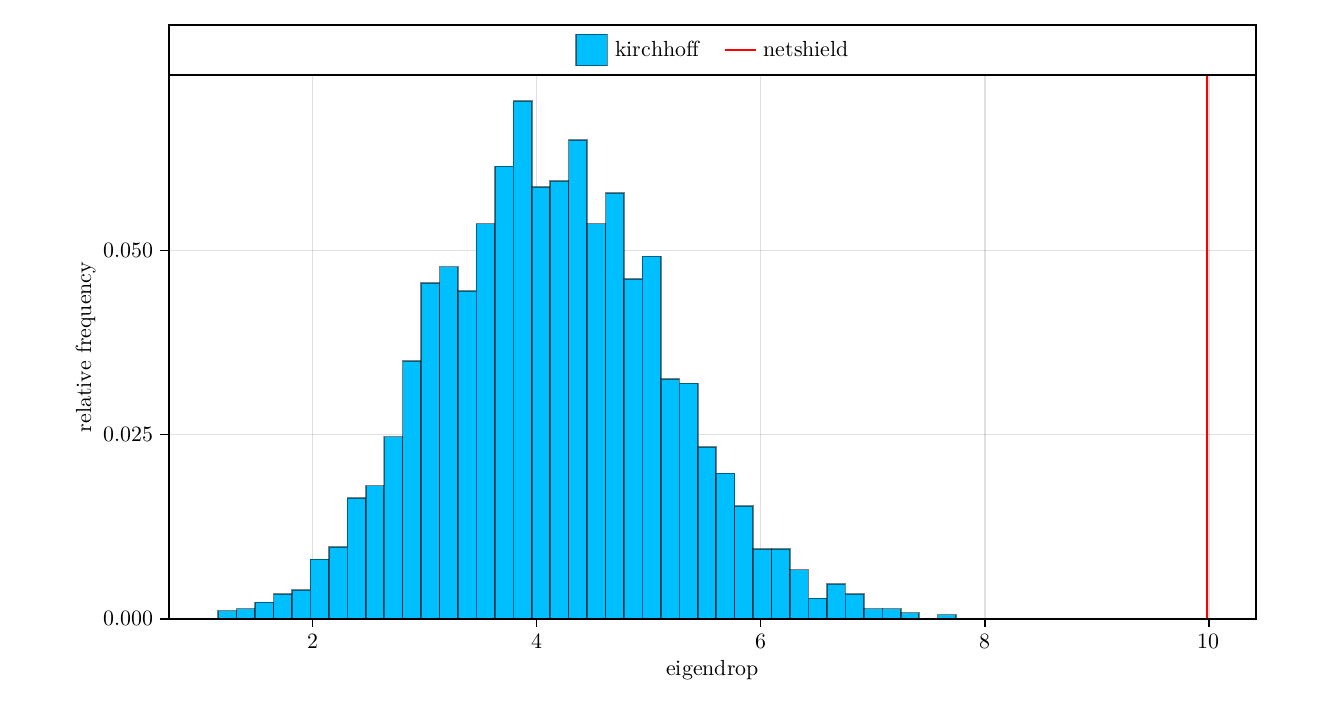}
          \caption{$k = 10$}
      \end{subcaptionblock}
      \begin{subcaptionblock}{0.45\textwidth}
          \includegraphics[width=\textwidth]{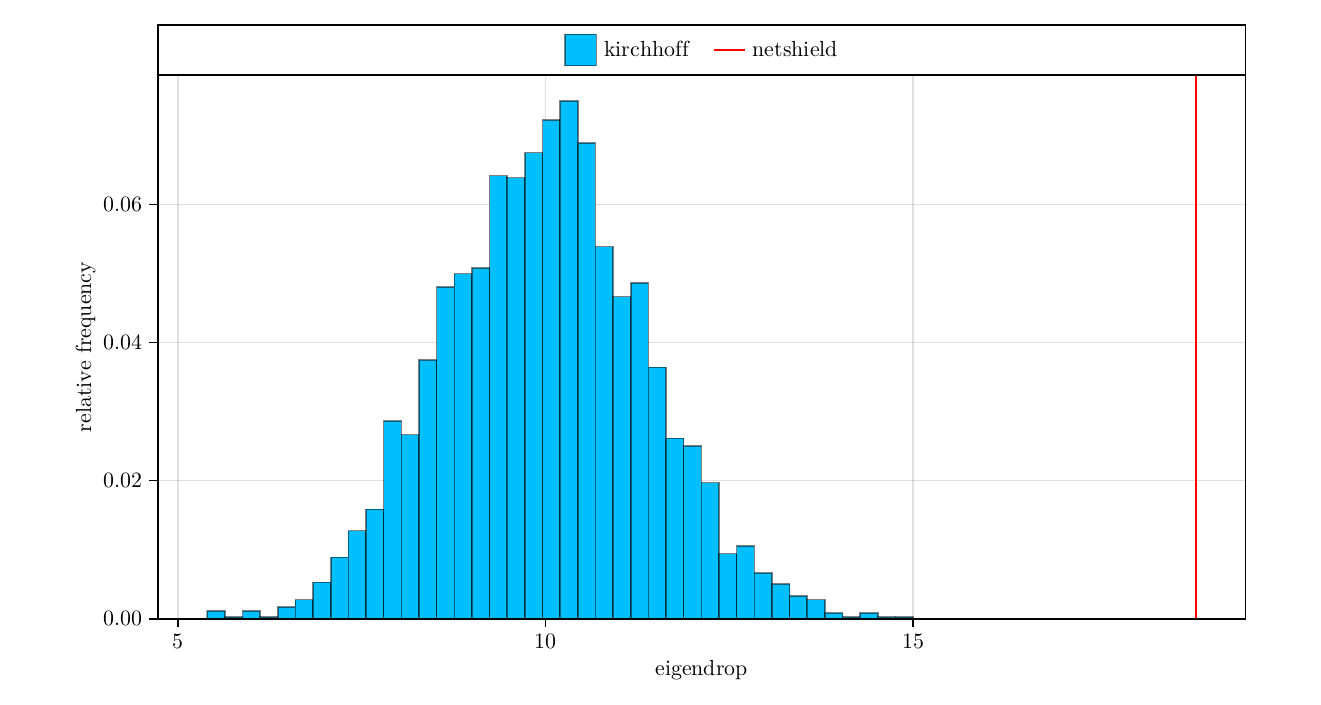}
          \caption{$k = 26$}
      \end{subcaptionblock}
\end{figure}

\begin{figure}
    \centering $ $
    \caption{Graph: ``conference 2'' (weighted) - eigendrop distribution}
      \begin{subcaptionblock}{0.45\textwidth}
          \includegraphics[width=\textwidth]{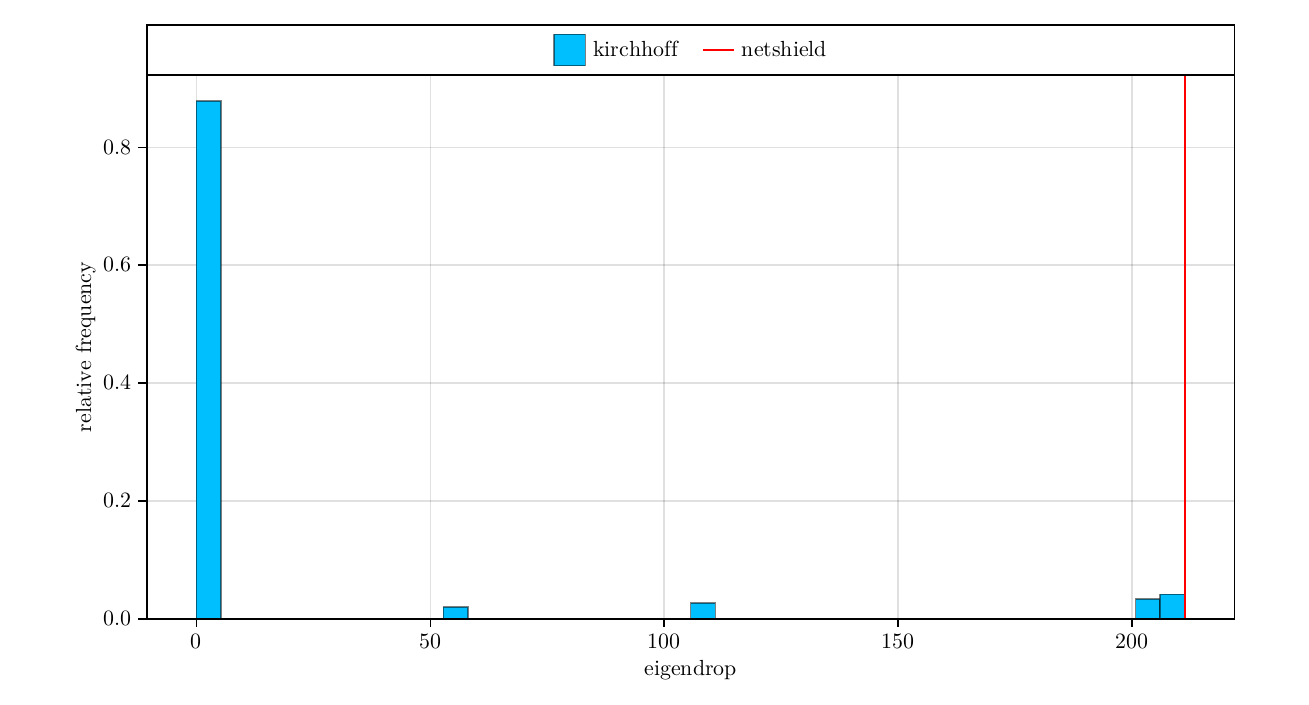}
          \caption{$k = 1$}
      \end{subcaptionblock}
      \begin{subcaptionblock}{0.45\textwidth}
          \includegraphics[width=\textwidth]{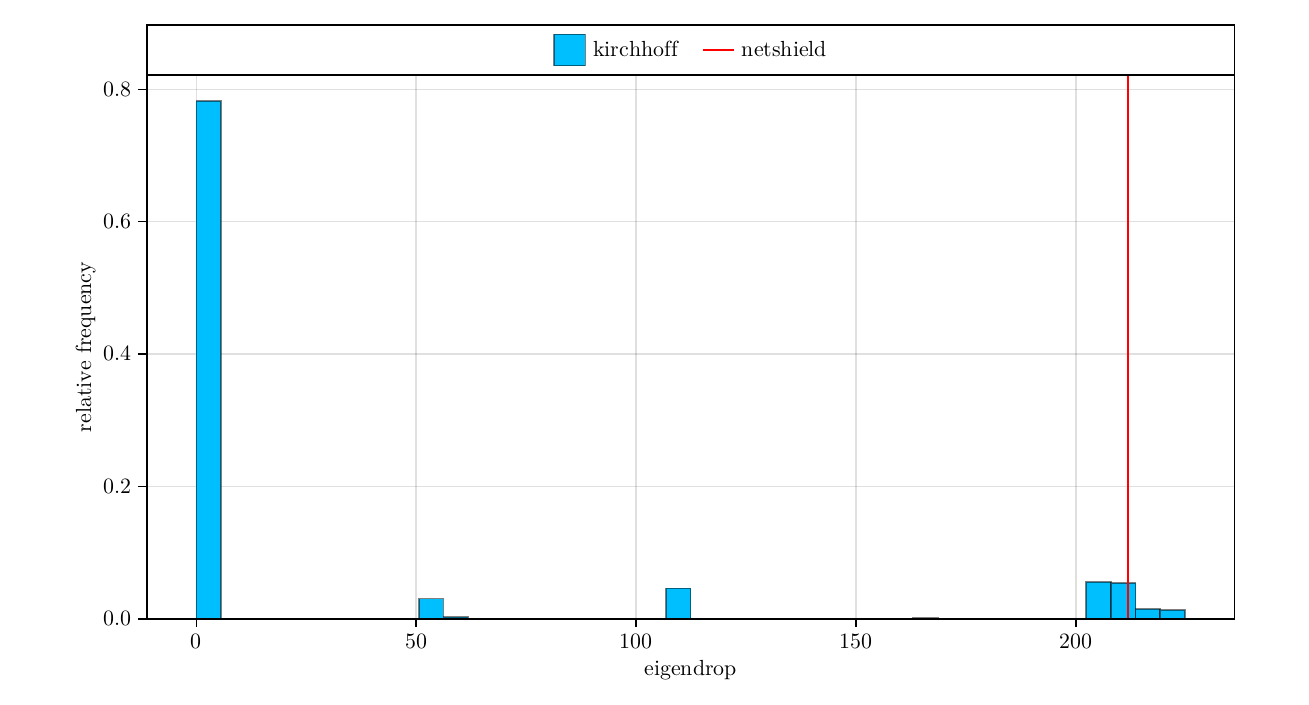}
          \caption{$k = 2$}
      \end{subcaptionblock}
      \\
      \begin{subcaptionblock}{0.45\textwidth}
          \includegraphics[width=\textwidth]{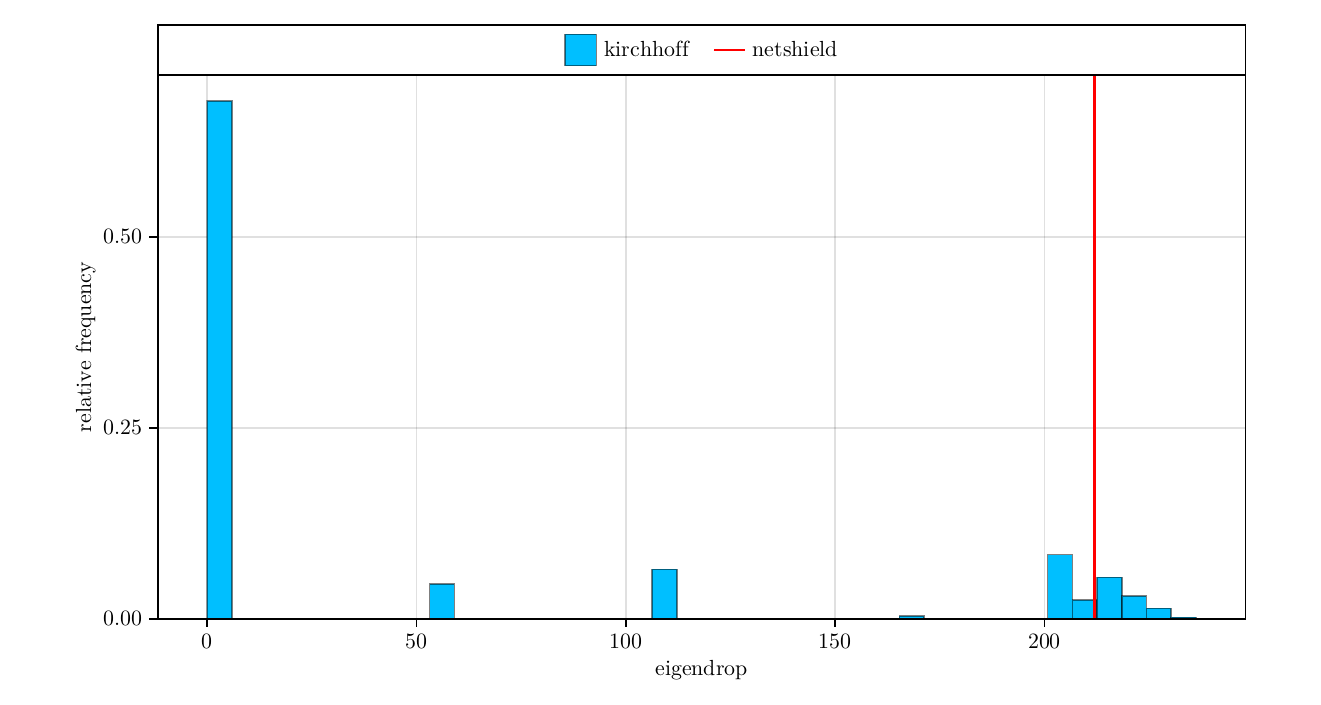}
          \caption{$k = 3$}
      \end{subcaptionblock}
      \begin{subcaptionblock}{0.45\textwidth}
          \includegraphics[width=\textwidth]{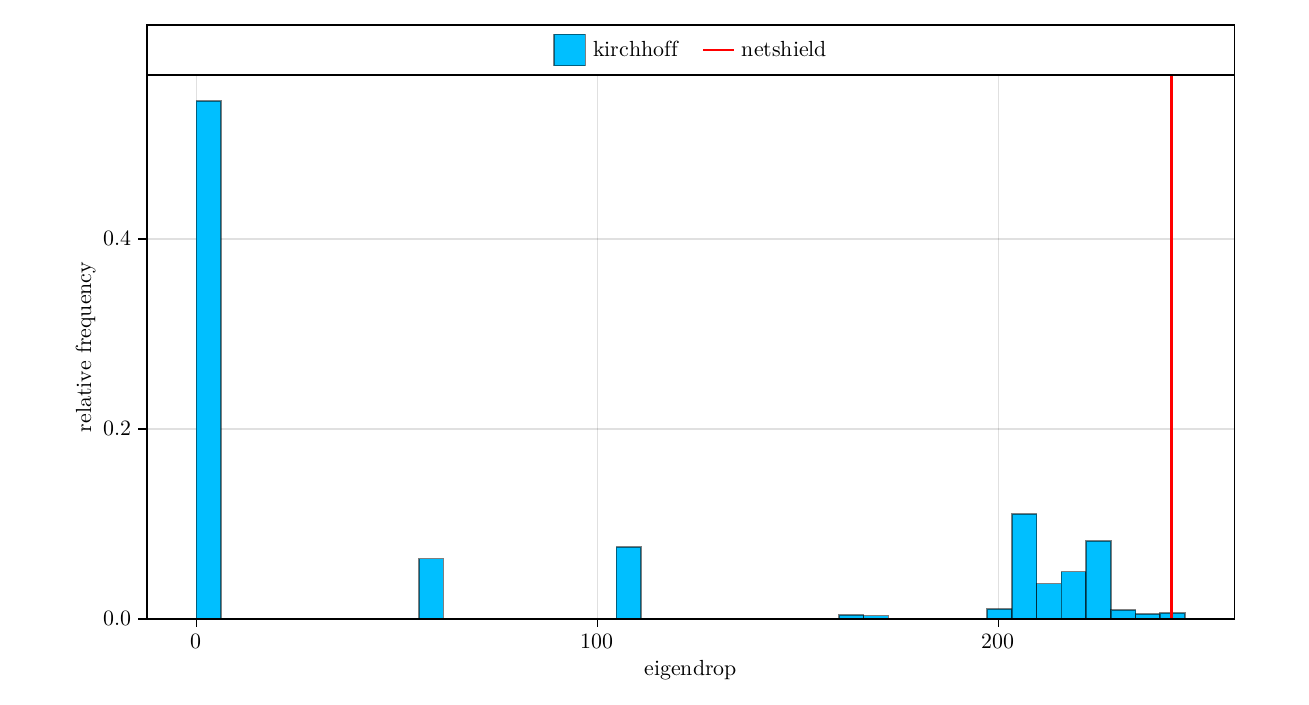}
          \caption{$k = 5$}
      \end{subcaptionblock}
      \\
      \begin{subcaptionblock}{0.45\textwidth}
          \includegraphics[width=\textwidth]{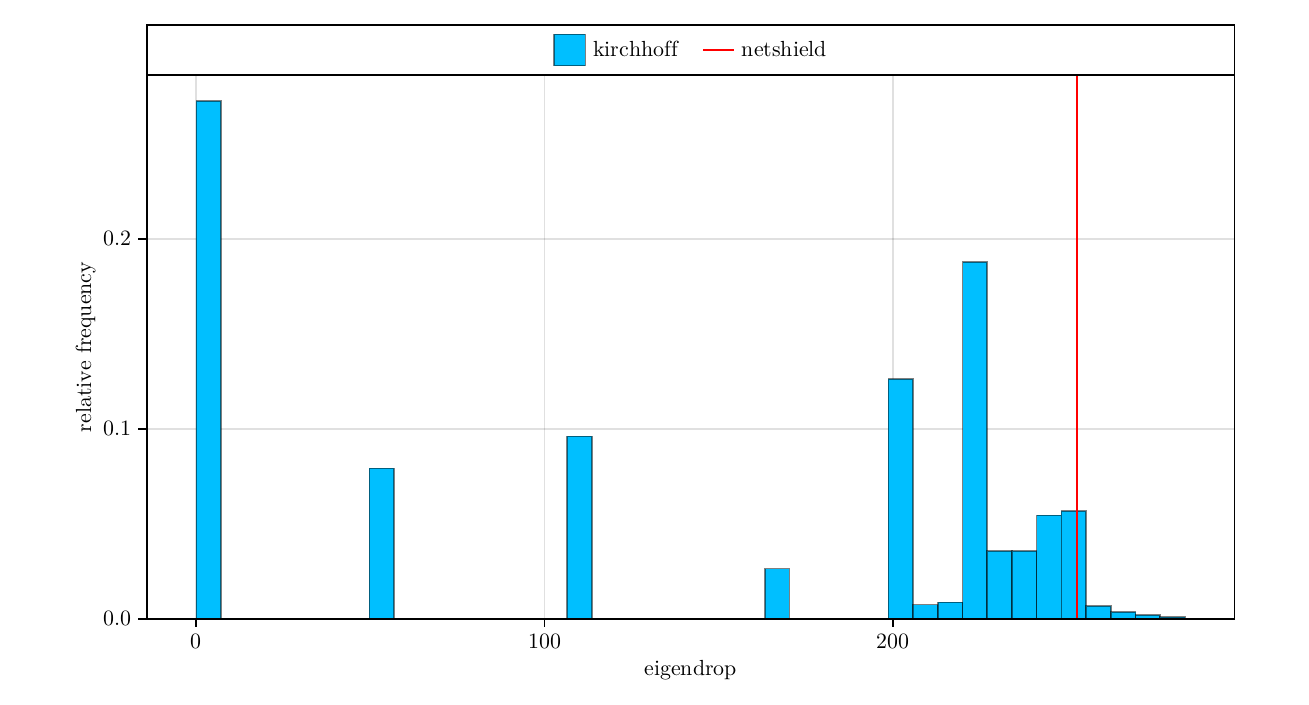}
          \caption{$k = 10$}
      \end{subcaptionblock}
      \begin{subcaptionblock}{0.45\textwidth}
          \includegraphics[width=\textwidth]{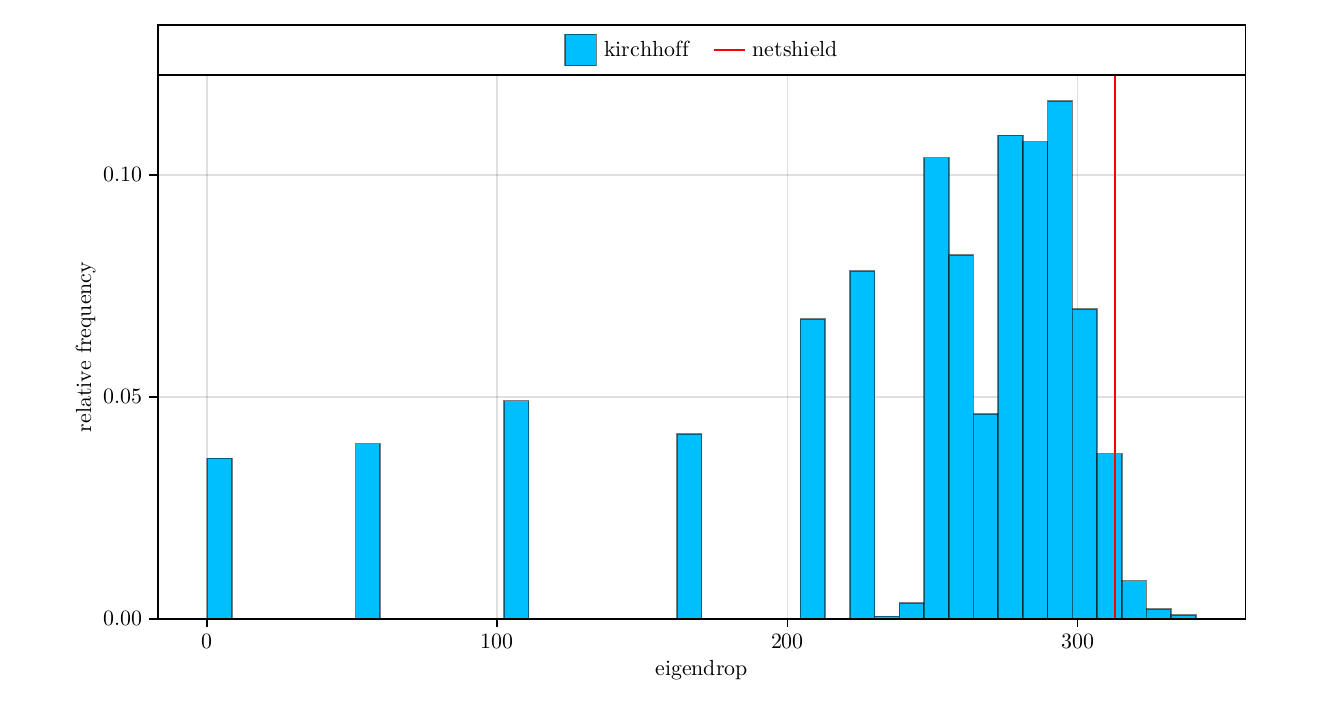}
          \caption{$k = 26$}
      \end{subcaptionblock}
\end{figure}

\begin{figure}
    \centering $ $
    \caption{Graph: ``conference 3'' (non-weighted) - eigendrop distribution}
      \begin{subcaptionblock}{0.45\textwidth}
          \includegraphics[width=\textwidth]{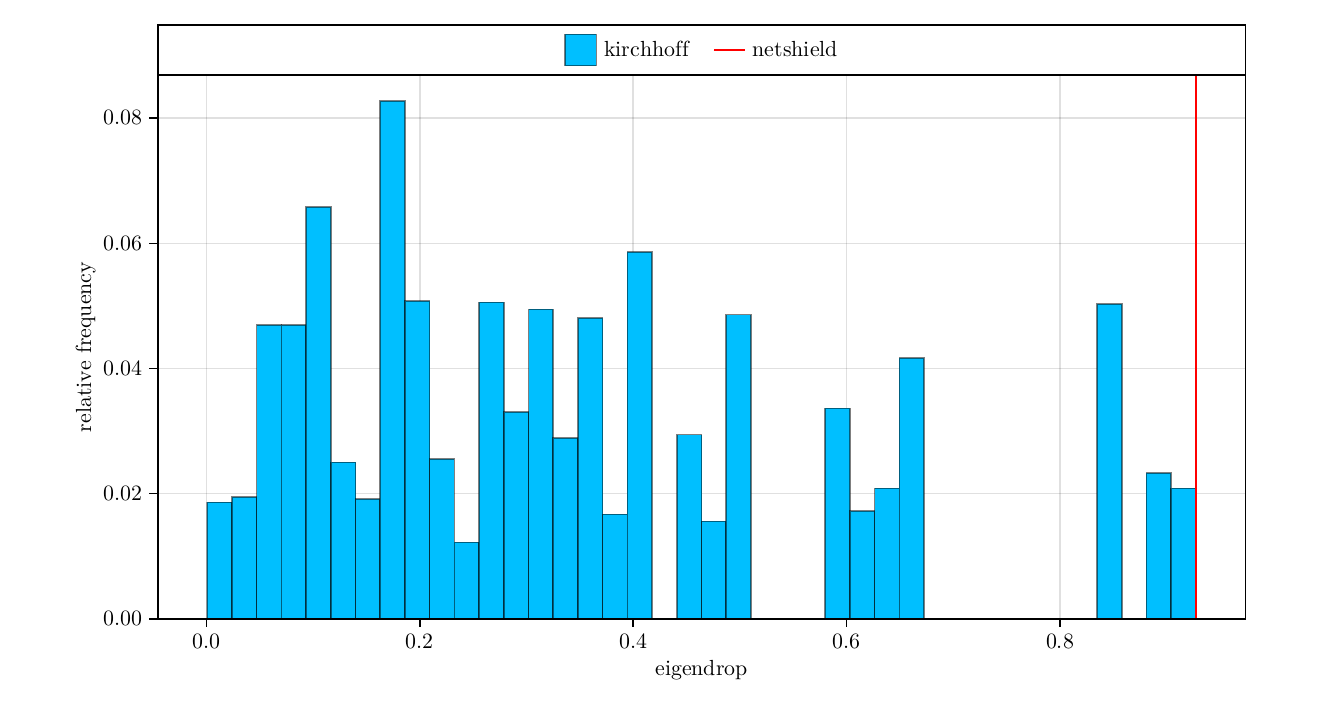}
          \caption{$k = 1$}
      \end{subcaptionblock}
      \begin{subcaptionblock}{0.45\textwidth}
          \includegraphics[width=\textwidth]{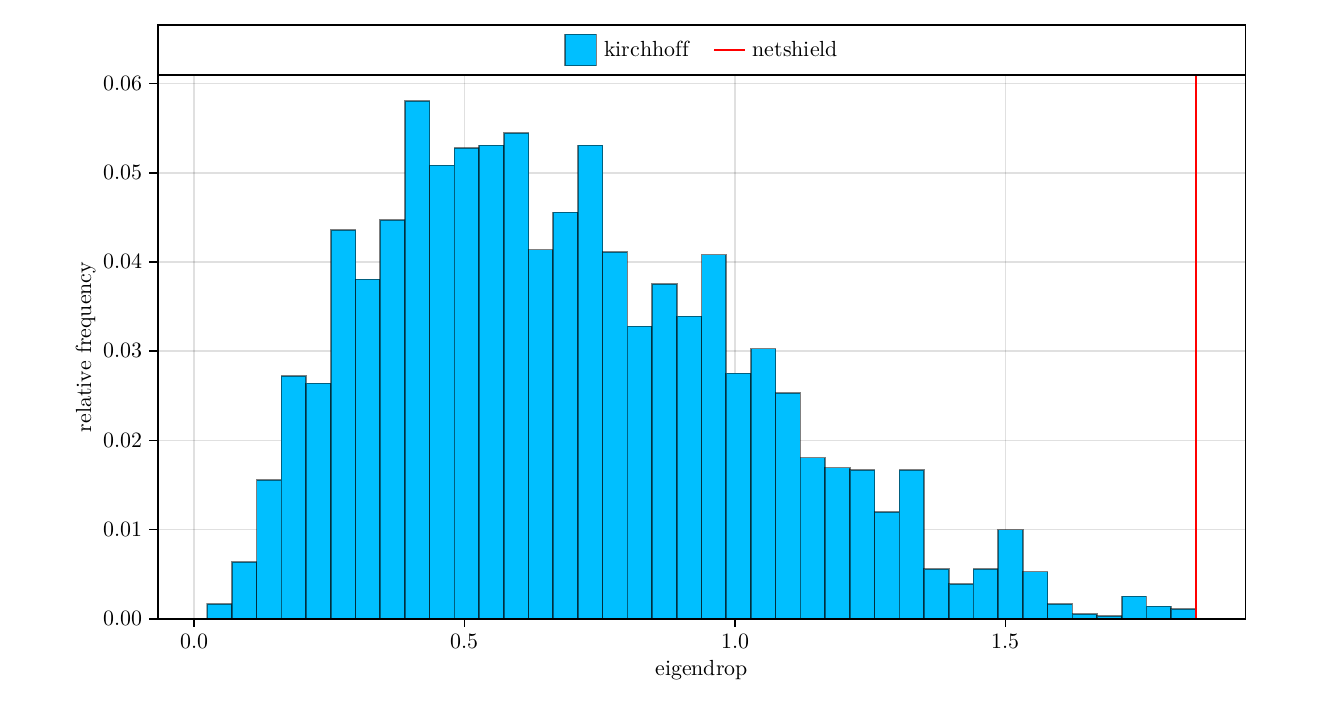}
          \caption{$k = 2$}
      \end{subcaptionblock}
      \\
      \begin{subcaptionblock}{0.45\textwidth}
          \includegraphics[width=\textwidth]{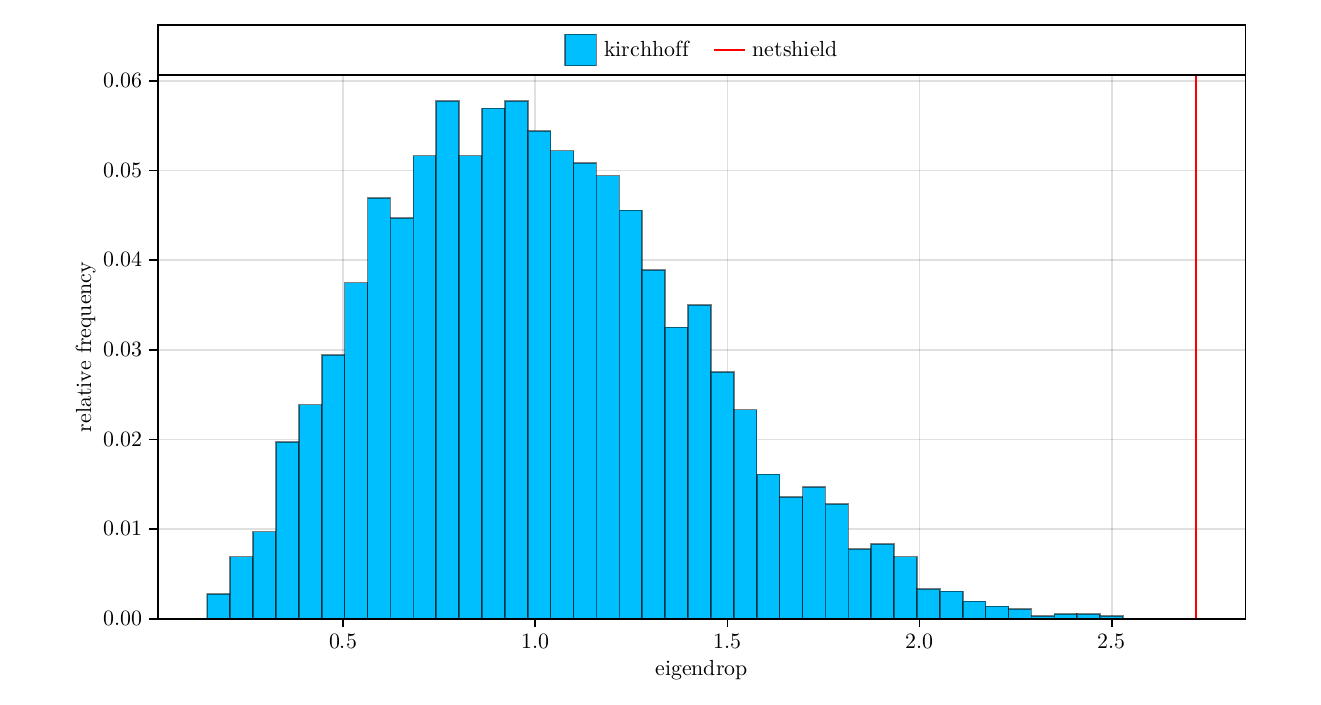}
          \caption{$k = 3$}
      \end{subcaptionblock}
      \begin{subcaptionblock}{0.45\textwidth}
          \includegraphics[width=\textwidth]{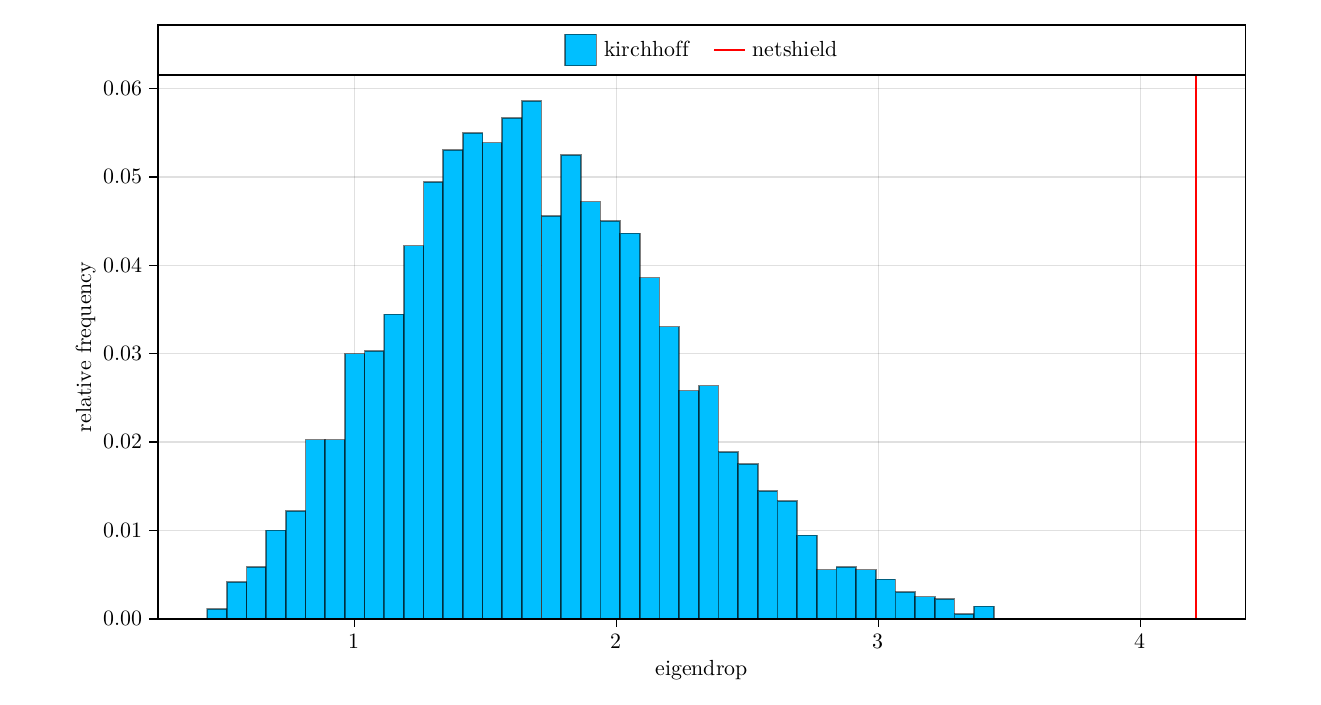}
          \caption{$k = 5$}
      \end{subcaptionblock}
      \\
      \begin{subcaptionblock}{0.45\textwidth}
          \includegraphics[width=\textwidth]{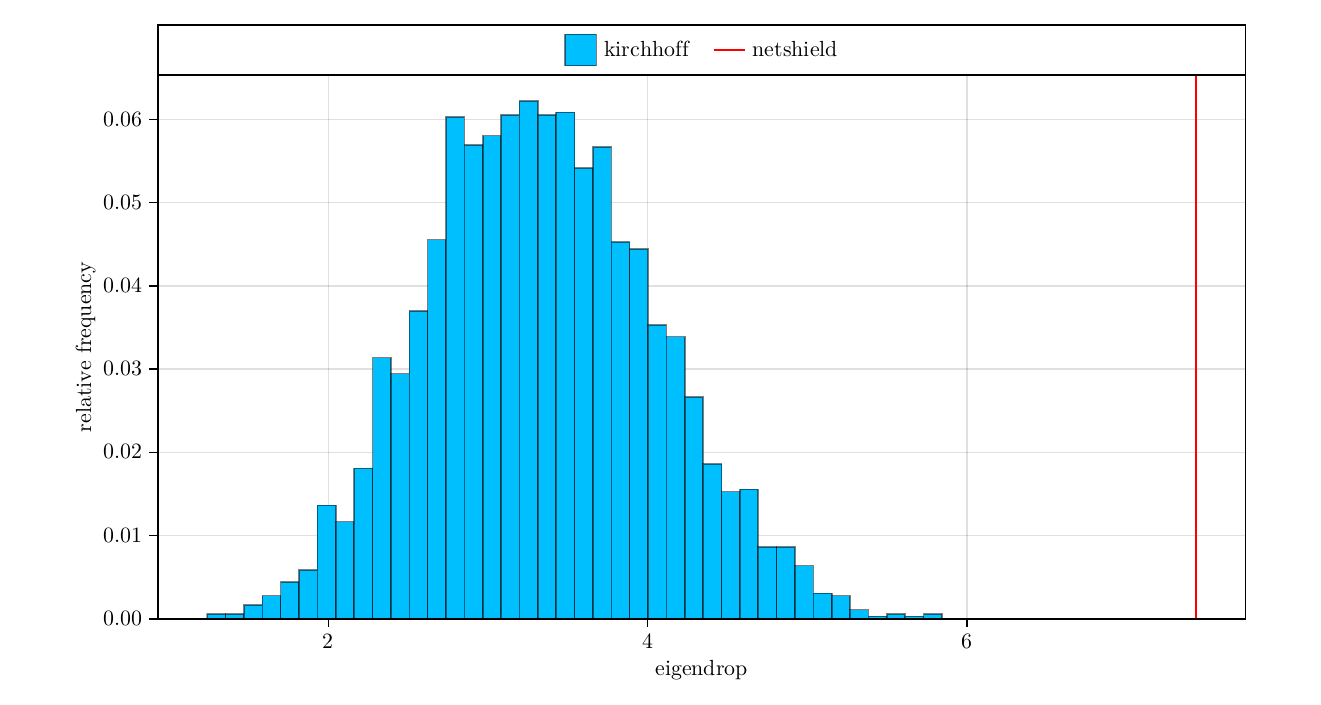}
          \caption{$k = 10$}
      \end{subcaptionblock}
      \begin{subcaptionblock}{0.45\textwidth}
          \includegraphics[width=\textwidth]{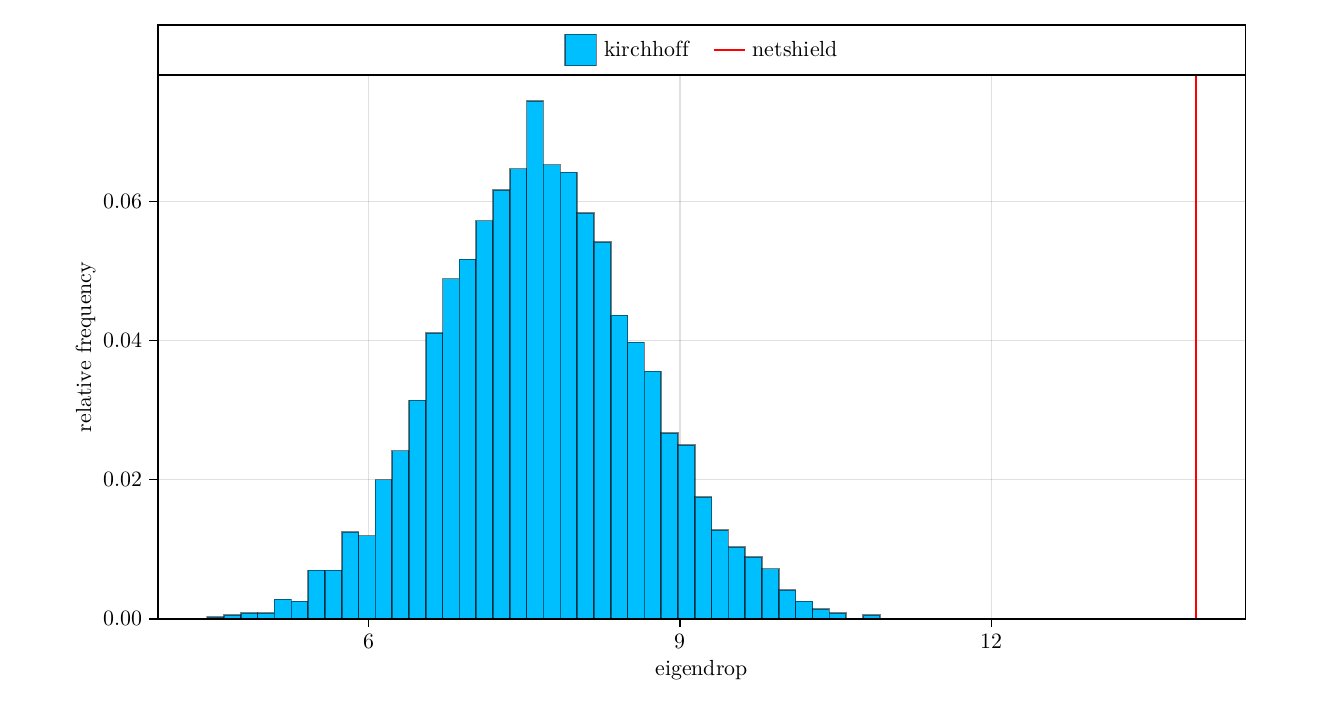}
          \caption{$k = 24$}
      \end{subcaptionblock}
\end{figure}

\begin{figure}
    \centering $ $
    \caption{Graph: ``conference 3'' (weighted) - eigendrop distribution}
      \begin{subcaptionblock}{0.45\textwidth}
          \includegraphics[width=\textwidth]{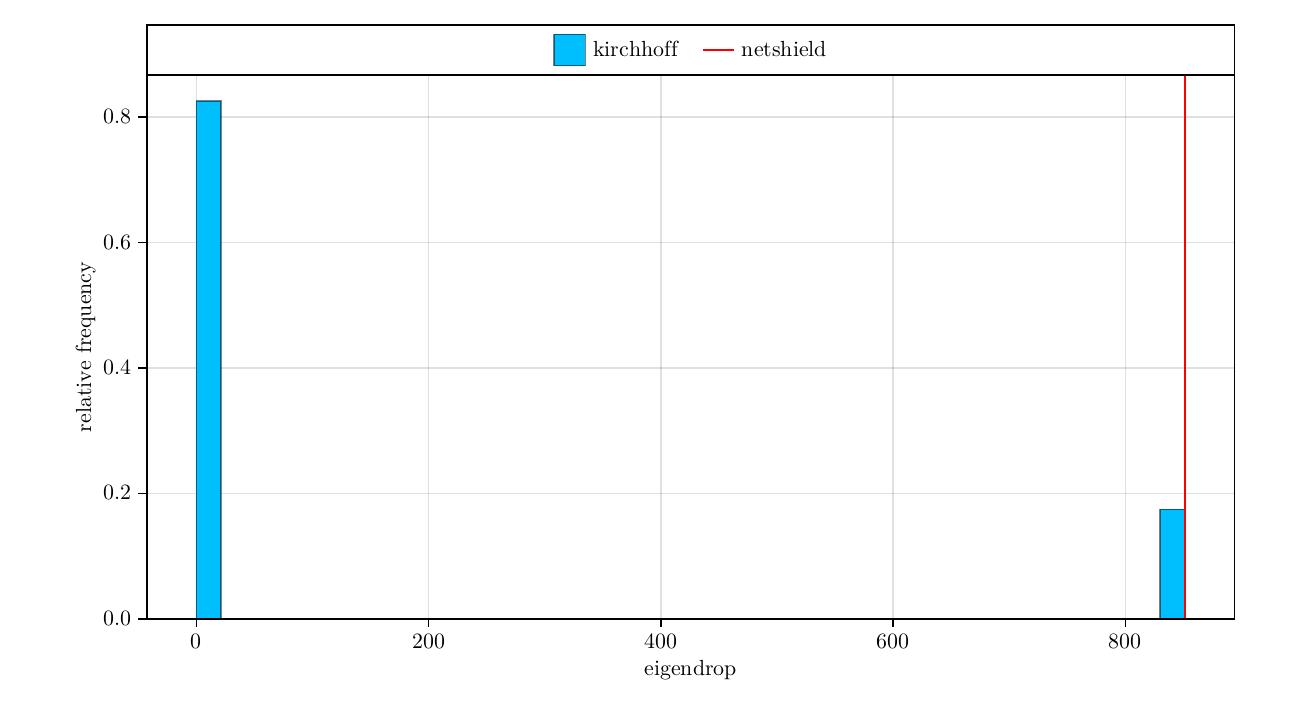}
          \caption{$k = 1$}
      \end{subcaptionblock}
      \begin{subcaptionblock}{0.45\textwidth}
          \includegraphics[width=\textwidth]{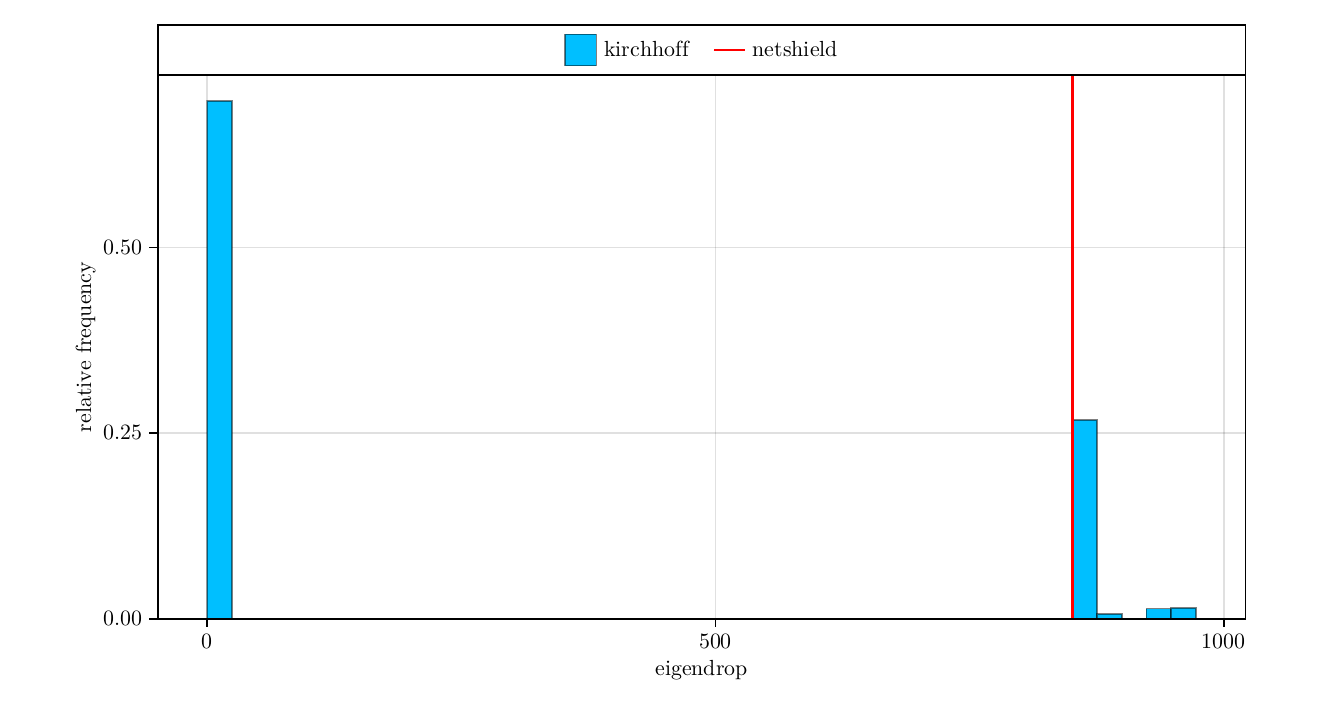}
          \caption{$k = 2$}
      \end{subcaptionblock}
      \\
      \begin{subcaptionblock}{0.45\textwidth}
          \includegraphics[width=\textwidth]{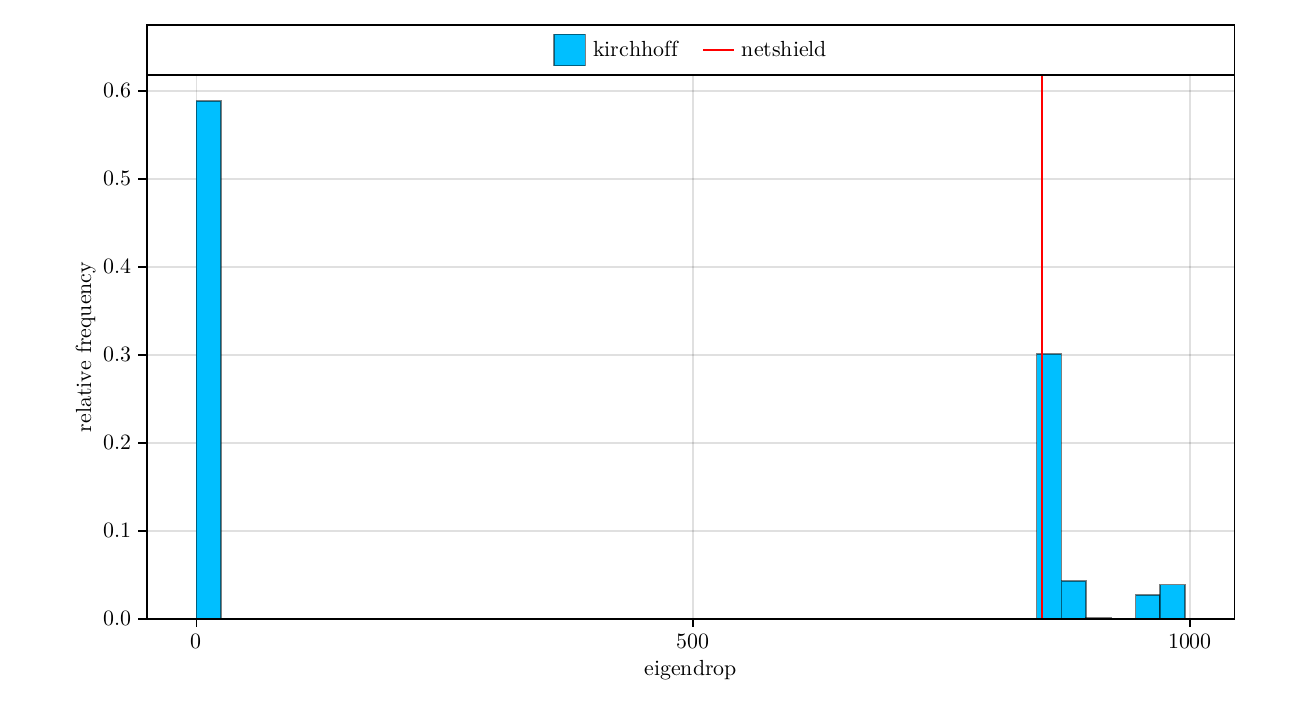}
          \caption{$k = 3$}
      \end{subcaptionblock}
      \begin{subcaptionblock}{0.45\textwidth}
          \includegraphics[width=\textwidth]{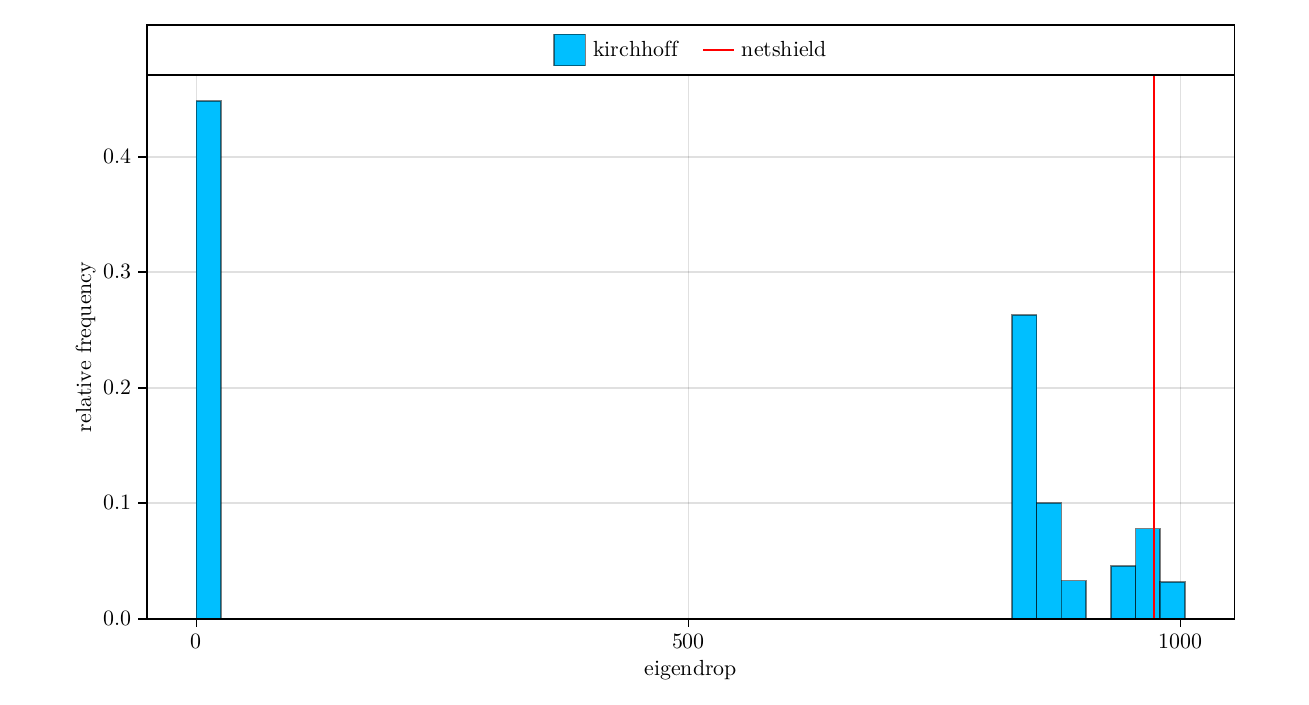}
          \caption{$k = 5$}
      \end{subcaptionblock}
      \\
      \begin{subcaptionblock}{0.45\textwidth}
          \includegraphics[width=\textwidth]{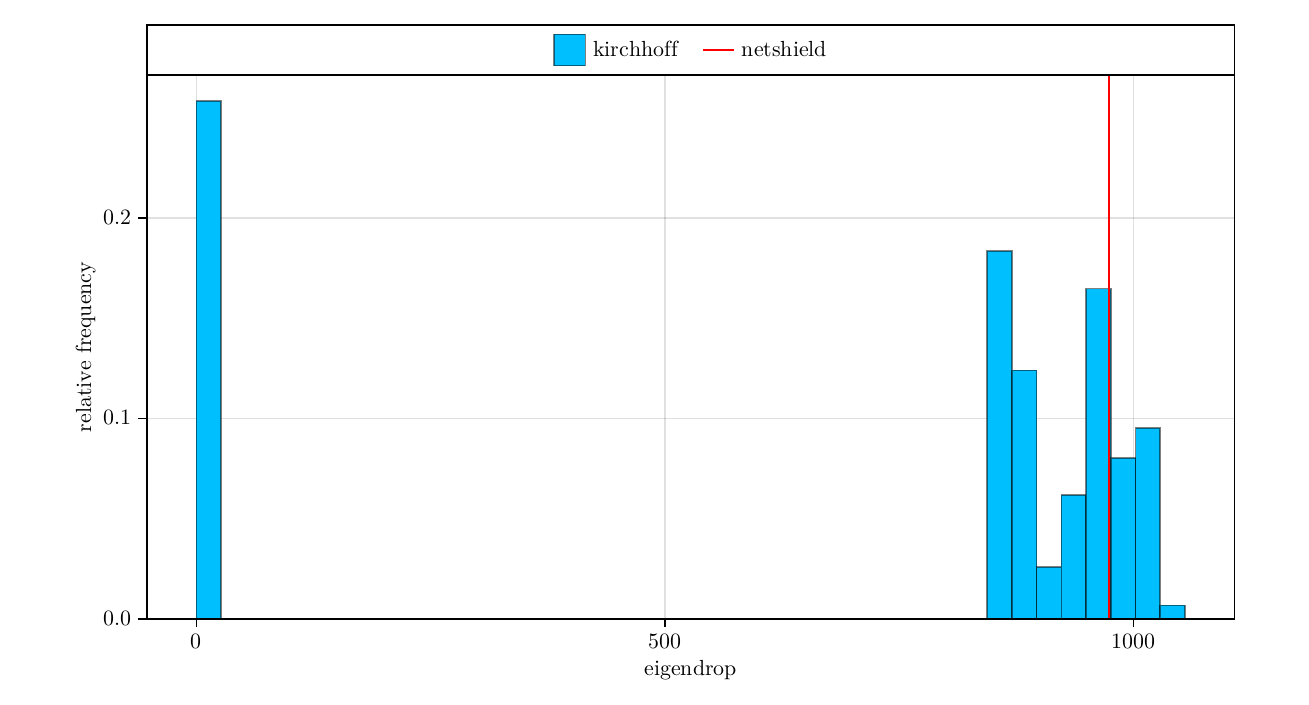}
          \caption{$k = 10$}
      \end{subcaptionblock}
      \begin{subcaptionblock}{0.45\textwidth}
          \includegraphics[width=\textwidth]{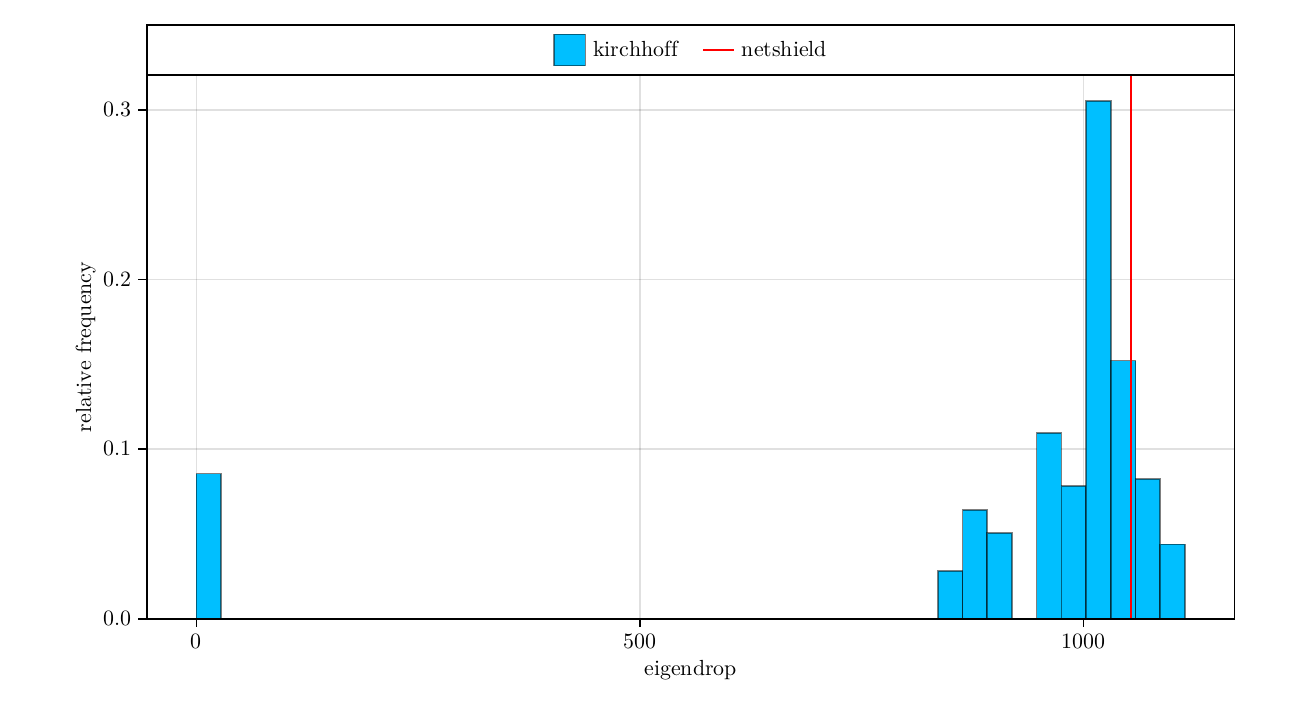}
          \caption{$k = 24$}
      \end{subcaptionblock}
\end{figure}

\begin{figure}
    \centering $ $
    \caption{Graph: ``airport 1'' (non-weighted) - eigendrop distribution}
      \begin{subcaptionblock}{0.45\textwidth}
          \includegraphics[width=\textwidth]{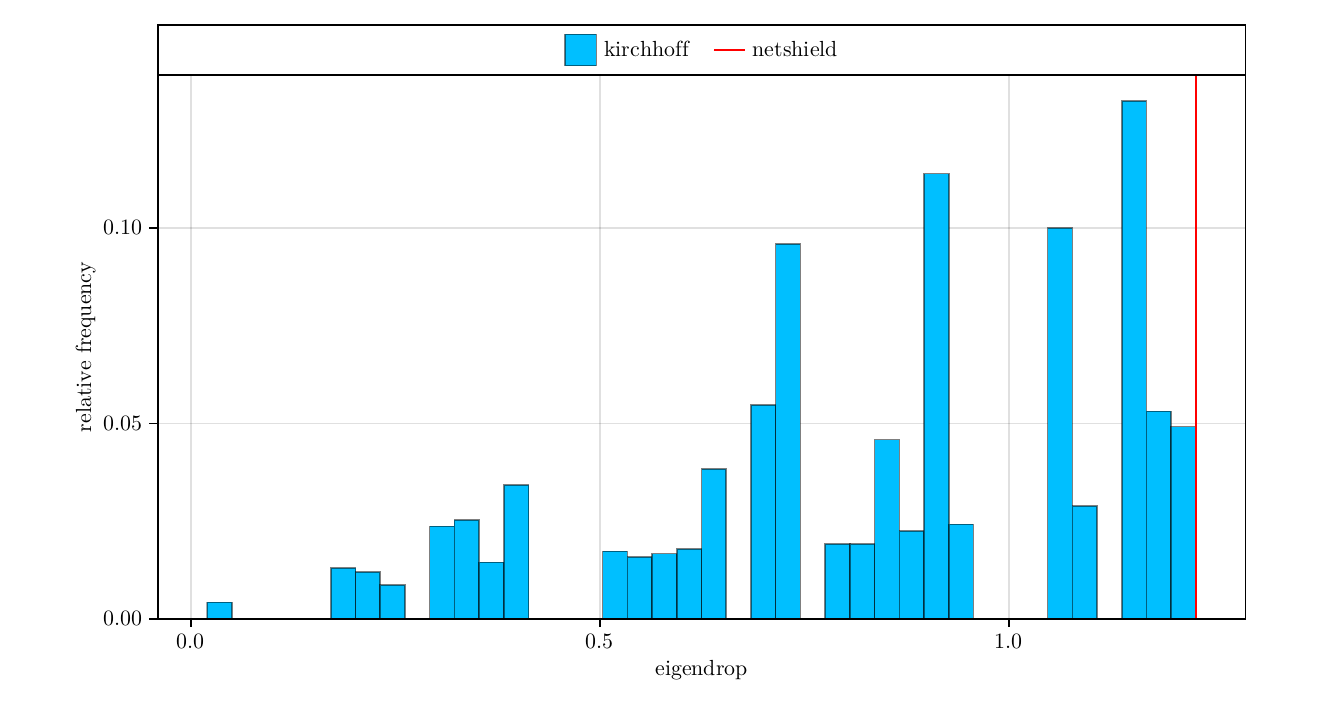}
          \caption{$k = 1$}
      \end{subcaptionblock}
      \begin{subcaptionblock}{0.45\textwidth}
          \includegraphics[width=\textwidth]{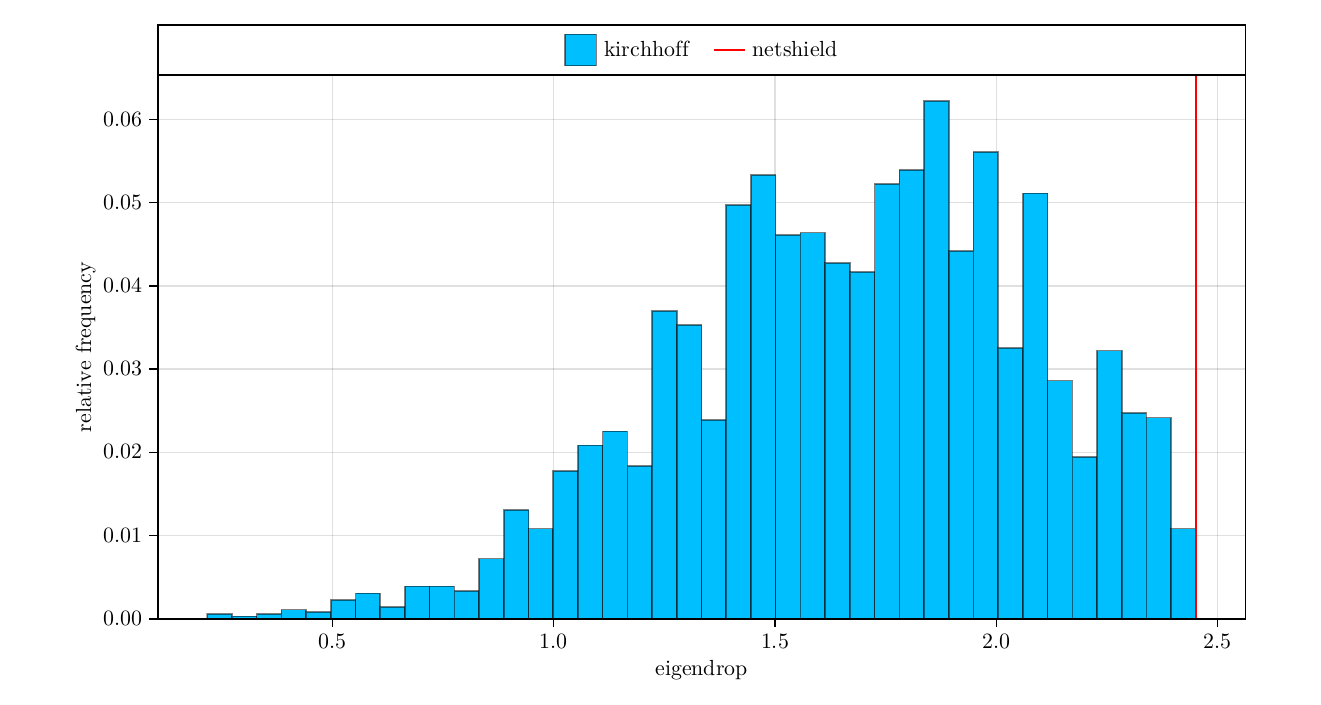}
          \caption{$k = 2$}
      \end{subcaptionblock}
      \\
      \begin{subcaptionblock}{0.45\textwidth}
          \includegraphics[width=\textwidth]{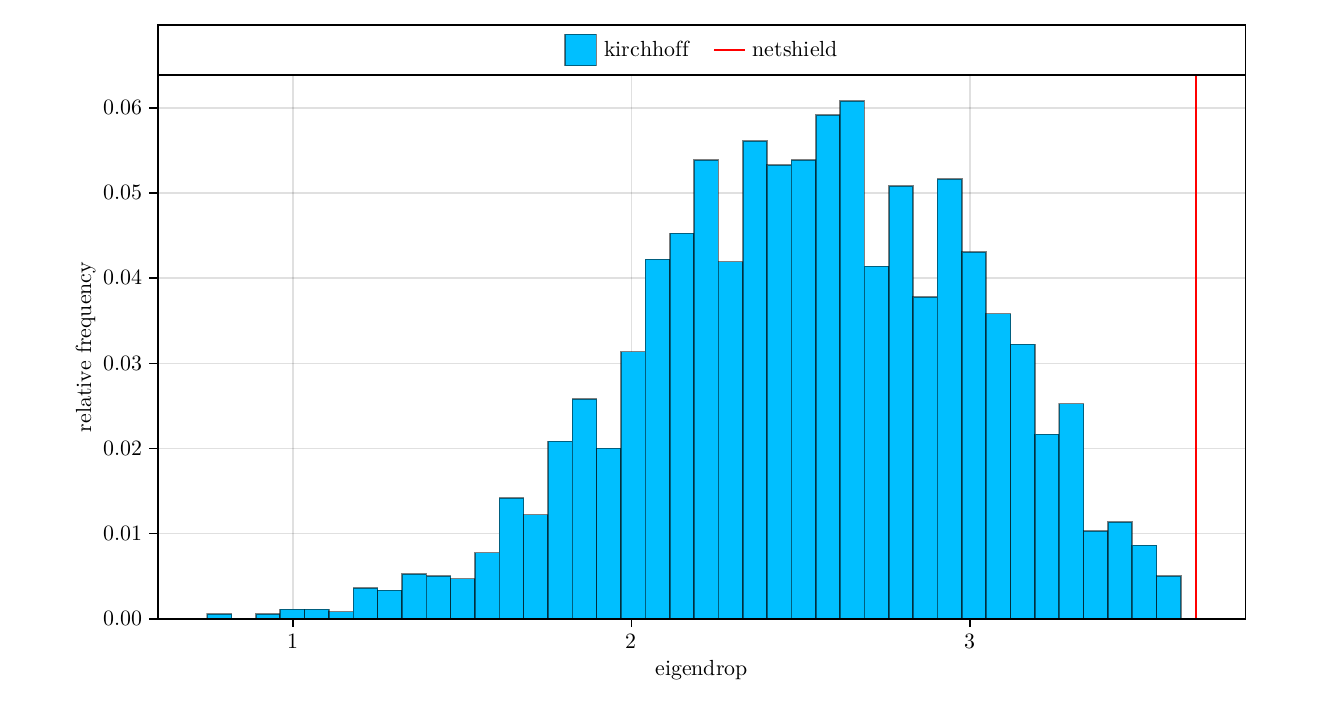}
          \caption{$k = 3$}
      \end{subcaptionblock}
      \begin{subcaptionblock}{0.45\textwidth}
          \includegraphics[width=\textwidth]{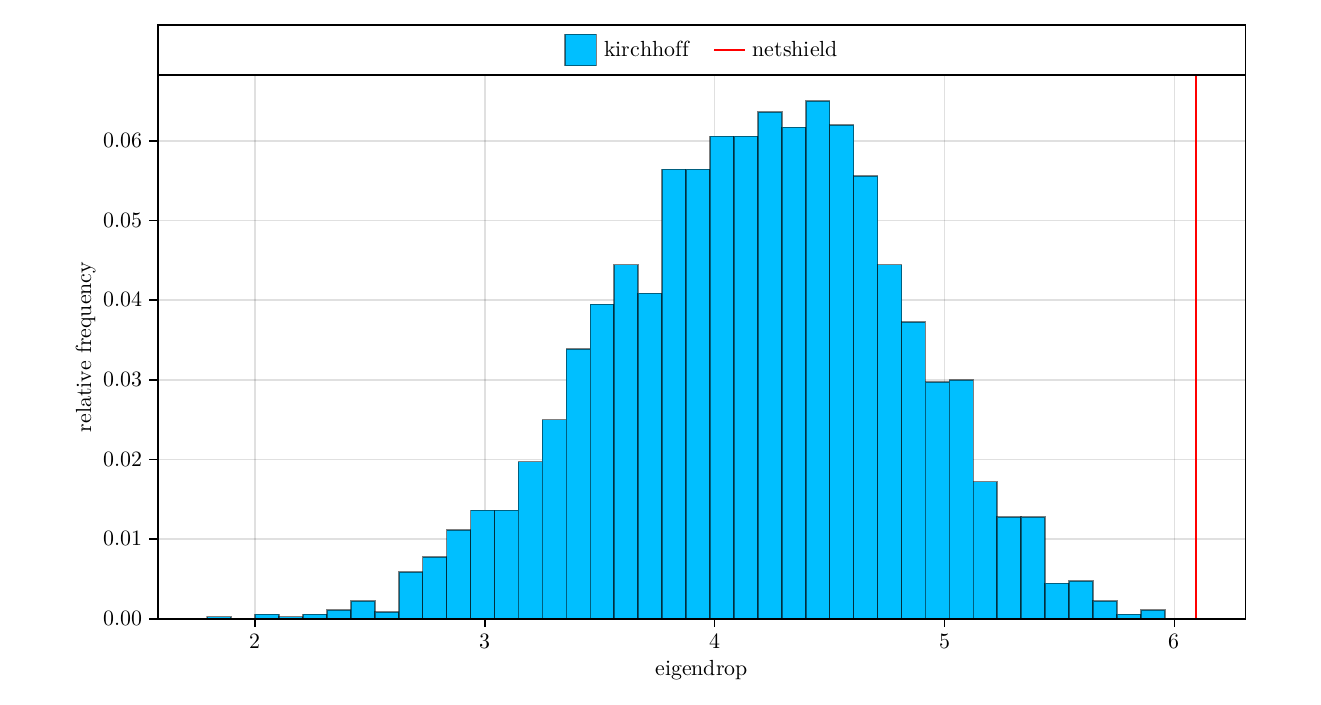}
          \caption{$k = 5$}
      \end{subcaptionblock}
      \\
      \begin{subcaptionblock}{0.45\textwidth}
          \includegraphics[width=\textwidth]{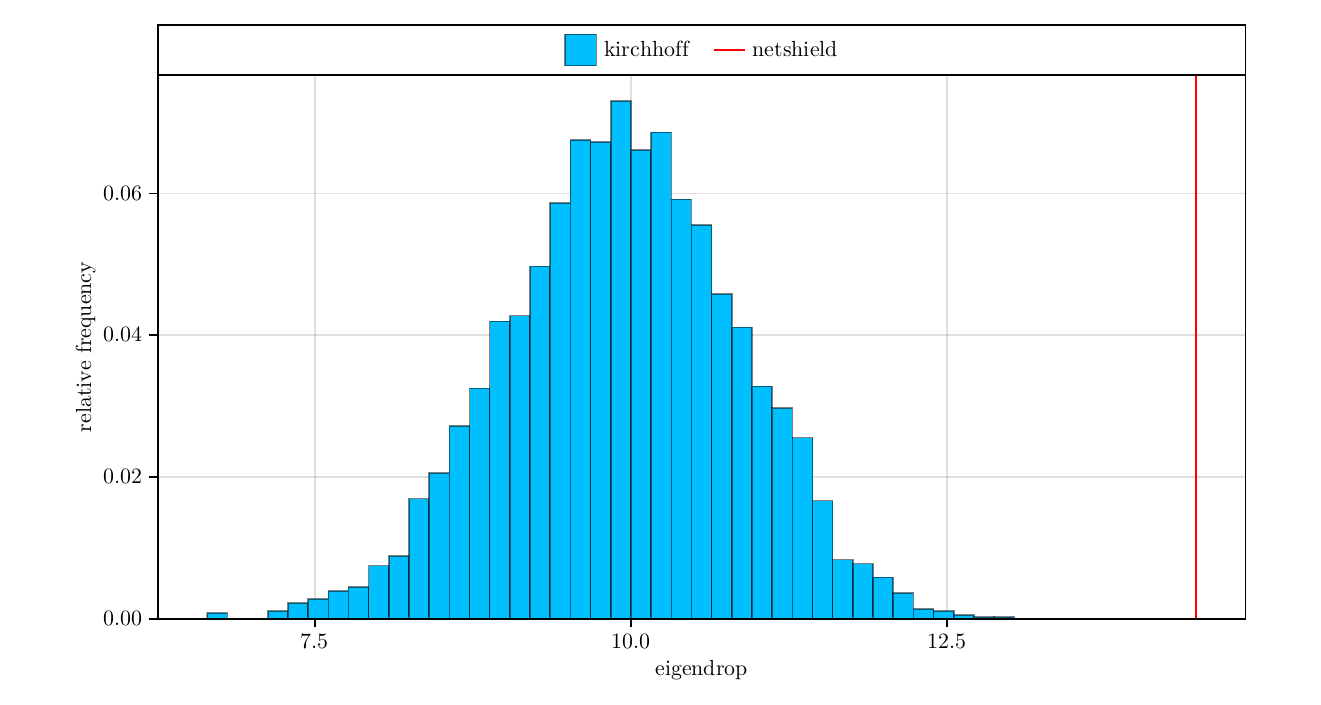}
          \caption{$k = 12$}
      \end{subcaptionblock}
\end{figure}

\begin{figure}
    \centering $ $
    \caption{Graph: ``airport 2'' (non-weighted) - eigendrop distribution}
      \begin{subcaptionblock}{0.45\textwidth}
          \includegraphics[width=\textwidth]{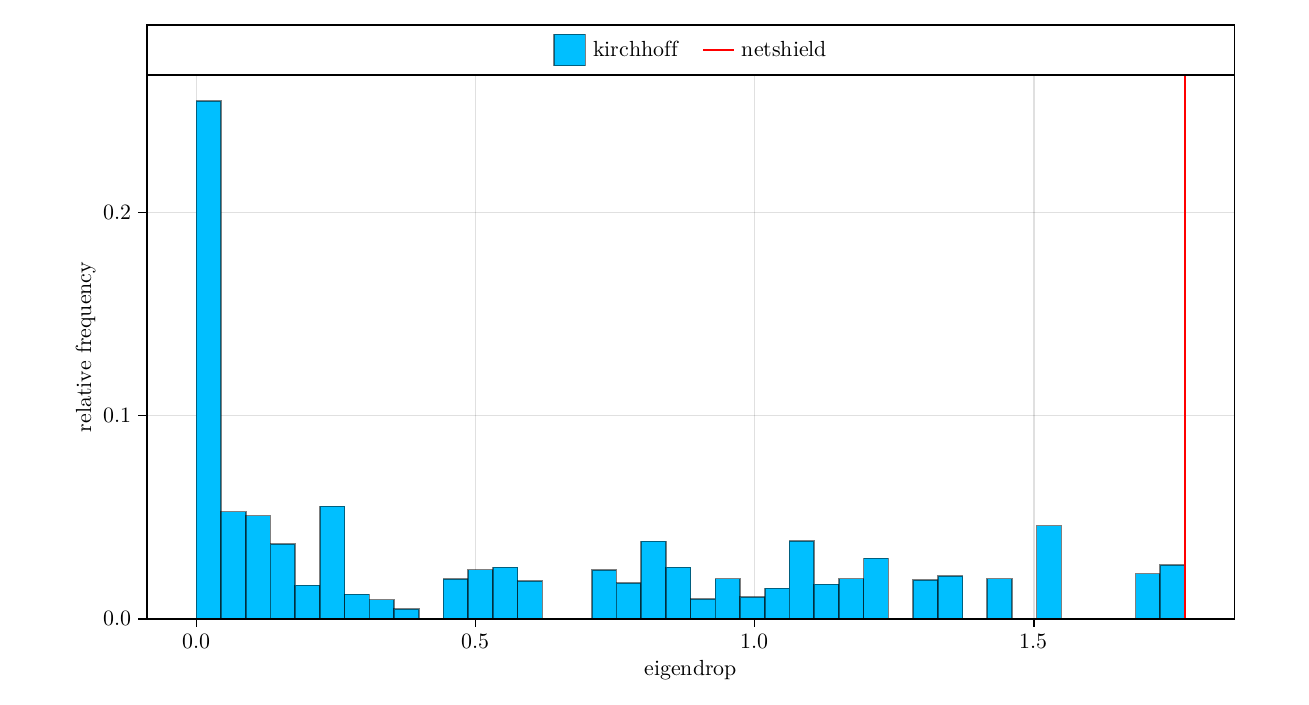}
          \caption{$k = 1$}
      \end{subcaptionblock}
      \begin{subcaptionblock}{0.45\textwidth}
          \includegraphics[width=\textwidth]{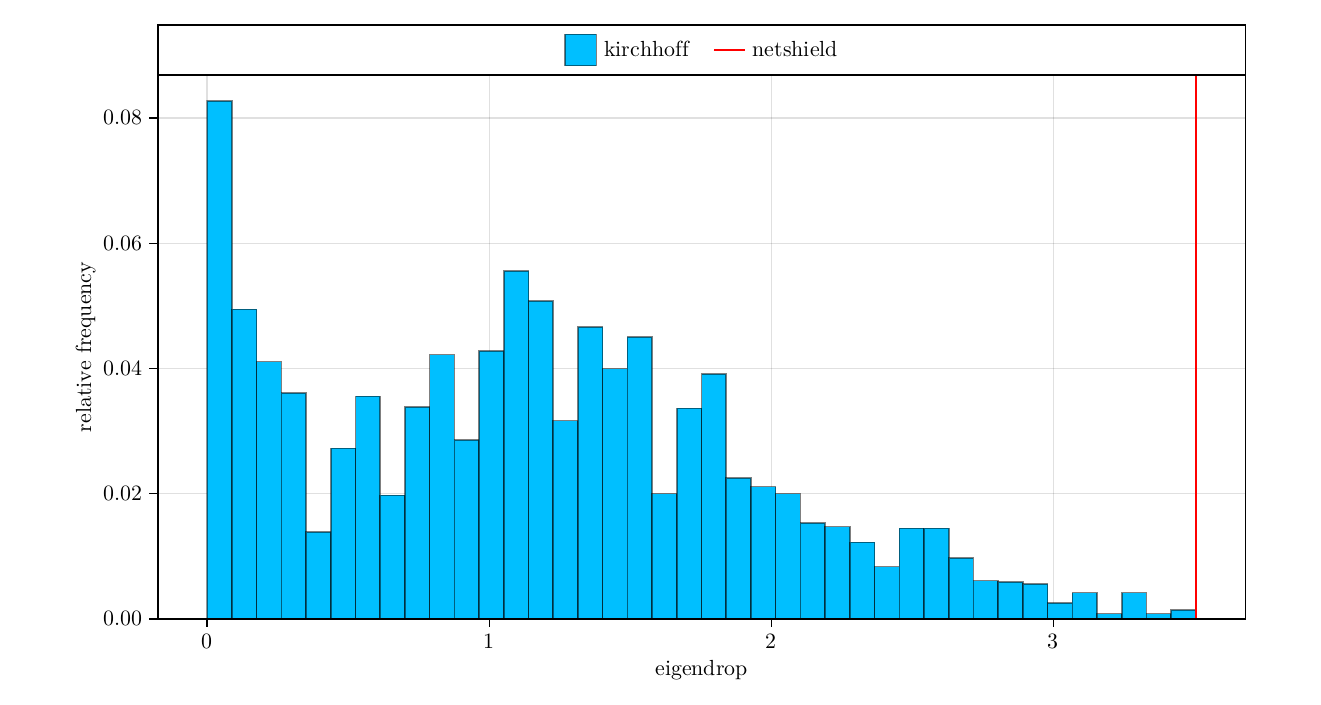}
          \caption{$k = 2$}
      \end{subcaptionblock}
      \\
      \begin{subcaptionblock}{0.45\textwidth}
          \includegraphics[width=\textwidth]{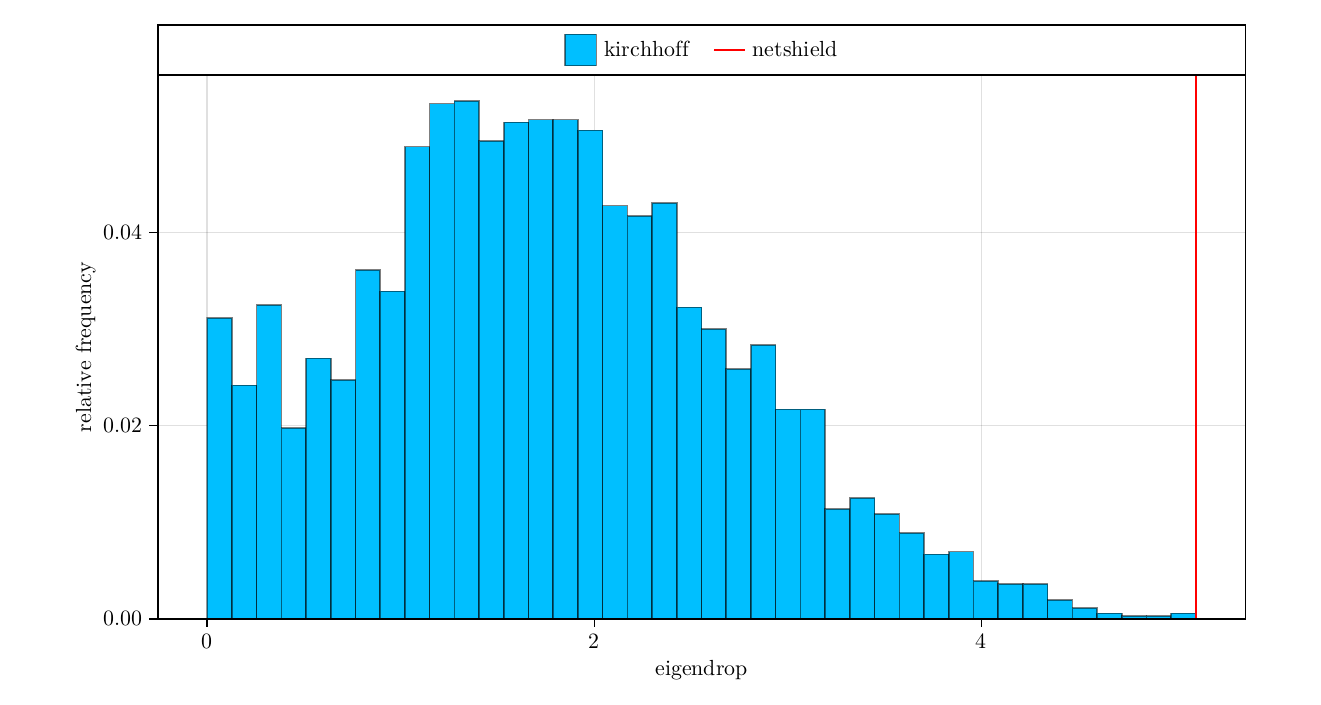}
          \caption{$k = 3$}
      \end{subcaptionblock}
      \begin{subcaptionblock}{0.45\textwidth}
          \includegraphics[width=\textwidth]{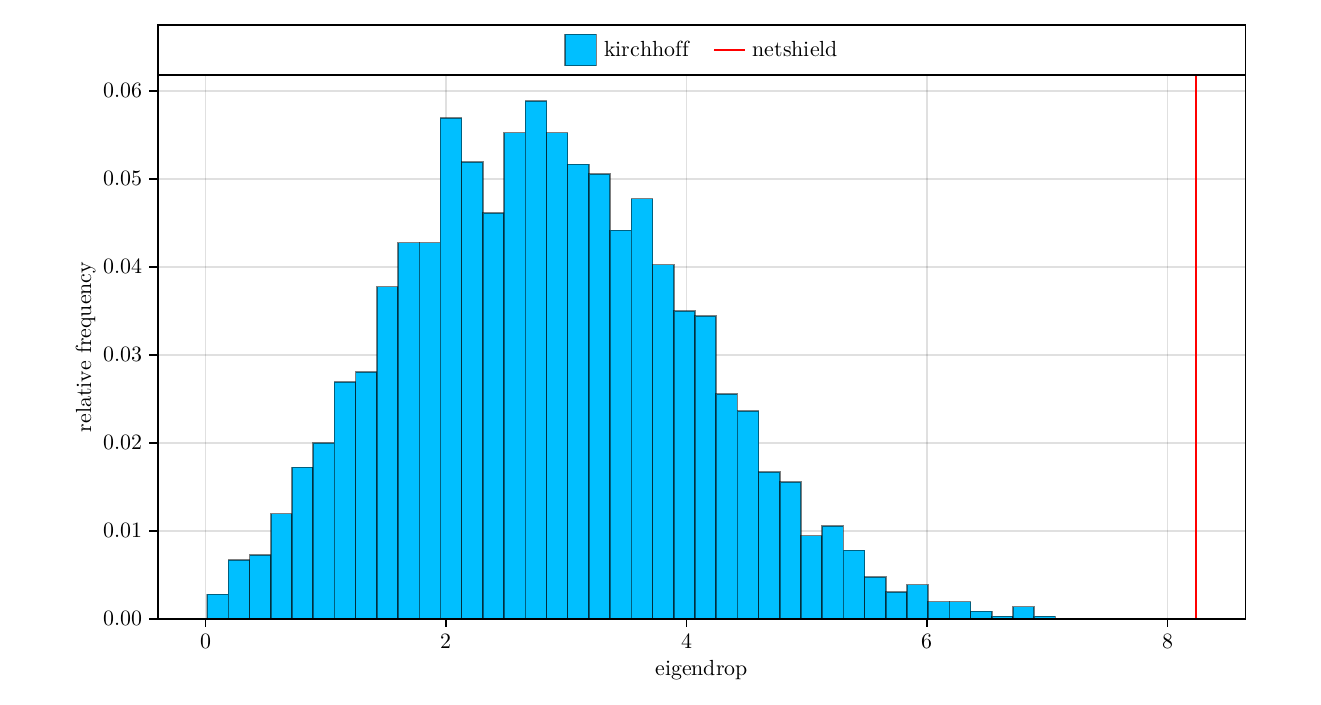}
          \caption{$k = 5$}
      \end{subcaptionblock}
      \\
      \begin{subcaptionblock}{0.45\textwidth}
          \includegraphics[width=\textwidth]{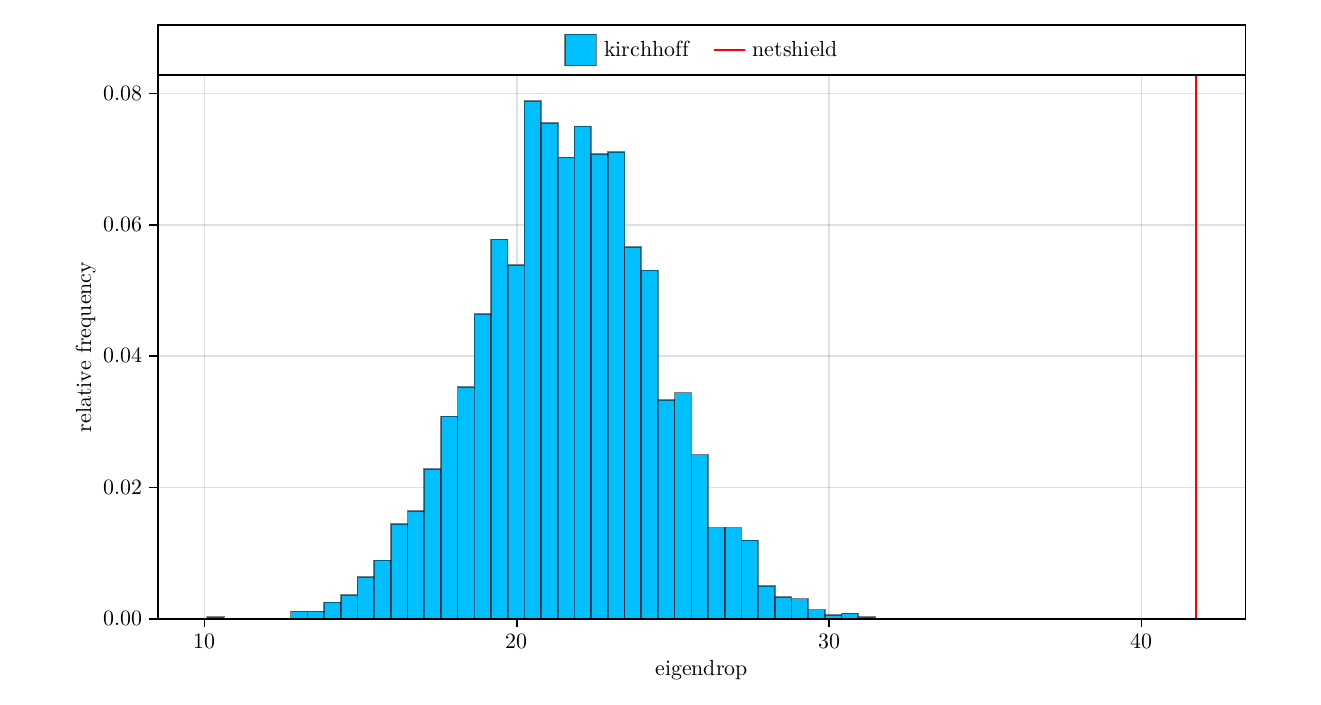}
          \caption{$k = 50$}
      \end{subcaptionblock}
      \begin{subcaptionblock}{0.45\textwidth}
          \includegraphics[width=\textwidth]{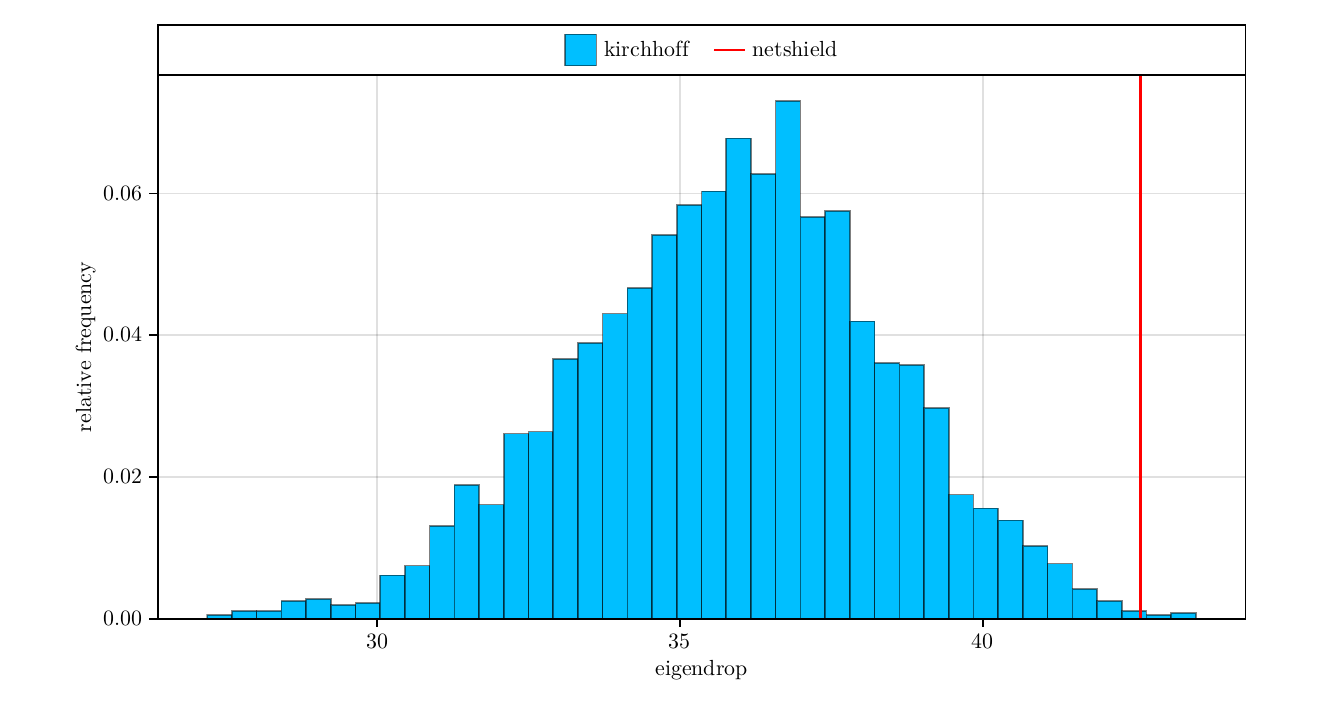}
          \caption{$k = 125$}
      \end{subcaptionblock}
\end{figure}

\begin{figure}
    \centering $ $
    \caption{Graph: ``airport 2'' (weighted) - eigendrop distribution}
      \begin{subcaptionblock}{0.45\textwidth}
          \includegraphics[width=\textwidth]{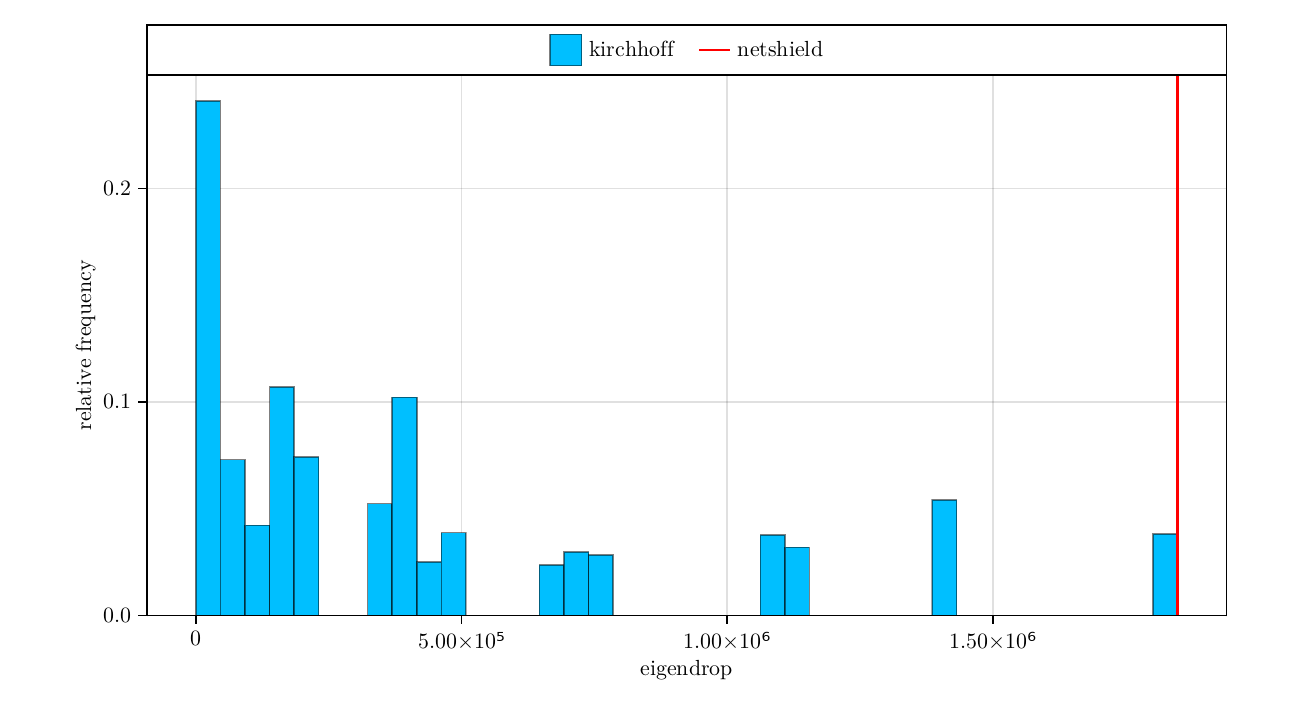}
          \caption{$k = 1$}
      \end{subcaptionblock}
      \begin{subcaptionblock}{0.45\textwidth}
          \includegraphics[width=\textwidth]{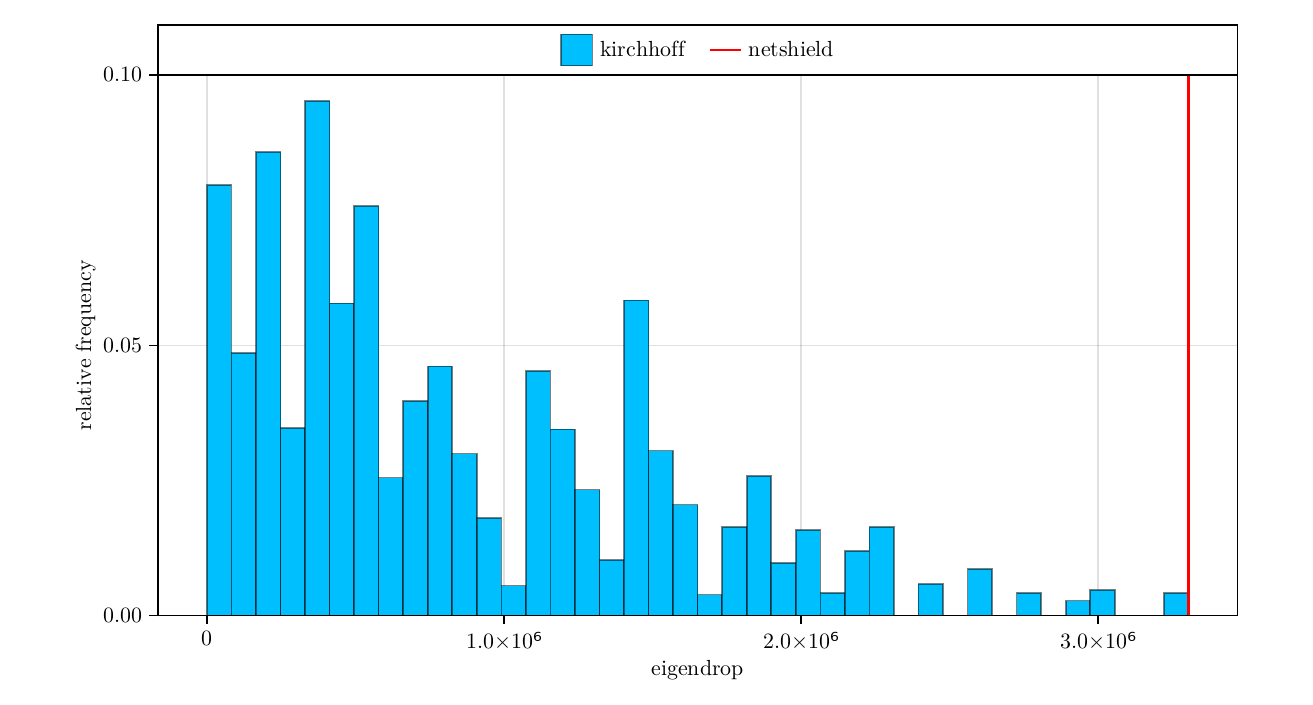}
          \caption{$k = 2$}
      \end{subcaptionblock}
      \\
      \begin{subcaptionblock}{0.45\textwidth}
          \includegraphics[width=\textwidth]{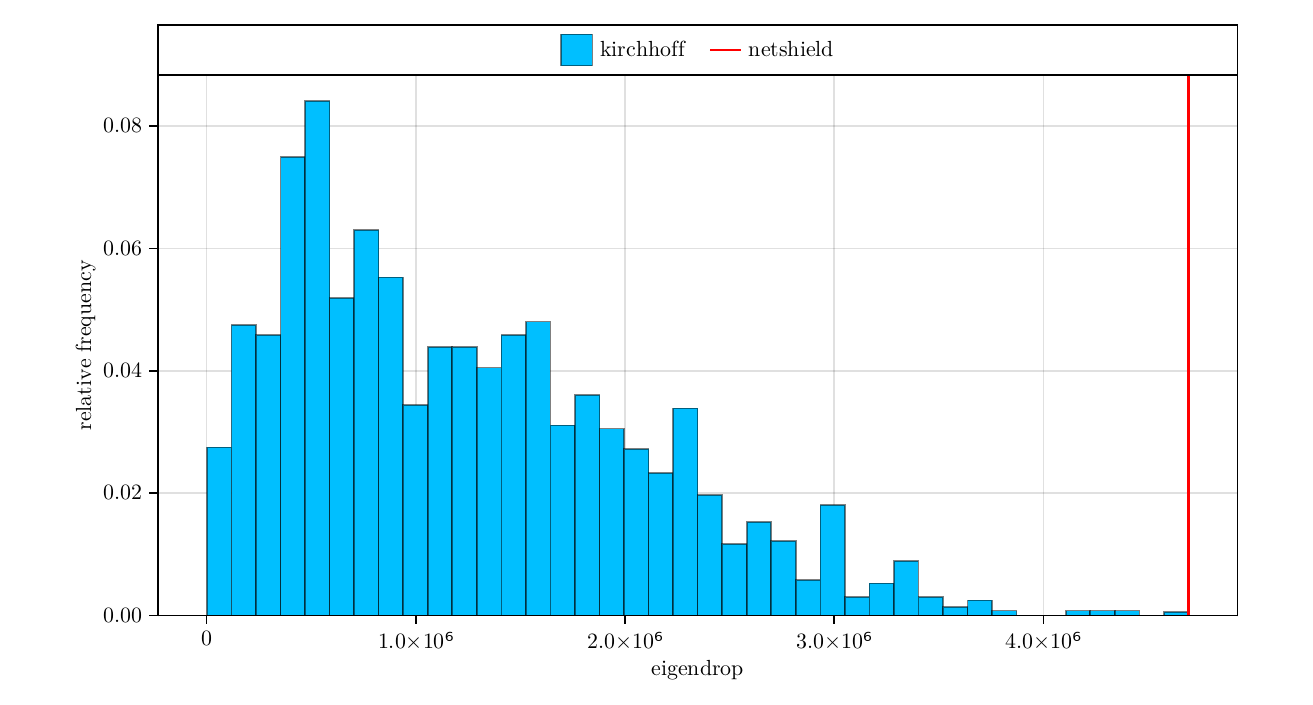}
          \caption{$k = 3$}
      \end{subcaptionblock}
      \begin{subcaptionblock}{0.45\textwidth}
          \includegraphics[width=\textwidth]{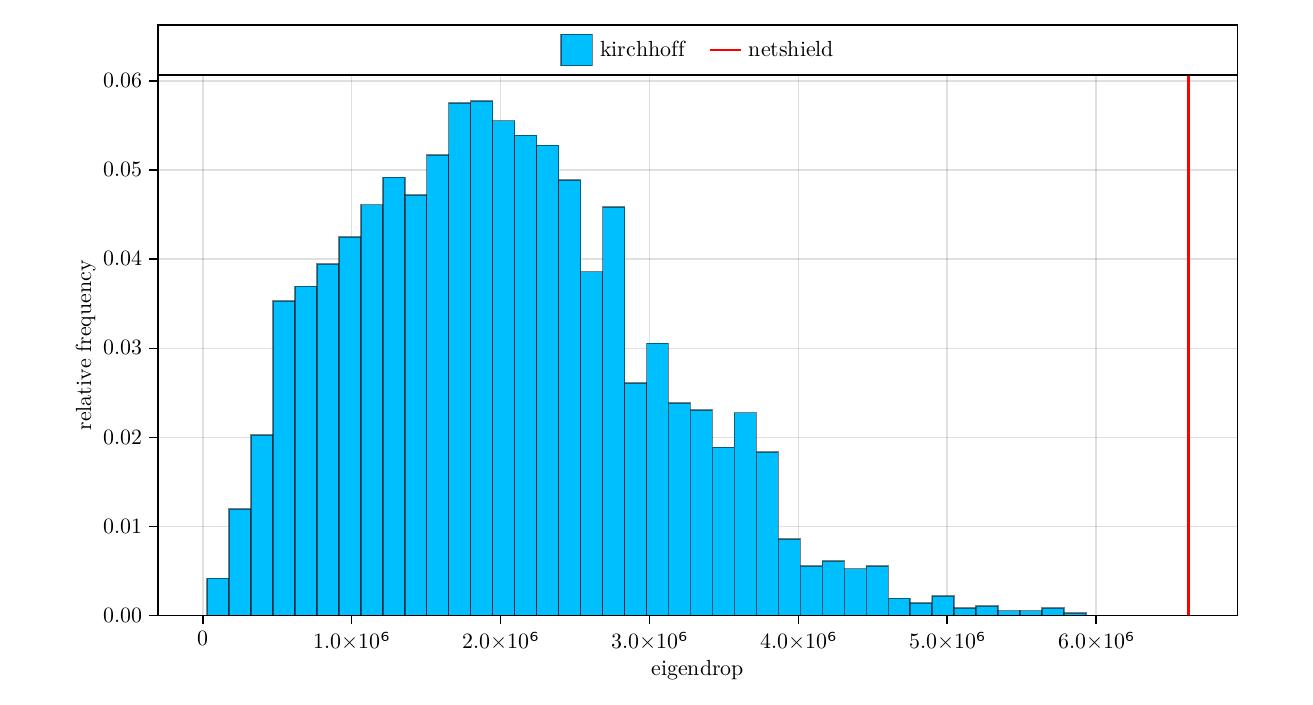}
          \caption{$k = 5$}
      \end{subcaptionblock}
      \\
      \begin{subcaptionblock}{0.45\textwidth}
          \includegraphics[width=\textwidth]{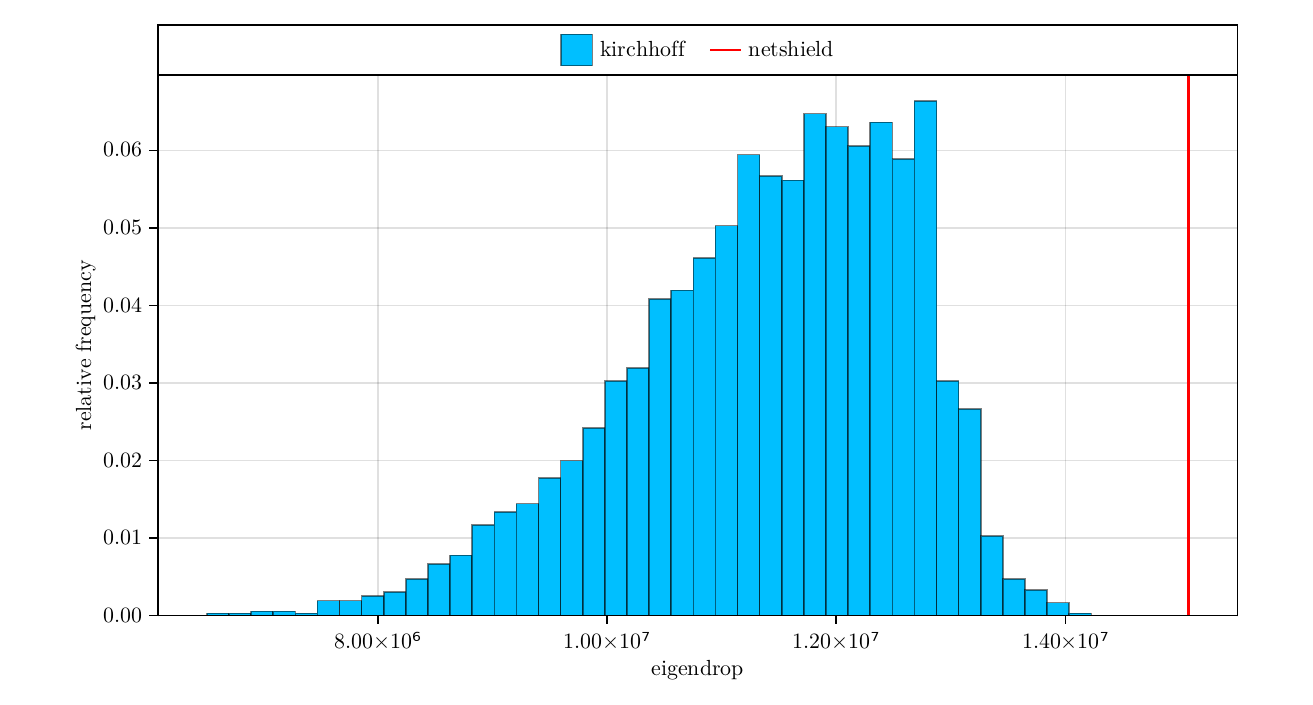}
          \caption{$k = 50$}
      \end{subcaptionblock}
      \begin{subcaptionblock}{0.45\textwidth}
          \includegraphics[width=\textwidth]{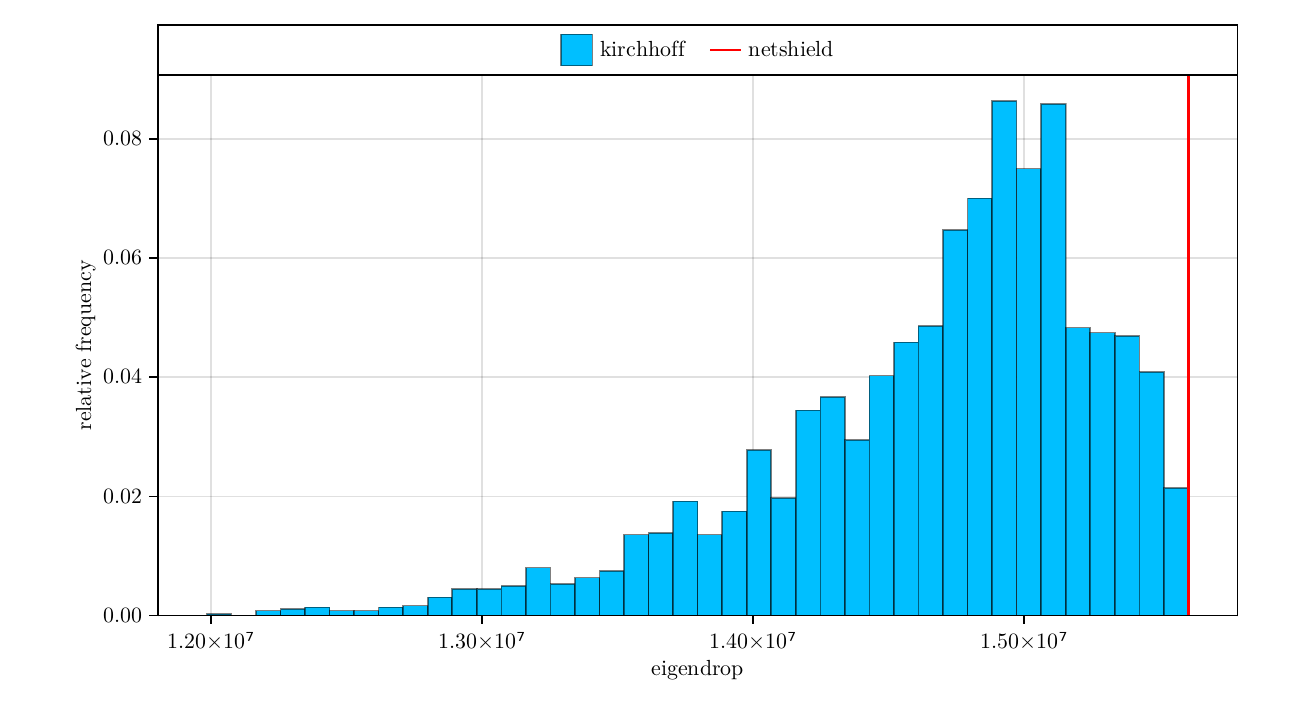}
          \caption{$k = 125$}
      \end{subcaptionblock}
\end{figure}

\begin{figure}
    \centering $ $
    \caption{Graph: ``airport 3'' (non-weighted) - eigendrop distribution}
      \begin{subcaptionblock}{0.45\textwidth}
          \includegraphics[width=\textwidth]{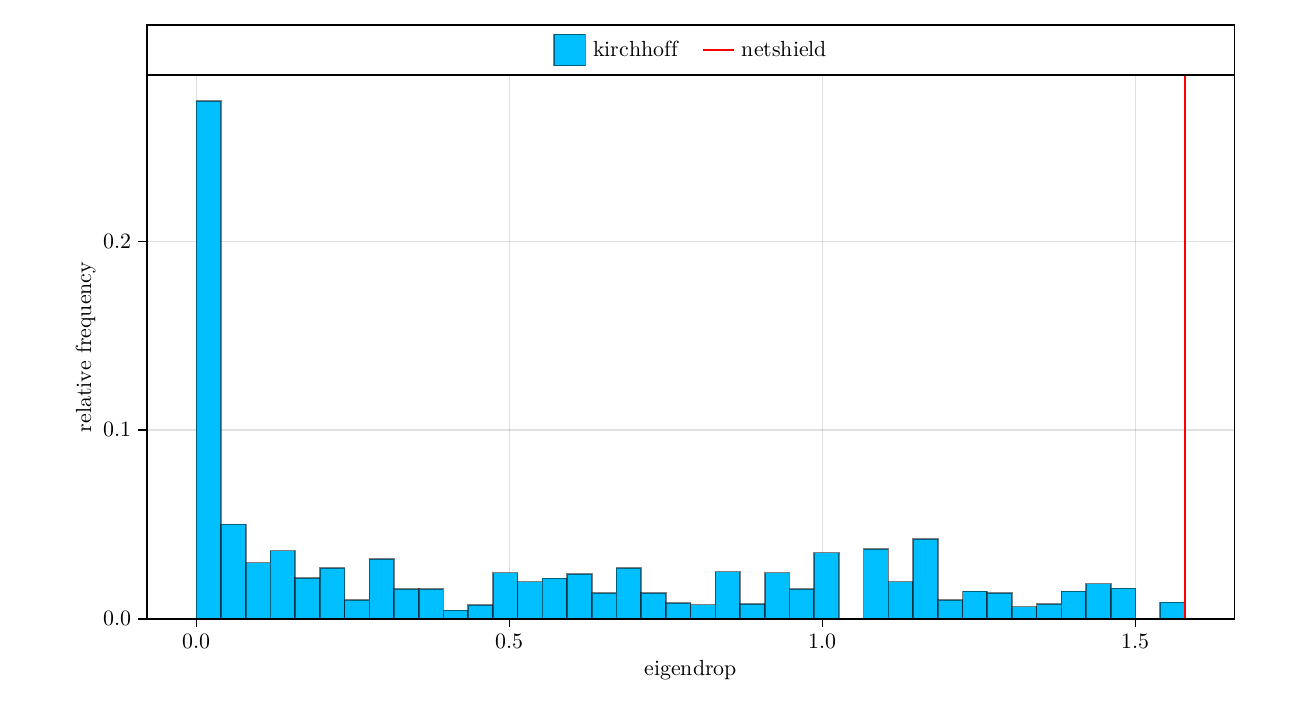}
          \caption{$k = 1$}
      \end{subcaptionblock}
      \begin{subcaptionblock}{0.45\textwidth}
          \includegraphics[width=\textwidth]{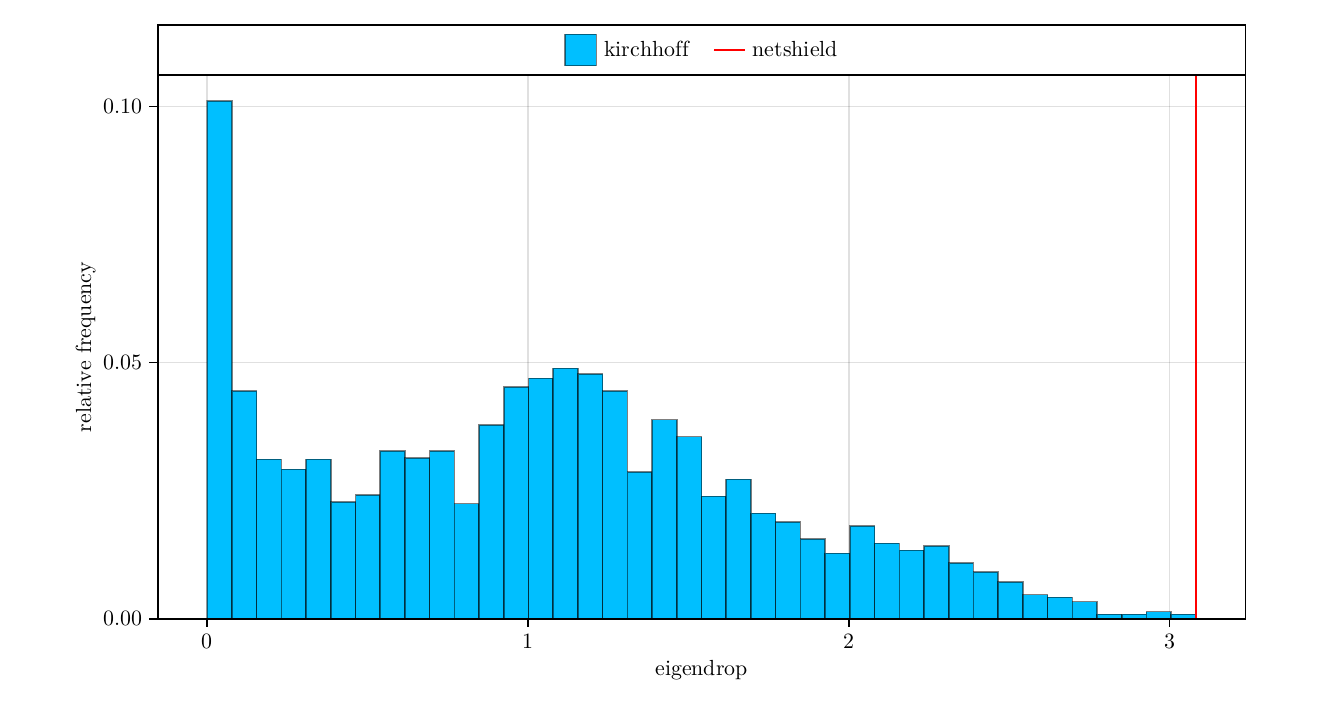}
          \caption{$k = 2$}
      \end{subcaptionblock}
      \\
      \begin{subcaptionblock}{0.45\textwidth}
          \includegraphics[width=\textwidth]{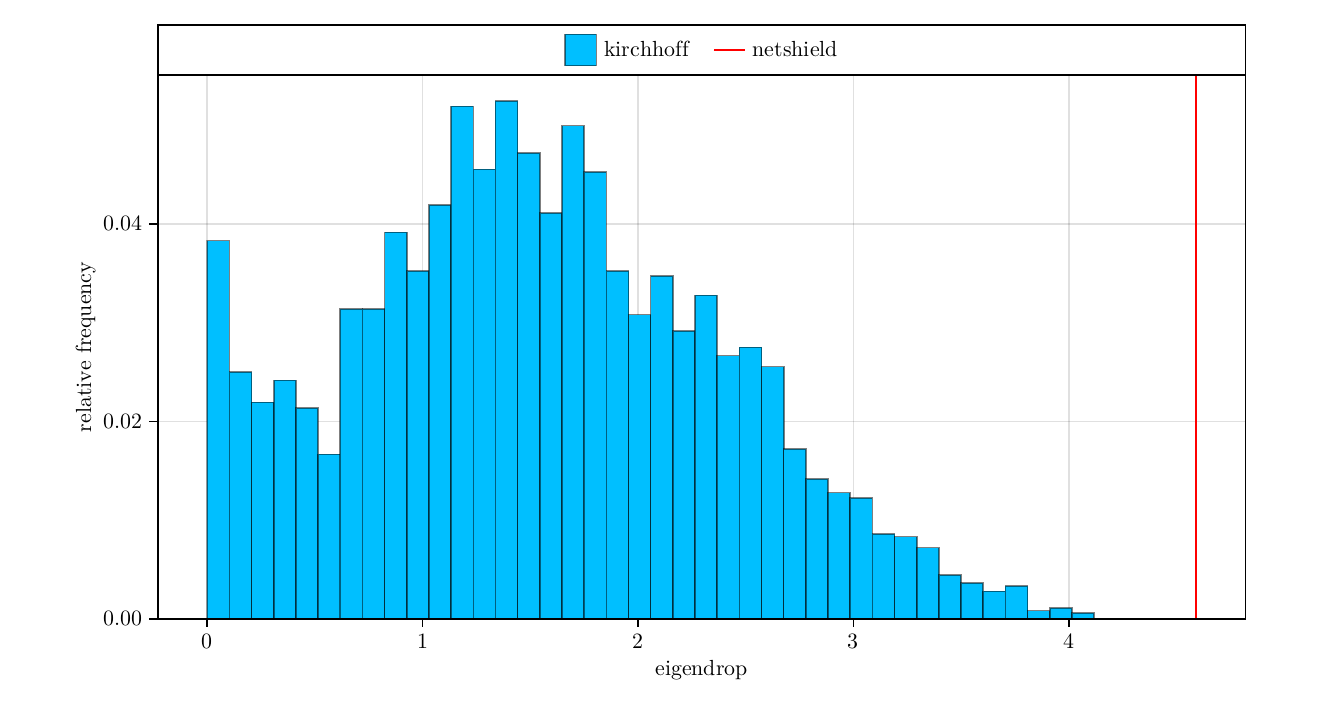}
          \caption{$k = 3$}
      \end{subcaptionblock}
      \begin{subcaptionblock}{0.45\textwidth}
          \includegraphics[width=\textwidth]{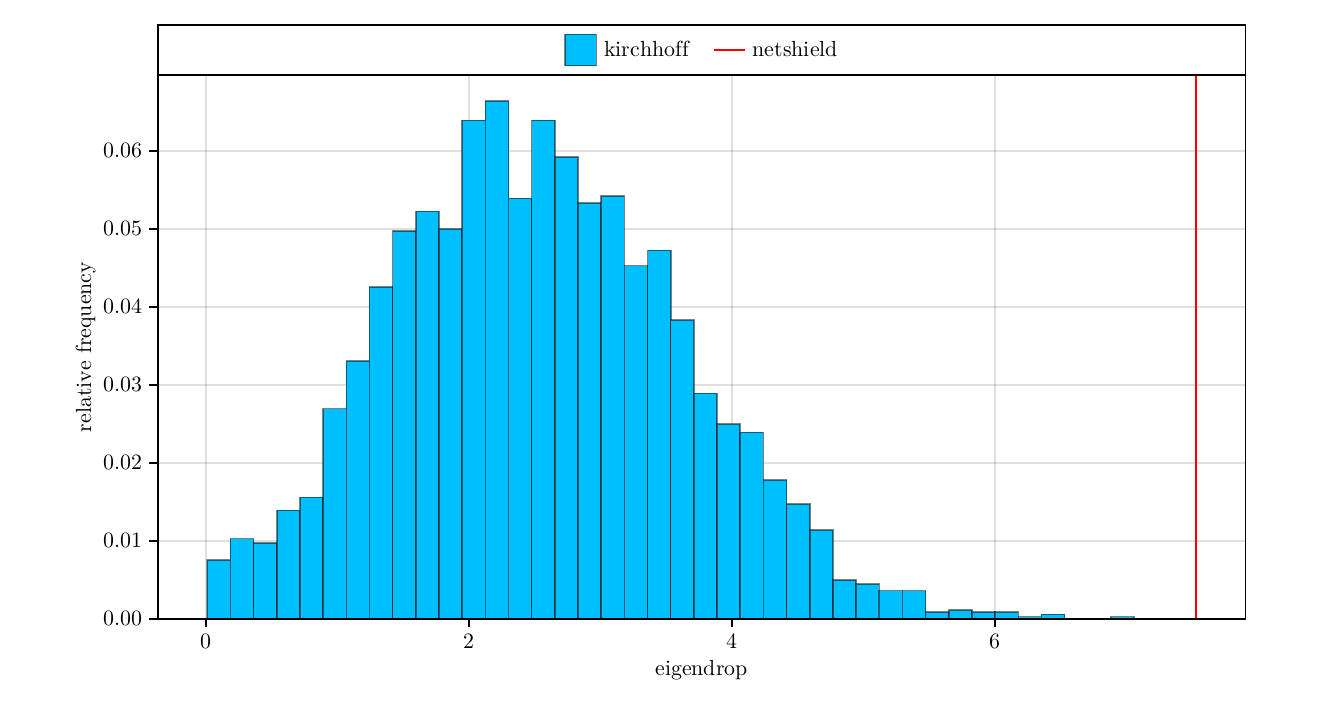}
          \caption{$k = 5$}
      \end{subcaptionblock}
      \\
      \begin{subcaptionblock}{0.45\textwidth}
          \includegraphics[width=\textwidth]{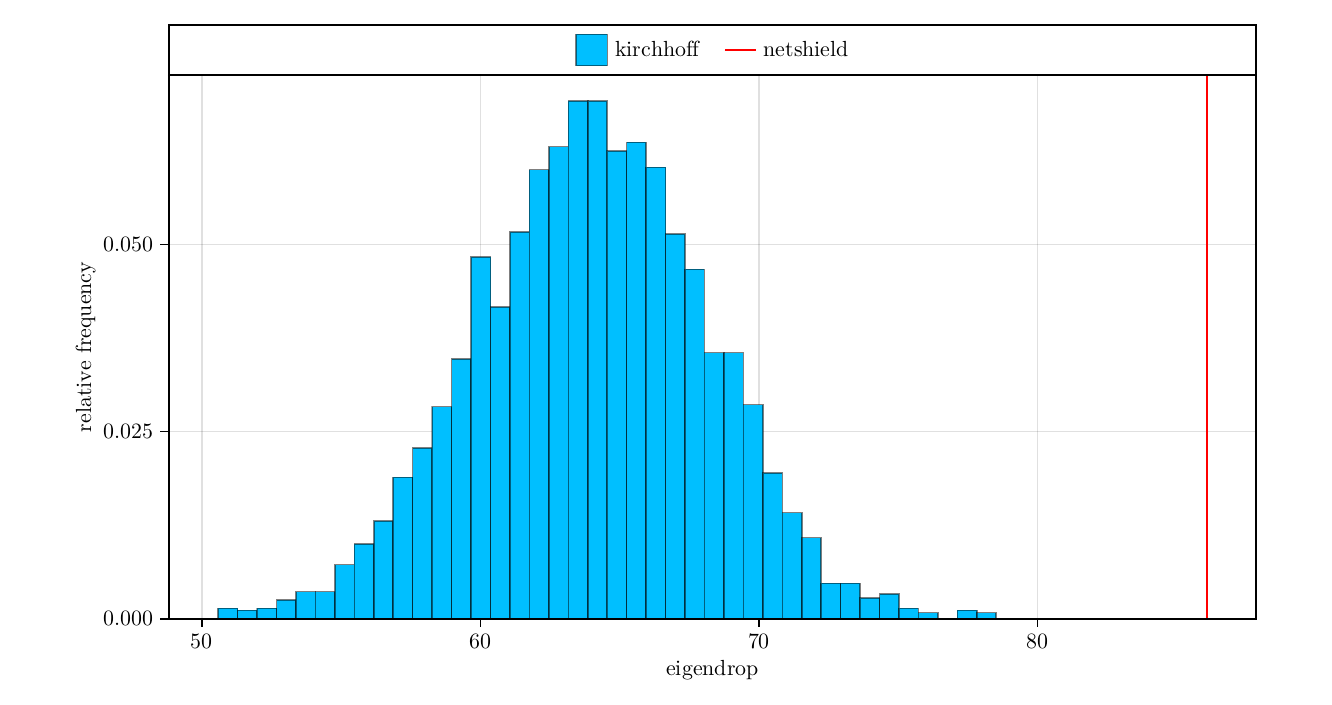}
          \caption{$k = 186$}
      \end{subcaptionblock}
      \begin{subcaptionblock}{0.45\textwidth}
          \includegraphics[width=\textwidth]{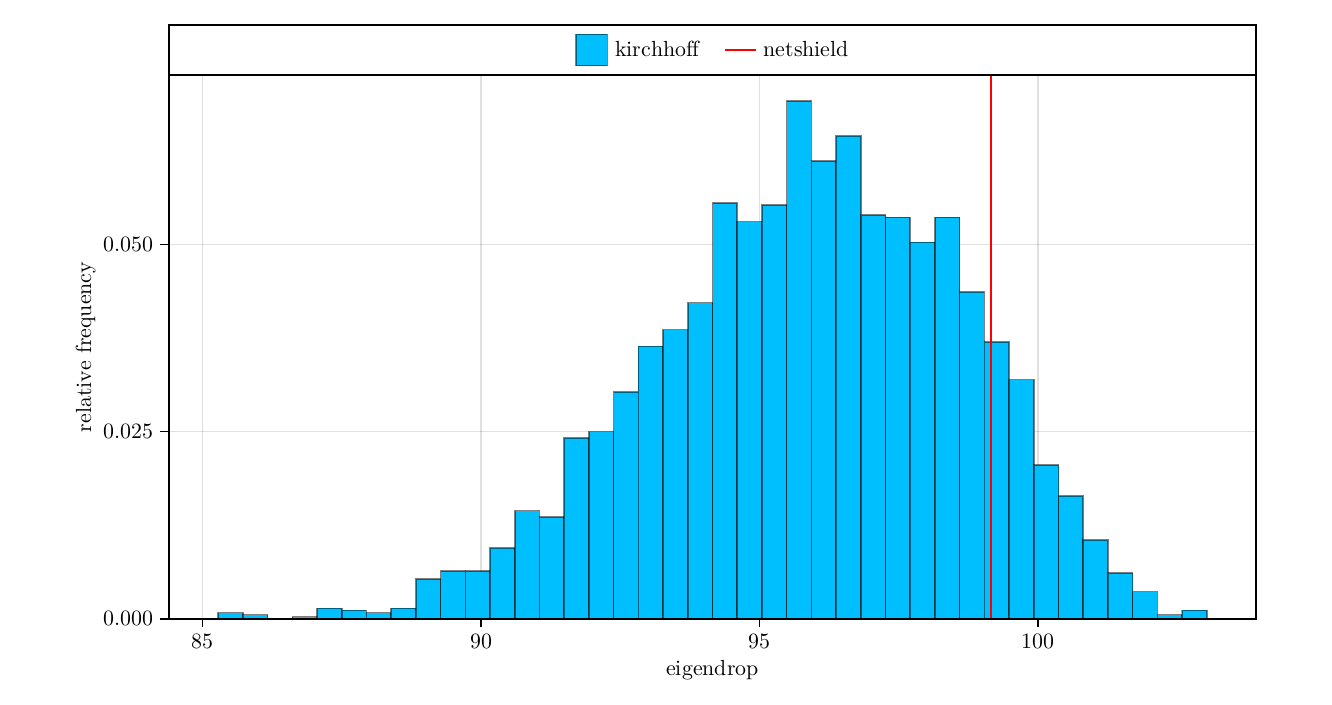}
          \caption{$k = 464$}
      \end{subcaptionblock}
\end{figure}

\begin{figure}
    \centering $ $
    \caption{Graph: ``airport 4'' (non-weighted) - eigendrop distribution}
      \begin{subcaptionblock}{0.45\textwidth}
          \includegraphics[width=\textwidth]{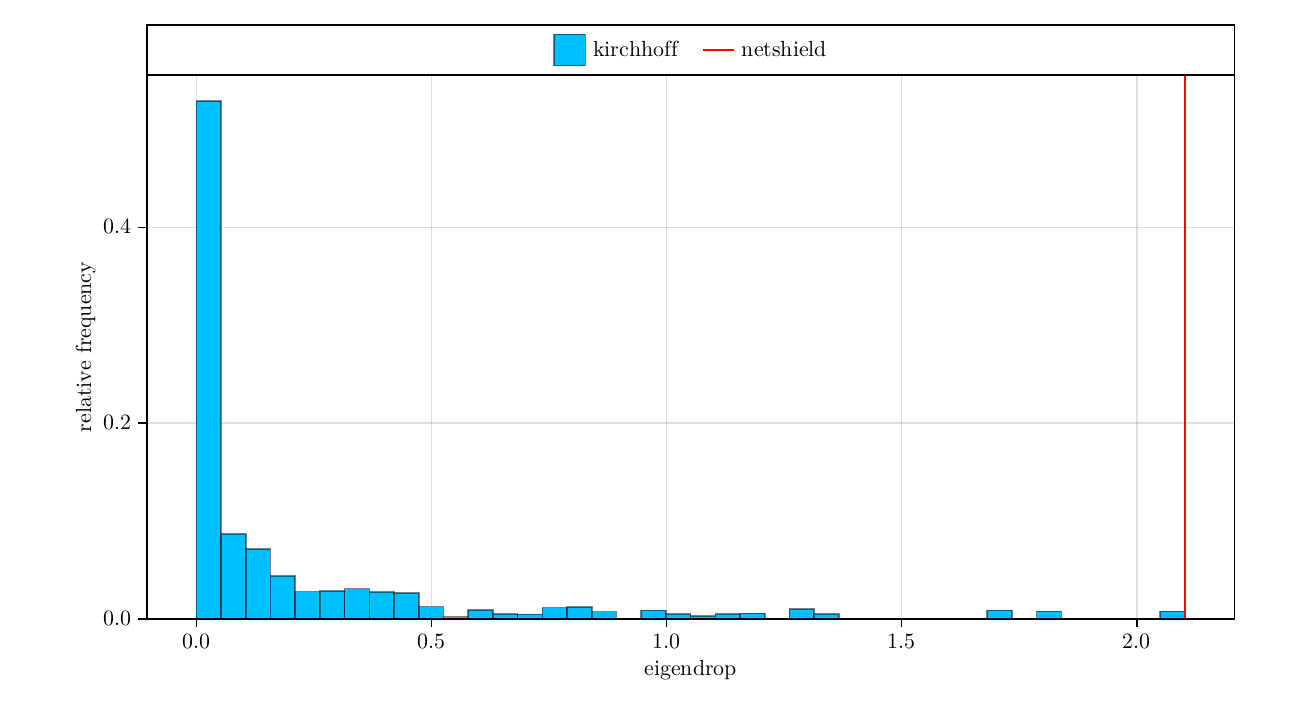}
          \caption{$k = 1$}
      \end{subcaptionblock}
      \begin{subcaptionblock}{0.45\textwidth}
          \includegraphics[width=\textwidth]{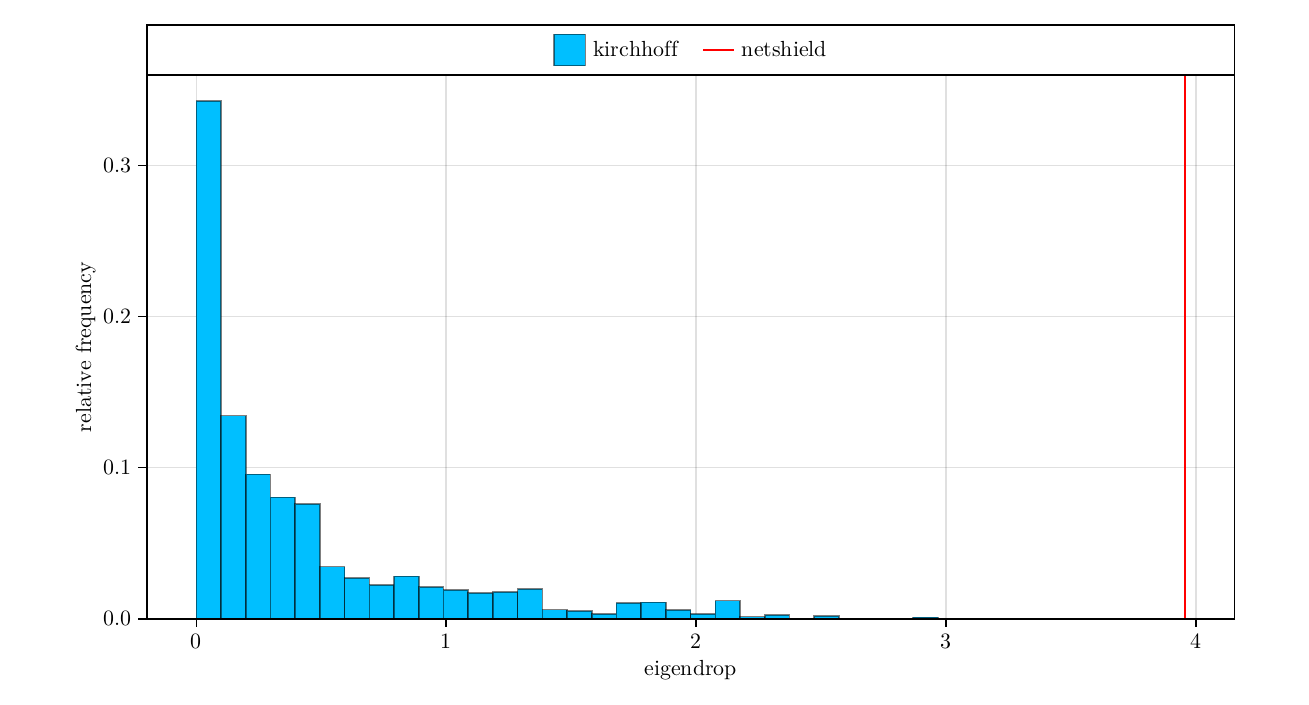}
          \caption{$k = 2$}
      \end{subcaptionblock}
      \\
      \begin{subcaptionblock}{0.45\textwidth}
          \includegraphics[width=\textwidth]{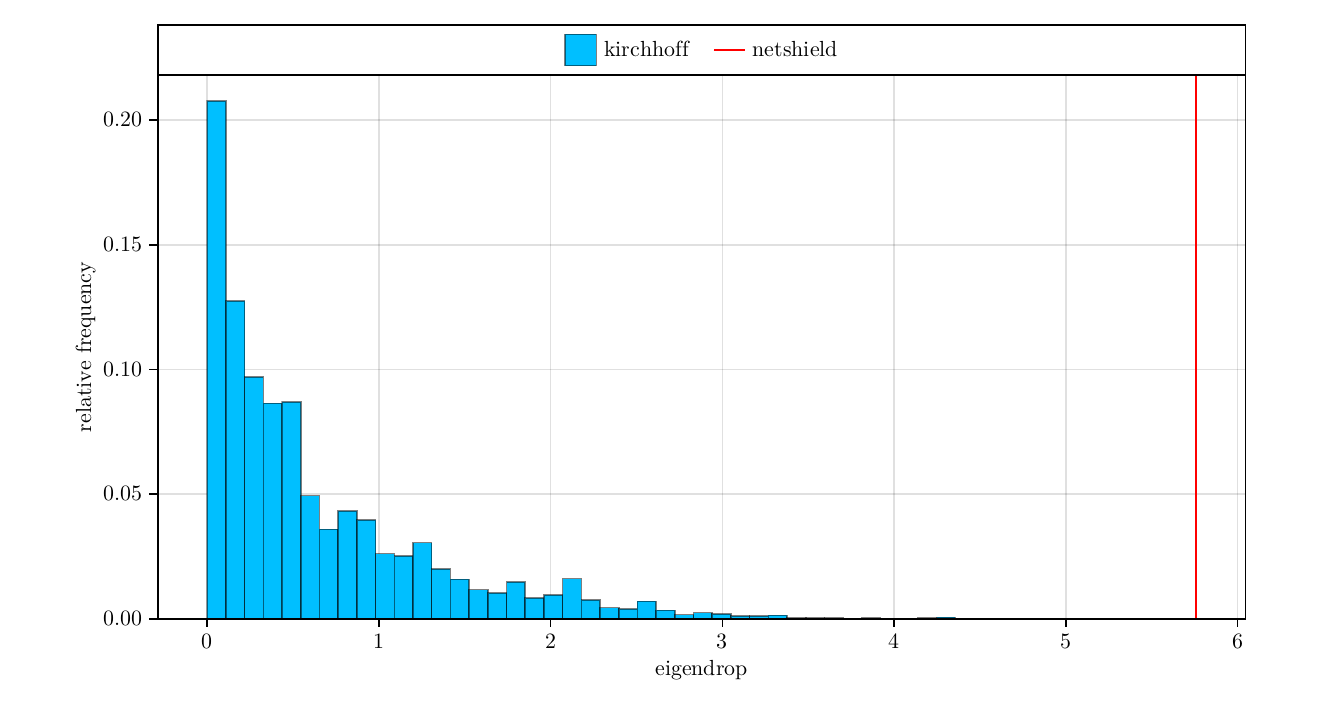}
          \caption{$k = 3$}
      \end{subcaptionblock}
      \begin{subcaptionblock}{0.45\textwidth}
          \includegraphics[width=\textwidth]{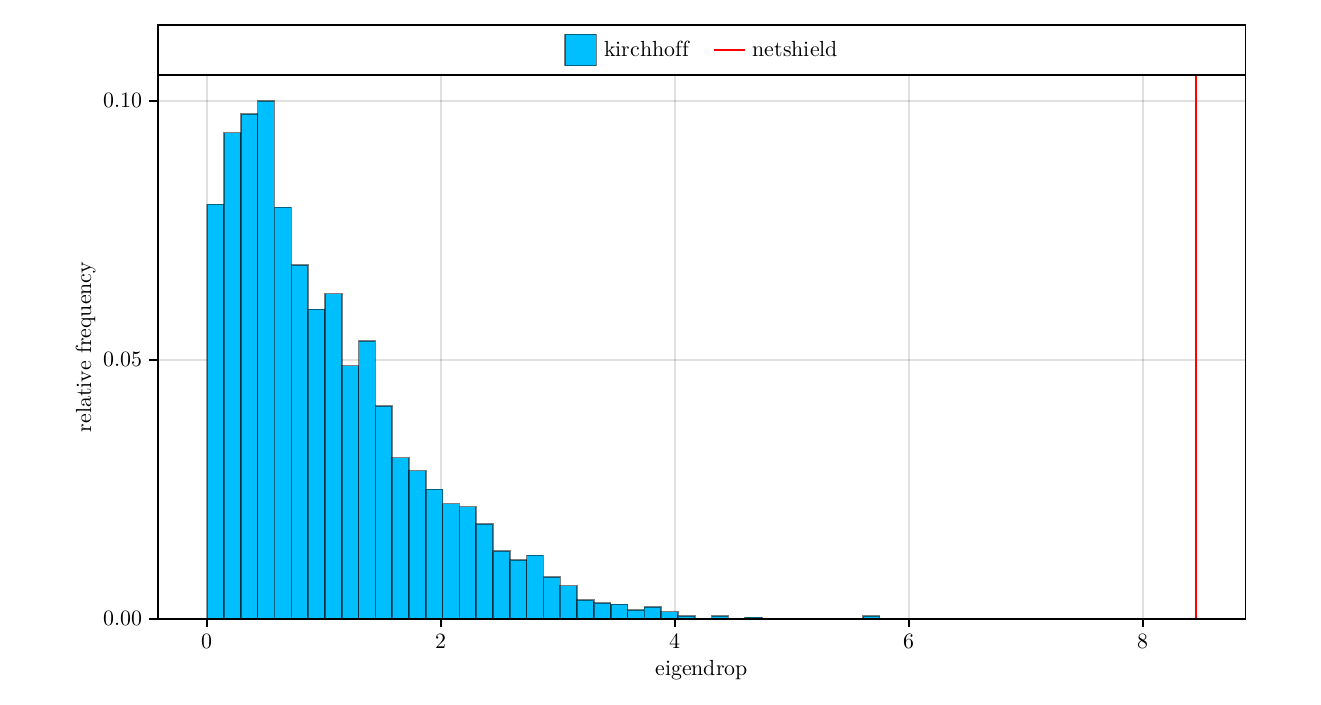}
          \caption{$k = 5$}
      \end{subcaptionblock}
      \\
      \begin{subcaptionblock}{0.45\textwidth}
          \includegraphics[width=\textwidth]{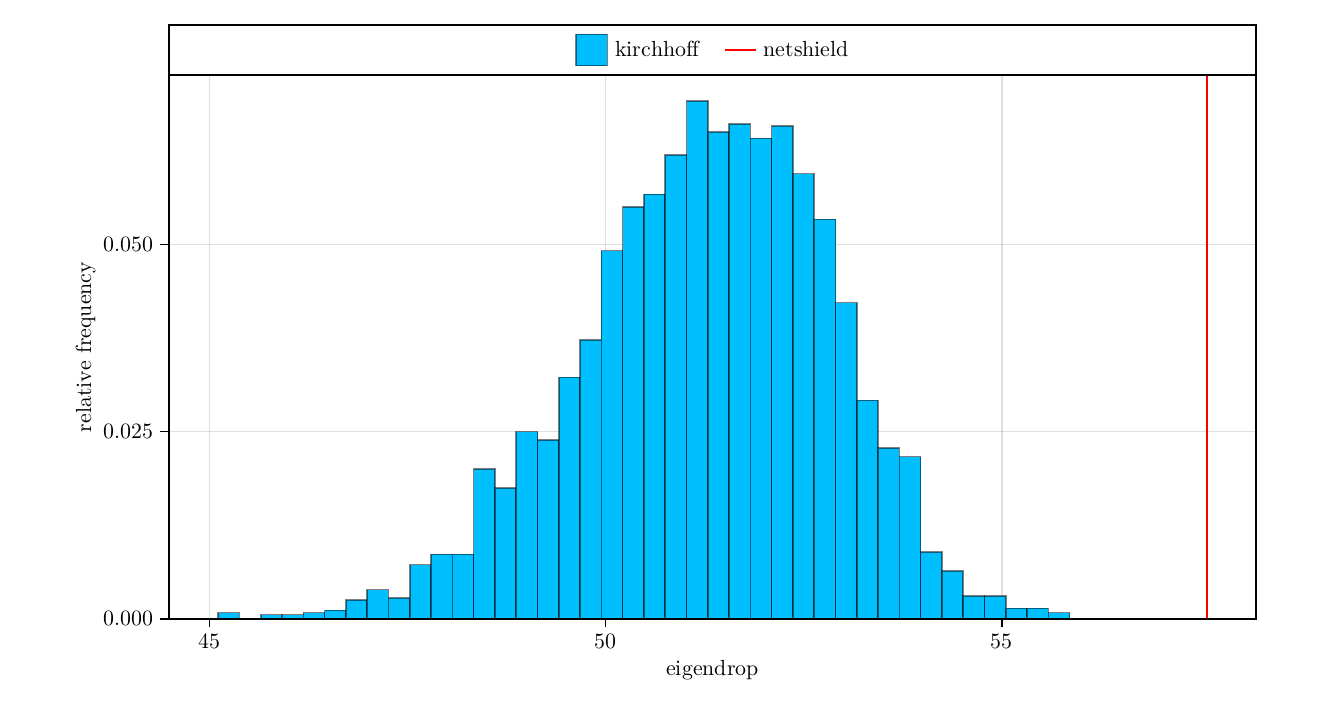}
          \caption{$k = 798$}
      \end{subcaptionblock}
      \begin{subcaptionblock}{0.45\textwidth}
          \includegraphics[width=\textwidth]{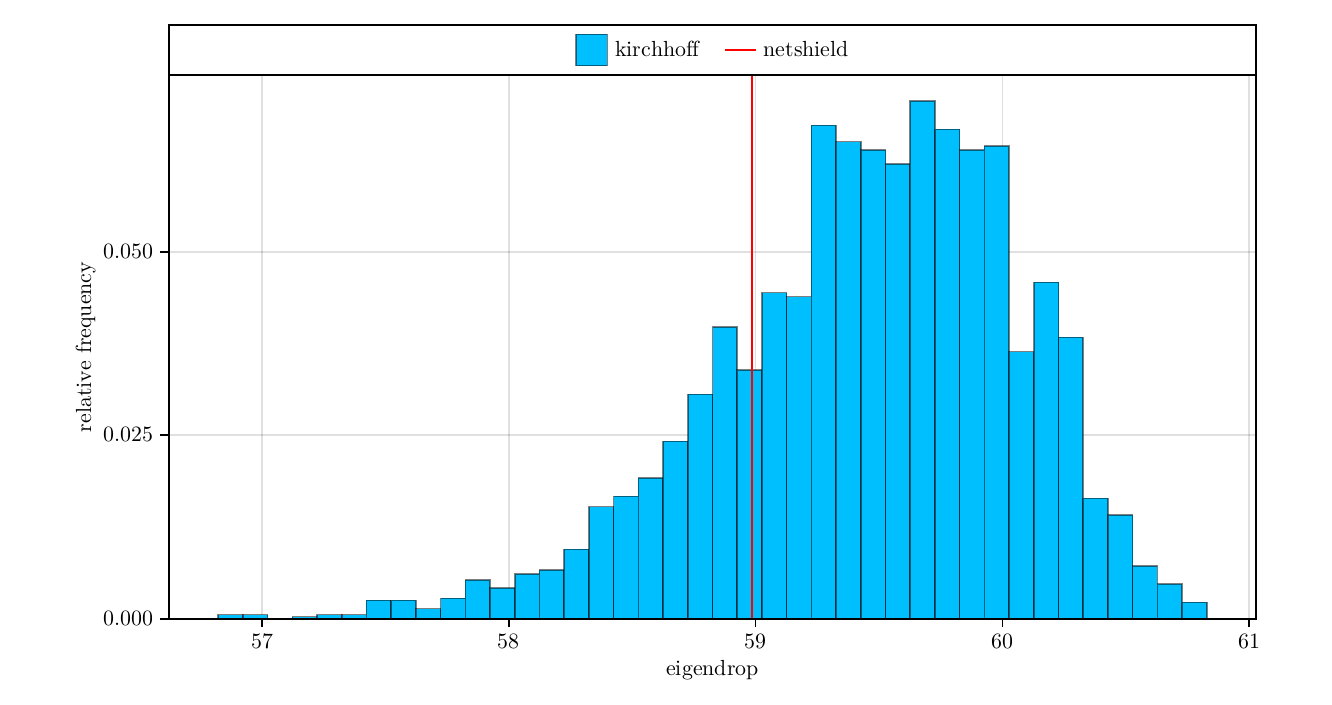}
          \caption{$k = 1994$}
      \end{subcaptionblock}
\end{figure}

\begin{figure}
    \centering $ $
    \caption{Graph: ``airport 4'' (weighted) - eigendrop distribution}
      \begin{subcaptionblock}{0.45\textwidth}
          \includegraphics[width=\textwidth]{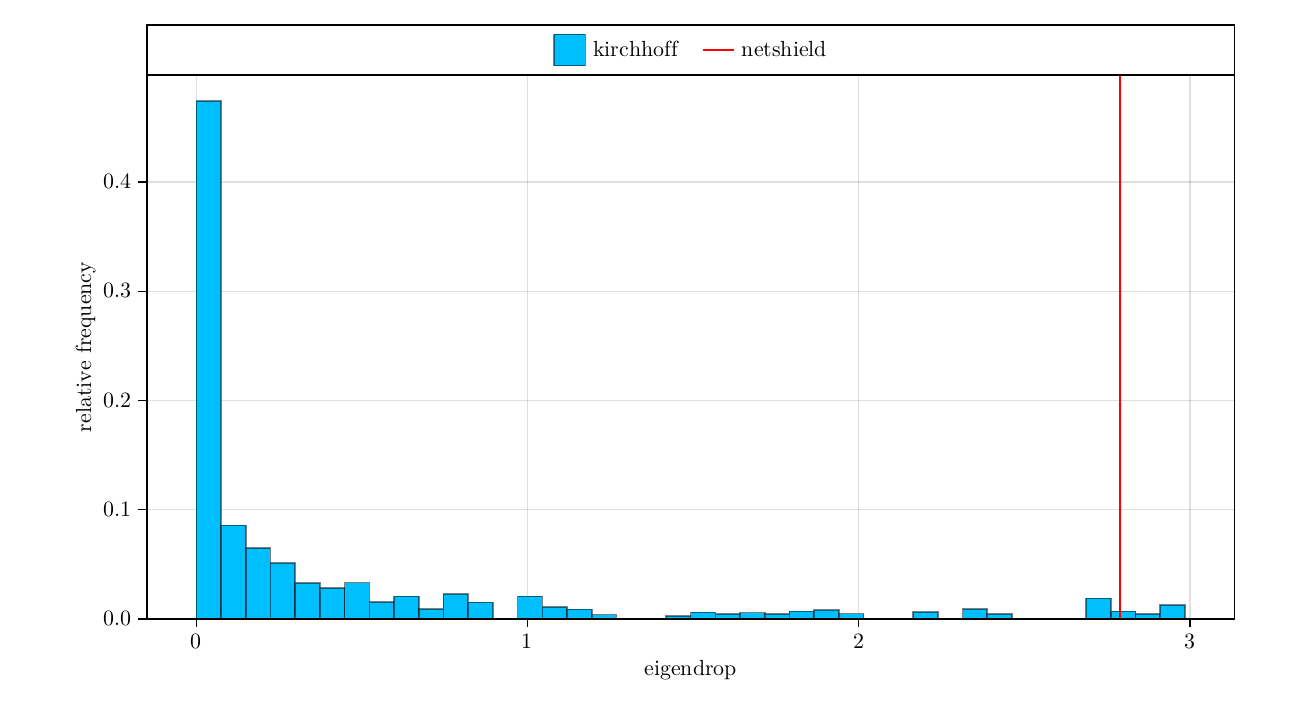}
          \caption{$k = 1$}
      \end{subcaptionblock}
      \begin{subcaptionblock}{0.45\textwidth}
          \includegraphics[width=\textwidth]{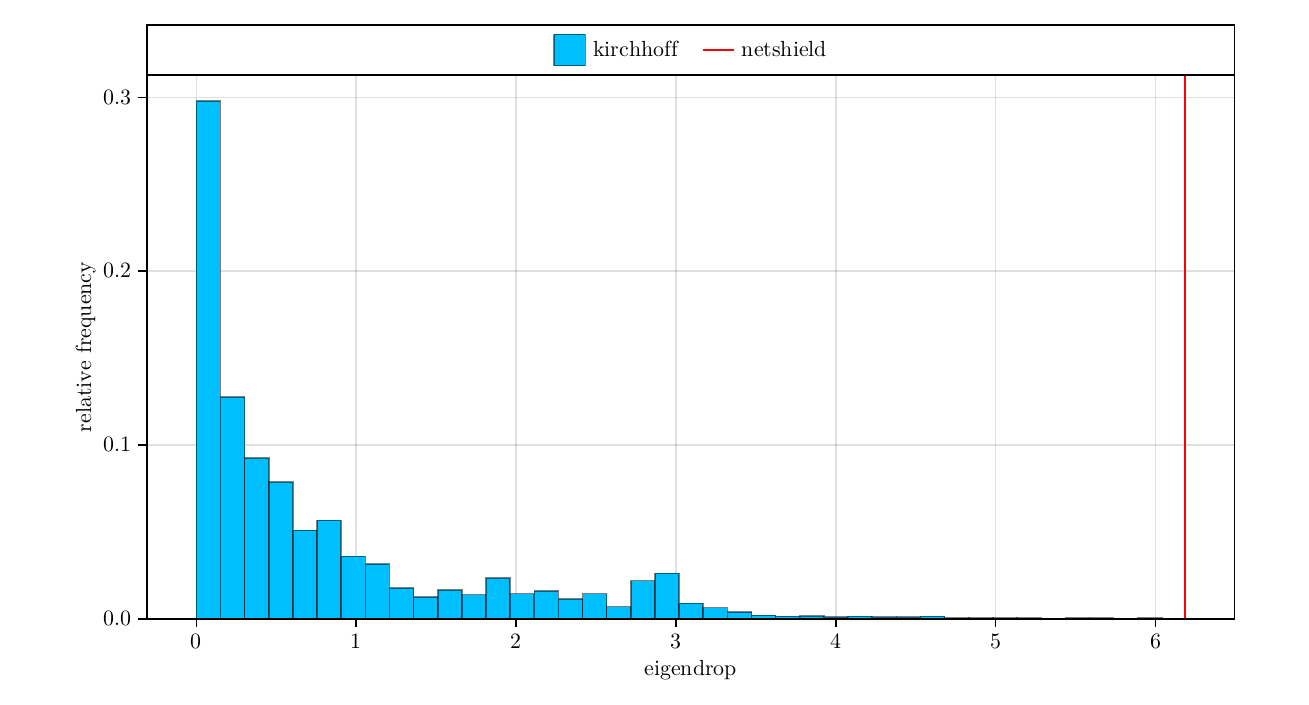}
          \caption{$k = 2$}
      \end{subcaptionblock}
      \\
      \begin{subcaptionblock}{0.45\textwidth}
          \includegraphics[width=\textwidth]{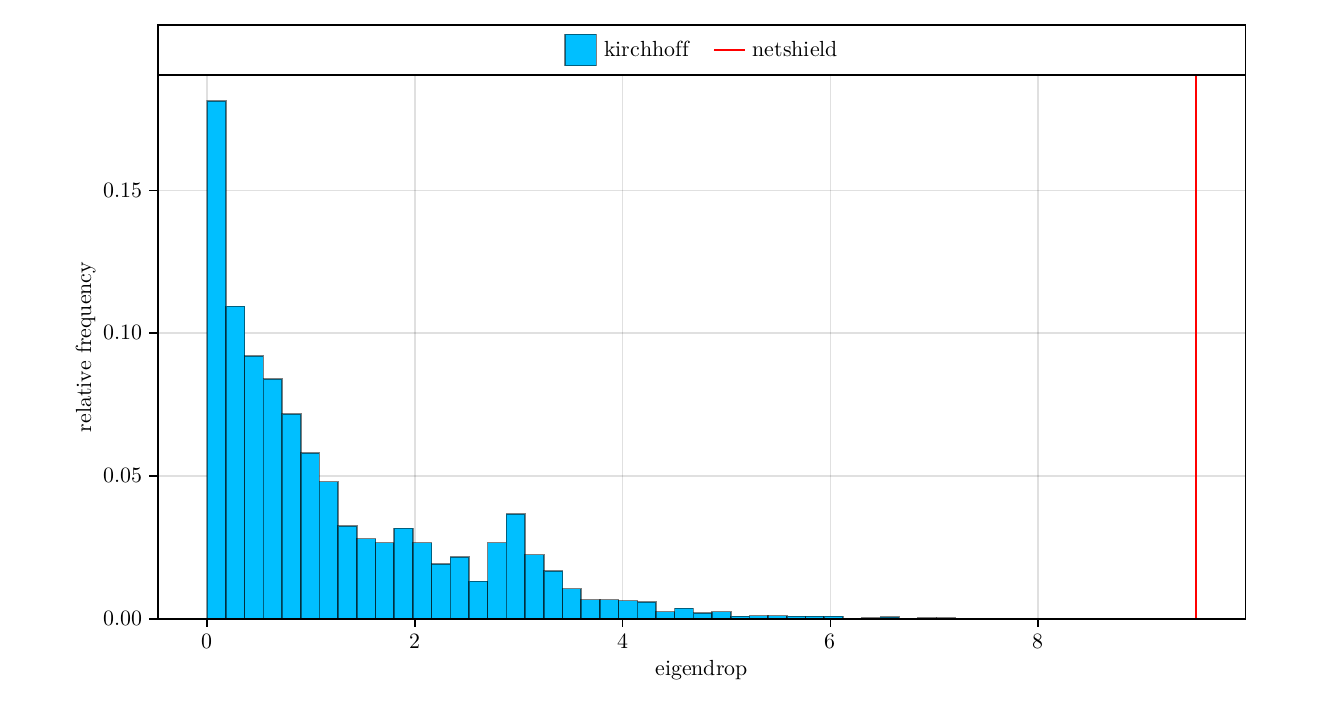}
          \caption{$k = 3$}
      \end{subcaptionblock}
      \begin{subcaptionblock}{0.45\textwidth}
          \includegraphics[width=\textwidth]{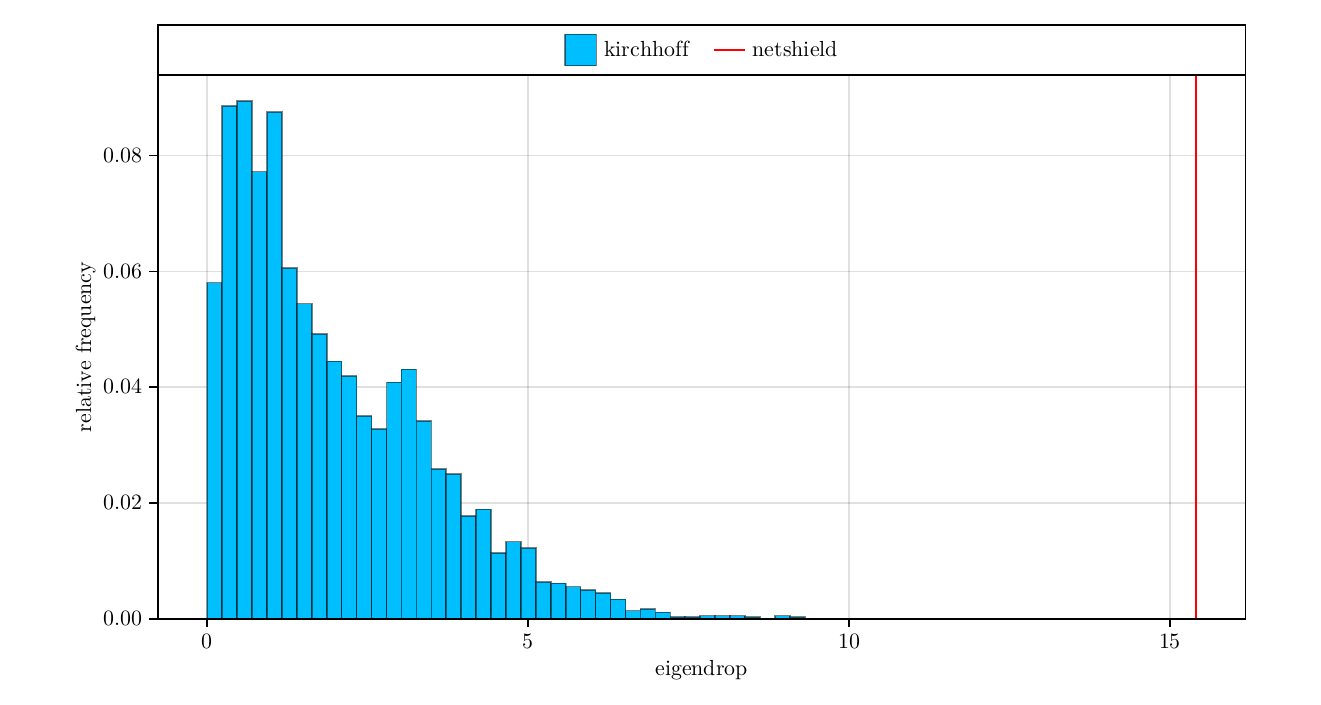}
          \caption{$k = 5$}
      \end{subcaptionblock}
      \\
      \begin{subcaptionblock}{0.45\textwidth}
          \includegraphics[width=\textwidth]{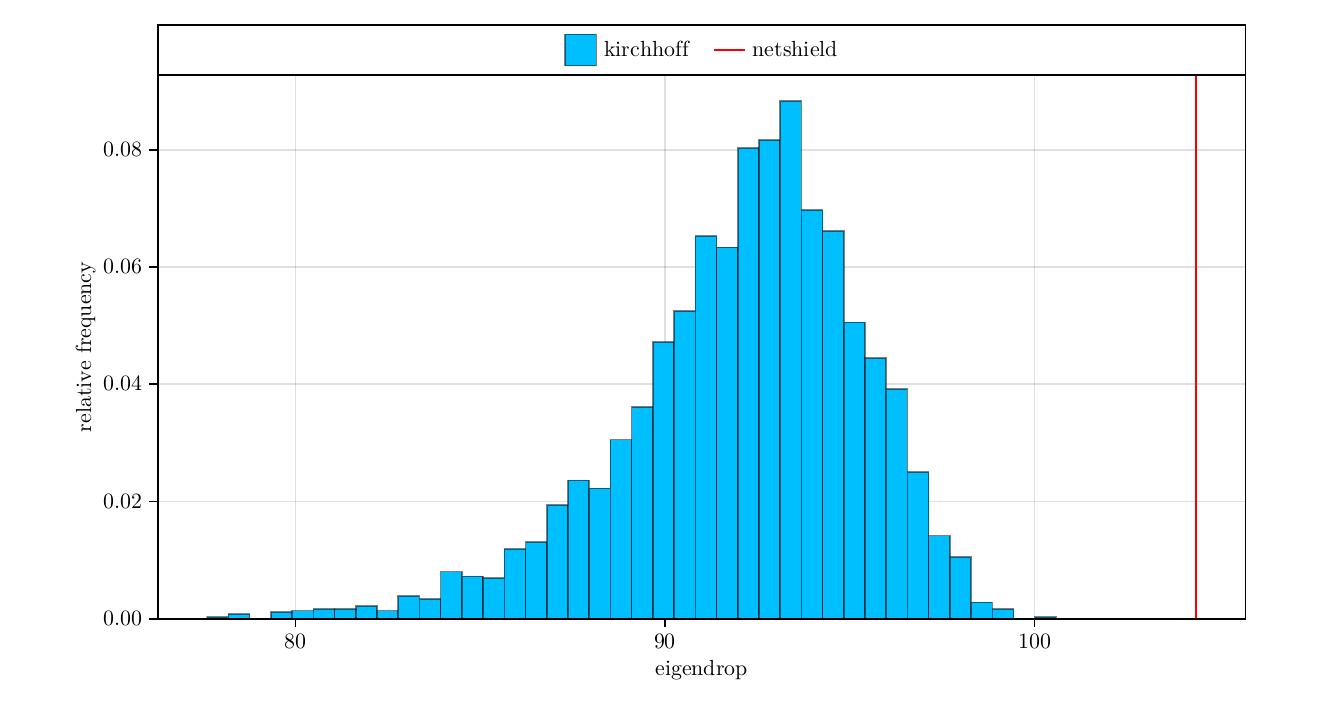}
          \caption{$k = 798$}
      \end{subcaptionblock}
      \begin{subcaptionblock}{0.45\textwidth}
          \includegraphics[width=\textwidth]{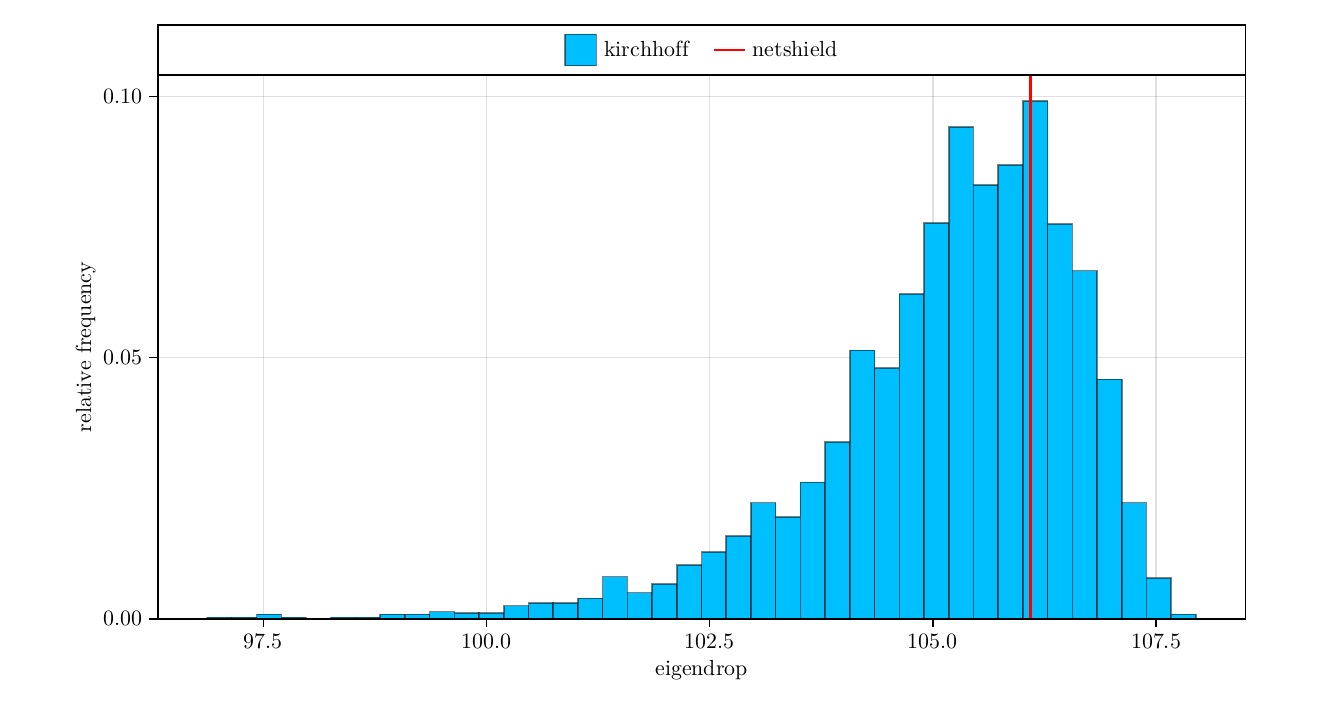}
          \caption{$k = 1994$}
      \end{subcaptionblock}
\end{figure}

\begin{figure}
    \centering $ $
    \caption{Graph: ``rfid'' (non-weighted) - eigendrop distribution}
      \begin{subcaptionblock}{0.45\textwidth}
          \includegraphics[width=\textwidth]{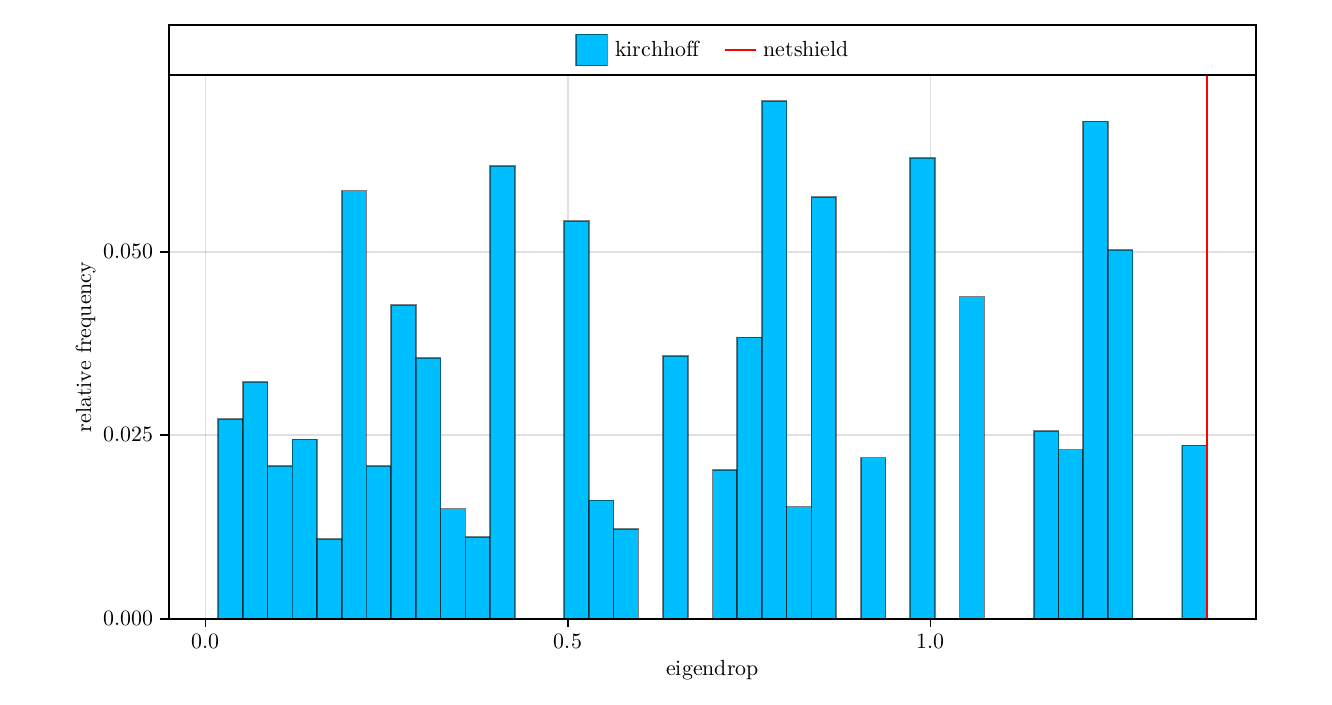}
          \caption{$k = 1$}
      \end{subcaptionblock}
      \begin{subcaptionblock}{0.45\textwidth}
          \includegraphics[width=\textwidth]{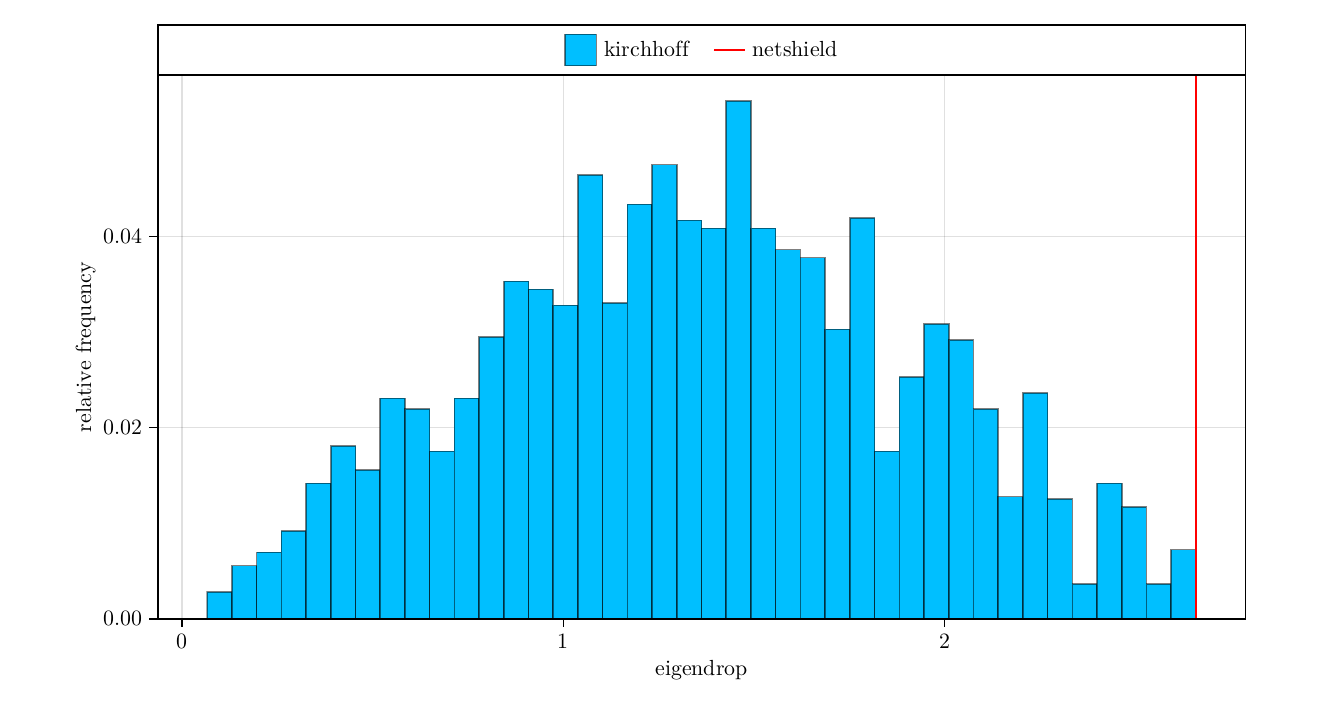}
          \caption{$k = 2$}
      \end{subcaptionblock}
      \\
      \begin{subcaptionblock}{0.45\textwidth}
          \includegraphics[width=\textwidth]{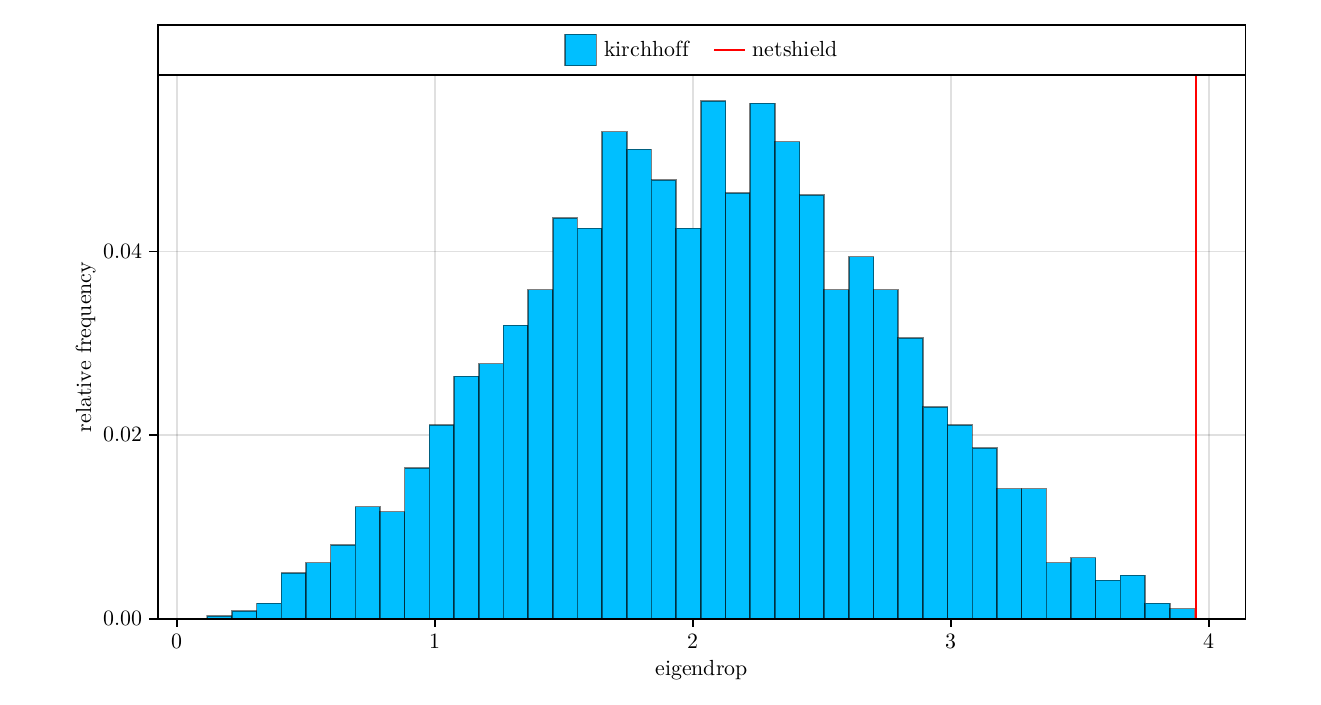}
          \caption{$k = 3$}
      \end{subcaptionblock}
      \begin{subcaptionblock}{0.45\textwidth}
          \includegraphics[width=\textwidth]{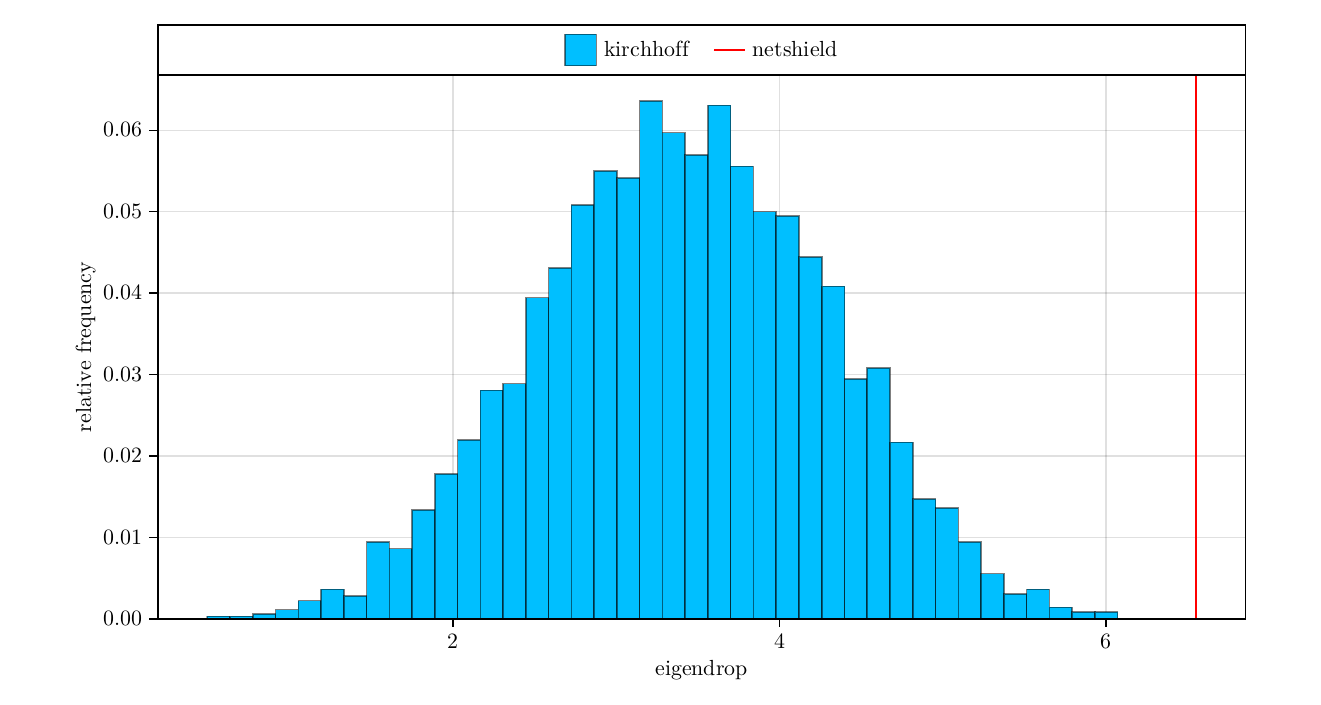}
          \caption{$k = 5$}
      \end{subcaptionblock}
      \\
      \begin{subcaptionblock}{0.45\textwidth}
          \includegraphics[width=\textwidth]{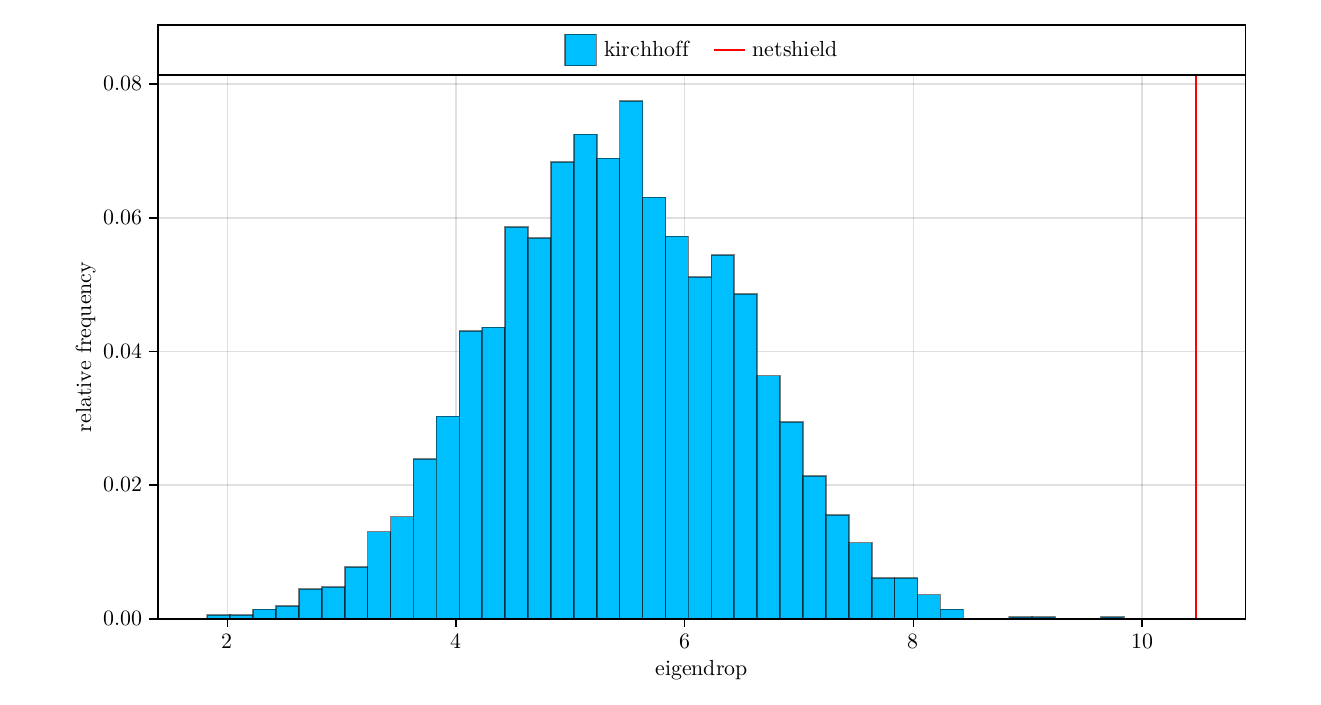}
          \caption{$k = 8$}
      \end{subcaptionblock}
      \begin{subcaptionblock}{0.45\textwidth}
          \includegraphics[width=\textwidth]{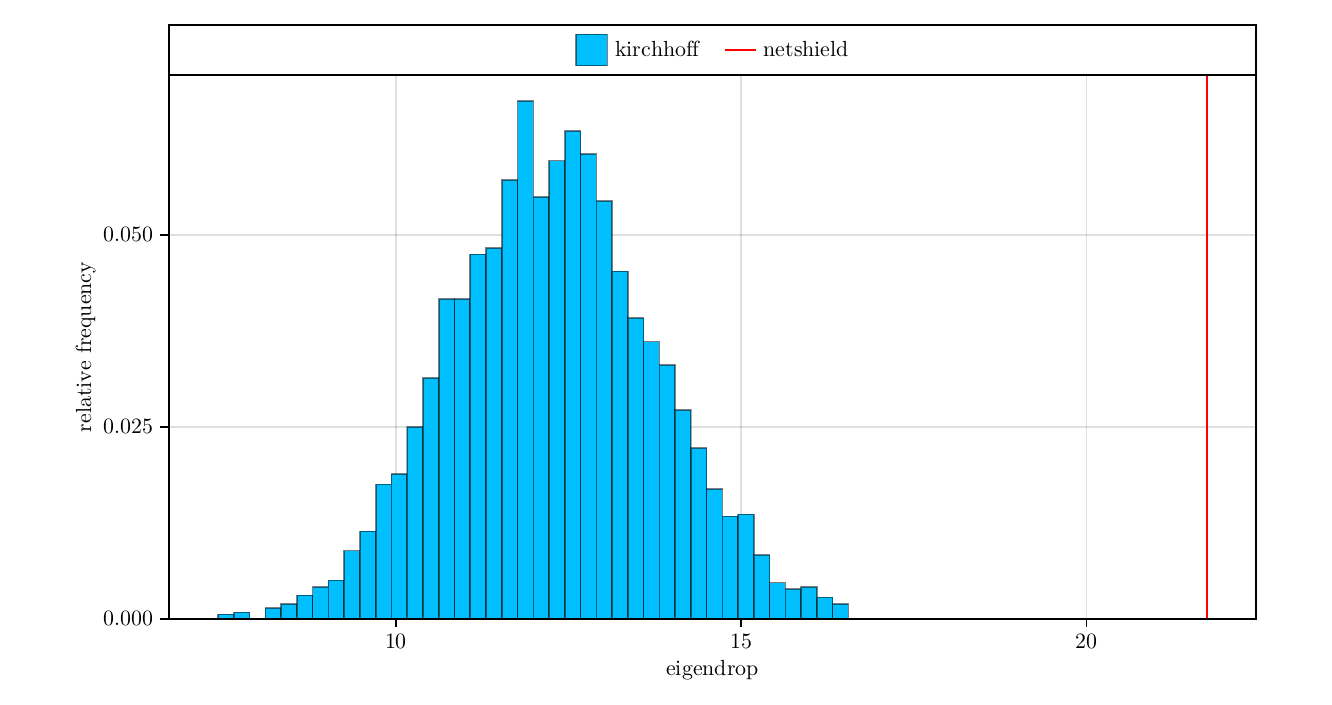}
          \caption{$k = 19$}
      \end{subcaptionblock}
\end{figure}

\begin{figure}
    \centering $ $
    \caption{Graph: ``rfid'' (weighted) - eigendrop distribution}
      \begin{subcaptionblock}{0.45\textwidth}
          \includegraphics[width=\textwidth]{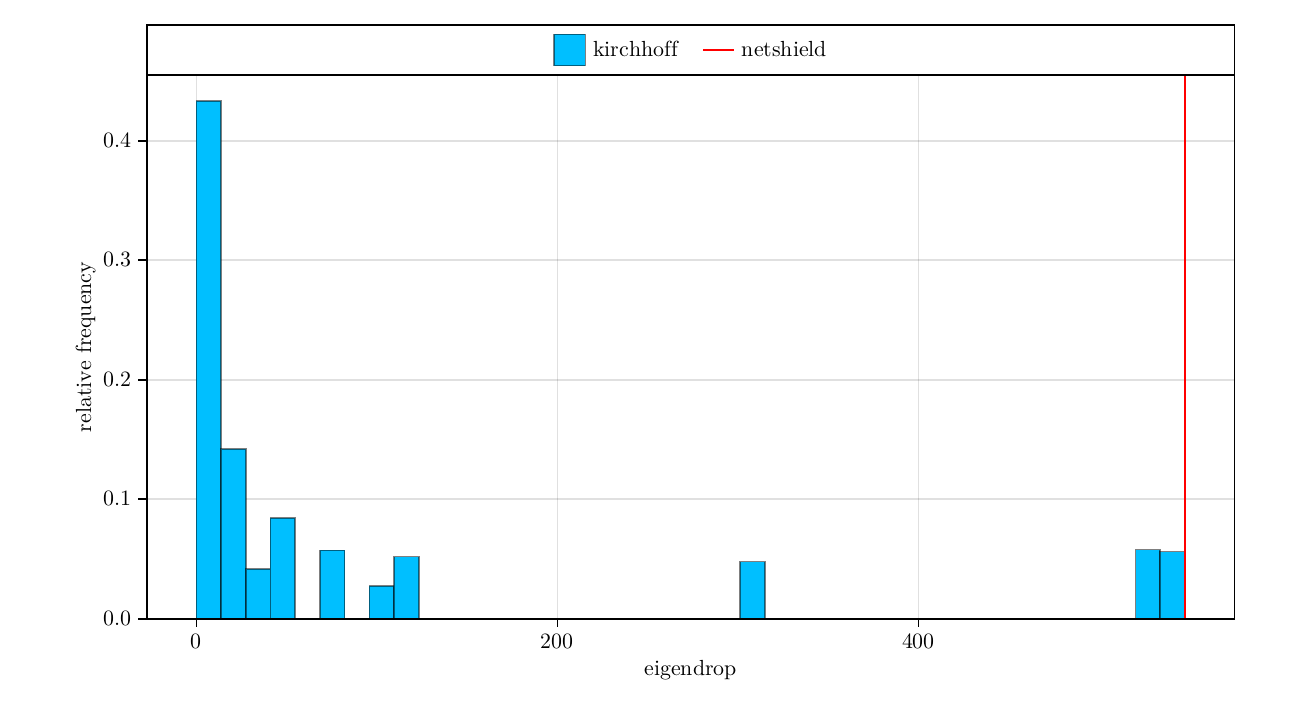}
          \caption{$k = 1$}
      \end{subcaptionblock}
      \begin{subcaptionblock}{0.45\textwidth}
          \includegraphics[width=\textwidth]{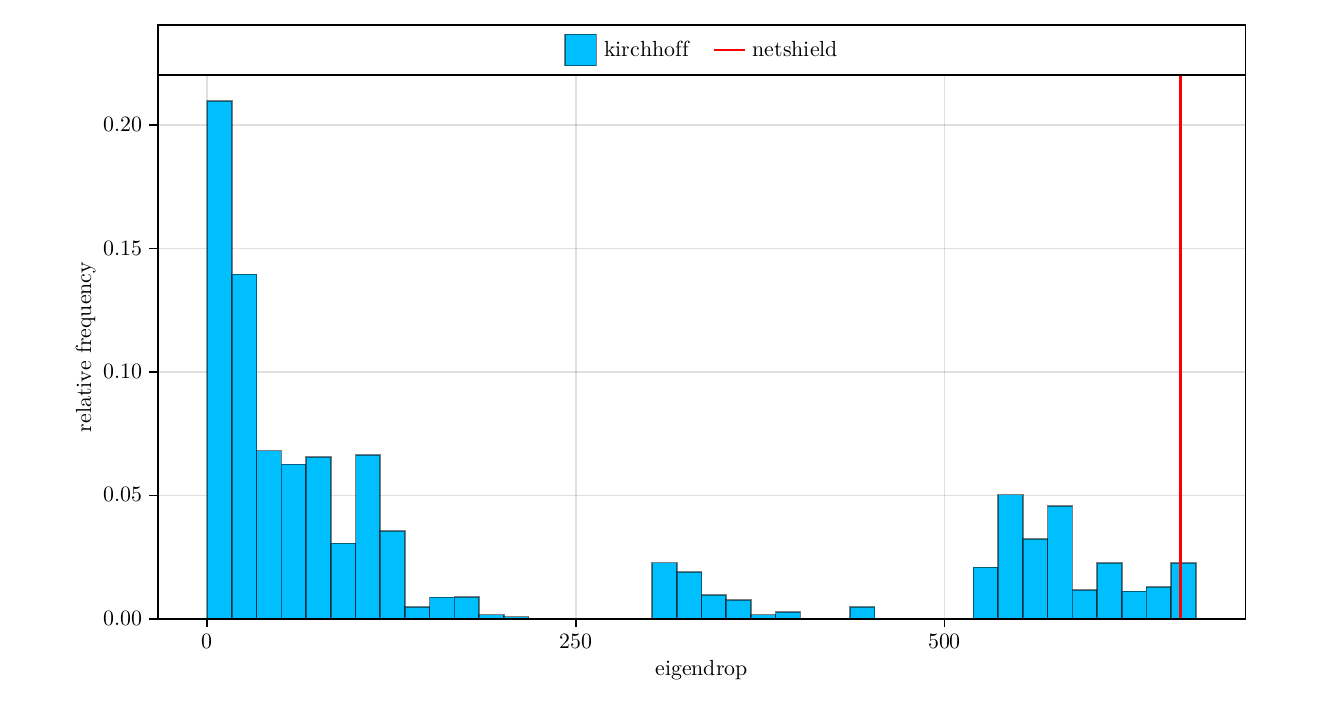}
          \caption{$k = 2$}
      \end{subcaptionblock}
      \\
      \begin{subcaptionblock}{0.45\textwidth}
          \includegraphics[width=\textwidth]{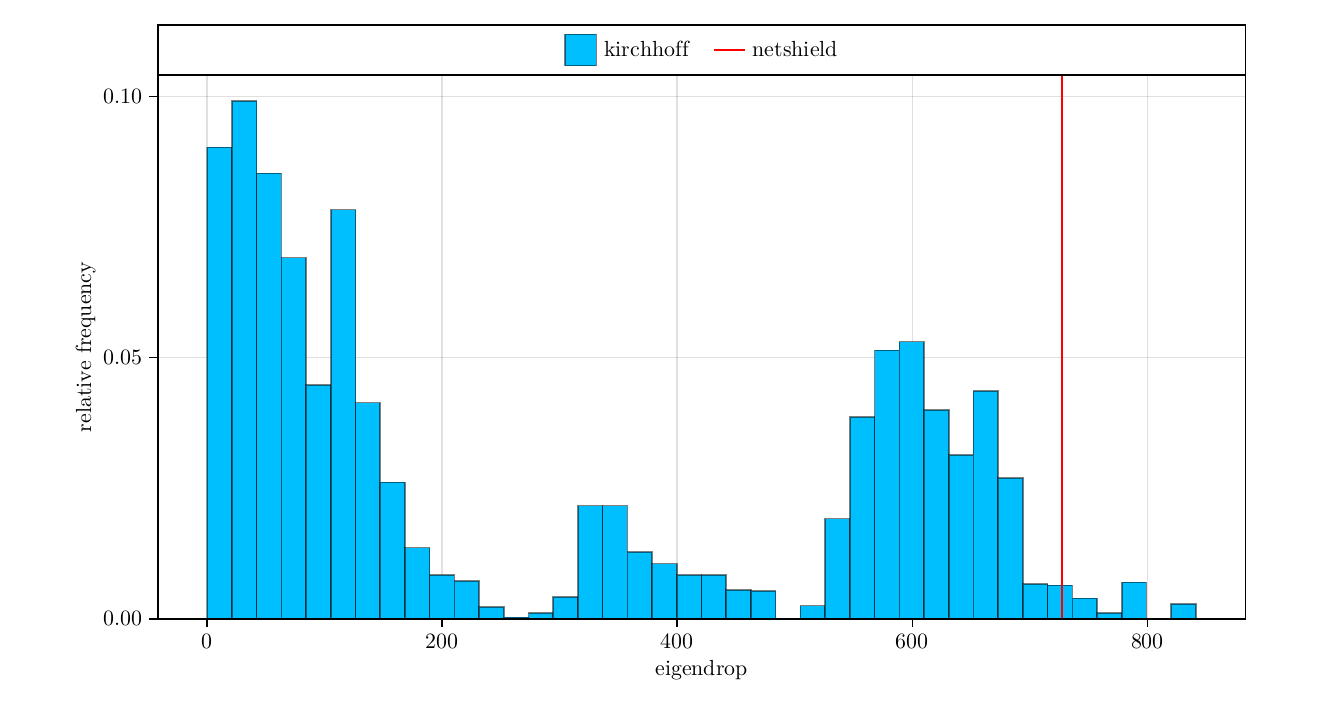}
          \caption{$k = 3$}
      \end{subcaptionblock}
      \begin{subcaptionblock}{0.45\textwidth}
          \includegraphics[width=\textwidth]{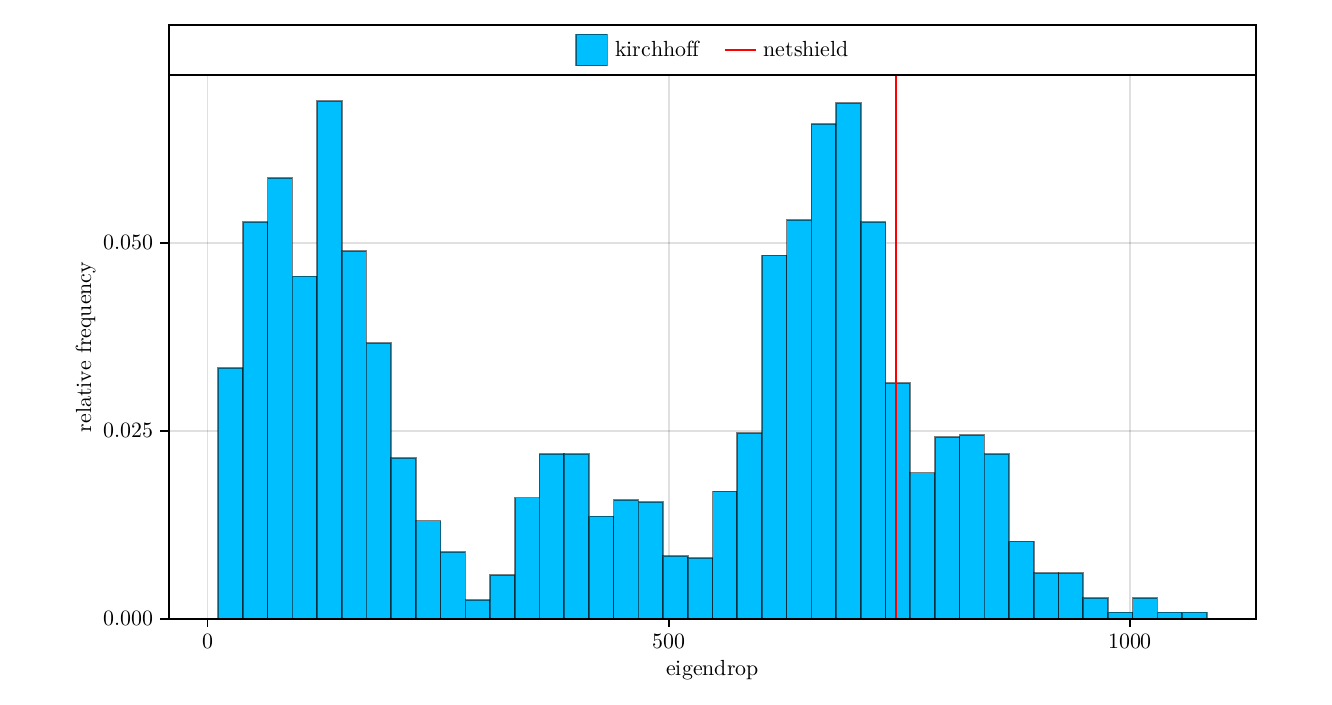}
          \caption{$k = 5$}
      \end{subcaptionblock}
      \\
      \begin{subcaptionblock}{0.45\textwidth}
          \includegraphics[width=\textwidth]{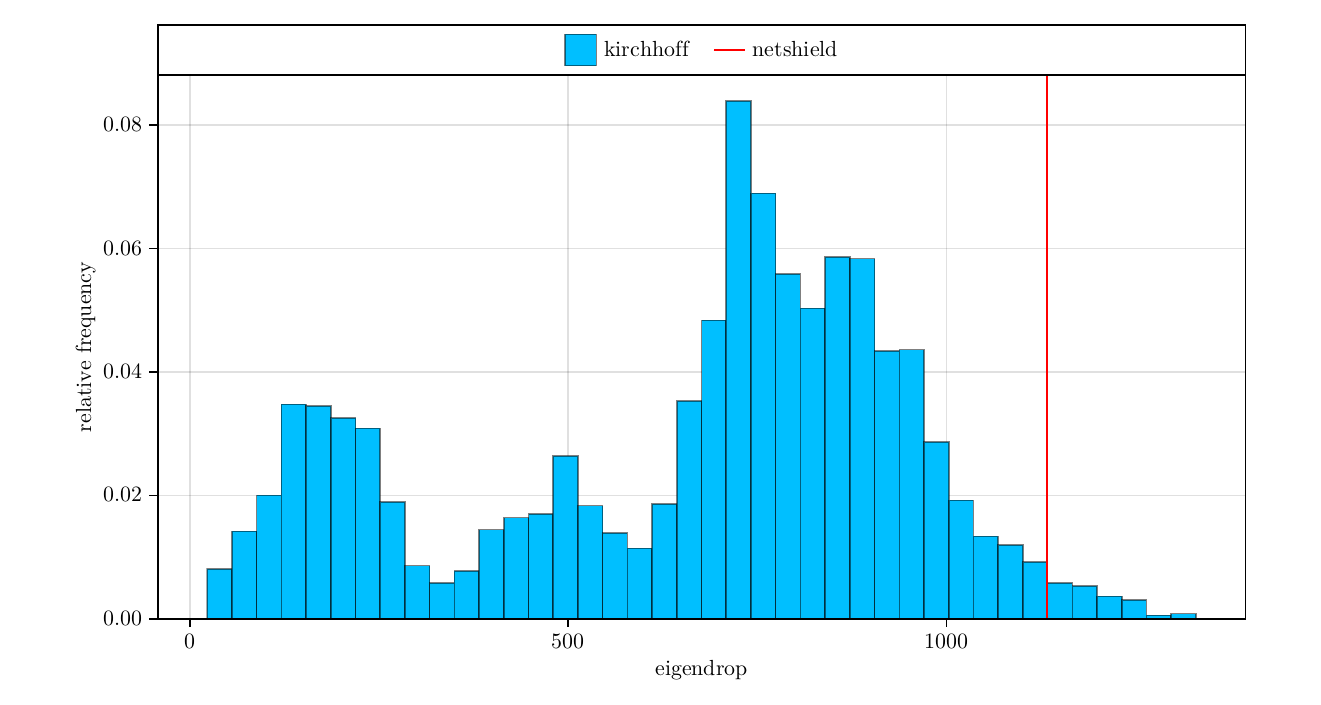}
          \caption{$k = 8$}
      \end{subcaptionblock}
      \begin{subcaptionblock}{0.45\textwidth}
          \includegraphics[width=\textwidth]{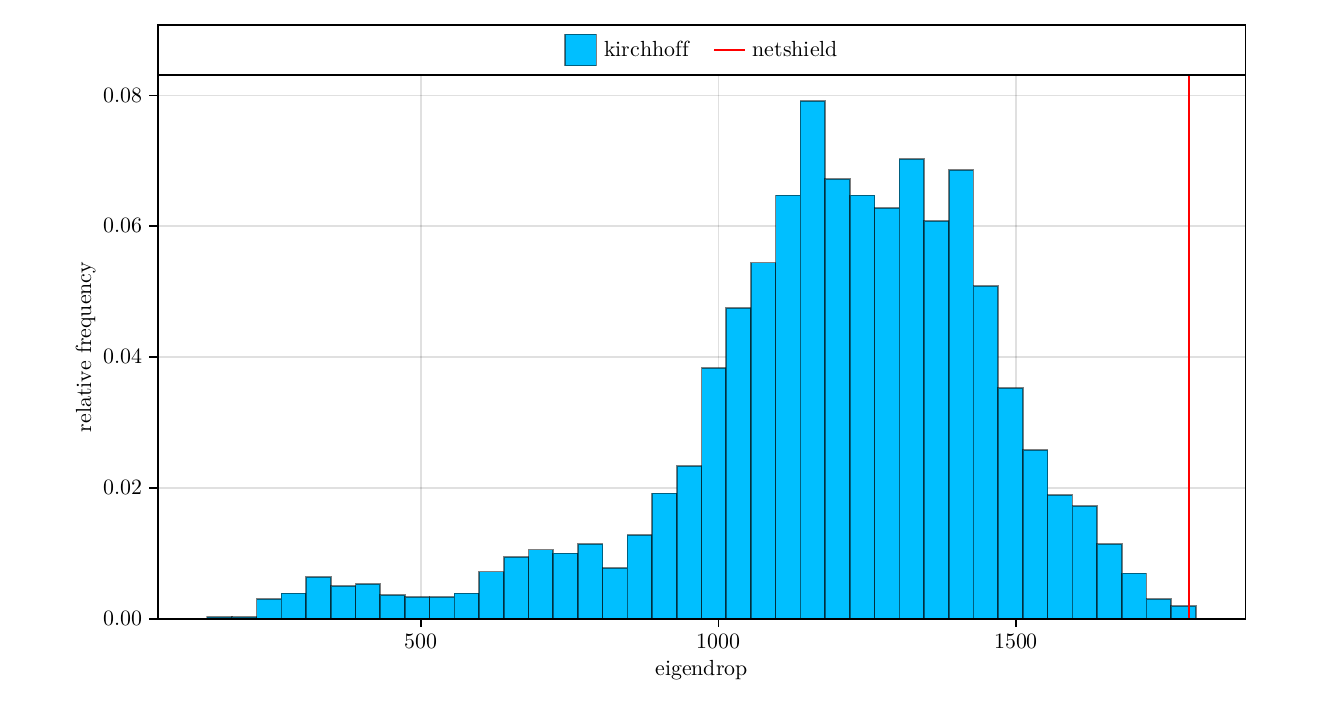}
          \caption{$k = 19$}
      \end{subcaptionblock}
\end{figure}

\begin{figure}
    \centering $ $
    \caption{Graph: ``UKfaculty'' (non-weighted) - eigendrop distribution}
      \begin{subcaptionblock}{0.45\textwidth}
          \includegraphics[width=\textwidth]{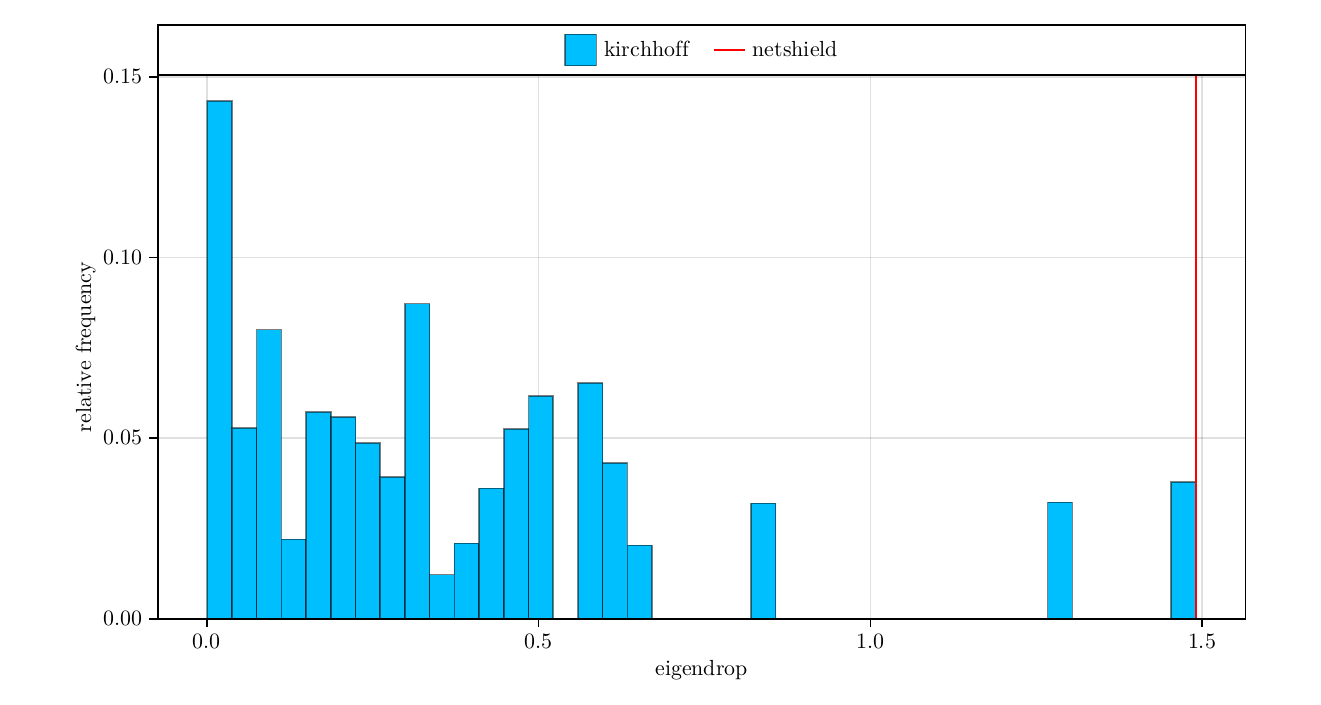}
          \caption{$k = 1$}
      \end{subcaptionblock}
      \begin{subcaptionblock}{0.45\textwidth}
          \includegraphics[width=\textwidth]{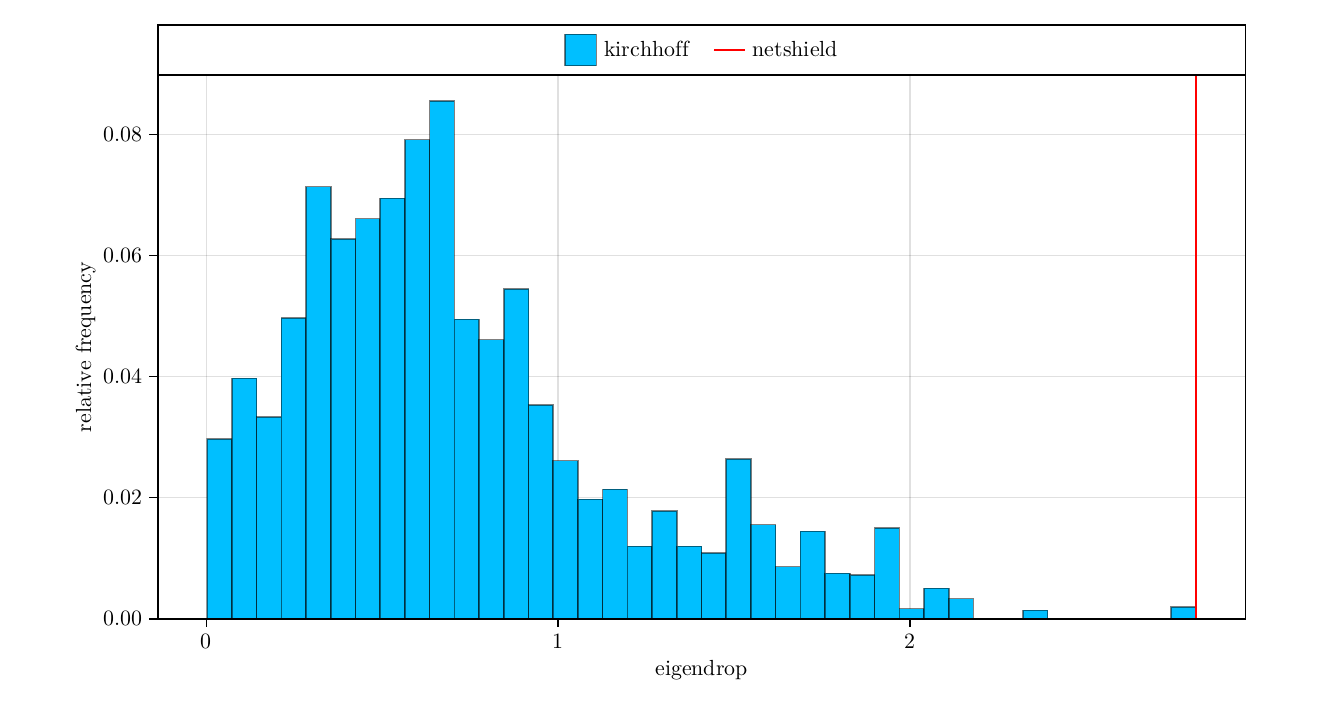}
          \caption{$k = 2$}
      \end{subcaptionblock}
      \\
      \begin{subcaptionblock}{0.45\textwidth}
          \includegraphics[width=\textwidth]{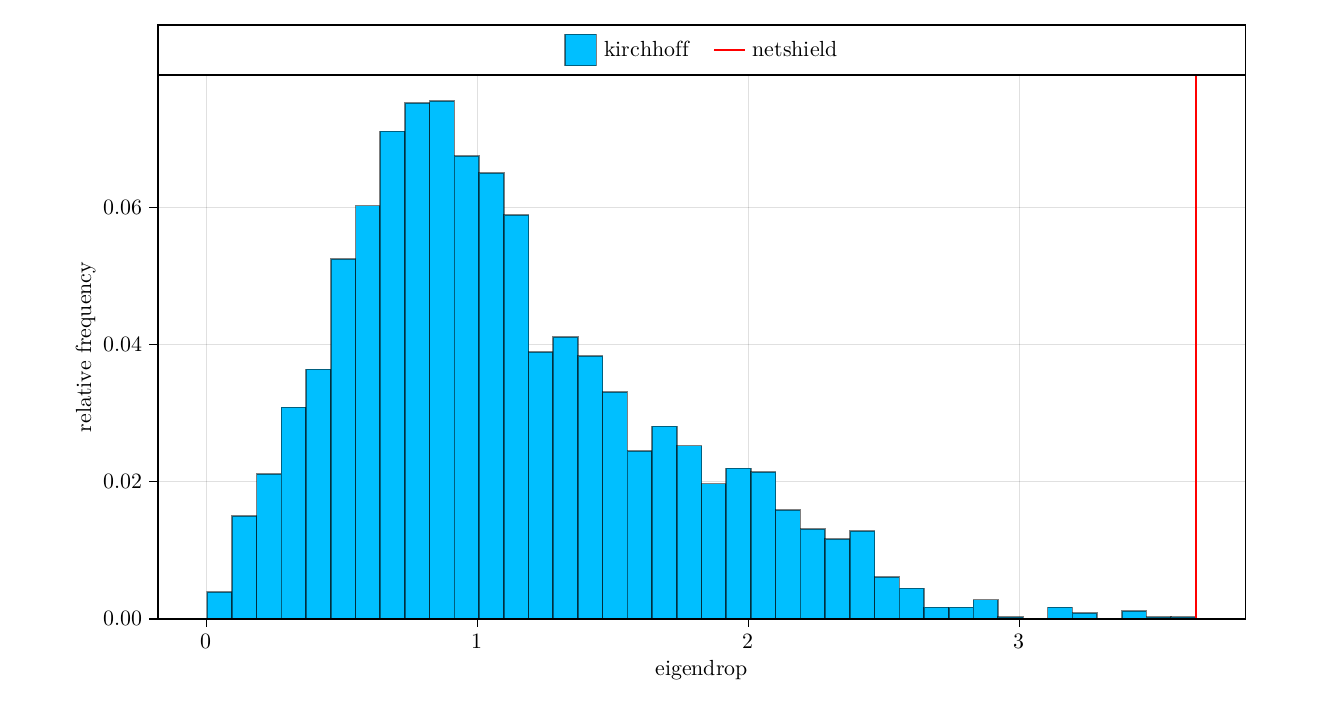}
          \caption{$k = 3$}
      \end{subcaptionblock}
      \begin{subcaptionblock}{0.45\textwidth}
          \includegraphics[width=\textwidth]{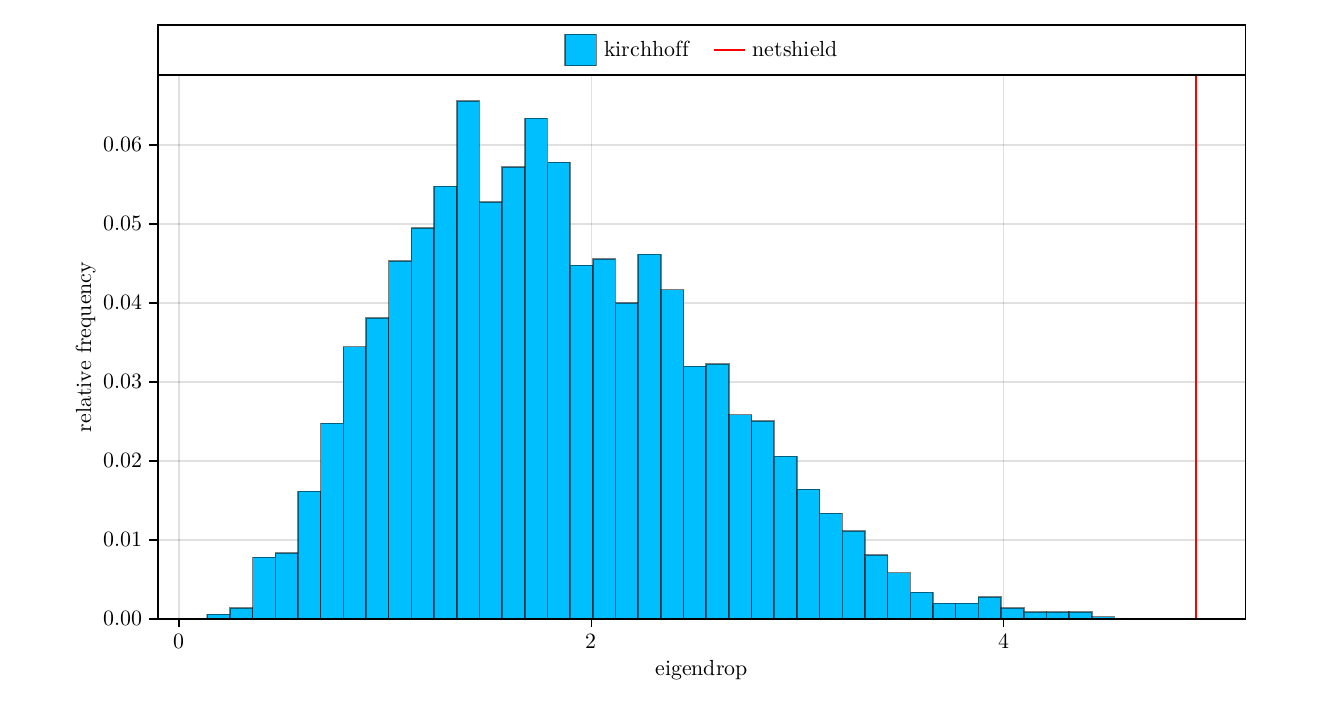}
          \caption{$k = 5$}
      \end{subcaptionblock}
      \\
      \begin{subcaptionblock}{0.45\textwidth}
          \includegraphics[width=\textwidth]{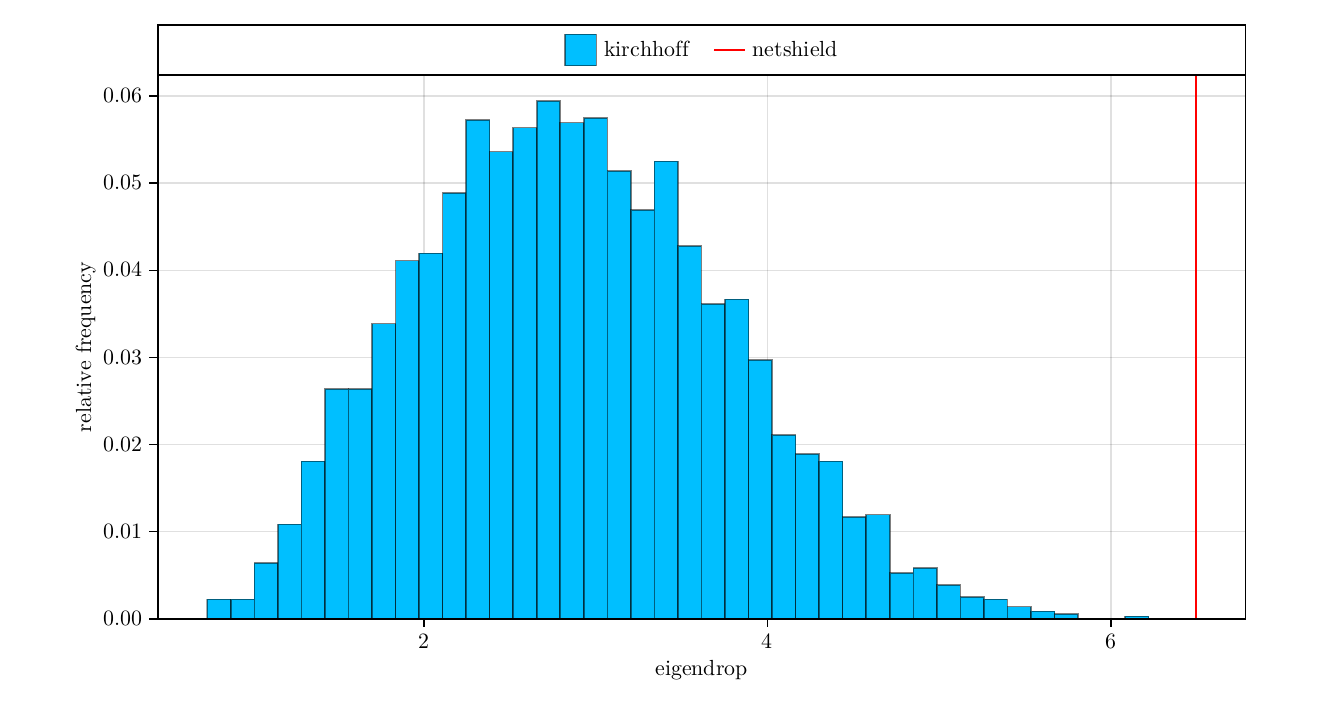}
          \caption{$k = 8$}
      \end{subcaptionblock}
      \begin{subcaptionblock}{0.45\textwidth}
          \includegraphics[width=\textwidth]{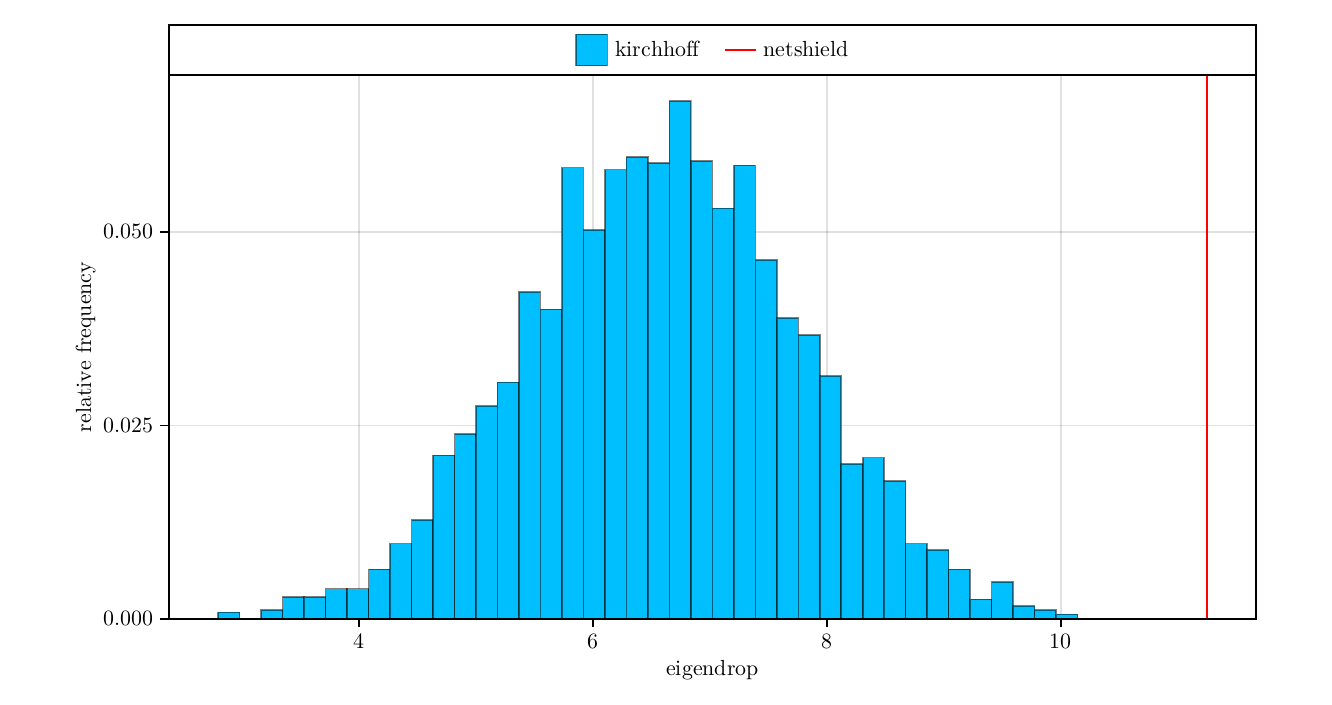}
          \caption{$k = 20$}
      \end{subcaptionblock}
\end{figure}

\begin{figure}
    \centering $ $
    \caption{Graph: ``UKfaculty'' (weighted) - eigendrop distribution}
      \begin{subcaptionblock}{0.45\textwidth}
          \includegraphics[width=\textwidth]{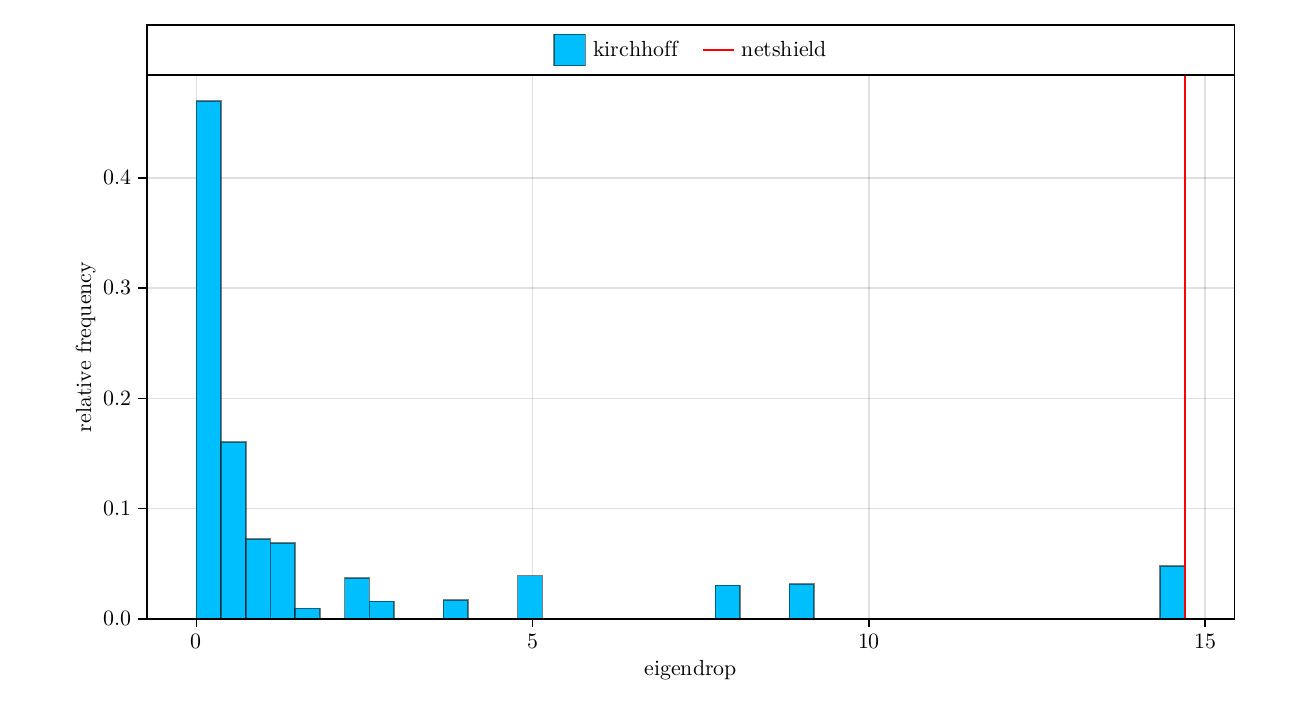}
          \caption{$k = 1$}
      \end{subcaptionblock}
      \begin{subcaptionblock}{0.45\textwidth}
          \includegraphics[width=\textwidth]{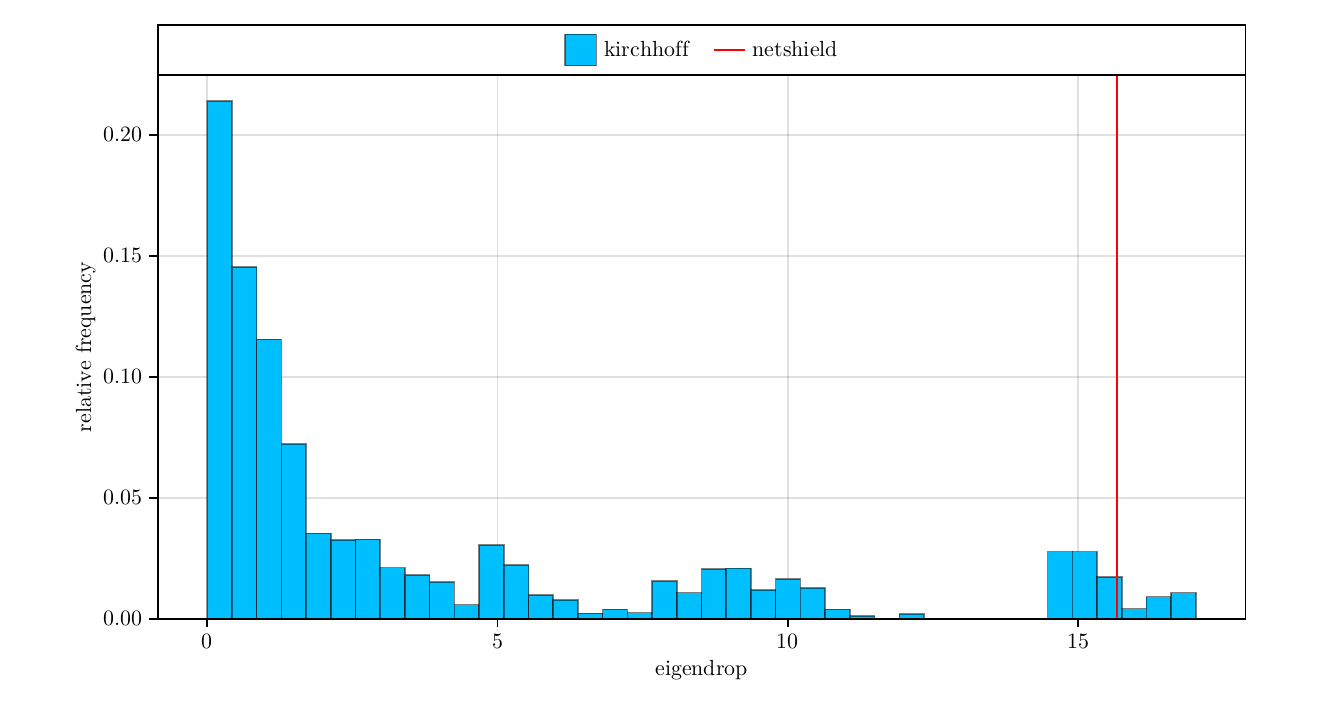}
          \caption{$k = 2$}
      \end{subcaptionblock}
      \\
      \begin{subcaptionblock}{0.45\textwidth}
          \includegraphics[width=\textwidth]{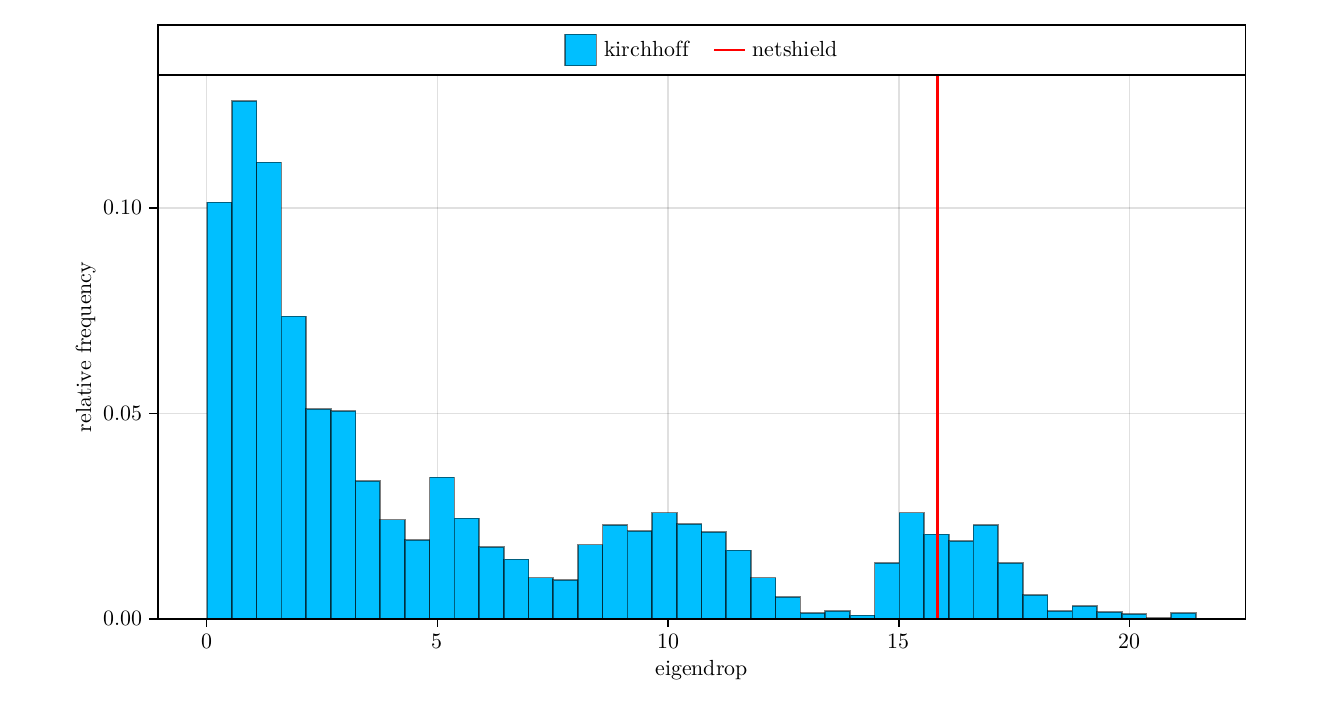}
          \caption{$k = 3$}
      \end{subcaptionblock}
      \begin{subcaptionblock}{0.45\textwidth}
          \includegraphics[width=\textwidth]{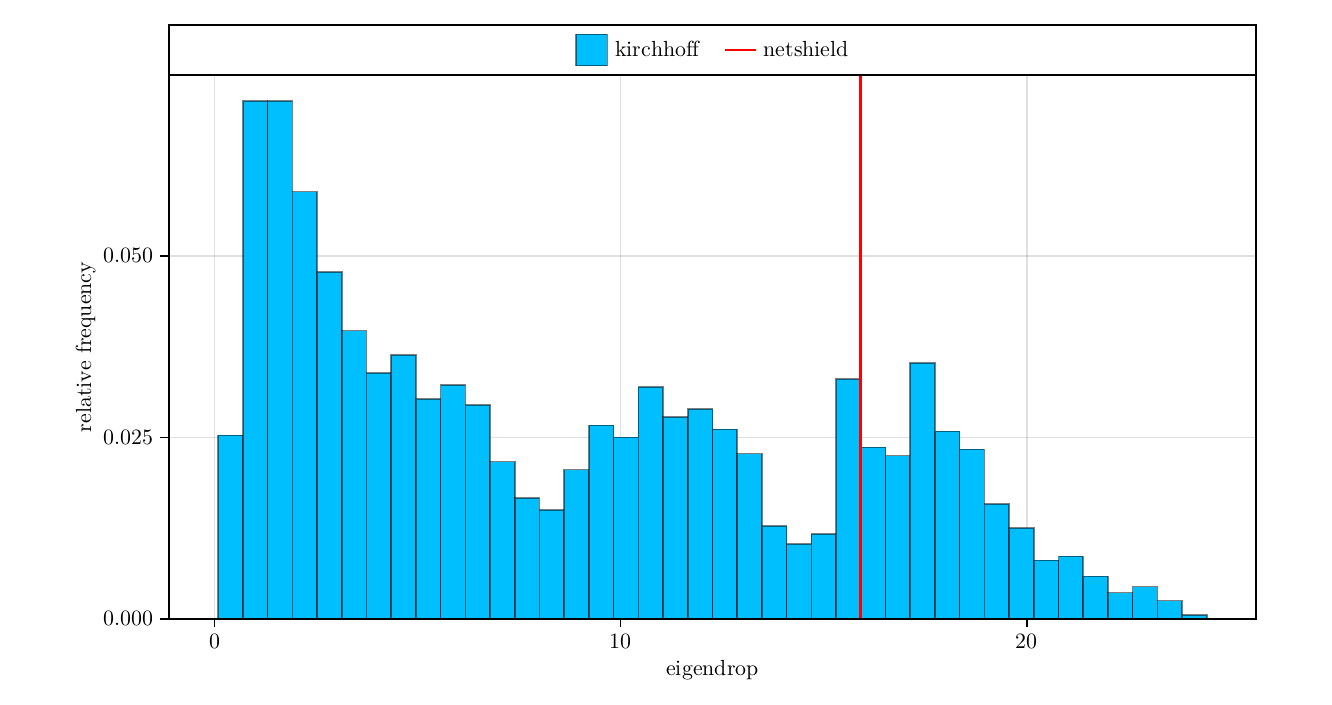}
          \caption{$k = 5$}
      \end{subcaptionblock}
      \\
      \begin{subcaptionblock}{0.45\textwidth}
          \includegraphics[width=\textwidth]{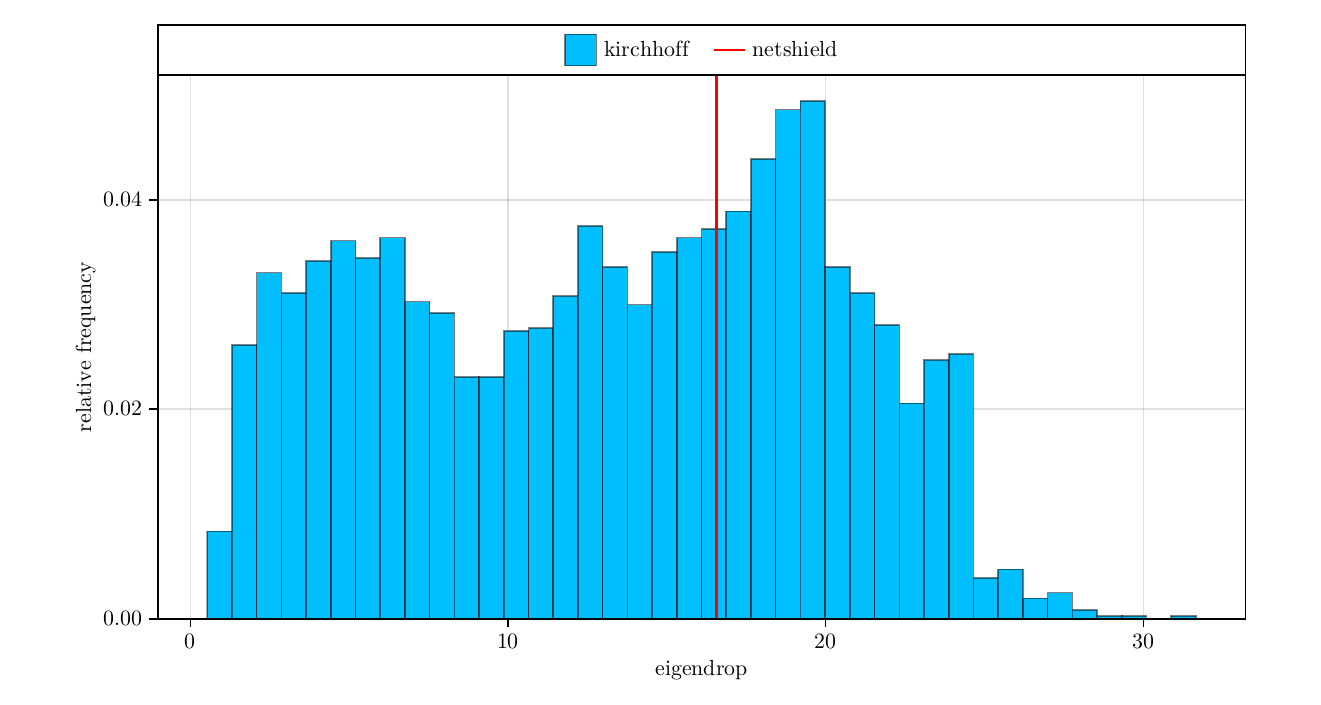}
          \caption{$k = 8$}
      \end{subcaptionblock}
      \begin{subcaptionblock}{0.45\textwidth}
          \includegraphics[width=\textwidth]{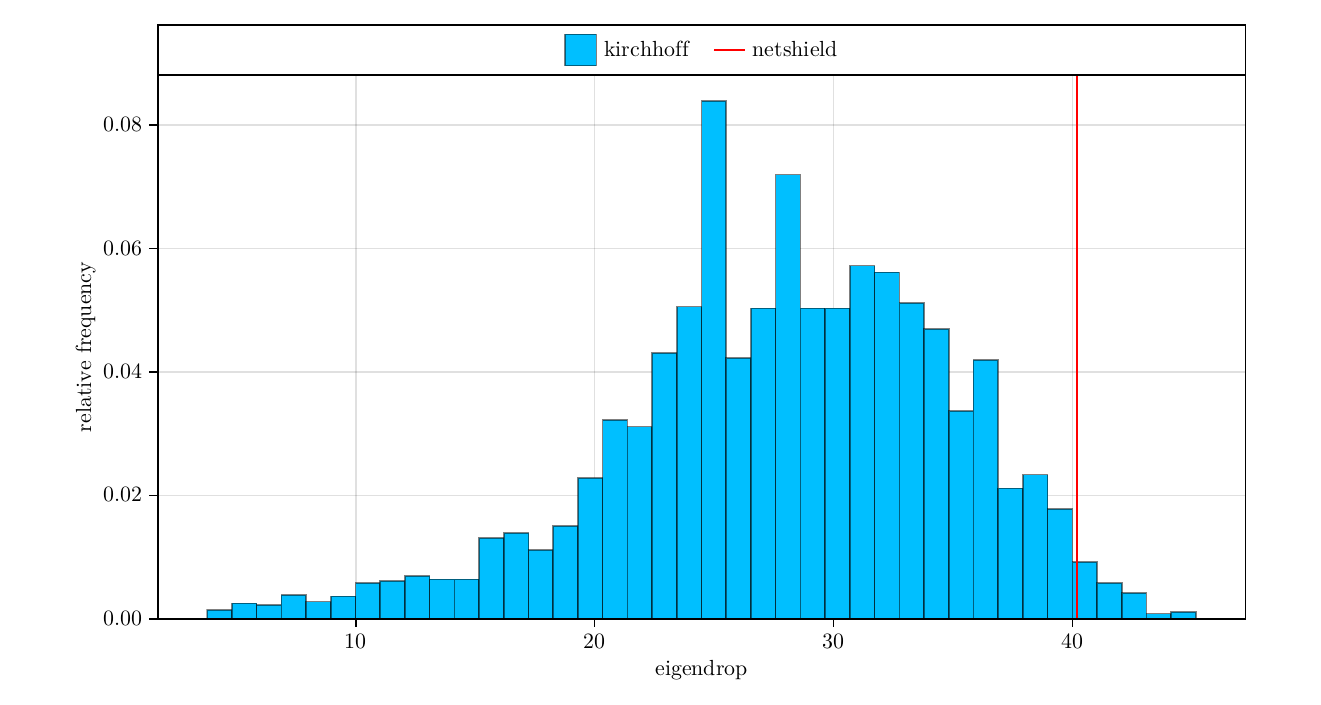}
          \caption{$k = 20$}
      \end{subcaptionblock}
\end{figure}

\begin{figure}
    \centering $ $
    \caption{Graph: ``enron'' (non-weighted) - eigendrop distribution}
      \begin{subcaptionblock}{0.45\textwidth}
          \includegraphics[width=\textwidth]{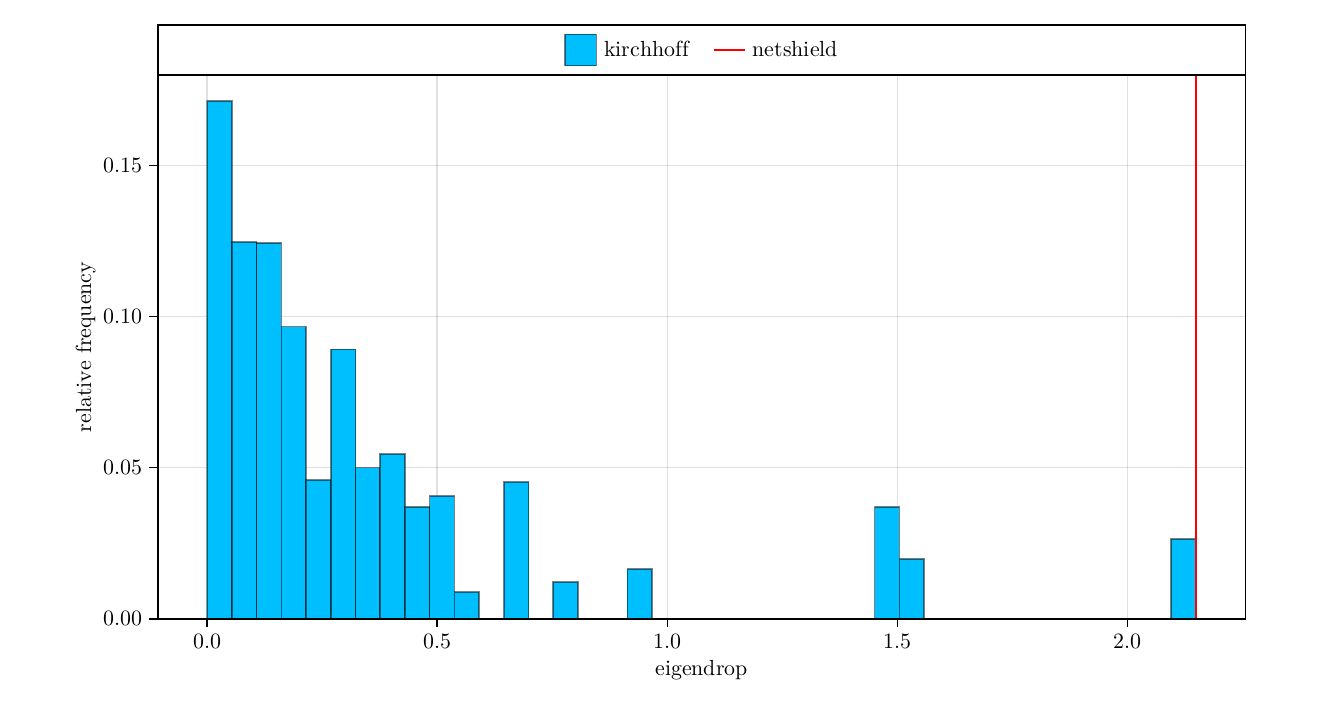}
          \caption{$k = 1$}
      \end{subcaptionblock}
      \begin{subcaptionblock}{0.45\textwidth}
          \includegraphics[width=\textwidth]{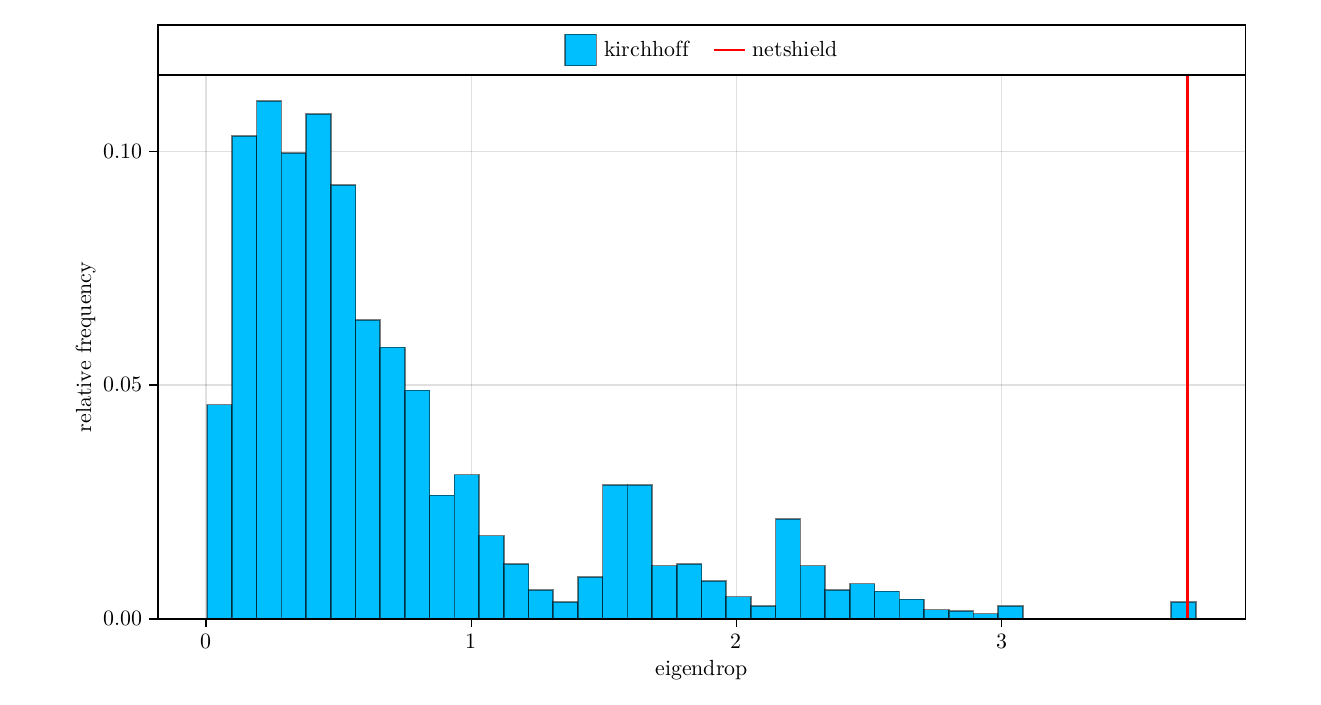}
          \caption{$k = 2$}
      \end{subcaptionblock}
      \\
      \begin{subcaptionblock}{0.45\textwidth}
          \includegraphics[width=\textwidth]{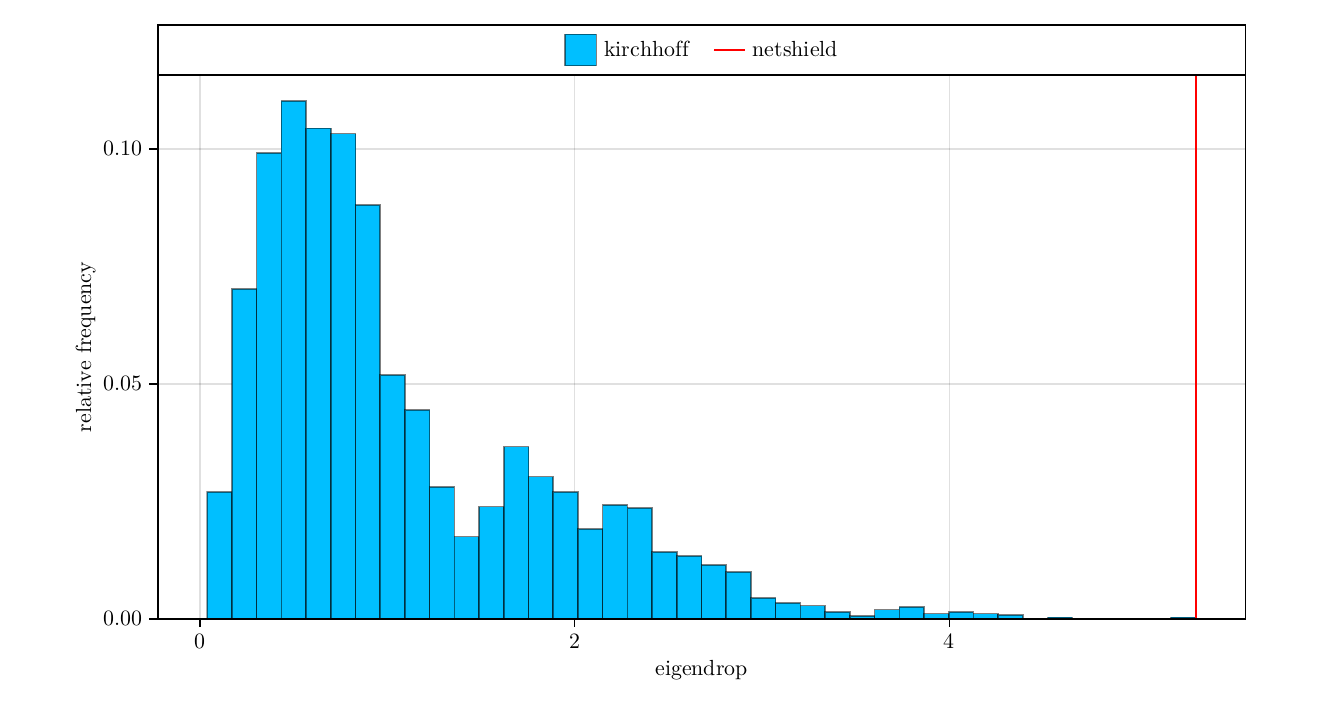}
          \caption{$k = 3$}
      \end{subcaptionblock}
      \begin{subcaptionblock}{0.45\textwidth}
          \includegraphics[width=\textwidth]{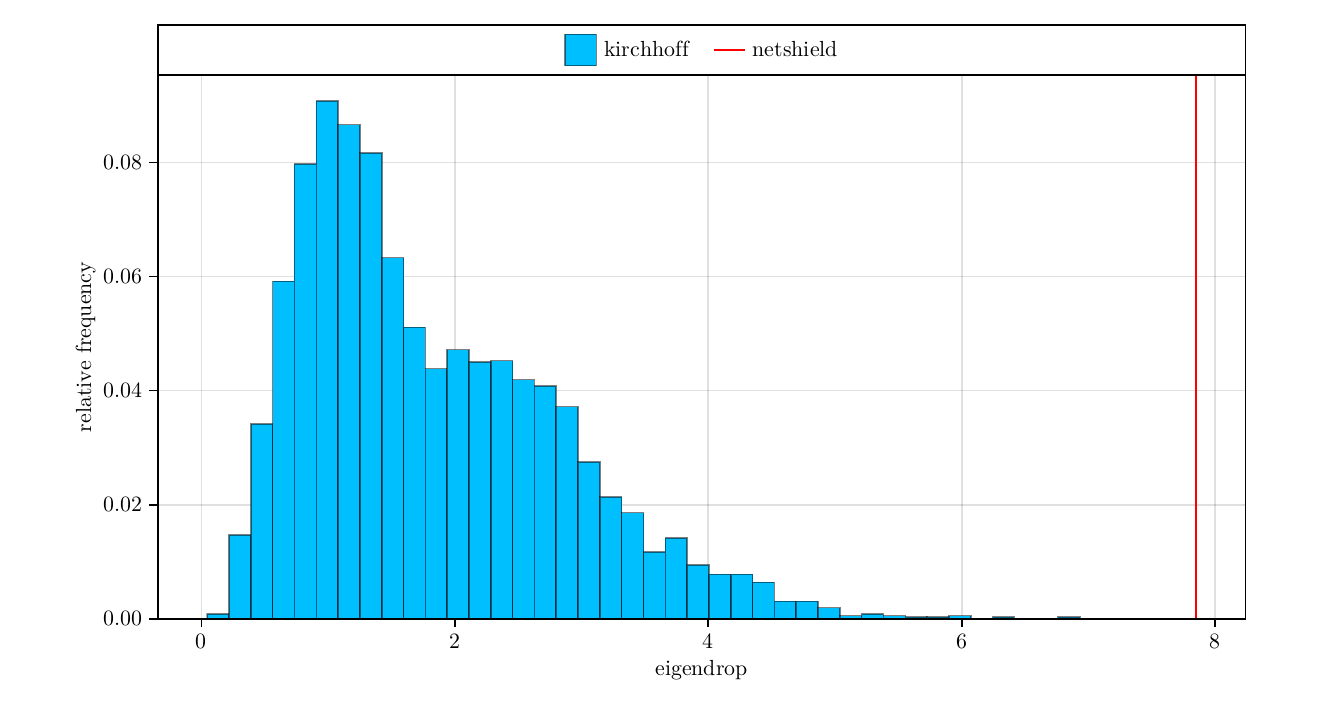}
          \caption{$k = 5$}
      \end{subcaptionblock}
      \\
      \begin{subcaptionblock}{0.45\textwidth}
          \includegraphics[width=\textwidth]{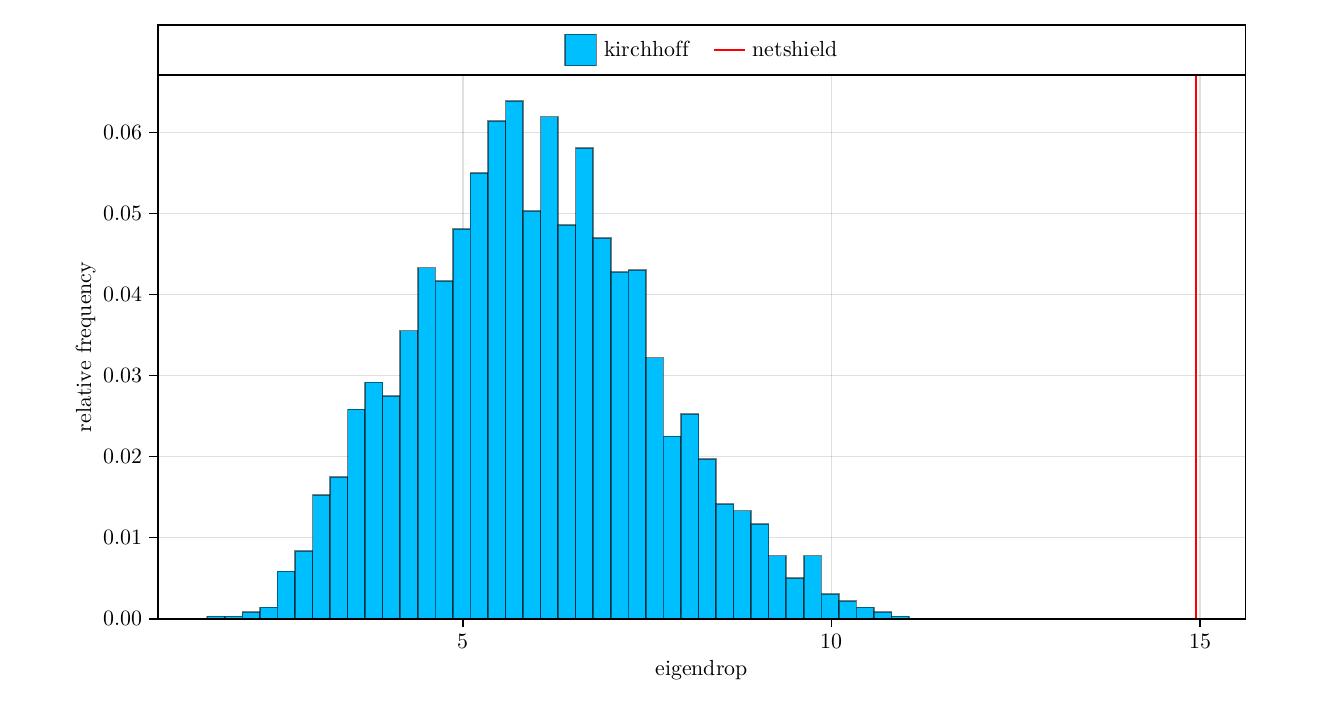}
          \caption{$k = 18$}
      \end{subcaptionblock}
      \begin{subcaptionblock}{0.45\textwidth}
          \includegraphics[width=\textwidth]{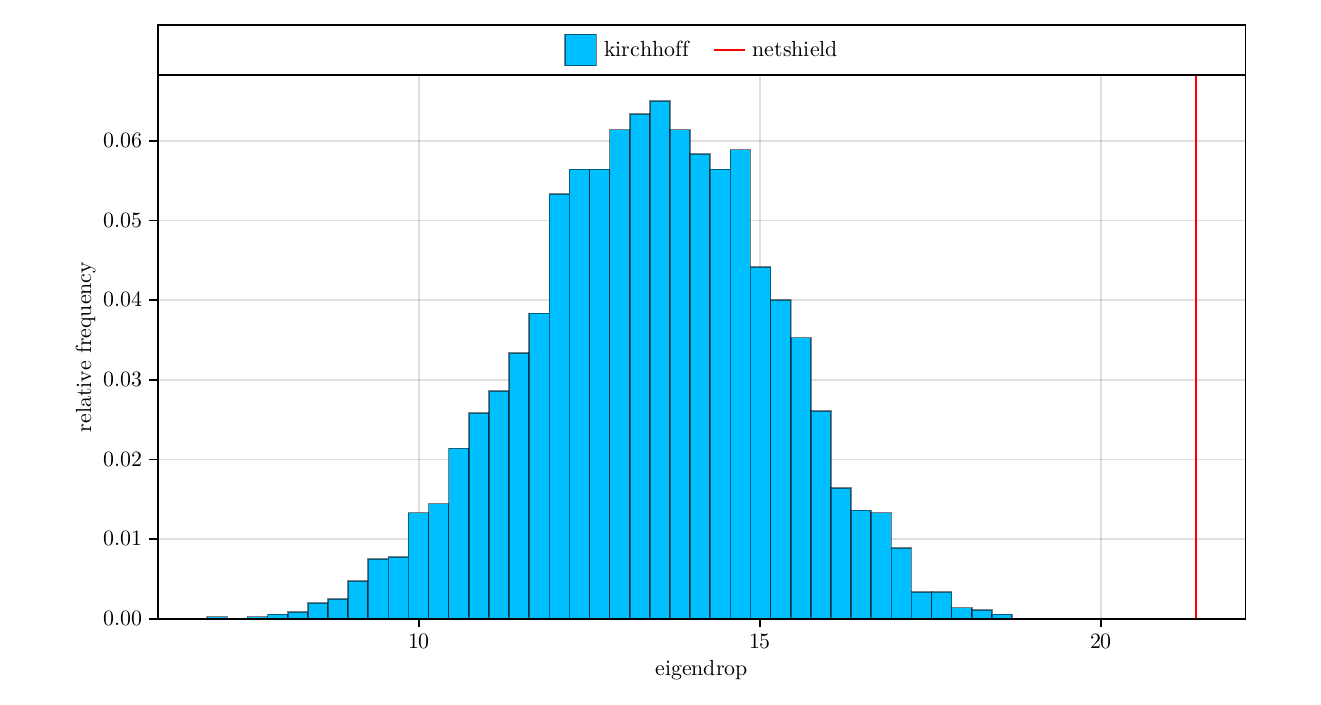}
          \caption{$k = 46$}
      \end{subcaptionblock}
\end{figure}

\end{document}